\documentclass{amsart}
\usepackage{fullpage,graphicx,color,comment,float}
\usepackage{mathrsfs,amssymb}
\newcommand{\pgram}{\text{$\mathfrak{p}$}}

\def\Xint#1{\mathchoice
   {\XXint\displaystyle\textstyle{#1}}%
   {\XXint\textstyle\scriptstyle{#1}}%
   {\XXint\scriptstyle\scriptscriptstyle{#1}}%
   {\XXint\scriptscriptstyle\scriptscriptstyle{#1}}%
   \!\int}
\def\XXint#1#2#3{{\setbox0=\hbox{$#1{#2#3}{\int}$}
     \vcenter{\hbox{$#2#3$}}\kern-.5\wd0}}

\def\dashint{\Xint-}

\def\eq{\begin{equation}}
\def\endeq{\end{equation}}
\def\bpm{\begin{pmatrix}}
\def\epm{\end{pmatrix}}
\def\bbm{\begin{bmatrix}}
\def\ebm{\end{bmatrix}}
\def\pp{\mathcal{P}}
\def\pq{\mathcal{Q}}
\def\pu{\mathcal{U}}
\def\pv{\mathcal{V}}

\def\Odot{{\bf \dot{O}}}

\newtheorem{rhp}{Riemann-Hilbert Problem}
\newtheorem{proposition}{Proposition}
\newtheorem{theorem}{Theorem}
\newtheorem{corollary}{Corollary}
\theoremstyle{definition}
\newtheorem{definition}{Definition}
\newtheorem{remark}{Remark}

\renewcommand{\theequation}{\arabic{section}-\arabic{equation}}
\makeatletter
\@addtoreset{equation}{section}
\makeatother

\title{Large-degree asymptotics of rational Painlev\'e-II functions. I.}
\author{Robert J. Buckingham}
\address[R. J. Buckingham]{Department of Mathematical Sciences\\ University of Cincinnati\\ PO Box 210025\\ Cincinnati, OH 45221.}
\email{buckinrt@uc.edu}
\urladdr{http://homepages.uc.edu/~buckinrt/}
\author{Peter D. Miller}
\address[P. D. Miller]{Department of Mathematics, University of Michigan\\East Hall\\530 Church St.\\Ann Arbor, MI 48109.}
\email{millerpd@umich.edu}
\urladdr{http://www.math.lsa.umich.edu/~millerpd/}
\thanks{R. J. Buckingham was partially supported by the National Science 
Foundation via grant DMS-1312458, by the 
Simons Foundation via award 245775, and the Charles Phelps Taft Research 
Foundation.  P. D. Miller was partially supported by the National Science 
Foundation under grant DMS-0807653.} 

\begin{document}
\begin{abstract}
Rational solutions of the inhomogeneous Painlev\'e-II equation and of a 
related coupled Painlev\'e-II system have recently arisen in studies of fluid 
vortices and of the sine-Gordon equation.  For the sine-Gordon 
application in particular it is of interest to understand the large-degree 
asymptotic behavior of the rational Painlev\'e-II functions.  We explicitly compute the leading-order large-degree 
asymptotics of these two families of rational functions 
valid in the whole complex plane with the exception of a neighborhood of a 
certain piecewise-smooth closed curve.  
We obtain rigorous error bounds by using the Deift-Zhou 
nonlinear steepest-descent method for Riemann-Hilbert problems.

\end{abstract}
\maketitle

\section{Introduction}

Solutions of the inhomogeneous Painlev\'e-II equation
\eq
\label{PII}
p_{yy} = 2p^3 + \frac{2}{3}y p-\frac{2}{3}m, \quad p:\mathbb{C}\to\mathbb{C} \text{ with parameter } m\in\mathbb{C}
\endeq
and the coupled Painlev\'e-II system
\eq
\label{PII-system}
\left.\begin{matrix}
\vspace{.1in}\displaystyle u_{yy} + 2u^2v + \frac{1}{3}yu = 0 \\
\displaystyle v_{yy} + 2uv^2 + \frac{1}{3}yv = 0
\end{matrix}\right\}, 
\quad u,v:\mathbb{C}\to\mathbb{C},
\endeq
are often referred to in the literature, along with solutions to other 
Painlev\'e-type equations, as \emph{Painlev\'e transcendents} since solutions 
cannot in general be expressed in terms of elementary functions.  However, 
both \eqref{PII} and \eqref{PII-system} admit important families of rational 
solutions.  To be precise, define 
\eq
\label{backlund-initial-condition}
\pu_0(y):=1 \quad \text{and} \quad \pv_0(y):=-\frac{1}{6}y.
\endeq
Then define the rational functions $\pu_m(y)$ and $\pv_m(y)$ iteratively 
for positive integers $m$ by 
\eq
\label{backlund-positive}
\pu_{m+1}(y) := -\frac{1}{6}y\pu_m(y) - \frac{\pu_m'(y)^2}{\pu_m(y)} + \frac{1}{2}\pu_m''(y) \quad \text{and} \quad \pv_{m+1}(y) := \frac{1}{\pu_m(y)},
\endeq
and for negative integers $m$ by
\eq
\label{backlund-negative}
\pu_{m-1}(y) := \frac{1}{\pv_m(y)} \quad \text{and} \quad \pv_{m-1}(y) := \frac{1}{2}\pv_m''(y) - \frac{\pv_m'(y)^2}{\pv_m(y)} - \frac{1}{6}y\pv_m(y).
\endeq
Then $\{u,v\}=\{\pu_m,\pv_m\}$ solves the coupled Painlev\'e-II system 
\eqref{PII-system} for each choice of $m\in\mathbb{Z}$.  Furthermore, if 
we define 
\eq
\label{log-derivative}
\pp_m(y) := \frac{\pu_m'(y)}{\pu_m(y)} \quad \text{and} \quad \pq_m(y) := \frac{\pv_m'(y)}{\pv_m(y)}, \quad m\in\mathbb{Z}
\endeq
then $\pp_m$ satisfies \eqref{PII} with parameter $m$ while $\pq_m$ satisfies 
\eq
\pq_m''(y)=2\pq_m(y)^3+\frac{2}{3}y\pq_m(y)+\frac{2}{3}(m-1).
\endeq
It is known that the Painlev\'e-II 
equation \eqref{PII} 
has a rational solution if and only if $m\in\mathbb{Z}$ \cite{Airault:1979}, and 
when this rational solution exists it is unique \cite{Murata:1985}.  
These rational solutions have arisen in the study of fluid 
vortices \cite{Clarkson:2009} and string theory \cite{Johnson:2006}.  
From the point of view of
the Flaschka-Newell inverse monodromy theory for the Painlev\'e-II equation 
\cite{Flaschka:1980} they play a role similar to that played by multisoliton solutions of integrable nonlinear wave equations (arising from determinantal formulae, corresponding to fully discrete scattering data, etc.).  

We became further interested in these functions when they appeared in the study of the sine-Gordon equation 
\cite{BuckinghamMcritical} in the semiclassical or small-dispersion limit.  It turns out that the functions $\pu_m$ are of crucial importance in an 
associated 
double-scaling limit describing the transition between librational and
rotational oscillations as a separatrix is traversed at a certain critical 
point.  Near the critical point, the solution of the sine-Gordon equation is accurately approximated in the limit
by a universal curvilinear grid of isolated kink-type solutions, with 
the location of the $m^\text{th}$ kink in the space-time plane determined by the graph of $\pu_m(y)$ for 
$y\in\mathbb{R}$.  Furthermore, the kinks collide in a grazing fashion (that can 
be modeled by a double soliton solution of the sine-Gordon equation) at space-time points determined by 
the real poles and zeros of $\pu_m$.  It remains an open problem to match the critical behavior near the transition point onto larger-time dynamics presumably described by hyperelliptic functions, and this 
problem motivates the current study of
the large-$m$ asymptotic behavior of $\pu_m(y)$ for $y\in\mathbb{R}$.  

Generalizing further to $y\in\mathbb{C}$, the 
idea that the rational Painlev\'e functions considered here may have a 
particularly interesting structure 
in the complex $y$-plane when $m$ is large 
is indicated 
by two previous results.  Clarkson and Mansfield \cite{Clarkson:2003} noted 
from numerically-generated plots that the $m(m+1)/2$ complex zeros of the related $m^\text{th}$ 
\emph{Yablonskii-Vorob'ev polynomial} form a highly regular pattern in a 
triangular-type region with curved sides of size proportional to $m^{2/3}$ as 
$m\to\infty$ (see also the work of Roffelson \cite{Roffelson:2010,Roffelson:2012} for other recent 
results on the Yablonskii-Vorob'ev polynomials, including interlacing properties of the roots).  The rational Painlev\'e-II 
functions $\pp_m(y)$ are the logarithmic 
derivatives of ratios of successive Yablonskii-Vorob'ev polynomials (or equivalently, the Painlev\'e-II functions $\pu_m(y)$ are themselves such ratios), and 
thus the zeros and poles of $\pu_m$ and $\pp_m$ exhibit the same qualitative 
behavior (see Figure \ref{um-zeros}, in which we present figures similar to those in \cite{Clarkson:2003} but in a rescaled independent variable).  The regular pattern of zeros and poles suggests modeling the rational functions with (doubly-periodic) elliptic functions of some local complex coordinate.  
This is compatible with a study of Kapaev \cite{Kapaev:1997} in which 
it is shown that general solutions 
of \eqref{PII} are asymptotically described by elliptic functions as 
$m\to\infty$ 
away from certain $m$-independent Stokes lines in the complex plane 
of the rescaled independent variable $m^{-1/3}y$.  
We call the region containing the poles and zeros of the Painlev\'e-II rational functions in the rescaled complex plane the \emph{elliptic region} $T$ (a precise definition including an expression for the boundary will be formulated later).
\begin{figure}[h]
\includegraphics[width=2.1in]{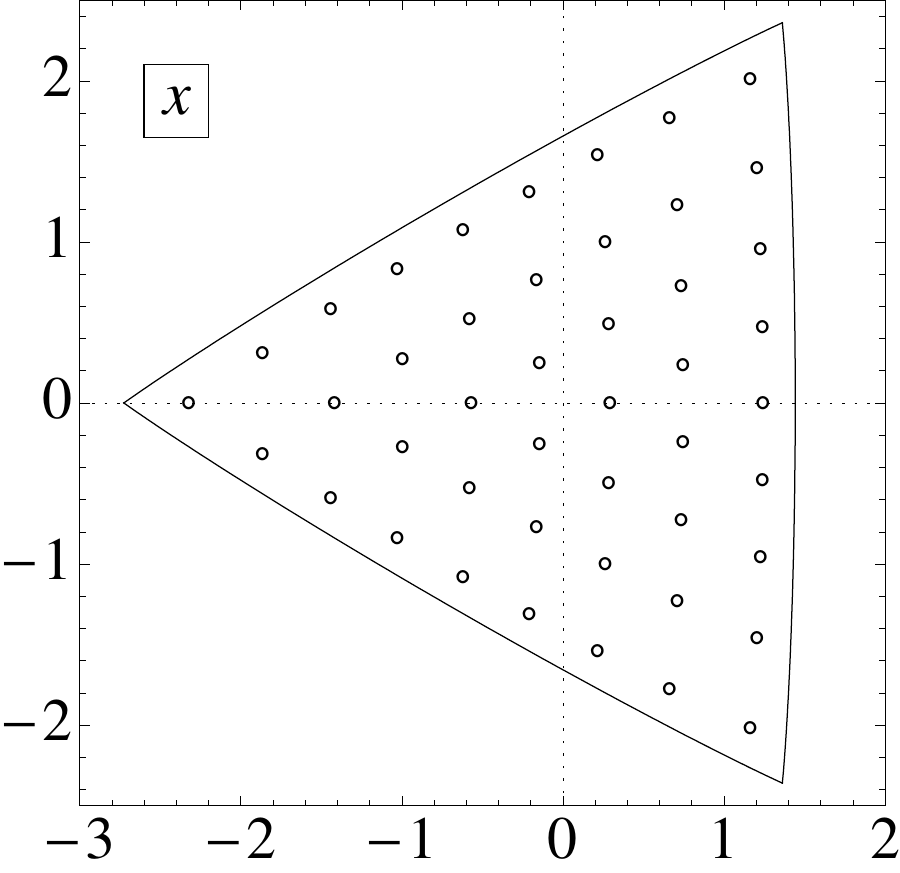}
\includegraphics[width=2.1in]{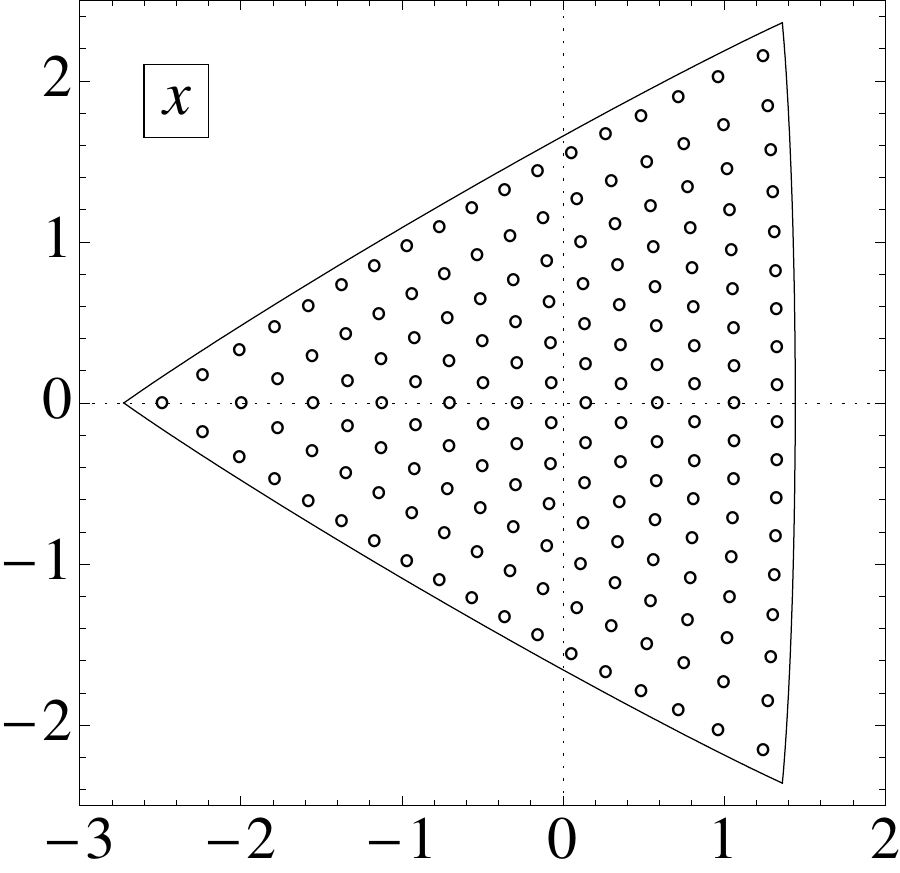}
\includegraphics[width=2.1in]{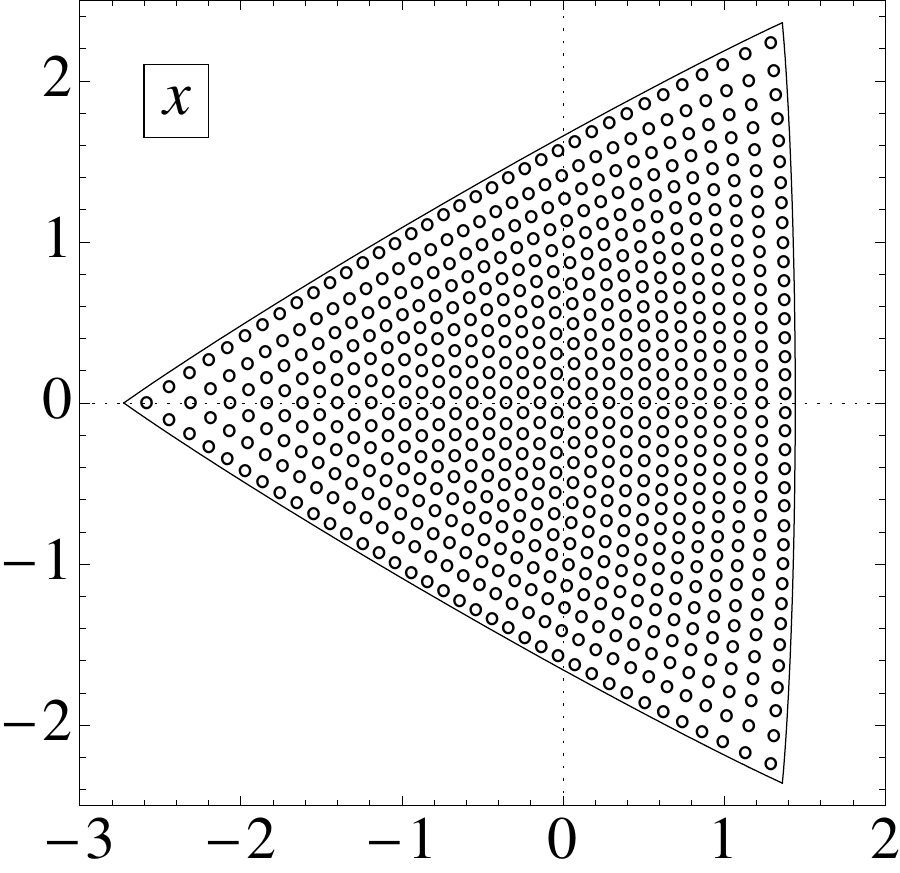}
\caption{\emph{The zeros of $\pu_m(y)$ in the complex $x$-plane, where 
$x=(m-\tfrac{1}{2})^{-2/3}y$ 
for $m=9$ (left), $m=18$ (center), and $m=36$ (right), along with the 
$m$-independent boundary of the elliptic region $T$.}}
\label{um-zeros}
\end{figure}
\begin{figure}[h]
\includegraphics[width=2.5in]{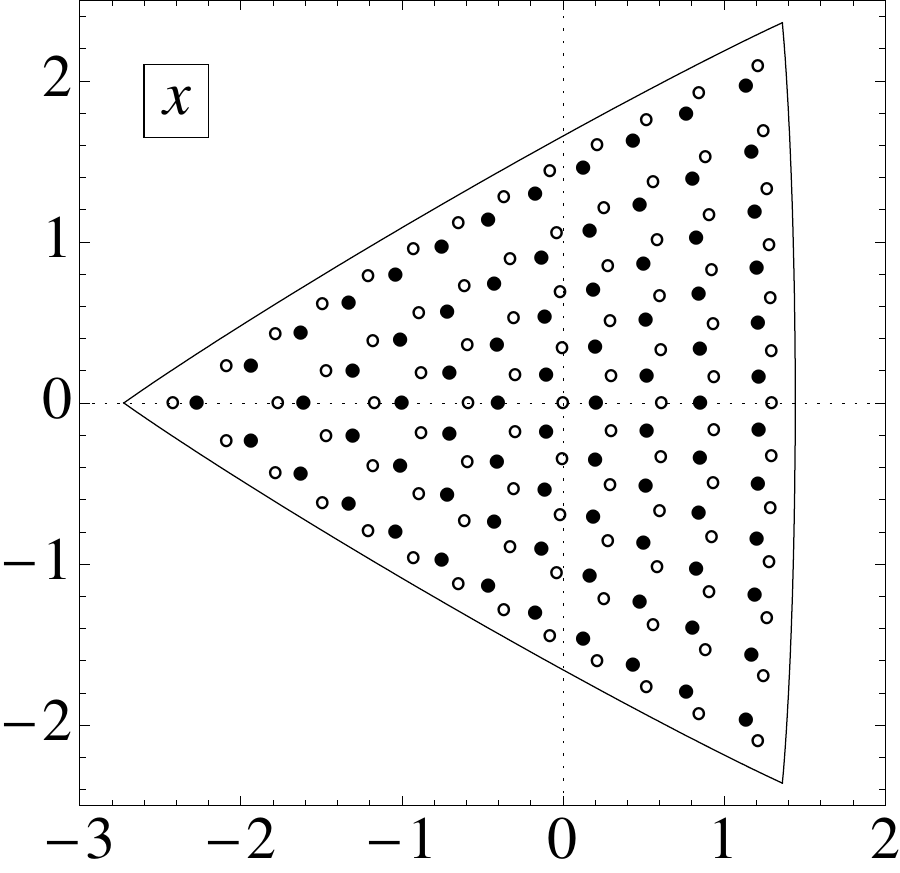}
\includegraphics[width=2.5in]{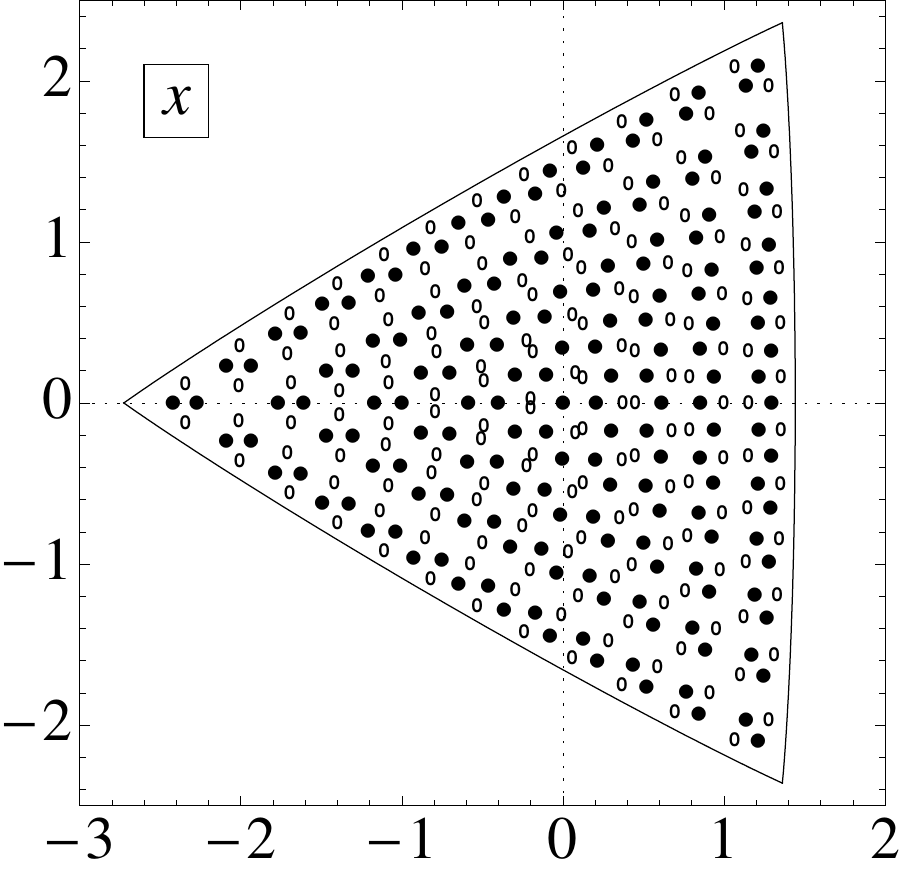}
\caption{\emph{The zeros (circles) and poles (dots) of $\pu_{13}(y)$ (left) 
and $\pp_{13}(y)$ (right) in the complex $x$-plane, where 
$x=(13-\tfrac{1}{2})^{-2/3}y$, along 
with the $m$-independent boundary of the region $T$.  
Note that every zero and every pole of $\pu_{13}(y)$ is a pole of 
$\pp_{13}(y)$.}}
\label{u13-p13-zeros-poles}
\end{figure}
\begin{figure}[h]
\includegraphics[width=2.5in]{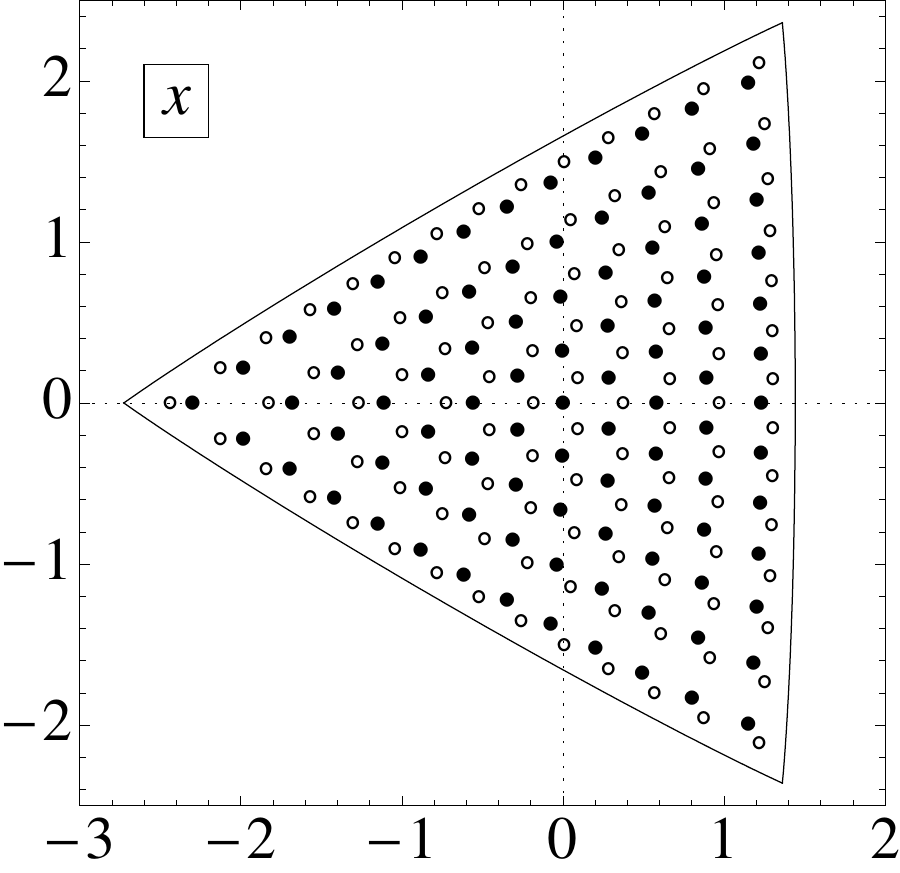}
\includegraphics[width=2.5in]{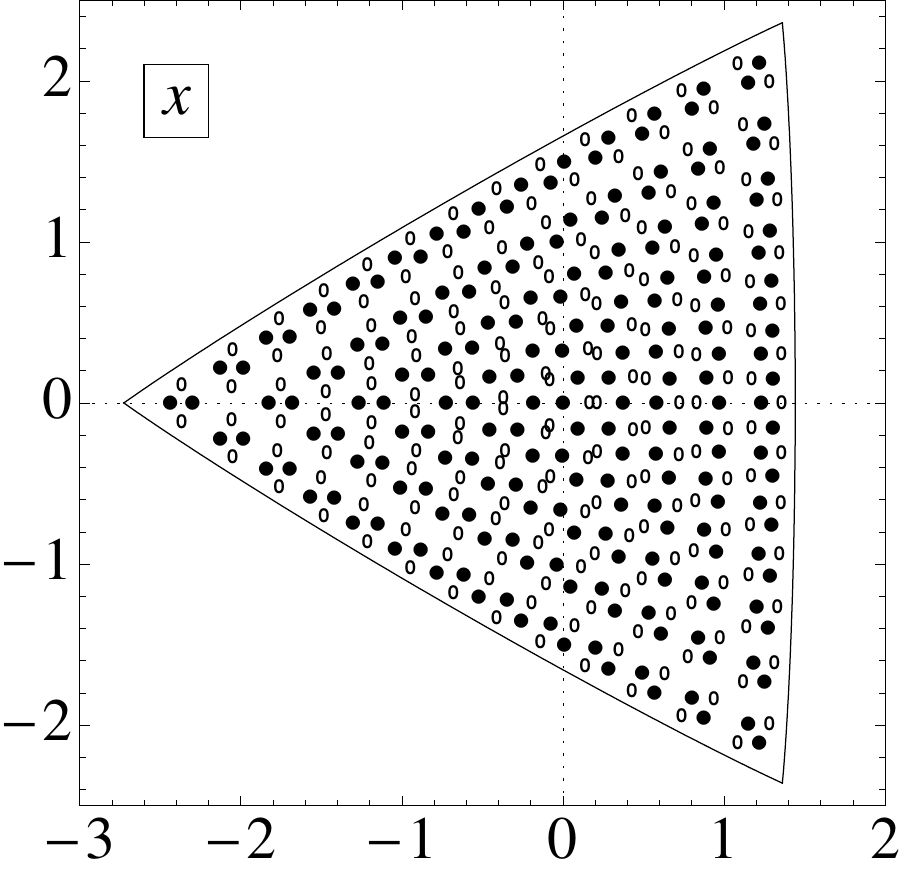}
\caption{\emph{The zeros (circles) and poles (dots) of $\pu_{14}(y)$ (left) 
and $\pp_{14}(y)$ (right) in the complex $x$-plane, where 
$x=(14-\tfrac{1}{2})^{-2/3}y$, 
along with the $m$-independent boundary of the elliptic region $T$.  
Note that each pole of $\pu_{14}(y)$ corresponds to a (shifted) zero of $\pu_{13}(y)$ in Figure \ref{u13-p13-zeros-poles}.  }}
\label{u14-p14-zeros-poles}
\end{figure}

\newpage
The purpose of our paper is to use Riemann-Hilbert analysis to rigorously and explicitly determine, with 
error estimates, the large-$m$ behavior of $\pu_m(y)$, $\pv_m(y)$, $\pp_m(y)$, 
and $\pq_m(y)$ for $y$ 
outside and inside the elliptic region.  The remaining cases, where 
$x=(m-\frac{1}{2})^{-2/3}y$ is near an edge or a corner of the elliptic 
region, will be handled in a subsequent work \cite{Buckingham-rational-crit}.  
Our starting point is Riemann-Hilbert 
Problem~\ref{rhp:DSlocalII} (see \S\ref{Riemann-Hilbert-section}) which arises naturally as a parametrix in 
the double-scaling semiclassical limit of the sine-Gordon equation 
\cite{BuckinghamMcritical}.  This 
Riemann-Hilbert problem is associated with the so-called Jimbo-Miwa 
Lax pair for the Painlev\'e-II equation \cite{Jimbo:1981a,Jimbo:1981b} (see also 
\cite[page 156]{Fokas:2006-book}).  The jump matrices for this Riemann-Hilbert 
problem are nontrivial and oscillatory for rational solutions, which allows us to use the 
machinery of the Deift-Zhou nonlinear steepest-descent method 
\cite{Deift:1993,Deift:1997}.  On the other hand, for the Riemann-Hilbert 
problem 
associated with the alternative so-called Flaschka-Newell Lax pair 
\cite{Flaschka:1980} the jump matrices corresponding to the rational solutions are trivial 
(i.e.\@ they all degenerate to the Identity matrix), 
and instead the monodromy data is encoded in the principal part expansion of 
a high-order pole at the origin.   
Therefore, the Deift-Zhou method does not apply to the latter problem without substantial modifications to exchange isolated pole singularities for jumps along contours.  

\begin{remark}
If $u(y)$ and $v(y)$ solve the coupled system \eqref{PII-system}, then 
$w(y):=u(y)v(y)$ solves the Painlev\'e-XXXIV equation 
\eq
w_{yy} = \frac{w_y^2}{2w} - 4w^2 - \frac{2}{3}yw - \frac{\nu^2}{2w}, \quad \nu:=vu_y-uv_y.
\endeq
Therefore the results we will present will imply corresponding asymptotic formulae for the rational solutions of the Painlev\'e-XXXIV equation.
\end{remark}

\begin{remark}
It is sufficient to assume that $m$ is a large \emph{positive} integer.  Indeed, a simple induction argument using the recursions \eqref{backlund-positive}--\eqref{backlund-negative} and initial conditions \eqref{backlund-initial-condition} shows
that
\begin{equation}
\mathcal{U}_{-m}(y)=\frac{1}{\mathcal{U}_m(y)}\quad\text{and}\quad
\mathcal{V}_{1-m}(y)=\frac{1}{\mathcal{V}_{m+1}(y)} \quad \text{for all } m\in\mathbb{Z}.
\end{equation}
\end{remark}

\subsection{A Boutroux-type ansatz}  

We now note a simple but nonrigorous 
computation that motivates some of our results.  Starting from the Painlev\'e-II 
equation \eqref{PII}, we would like to rescale $p$ and $y$ so that the 
zeros and poles of the rational solutions are (approximately) equally spaced.  
It is known that the maximum modulus of the zeros of $\pp_m(y)$ grows as 
$m^{2/3}$ \cite{Kametaka:1983}.  This suggests the fact (which we will prove 
later) that the large-$m$ asymptotic boundaries of the elliptic region $T$ 
are fixed in the $x$-plane, where $x=(m-\tfrac{1}{2})^{-2/3}y$.  (We include the shift 
by $\tfrac{1}{2}$ to be consistent with our later calculations.)  
To zoom in on the local behavior near a point $x=x_0\in T$ we need a local 
coordinate $w$ that behaves like $m(x-x_0)$ since the number of 
roots (or poles) of the rational solution is of order $m^2$.  These arguments motivate the rescalings: 
\eq
\label{Boutroux-rescalings}
y=(m-\tfrac{1}{2})^{2/3}x = (m-\tfrac{1}{2})^{2/3}x_0 + (m-\tfrac{1}{2})^{-1/3}w, \quad p(y)=(m-\tfrac{1}{2})^{1/3}\dot{\mathcal{P}}(w).
\endeq
These rescalings render \eqref{PII} in the equivalent form
\eq
\dot{\mathcal{P}}'' = 2\dot{\mathcal{P}}^3 + \frac{2}{3(m-\tfrac{1}{2})}w\dot{\mathcal{P}} + \frac{2}{3}x_0\dot{\mathcal{P}} - \frac{2}{3}, \quad {}^\prime=\frac{d}{dw}.
\endeq
Now formally disregarding the term whose coefficient becomes small as 
$m\to\infty$ gives the model equation
\eq
\dot{\mathcal{P}}'' = 2\dot{\mathcal{P}}^3 + \frac{2}{3}x_0\dot{\mathcal{P}} - \frac{2}{3}.
\endeq
Multiplication by $\dot{\mathcal{P}}'$ and integrating with respect to $w$ gives 
\eq
\frac{1}{2}(\dot{\mathcal{P}}')^2-\frac{1}{2}\dot{\mathcal{P}}^4-\frac{1}{3}x_0\dot{\mathcal{P}}^2 + \frac{2}{3}\dot{\mathcal{P}} = \frac{1}{2}\Pi
\endeq
for some integration constant $\Pi=\Pi(x_0)$.  Thus $\dot{\mathcal{P}}$ satisfies
\eq
\dot{\mathcal{P}}' = \left(\dot{\mathcal{P}}^4 + \frac{2}{3}x_0 \dot{\mathcal{P}}^2 - \frac{4}{3}\dot{\mathcal{P}} + \Pi\right)^{1/2}.
\label{eq:BoutrouxElliptic}
\endeq

The remaining argument now breaks into two cases.  First, suppose one can find 
functions $S=S(x_0)$ and $\Delta=\Delta(x_0)$ satisfying 
$6S^2+3\Delta^2=-8x_0$ and $3S\Delta^2=16$.  Then \eqref{eq:BoutrouxElliptic} 
can be written as 
\eq
\label{eq:BoutrouxGen0}
(\dot\pp')^2 = \left(\dot\pp-\frac{S-\Delta}{2}\right)\left(\dot\pp-\frac{S+\Delta}{2}\right)\left(\dot\pp+\frac{S}{2}\right)^2
\endeq
provided that also $\Pi=S^2(S^2-\Delta^2)/8$.  Then it is evident that the 
constant function 
$\dot\pp=-S/2$ (independent of $w$) satisfies 
\eqref{eq:BoutrouxGen0}.  In fact, this will turn out to be the correct 
leading-order approximation to $\pp$ in most of the complex $x$-plane (see 
Theorem~\ref{main-genus-zero-thm}).

The second case is the generic one, in which the quartic on the 
right-hand side of \eqref{eq:BoutrouxElliptic} has distinct roots.  In this case, the solution to the differential equation \eqref{eq:BoutrouxElliptic} is a certain elliptic function of $w$ (unique up to translation in $w$), and since $x_0$ enters explicitly into one of the coefficients of the quartic on the right-hand side, the elliptic function will be slowly modulated as $x_0$ varies in the complex plane.  Thus one expects that, near some points $x_0$ in the complex $x$-plane at least, the rational function $\mathcal{P}_m$ is modeled by an elliptic function of a local variable $w$ satisfying the differential equation 
\eqref{eq:BoutrouxElliptic}.  A similar line of reasoning was followed by Boutroux \cite{Boutroux13} in his analysis of solutions of the first and second Painlev\'e equations in the limit that the independent variable tends to infinity.  It turns out that this result is essentially correct for $x_0$ inside 
the region $T$.  In this case, the constant of integration $\Pi$ will be tied to $x_0$ in a precise manner that is difficult to motivate from the simple 
reasoning given here.  

\subsection{Summary of Main Results}
Here we summarize our key results in words and provide references to the 
mathematically precise statements that follow.  We present our results in 
terms of the scaled variable 
\eq
x:=\left(m-\tfrac{1}{2}\right)^{-2/3}y
\endeq
for which the zeros and poles of the rational Painlev\'e-II functions densely 
fill out (as $m\to\infty$) the fixed, open, bounded, and simply-connected 
subset $T$ of the plane (the \emph{elliptic region}).  The behavior of the rational 
Painlev\'e-II functions is particularly difficult to formulate 
in a compact fashion for $x\in T$, so we delay the presentation of our results 
in full mathematical detail until the required machinery has been developed in 
\S\ref{section-gen0} and \S\ref{bulk-section}.  
We obtain our results by applying the Deift-Zhou steepest descent 
techniques to an appropriate Riemann-Hilbert problem specified in 
\S\ref{Riemann-Hilbert-section}.

\subsubsection{Asymptotic behavior for $x$ outside of the elliptic region $T$}
The asymptotic analysis of the rational Painlev\'e-II functions for $x$ 
outside of the elliptic region $T$ is more straightforward than the analysis 
in the elliptic region since, with the exception of the need to employ Airy 
functions to construct two parametrices in the standard fashion, all of the work involves elementary functions.  We obtain the following results.
\begin{itemize}
\item We characterize the boundary $\partial T$ of the elliptic region $T$ as explicitly as possible, expressing it as a level set of a function built from explicit elementary functions and the solution of a cubic equation (see \eqref{cubic-equation}, \eqref{S-large-x}, and \eqref{eq:EdgeCondition}).
\item We obtain explicit (up to the solution of the cubic equation \eqref{cubic-equation}) asymptotic formulae in the limit of large $m\in\mathbb{Z}_+$ for the rational Painlev\'e-II functions $\pu_m$, $\pv_m$, $\pp_m$, and $\pq_m$ for $x\notin\overline{T}$ (see Theorem~\ref{main-genus-zero-thm}).  The result for $\pp_m$ confirms that 
the (non-elliptic) Boutroux ansatz is correct for $x\notin\overline{T}$.  The order of accuracy is uniformly $\mathcal{O}(m^{-1})$ and the rescaled variable $x$ is allowed to approach the edges (but not the corners) of the boundary $\partial T$ of the elliptic region from the outside at a distance proportional to $\log(m)/m$ with no loss in the rate of decay of the error terms.  
\end{itemize}
The accuracy of the asymptotic approximations we construct for $x\in\mathbb{C}\setminus\overline{T}$ is illustrated (for $x\in\mathbb{R}$) in Figures~\ref{exact-vs-asymp-um-genus0} and \ref{exact-vs-asymp-pm-genus0}, wherein
\eq
x_c:=\inf (T\cap\mathbb{R}) \quad \text{and} \quad x_e:=\sup (T\cap\mathbb{R}).
\endeq
To generate these plots requires only the solution of the cubic equation \eqref{cubic-equation} for every $x$ of interest, a task easily accomplished numerically.
\begin{figure}[h]
\includegraphics[width=2in]{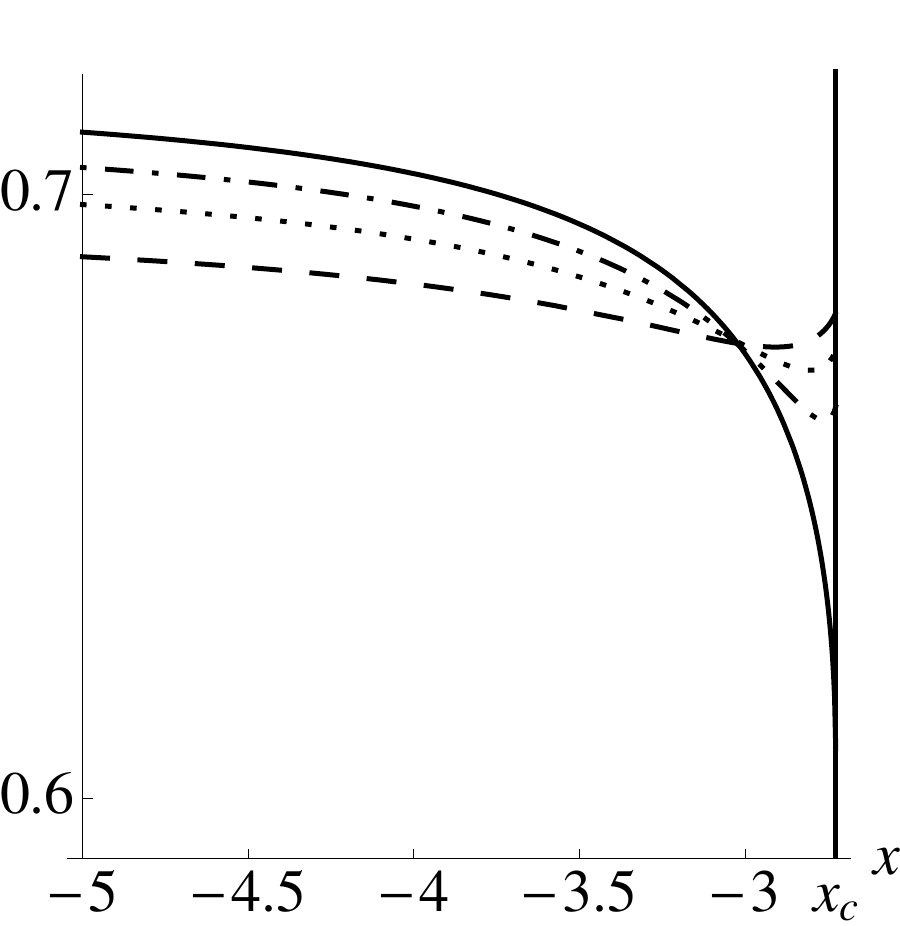}
\includegraphics[width=2in]{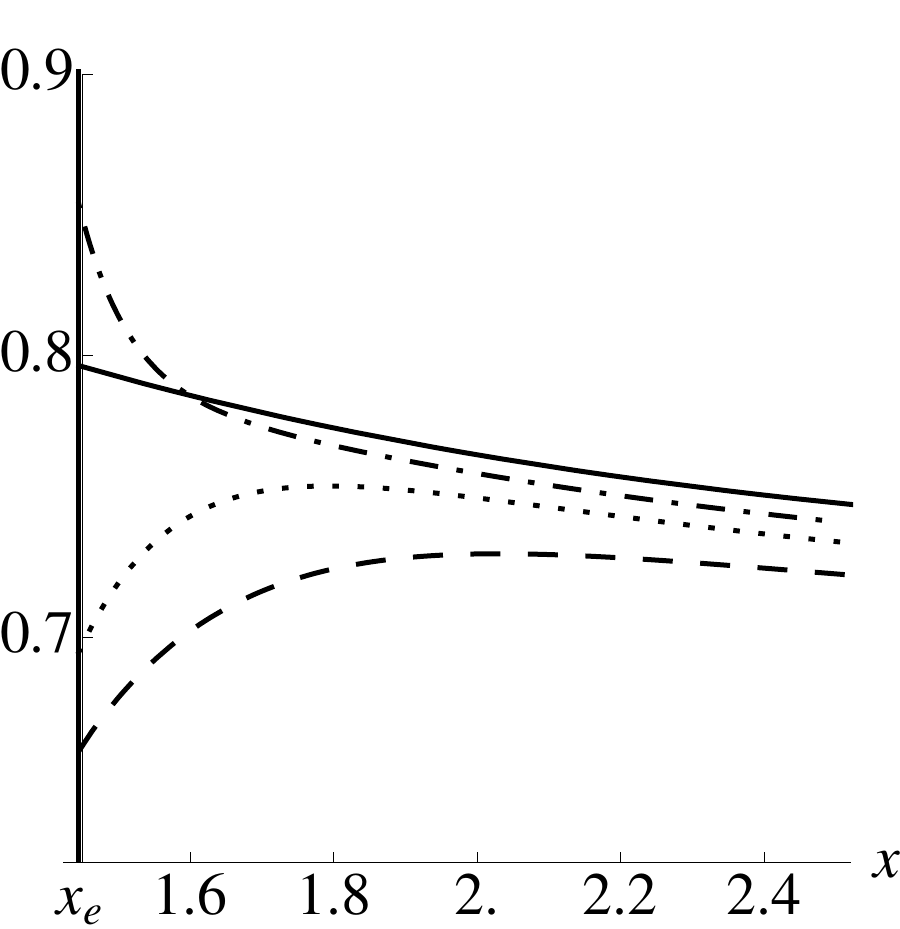}
\caption{\emph{Illustration of the approximation of $m^{-2m/3}e^{-m\lambda(x)}\pu_m((m-\tfrac{1}{2})^{2/3}x)$ (dashed lines: $m=3$, dotted lines: $m=5$, dot-dashed lines: $m=10$) by $\dot{\pu}(x)$ (solid lines) for $x<x_c$ (left plot) and $x>x_e$ (right plot).  The function $\lambda(x)$ is defined in \eqref{eq:lambdadef} and $\dot{\pu}(x)$ in \eqref{g0-pudot-pvdot-ppdot-pqdot}.  The right plot illustrates that the nature of the convergence at $x=x_e$ depends on if $m$ is even or odd (due to the behavior of $\pu_m(y)$ at its largest positive pole). }}
\label{exact-vs-asymp-um-genus0}
\end{figure}
\begin{figure}[h]
\includegraphics[width=2in]{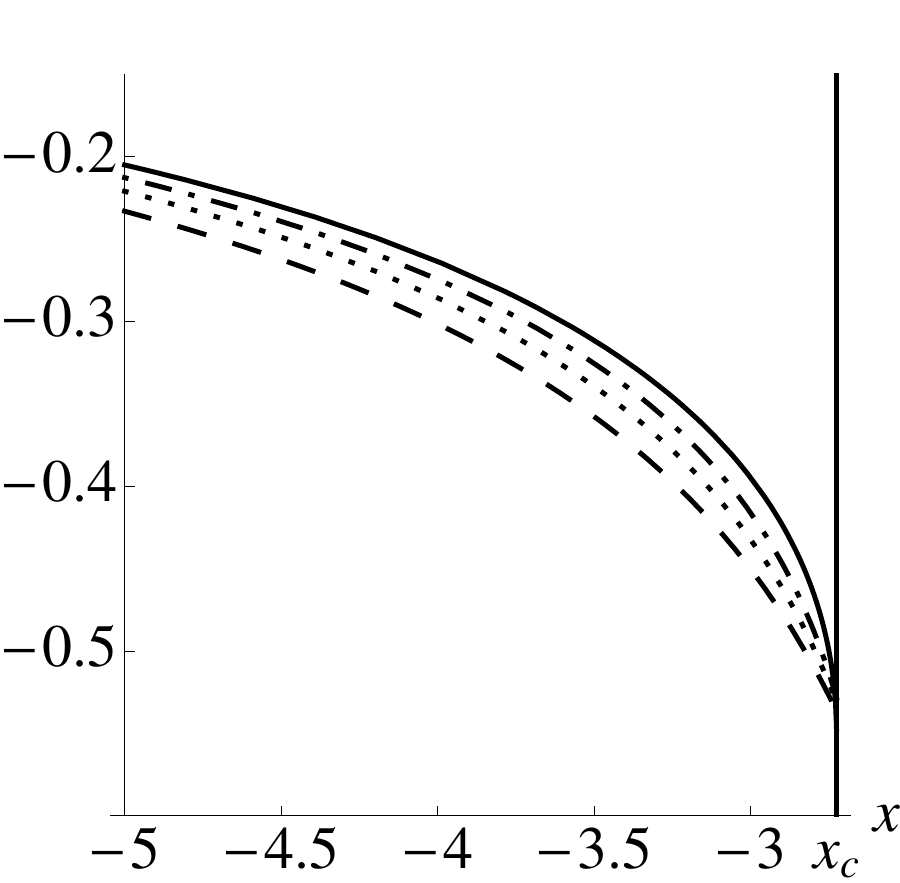}
\includegraphics[width=2in]{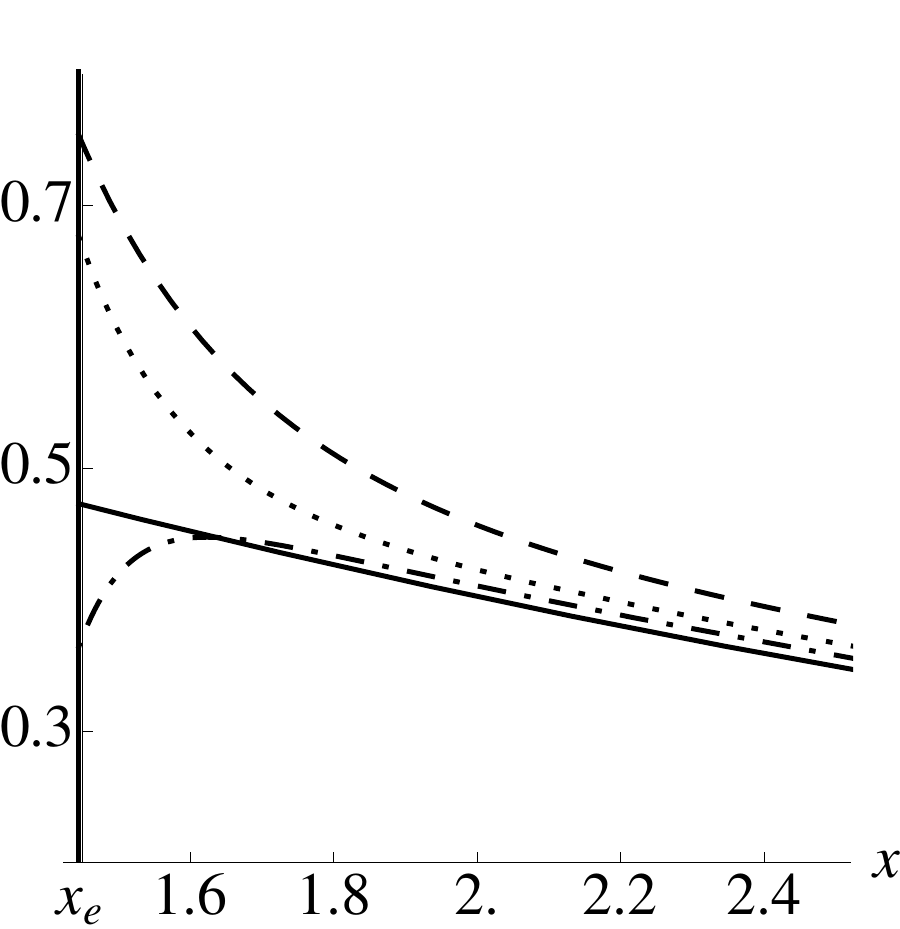}
\caption{\emph{Plots showing the asymptotic approximation of $m^{-1/3}\pp_m((m-\tfrac{1}{2})^{2/3}x)$ (dashed lines: $m=3$, dotted lines: $m=5$, dot-dashed lines: $m=10$) by $\dot{\pp}(x)$ (solid lines) for $x<x_c$ (left plot) and $x>x_e$ (right plot).  The function $\dot{\pp}(x)$ is defined in \eqref{g0-pudot-pvdot-ppdot-pqdot}.  The right plot illustrates that the nature of the convergence at $x=x_e$ depends on if $m$ is even or odd (due to the behavior of $\pp_m(y)$ at its largest positive pole). }}
\label{exact-vs-asymp-pm-genus0}
\end{figure}

\subsubsection{Asymptotic behavior for $x$ inside the elliptic region $T$}  

We provide a rigorous justification to the formal argument of the Boutroux 
ansatz method in the generic (elliptic) case when $x\in T$. 
The form of the quartic $z^4+\tfrac{2}{3}x_0z^2-\tfrac{4}{3}z+\Pi$ appearing in the Boutroux ansatz method (see the right-hand side of \eqref{eq:BoutrouxElliptic}) arises in a completely different way, through the imposition of certain \emph{moment conditions} (see \eqref{eq:g1-moments}) needed to construct an appropriate $g$-function (a key ingredient in the Deift-Zhou method).  The integration constant $\Pi$ is determined as a function of $x_0$ in order that certain constant exponents that occur in the jump matrices of a model Riemann-Hilbert problem are purely imaginary, ensuring that the solution of the model problem is suitably bounded and that errors can be controlled.  The precise conditions that determine $\Pi$ we call \emph{Boutroux conditions}, and they take the form \eqref{eq:g1-Boutroux}, or equivalently, \eqref{eq:BoutrouxConditions}.  
It turns out that $\Pi$ is a complex-valued function of $x_0$ that is smooth 
(i.e., $\text{Re}(\Pi)$ and $\text{Im}(\Pi)$ are infinitely differentiable 
functions of $\text{Re}(x_0)$ and $\text{Im}(x_0)$) but nowhere analytic for 
$x_0\in T$.  This fact leads to certain challenges in interpreting the asymptotic formulae for the rational Painlev\'e-II functions that we discuss at length in \S\ref{sec:g1-approximate-formulae}.  Our main results are the following.
\begin{itemize}
\item We obtain asymptotic formulae for all four rational Painlev\'e-II functions $\pu_m$, $\pv_m$, $\pp_m$, and $\pq_m$ in terms of the solution of certain $m$-independent algebraic equations that we prove exists (and that we are able to effectively implement numerically) and Riemann theta functions in whose arguments $m$ appears explicitly.  See Theorem~\ref{theorem-g1-basic}.  Significantly, we are able to obtain accuracy in a suitable reciprocal sense even near points of $T$ at which there exist (necessarily simultaneous and simple) poles of $\pu_m$ and $\pv_m$, and hence there is no solution whatsoever to the original Riemann-Hilbert problem formulated in \S\ref{Riemann-Hilbert-section} below.  We achieve this using B\"acklund transformations to essentially turn each pole into a zero of a related function that we can analyze.
\item As a corollary (see Corollary~\ref{cor:g1-pole-zero-approx}) we prove that uniformly for $x$ in compact subsets of $T$ (that is, avoiding $\partial T$) the zeros and poles of the rational Painlev\'e-II functions each lie within a distance proportional to $m^{-2}$ in the $x=(m-\tfrac{1}{2})^{-2/3}y$ independent variable from exactly one corresponding zero or pole of the (mostly) explicit approximating functions.  The distance between nearest neighbor poles or zeros scales as $m^{-1}$ in the $x$-plane.
\item We prove a distributional convergence result for the rescaled rational Painlev\'e-II function 
$m^{-1/3}\pp_m$ considered as a function of $x$, in which the rapid fluctuations of the rational function modeled by elliptic functions within $T$ are locally averaged in two dimensions to produce a genuine $m$-independent limit function that we call $\langle\dot{\pp}\rangle(x)$.  See Theorem~\ref{theorem:g1-weak-limit}.  Combining this result with the strong convergence result we obtain in Theorem~\ref{main-genus-zero-thm} for $x$ outside $T$, we define a ``macroscopic limit'' formula for $m^{-1/3}\pp_m$ that we call $\dot{\pp}_\mathrm{macro}(x)$, and that is a valid distributional limit for all $x$ away from $\partial T$.  See Corollary~\ref{corollary:g1-global-weak-convergence} for this global weak convergence result.
\item We obtain a similar distributional convergence result for $m^{-1/3}\pp_m((m-\tfrac{1}{2})^{2/3}x)$ now considered as a real-valued function of a real variable.  Here due to simple poles the integrals against test functions are defined in the principal value sense.  See Theorem~\ref{theorem:g1-weak-limit-real}.
\item We calculate the asymptotic density of poles of $\pu_m$ in the complex $x$-plane near an arbitrary point $x\in T$ and express it in terms of $m$-independent quantities that are easy to calculate numerically as functions of $x\in T$.  We also calculate the linear density of real poles of $\pu_m$ for $x\in T\cap \mathbb{R}$.  See Theorem~\ref{theorem:g1-density-equalities}.
\end{itemize}
We wish to emphasize that our asymptotic formulae are effective for numerical computations.
In fact, we were surprised at the accuracy of the approximate formulae; to the eye they are remarkably accurate for $m$ as small as $m=2$ or $m=3$ even though we only prove their accuracy in the asymptotic limit $m\to +\infty$.  
In Figure~\ref{fig:g1-U-compare-T} we plot the leading term of a suitably exponentially renormalized version of $\pu_m$ that is a valid approximation for $x\in T$ and superimpose the actual locations of the poles and zeros as numerically calculated by root-finding applied to formulae generated from the B\"acklund transformations \eqref{backlund-positive}--\eqref{backlund-negative}.  In Figures~\ref{fig:g1-U-compare-real}--\ref{fig:g1-U-compare-PiBySix-ImagParts} we compare the renormalized version of $\pu_m$ with its leading-order approximation for various $m$ on $\mathbb{R}\cap T$ and $e^{i\pi/6}\mathbb{R}\cap T$ (on which the approximation is pole-free).  
Figures~\ref{fig:g1-P-compare-real}--\ref{fig:g1-P-compare-PiBySix-ImagParts} do the same for the function $\pp_m$ and its leading-order approximation.  In this case when we take a real section of $T$ we can also compare to the distributional (weak) limit described by Theorem~\ref{theorem:g1-weak-limit-real}.  In Figure~\ref{fig:g1-densities} we plot the planar and linear pole density functions for $\pu_m$.  A Mathematica code for producing these and other figures is available from the authors upon request.
\begin{figure}[H]
\begin{center}
\includegraphics[width=2.3 in]{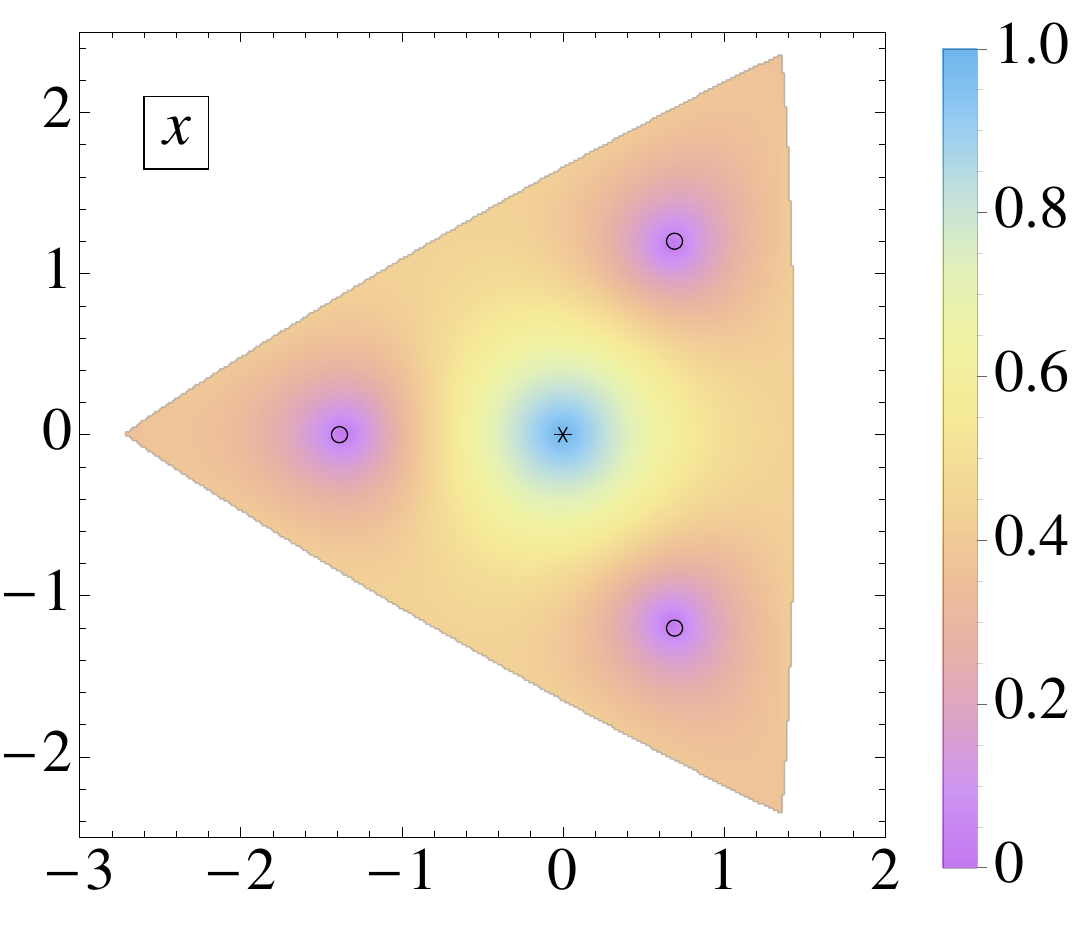}\hspace{0.5 in}
\includegraphics[width=2.3 in]{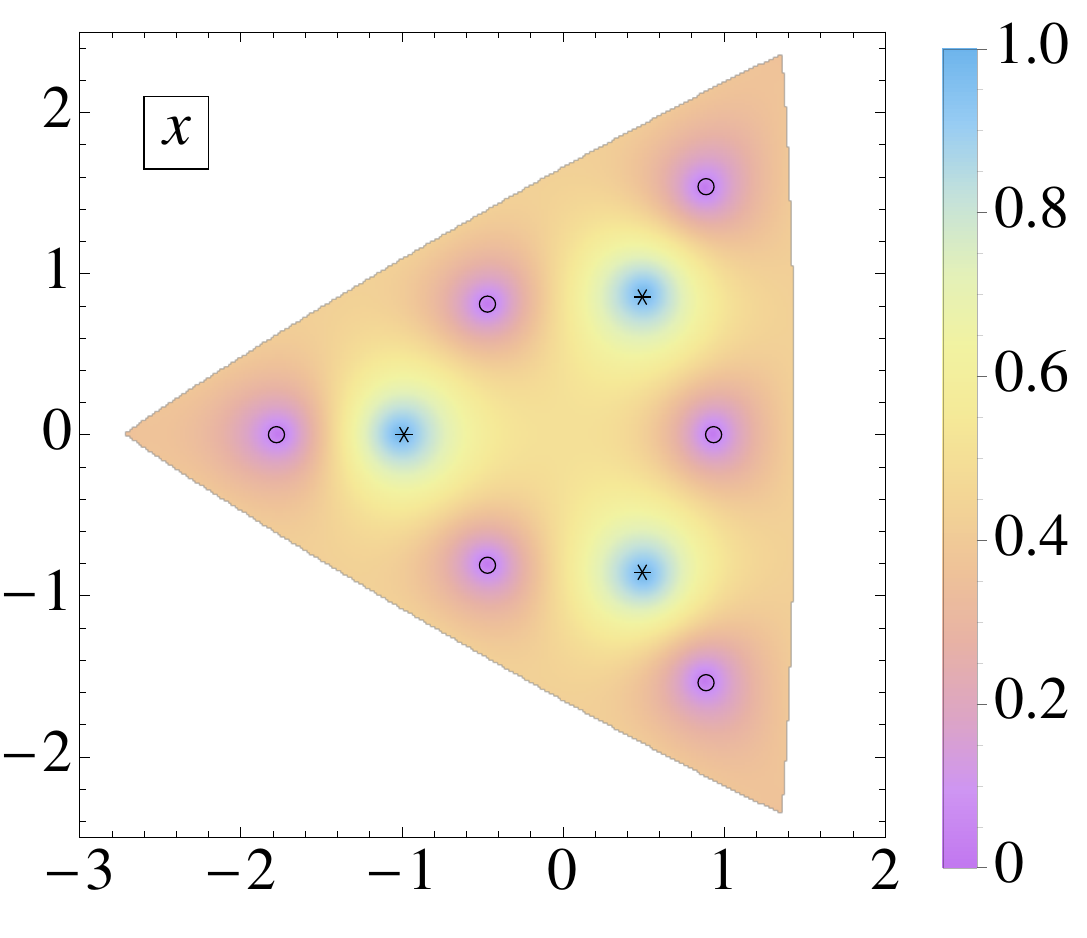}\\
\includegraphics[width=2.3 in]{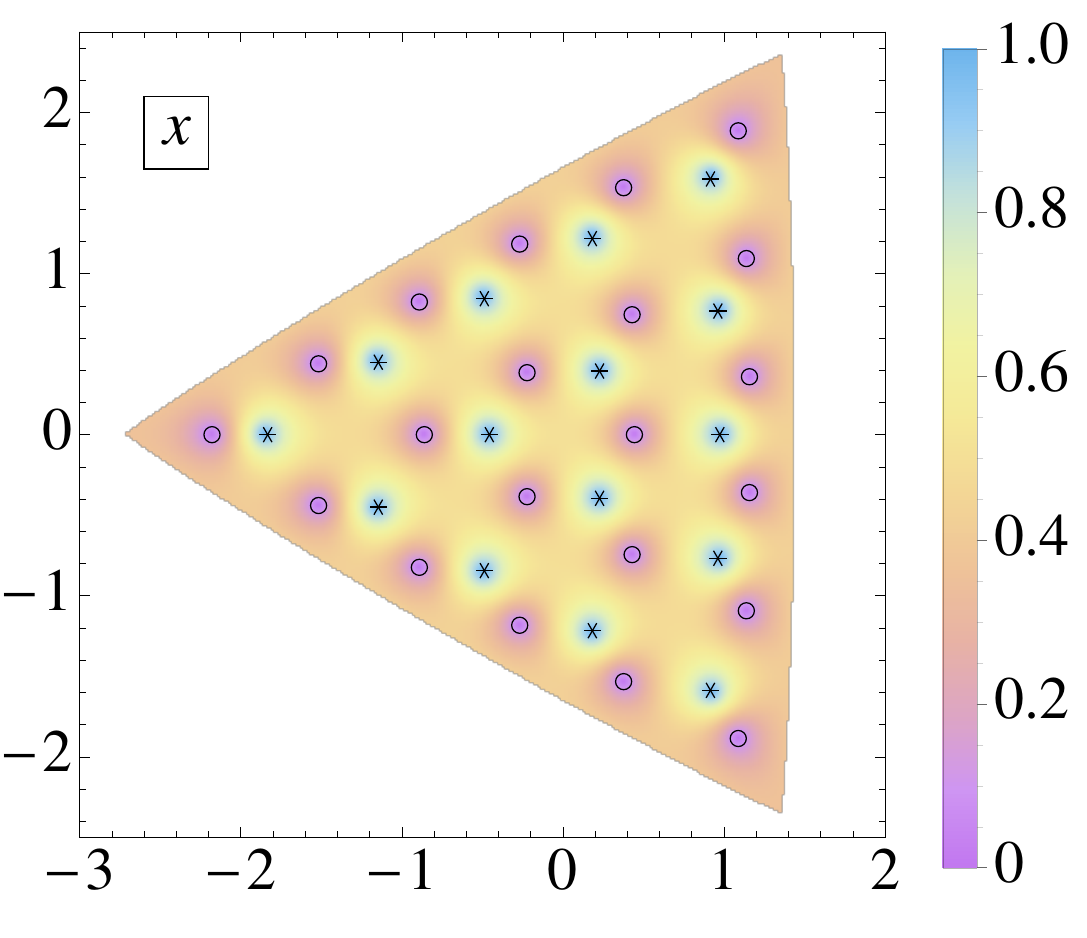}\hspace{0.5 in}
\includegraphics[width=2.3 in]{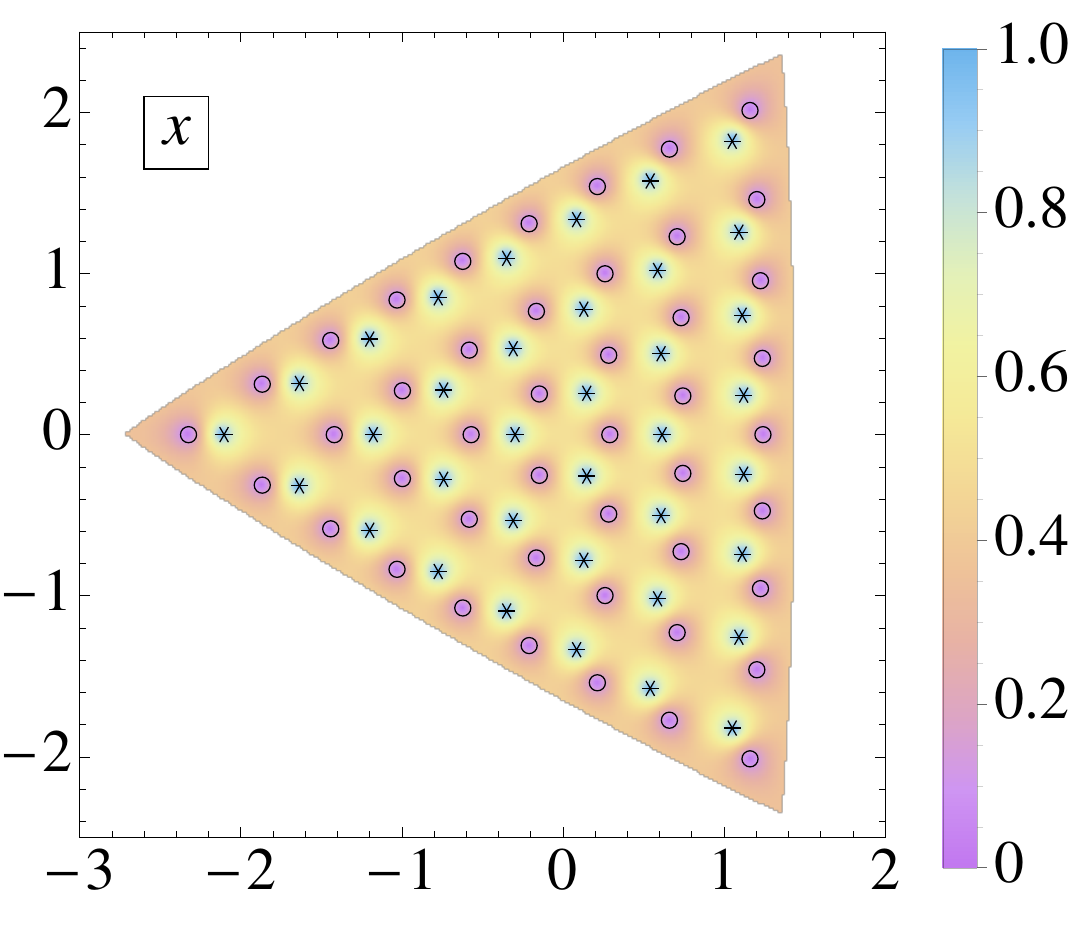}
\end{center}
\caption{\emph{Density plots of $\tfrac{2}{\pi}\arctan(|\dot{\mathcal{U}}_m(0;x)|)$ for $m=2$ (upper left), $m=3$ (upper right), $m=6$ (lower left), and $m=9$ (lower right).  Superimposed with circles and asterisks respectively are the exact locations of the zeros and poles of $\mathcal{U}_m((m-\tfrac{1}{2})^{2/3}x)$.  The function $\dot{\pu}_m(w;x_0)$ is defined in \eqref{eq:g1-dotU-formula} and 
\eqref{eq:g1-UVdot-redefine}.}}
\label{fig:g1-U-compare-T}
\end{figure}

\begin{figure}[H]
\begin{center}
\hspace{-.1in}
\includegraphics[width=2 in]{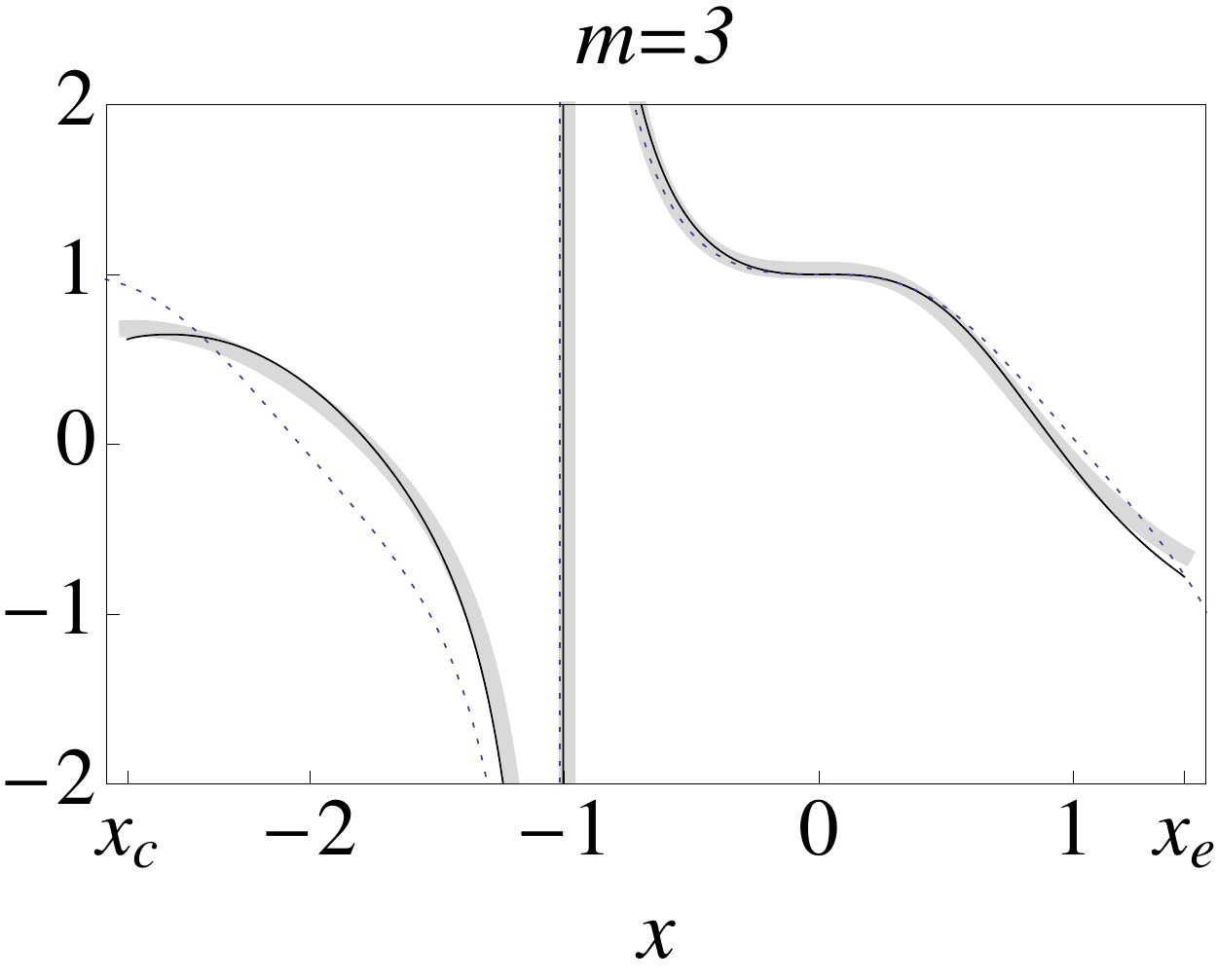}\hspace{0.5 in}
\includegraphics[width=2 in]{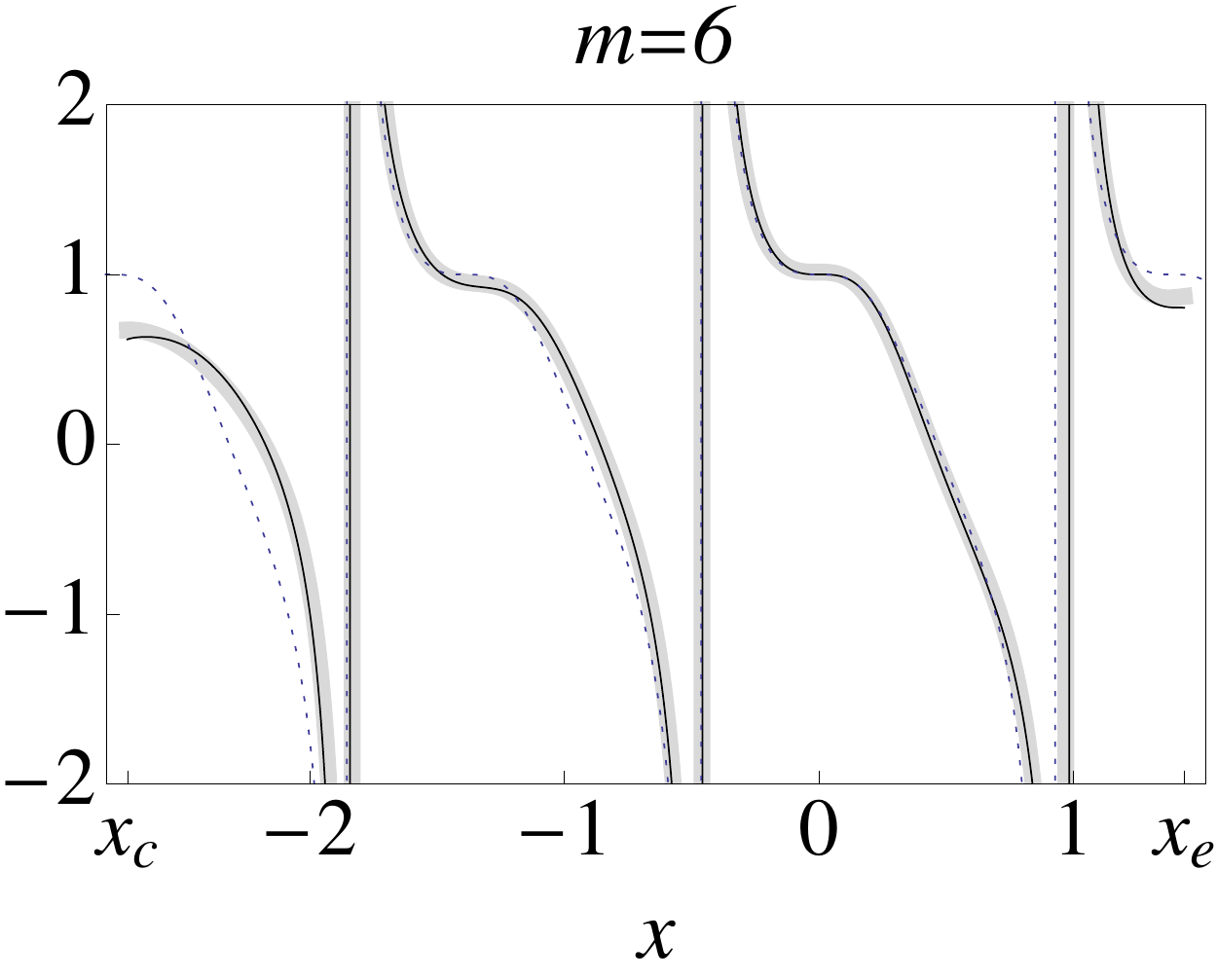}\hspace{0.2in}\\
\includegraphics[width=2 in]{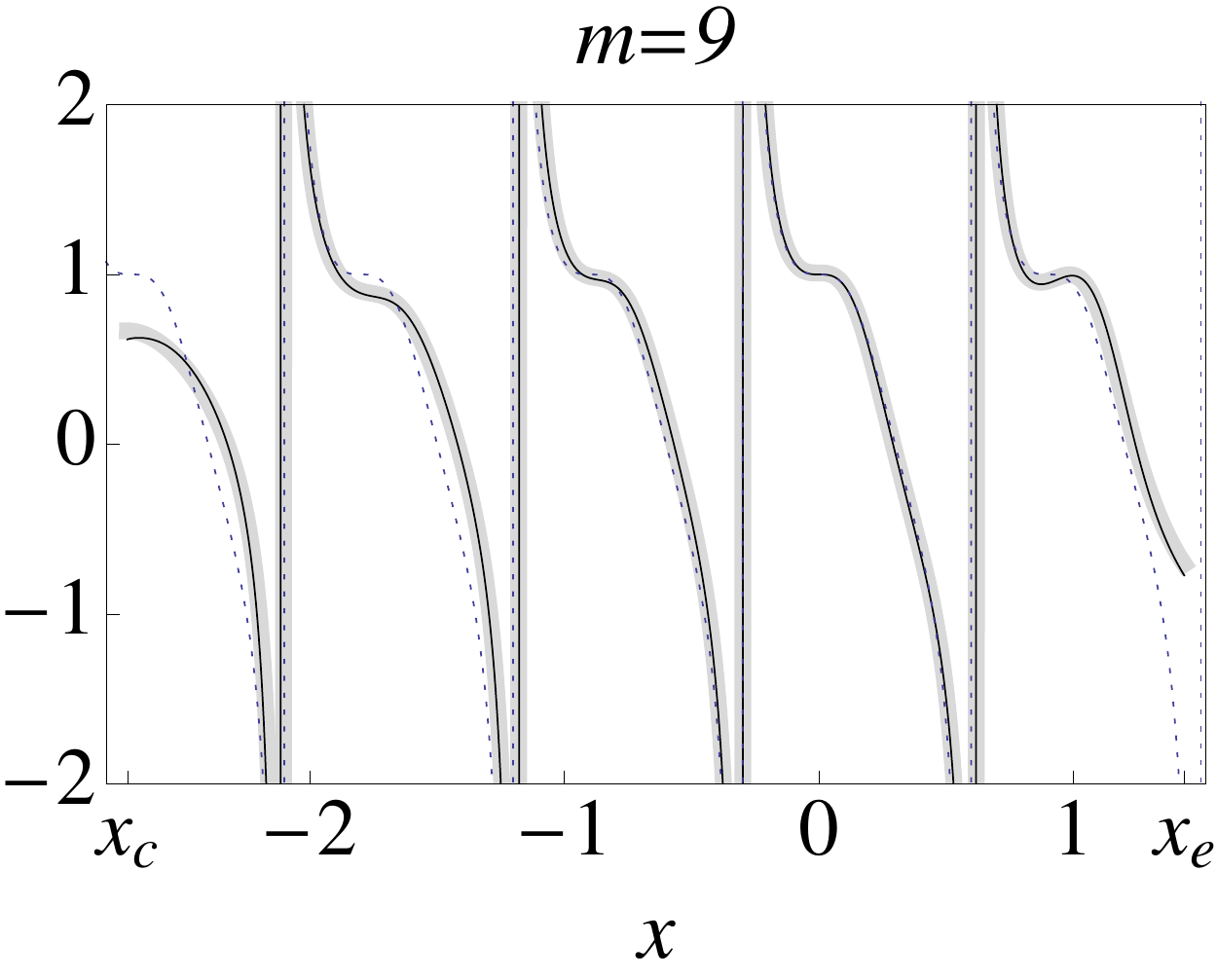}\hspace{0.5 in}
\raisebox{0.5 in}{\includegraphics[width=2 in]{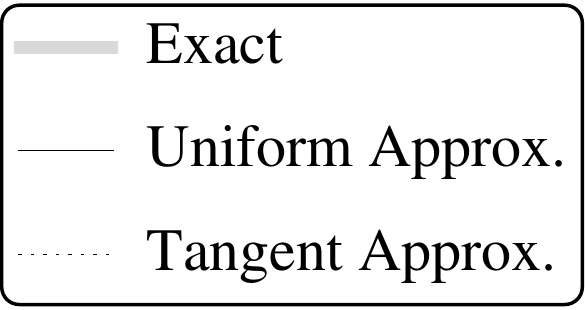}}
\end{center}
\caption{\emph{Comparing $m^{-2m/3}e^{-m\Lambda(x)}\mathcal{U}_m((m-\tfrac{1}{2})^{2/3}x)$ (thick gray curves), its uniform approximation $\dot{\mathcal{U}}_m(0;x)$ (black curves), and a tangent approximation based at the origin $\dot{\mathcal{U}}_m((m-\tfrac{1}{2})x,0)$ (dotted curves) as real functions of $x=x_0\in[x_c,x_e]$ for various values of $m$.  The function $\Lambda(x)$ is defined in \eqref{eq:g1-Lambda} and the function $\dot{\pu}_m(w;x_0)$ is defined in 
\eqref{eq:g1-dotU-formula} and \eqref{eq:g1-UVdot-redefine}.}}
\label{fig:g1-U-compare-real}
\end{figure}
\vspace{-.1in}
\begin{figure}[H]
\begin{center}
\hspace{-.1in}
\includegraphics[width=2 in]{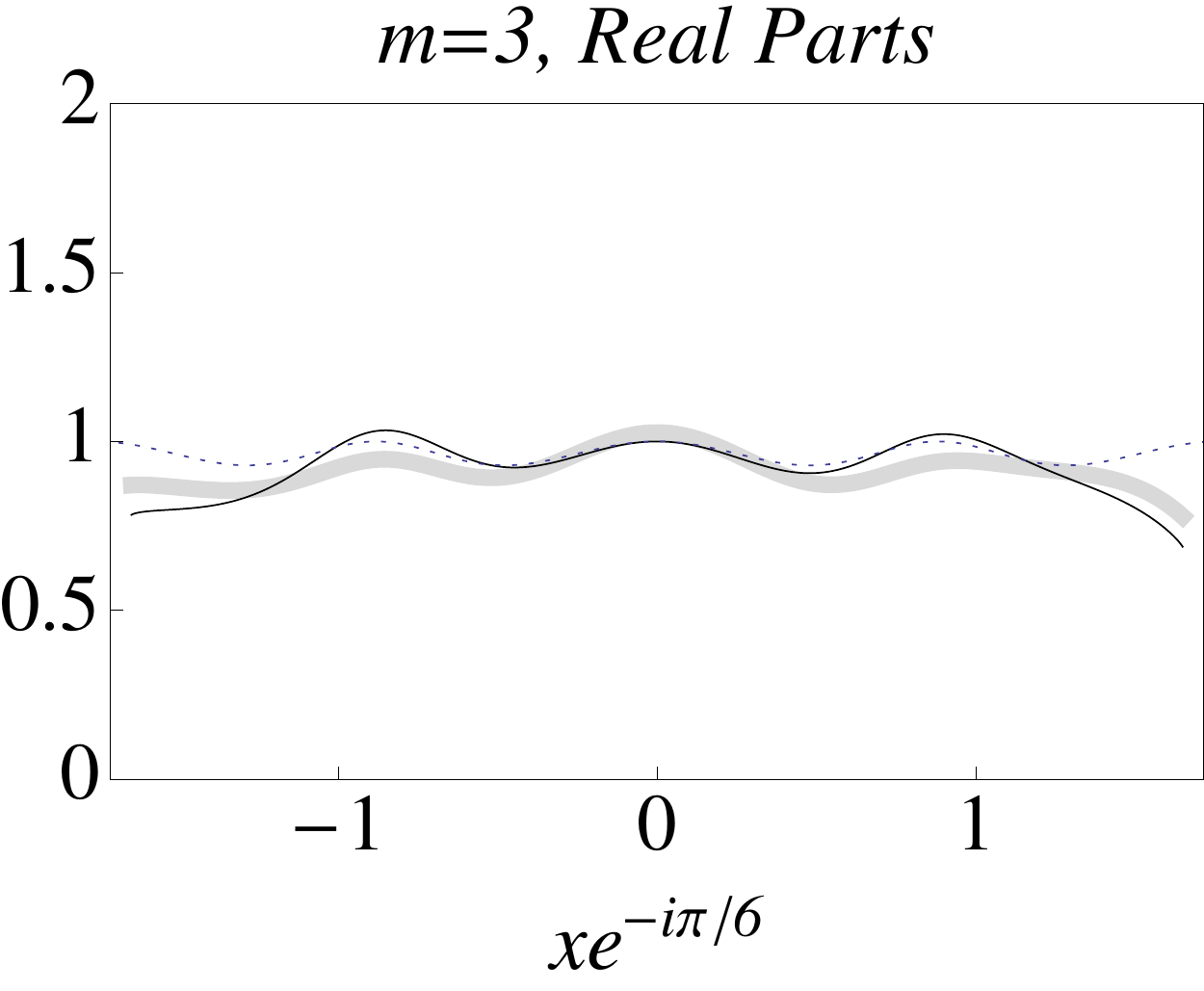}\hspace{0.5 in}%
\includegraphics[width=2 in]{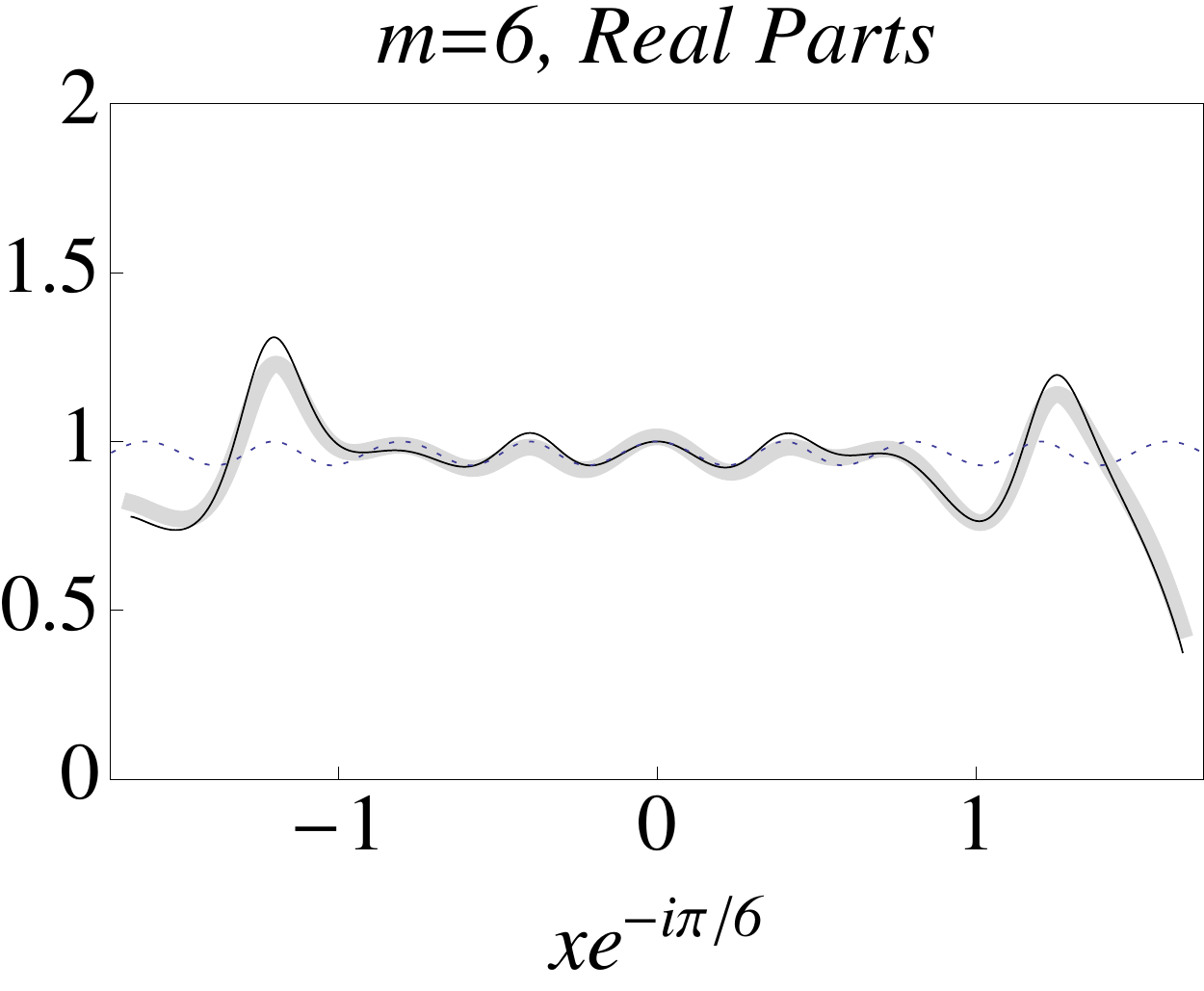}\hspace{0.2in}\\
\includegraphics[width=2 in]{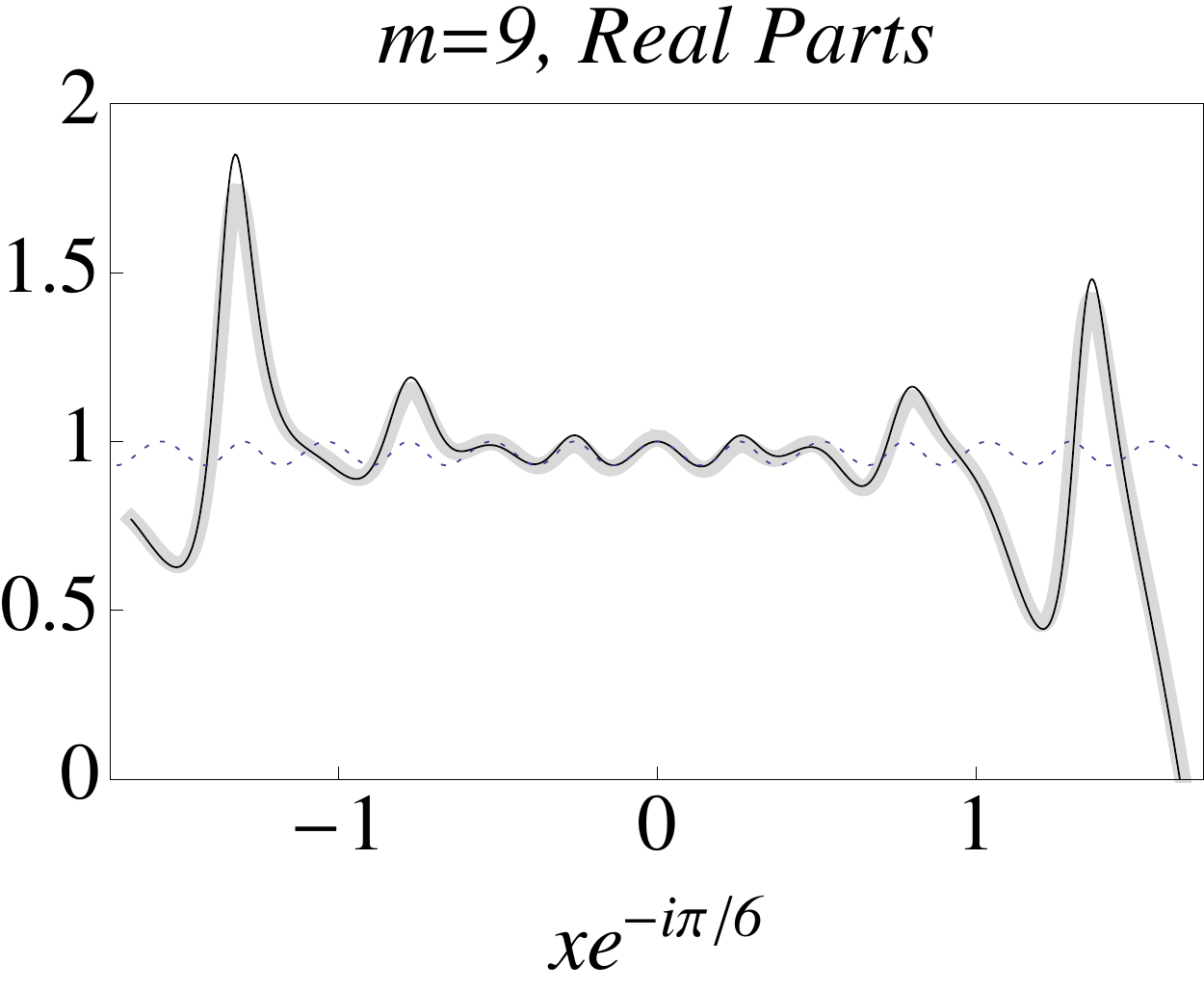}\hspace{0.5 in}
\raisebox{0.5 in}{\includegraphics[width=2 in]{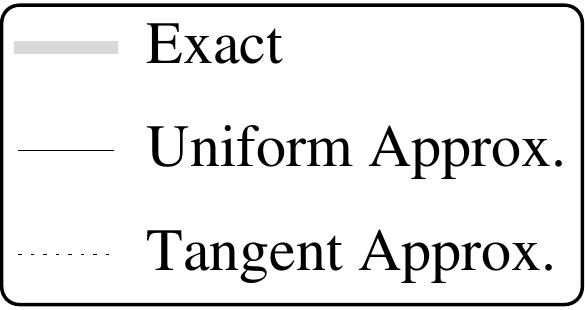}}
\end{center}
\caption{\emph{Comparing the real parts of $m^{-2m/3}e^{-m\Lambda(x)}\mathcal{U}_m((m-\tfrac{1}{2})^{2/3}x)$ (thick gray curves), the corresponding uniform approximation $\dot{\mathcal{U}}_m(0;x)$ (black curves), and tangent approximation $\dot{\mathcal{U}}_m((m-\tfrac{1}{2})x,0)$ (dotted curves) evaluated for $x\in e^{i\pi/6}\mathbb{R}\cap T$ for various values of $m$.  Here $\Lambda(x)$ is defined in \eqref{eq:g1-Lambda} and $\dot{\pu}_m(w;x_0)$ is defined in \eqref{eq:g1-dotU-formula} and \eqref{eq:g1-UVdot-redefine}.}}
\label{fig:g1-U-compare-PiBySix-RealParts}
\end{figure}
\begin{figure}[H]
\begin{center}
\hspace{-.1in}
\includegraphics[width=2 in]{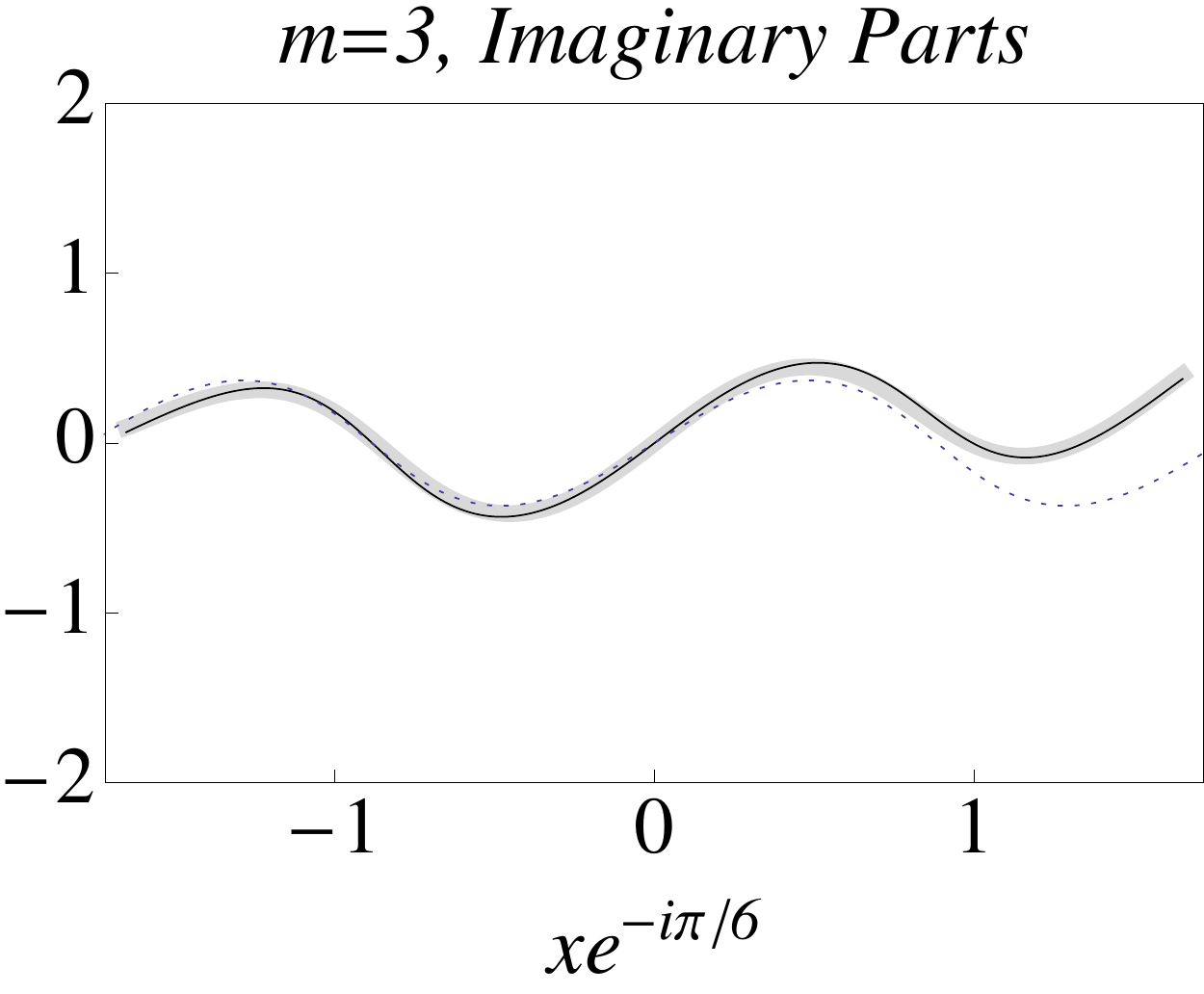}\hspace{0.5 in}%
\includegraphics[width=2 in]{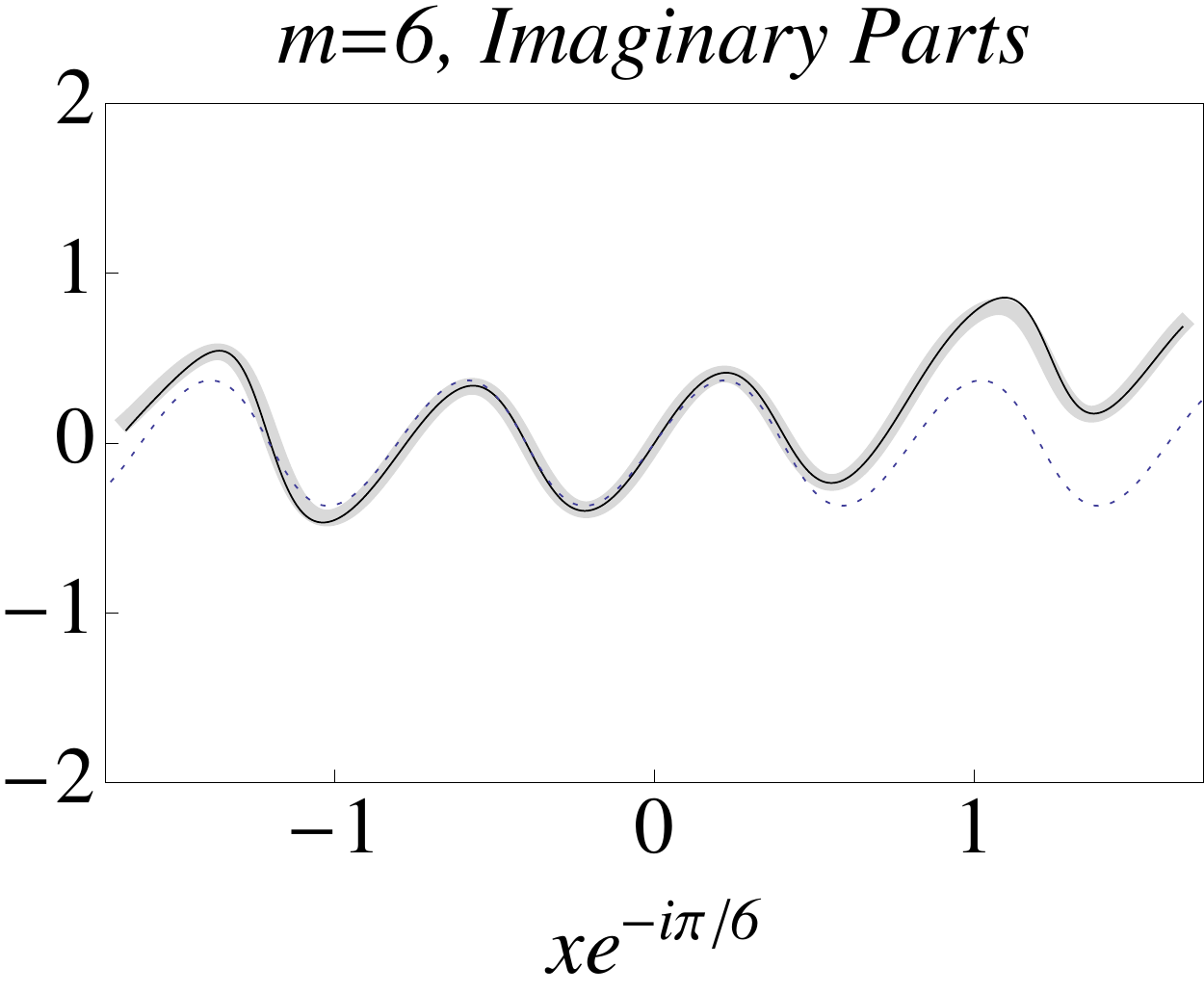}\hspace{0.2in}\\
\includegraphics[width=2 in]{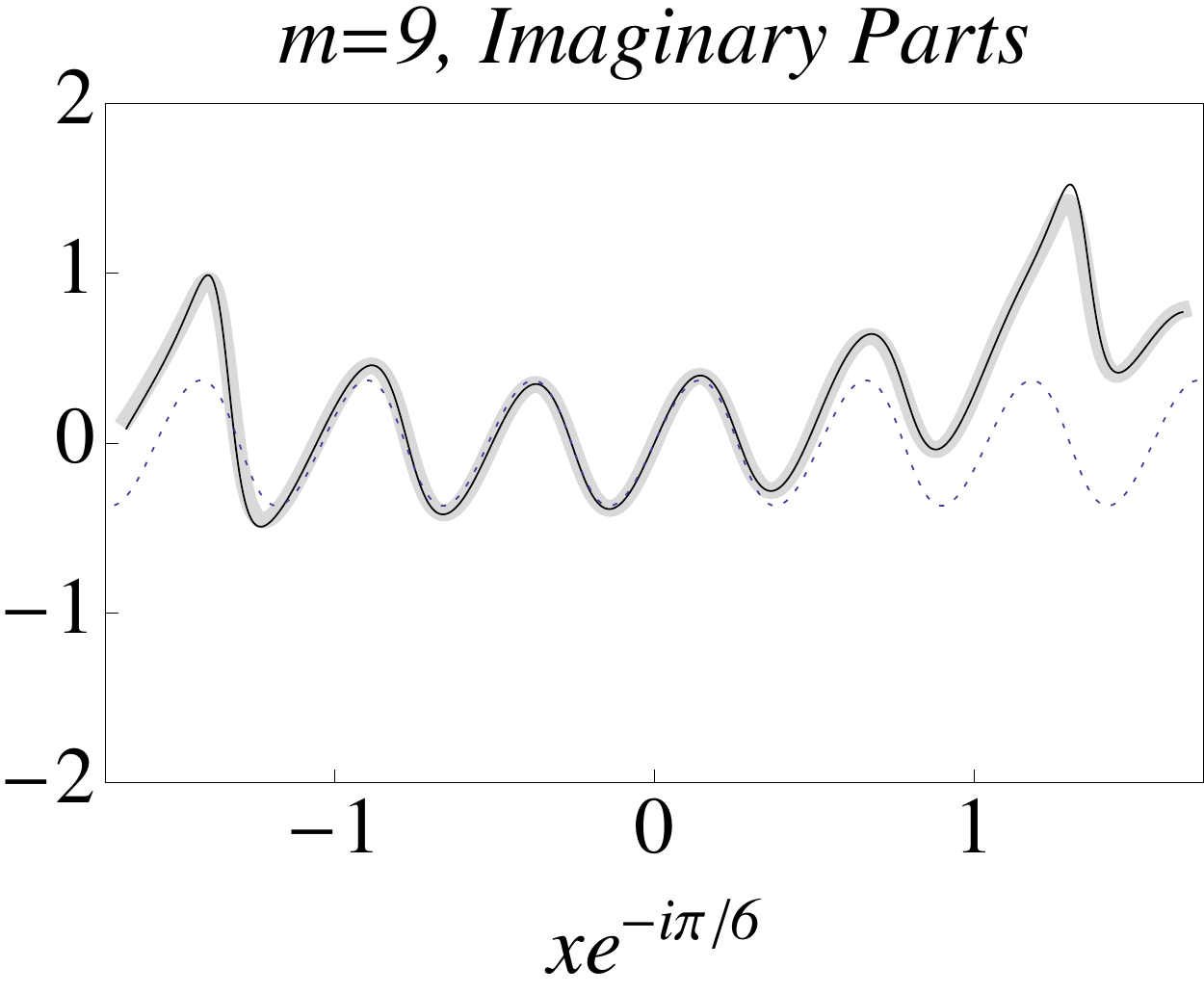}\hspace{0.5 in}
\raisebox{0.5 in}{\includegraphics[width=2 in]{PiBySixULegend.pdf}}
\end{center}
\caption{\emph{Same as Figure~\ref{fig:g1-U-compare-PiBySix-RealParts} except now the imaginary parts of the three functions are compared.}}
\label{fig:g1-U-compare-PiBySix-ImagParts}
\end{figure}
\begin{figure}[H]
\begin{center}
\hspace{-.1in}
\includegraphics[width=2 in]{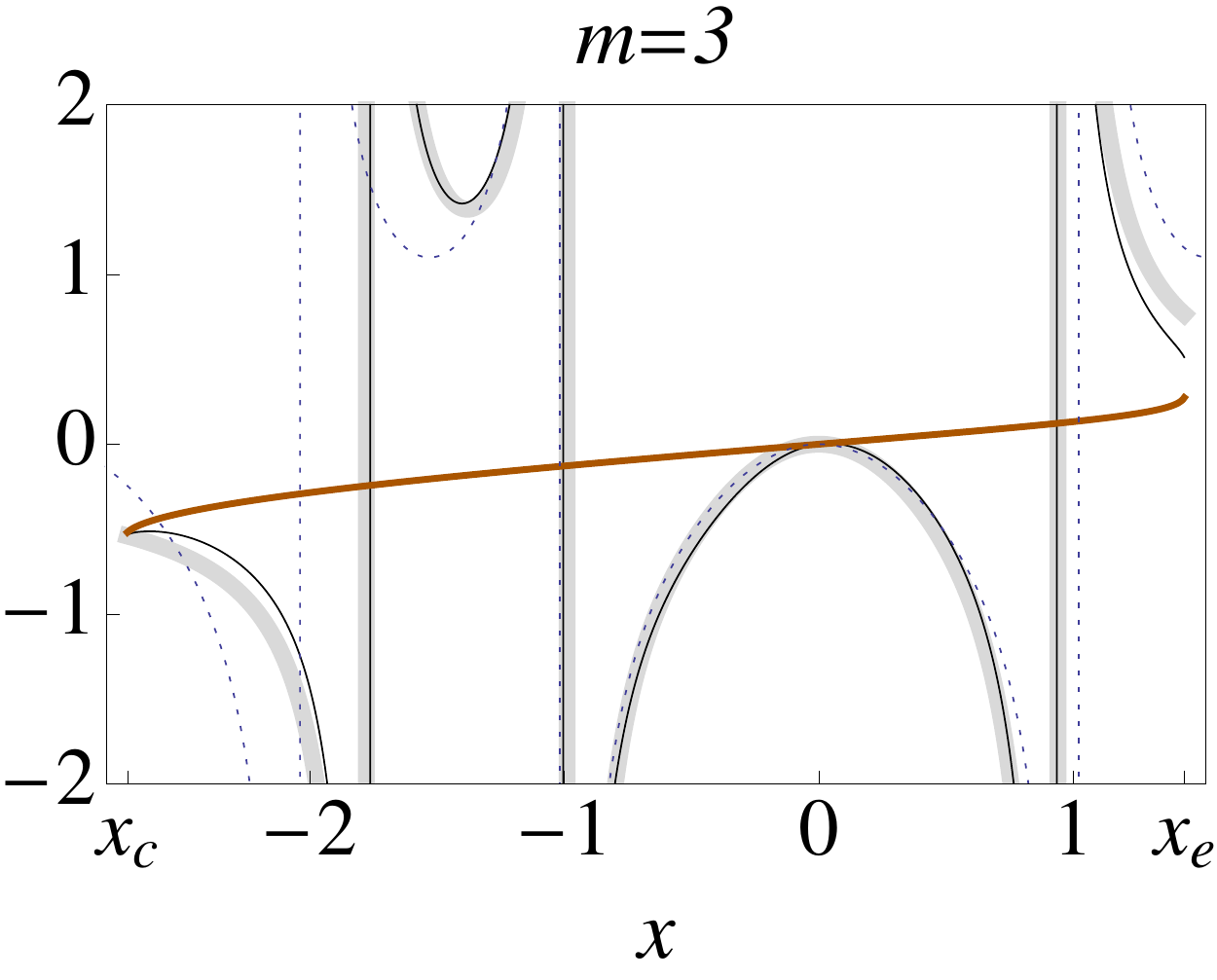}\hspace{0.5 in}%
\includegraphics[width=2 in]{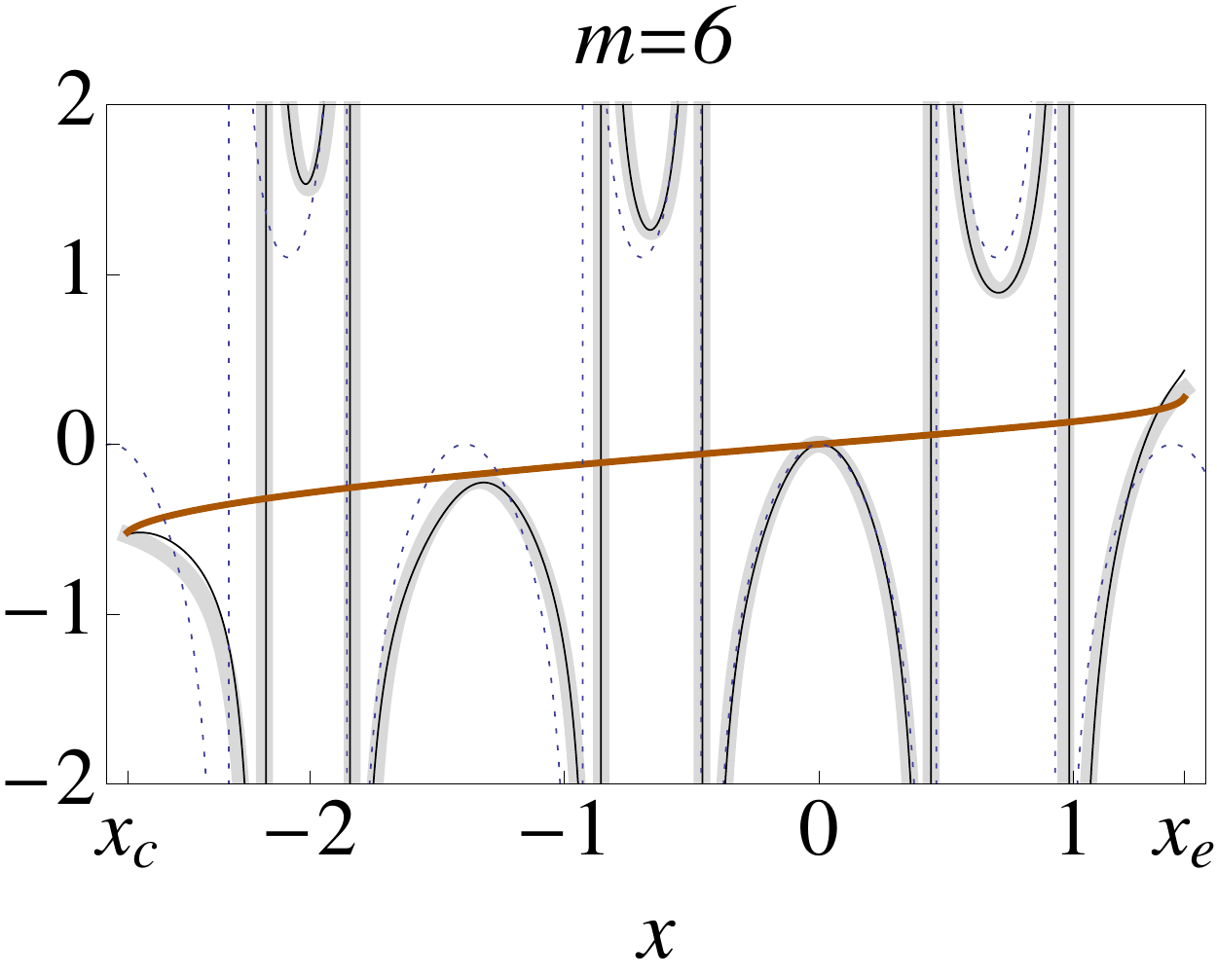}\hspace{0.2in}\\
\includegraphics[width=2 in]{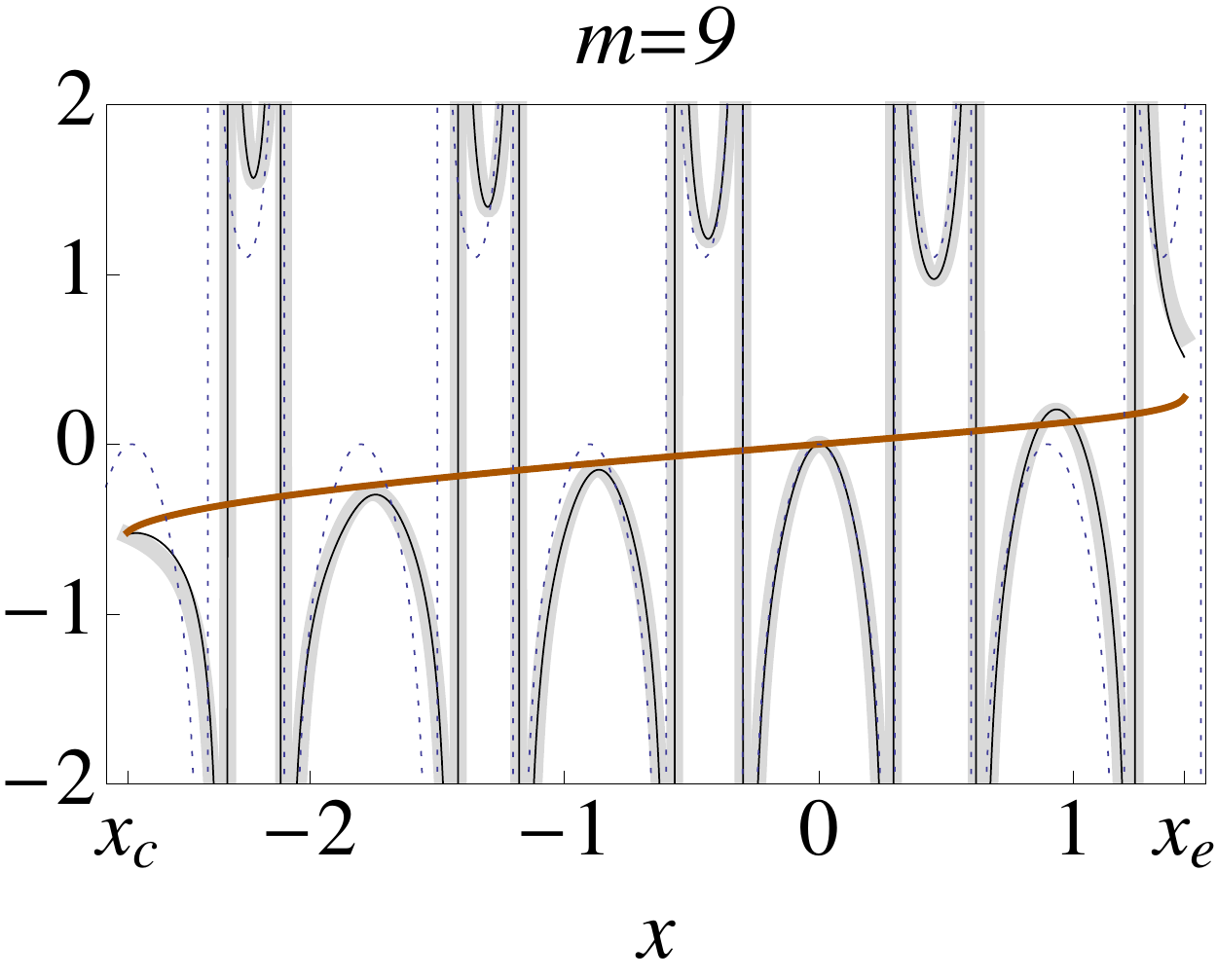}\hspace{0.5 in}
\raisebox{0.3 in}{\includegraphics[width=2 in]{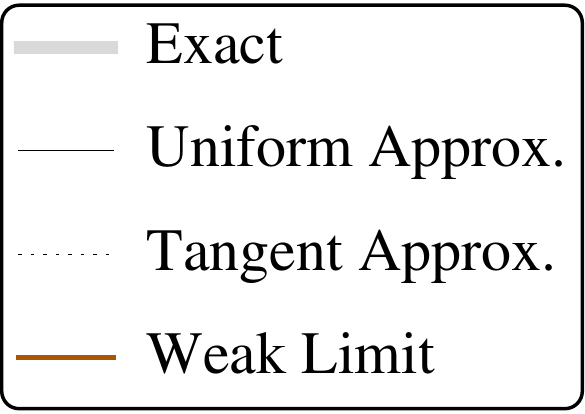}}
\end{center}
\caption{\emph{Comparing $m^{-1/3}\mathcal{P}_m((m-\tfrac{1}{2})^{2/3}x)$ (thick gray curves), its uniform approximation $\dot{\mathcal{P}}_m(0;x)$ (thin black curves), a tangent approximation based at the origin $\dot{\mathcal{P}}_m((m-\tfrac{1}{2})x,0)$ (dotted curves), and the weak limit $\langle\dot{\mathcal{P}}\rangle_\mathbb{R}(x)$ ($m$-independent medium brown curve) as real functions of $x=x_0\in[x_c,x_e]$ for various values of $m$.  The function
$\dot{\pp}_m(w;x_0)$ is defined in \eqref{eq:g1-dotP-formula} and
$\langle\dot{\pp}\rangle_\mathbb{R}(x)$ in \eqref{eq:g1-dotP-real-average}.}}
\label{fig:g1-P-compare-real}
\end{figure}
\begin{figure}[H]
\begin{center}
\hspace{-.1in}
\includegraphics[width=2 in]{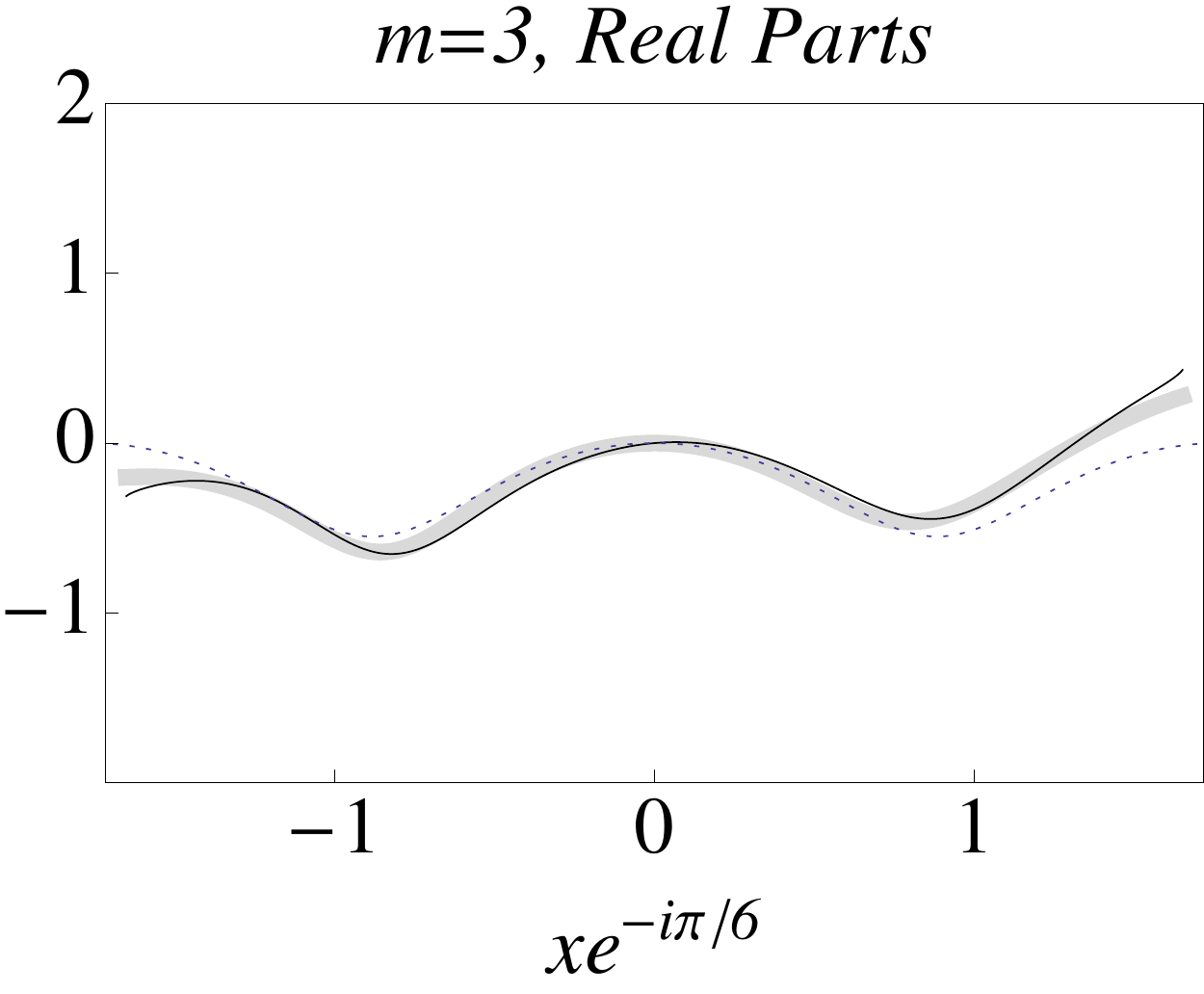}\hspace{0.5 in}%
\includegraphics[width=2 in]{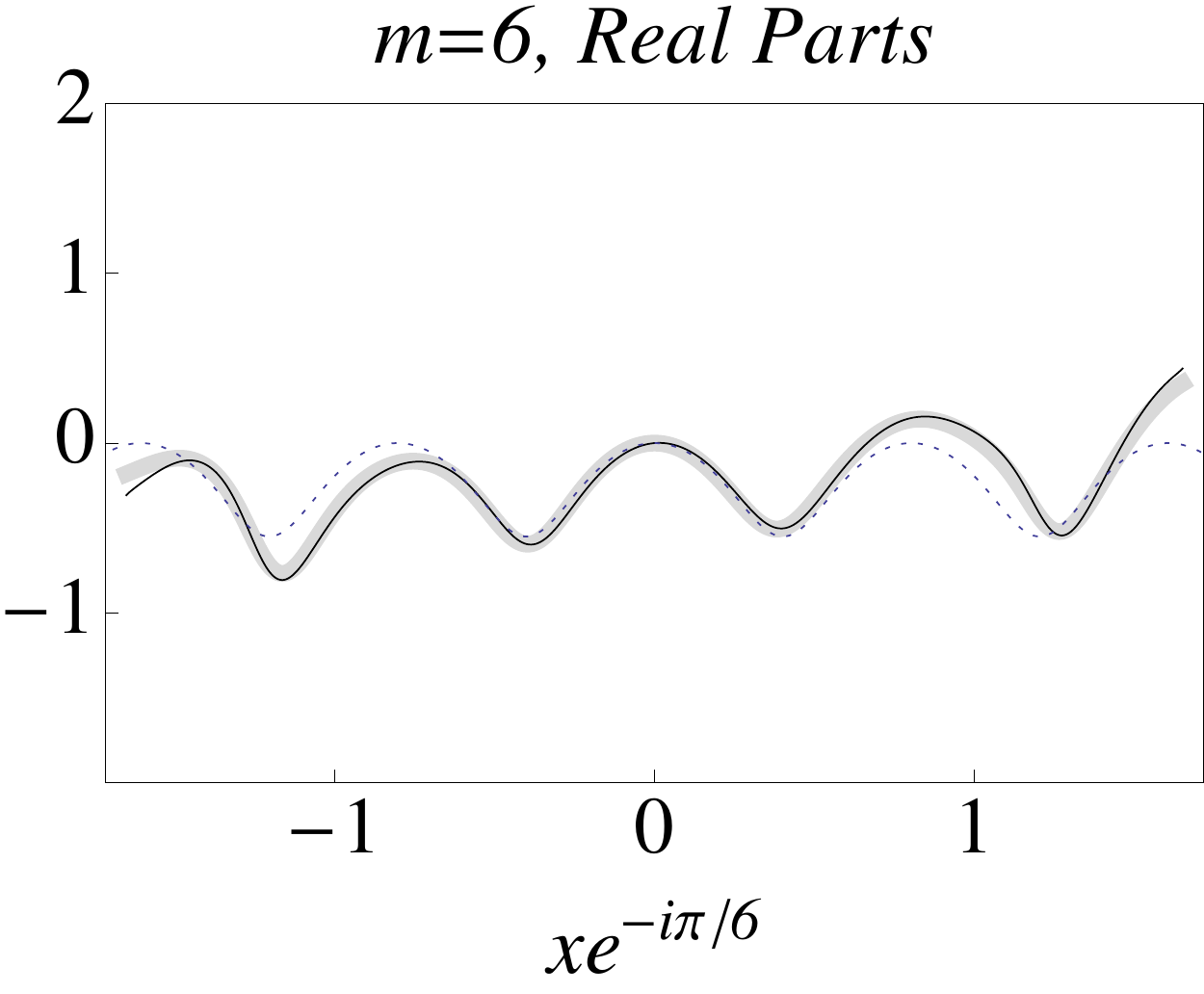}\hspace{0.2in}\\
\includegraphics[width=2 in]{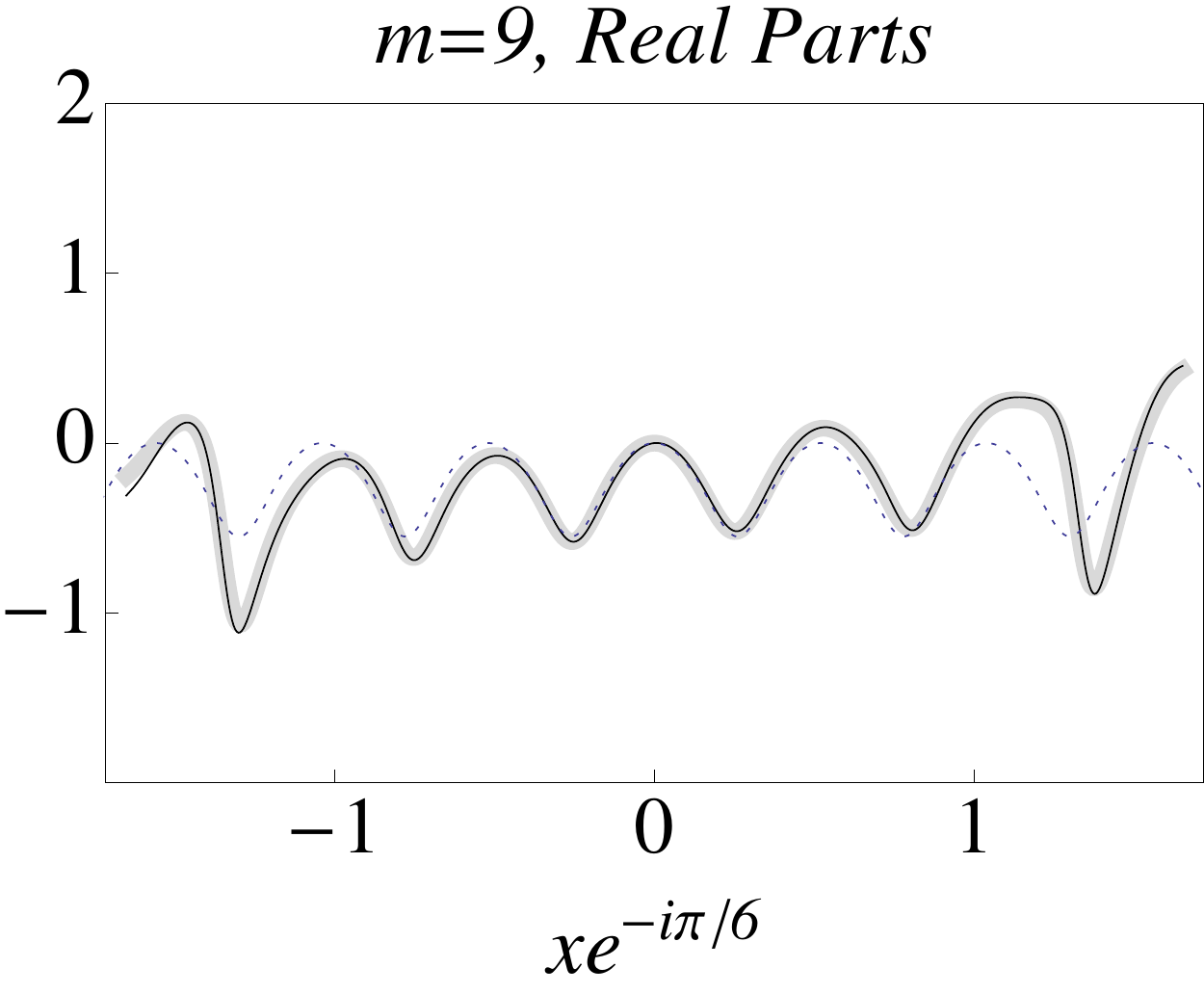}\hspace{0.5 in}
\raisebox{0.5 in}{\includegraphics[width=2 in]{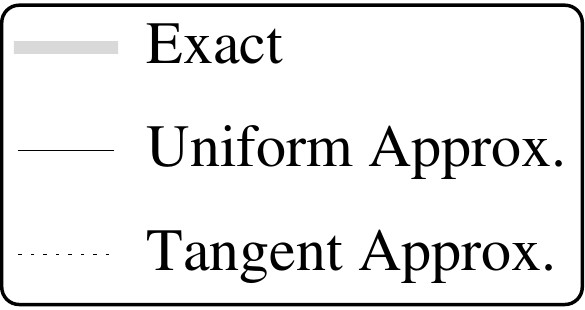}}
\end{center}
\caption{\emph{Comparing the real parts of $m^{-1/3}\mathcal{P}_m((m-\tfrac{1}{2})^{2/3}x)$ (thick gray curves), its uniform approximation $\dot{\mathcal{P}}_m(0;x)$ (thin black curves), and a tangent approximation based at the origin $\dot{\mathcal{P}}_m((m-\tfrac{1}{2})x,0)$ (dotted curves) evaluated for $x\in e^{i\pi/6}\mathbb{R}$ for various values of $m$.  The function $\dot{\pp}_m(w;x_0)$ is defined in \eqref{eq:g1-dotP-formula}.}}
\label{fig:g1-P-compare-PiBySix-RealParts}
\end{figure}
\begin{figure}[H]
\begin{center}
\hspace{-.1in}
\includegraphics[width=2 in]{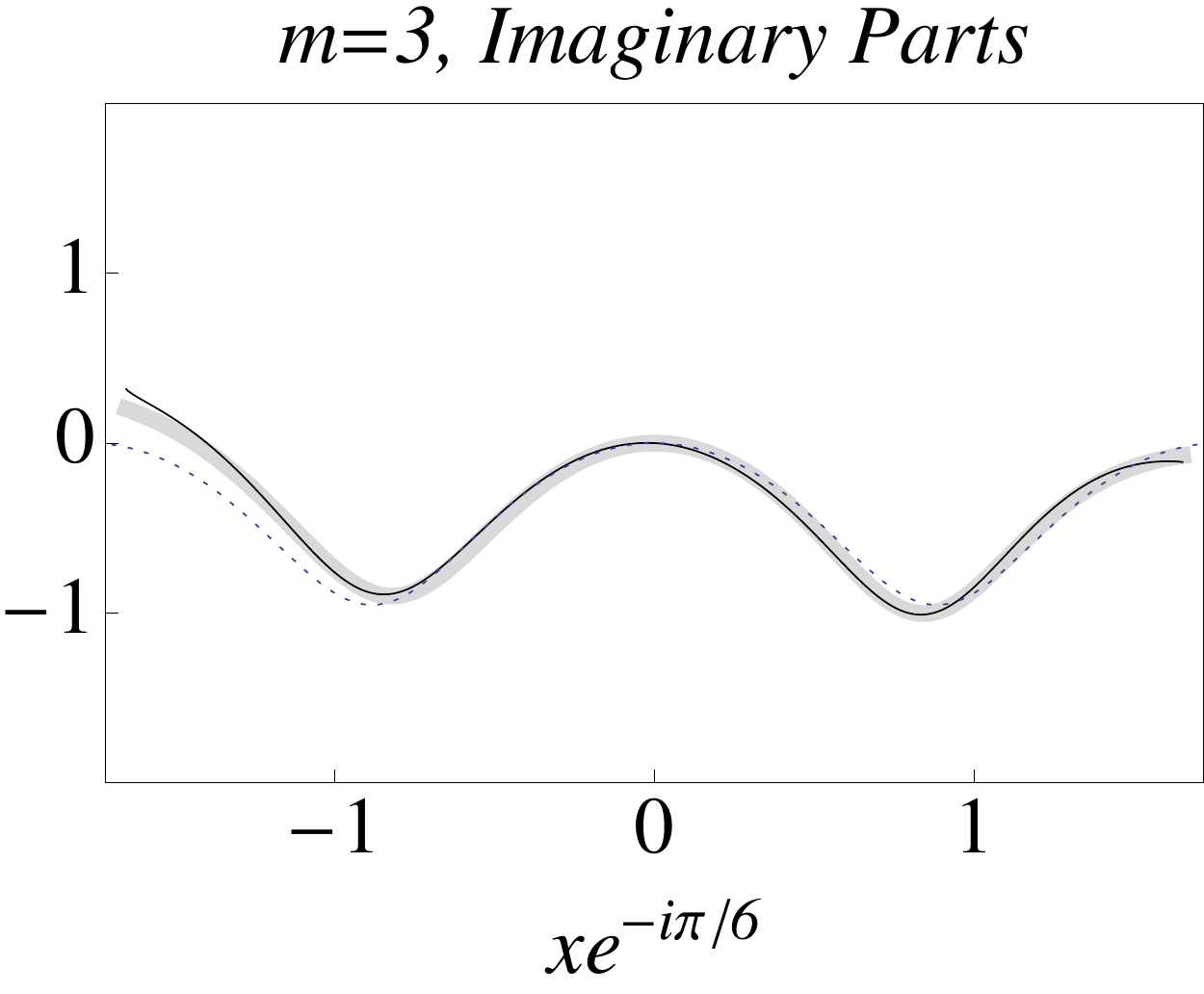}\hspace{0.5 in}%
\includegraphics[width=2 in]{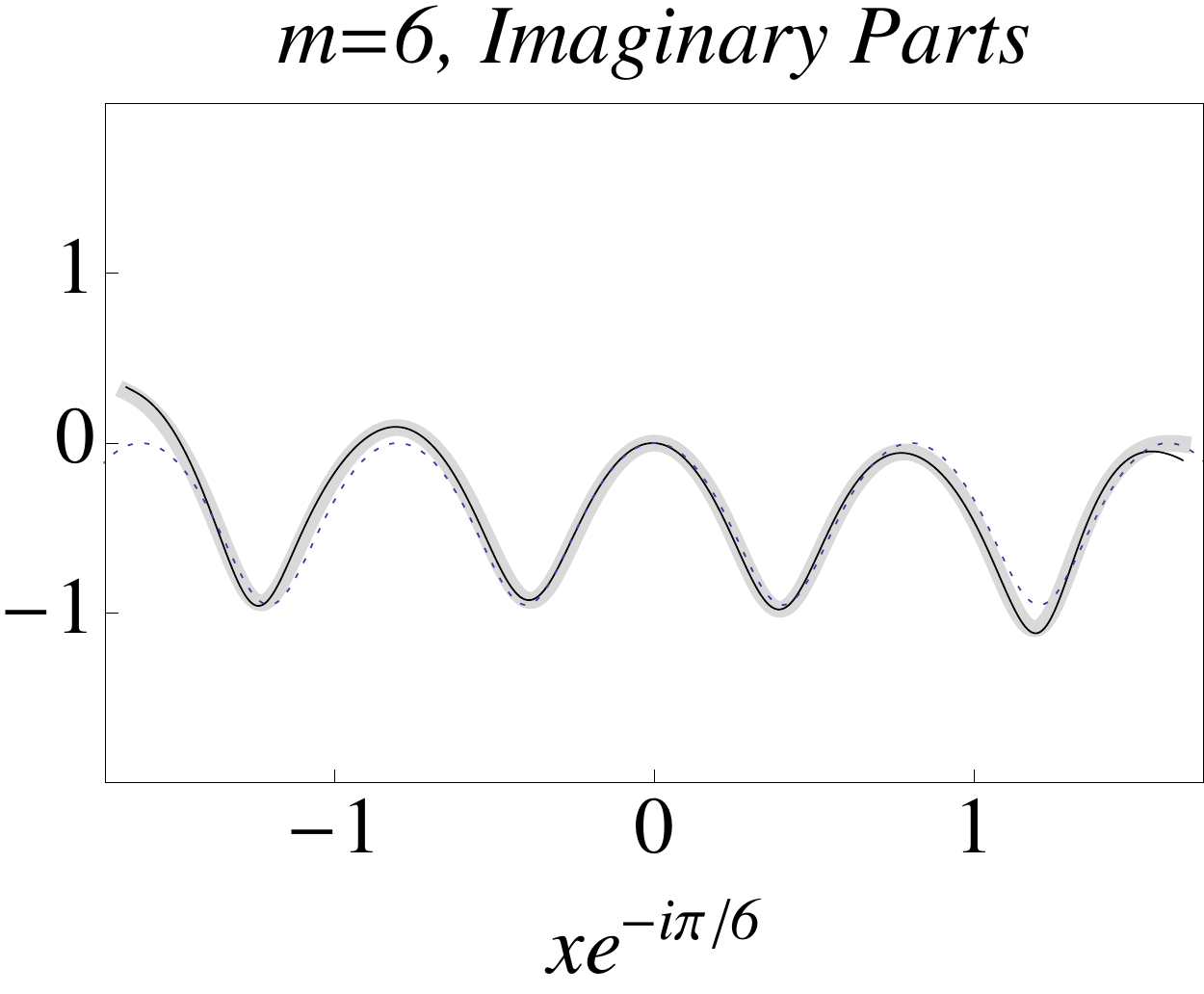}\hspace{0.2in}\\
\includegraphics[width=2 in]{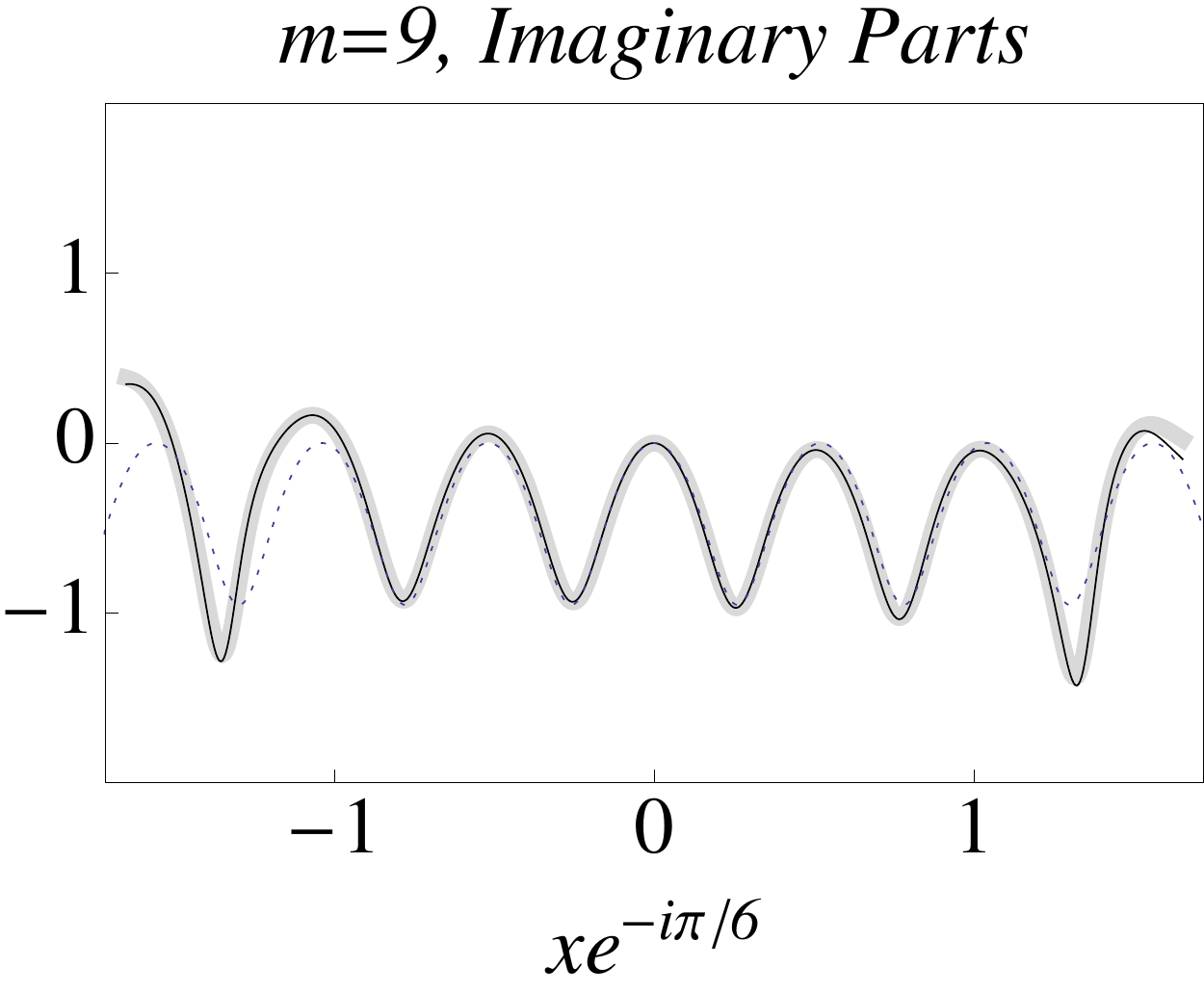}\hspace{0.5 in}
\raisebox{0.5 in}{\includegraphics[width=2 in]{PiBySixPLegend.pdf}}
\end{center}
\caption{\emph{Same as Figure~\ref{fig:g1-P-compare-PiBySix-RealParts} except now the imaginary parts of the three functions are compared.}}
\label{fig:g1-P-compare-PiBySix-ImagParts}
\end{figure}
\begin{figure}[H]
\begin{center}
\includegraphics[width=2.5 in]{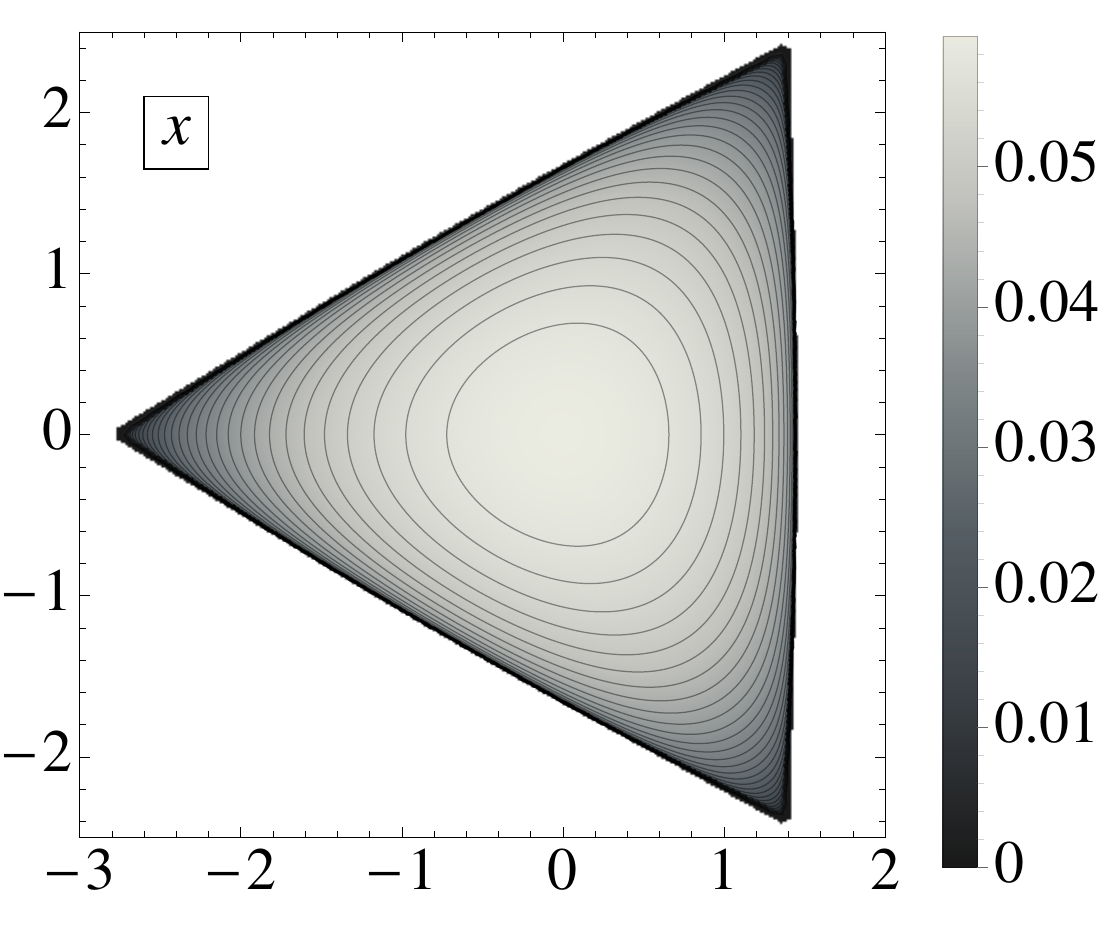}\hspace{0.5 in}%
\includegraphics[width=2.5 in]{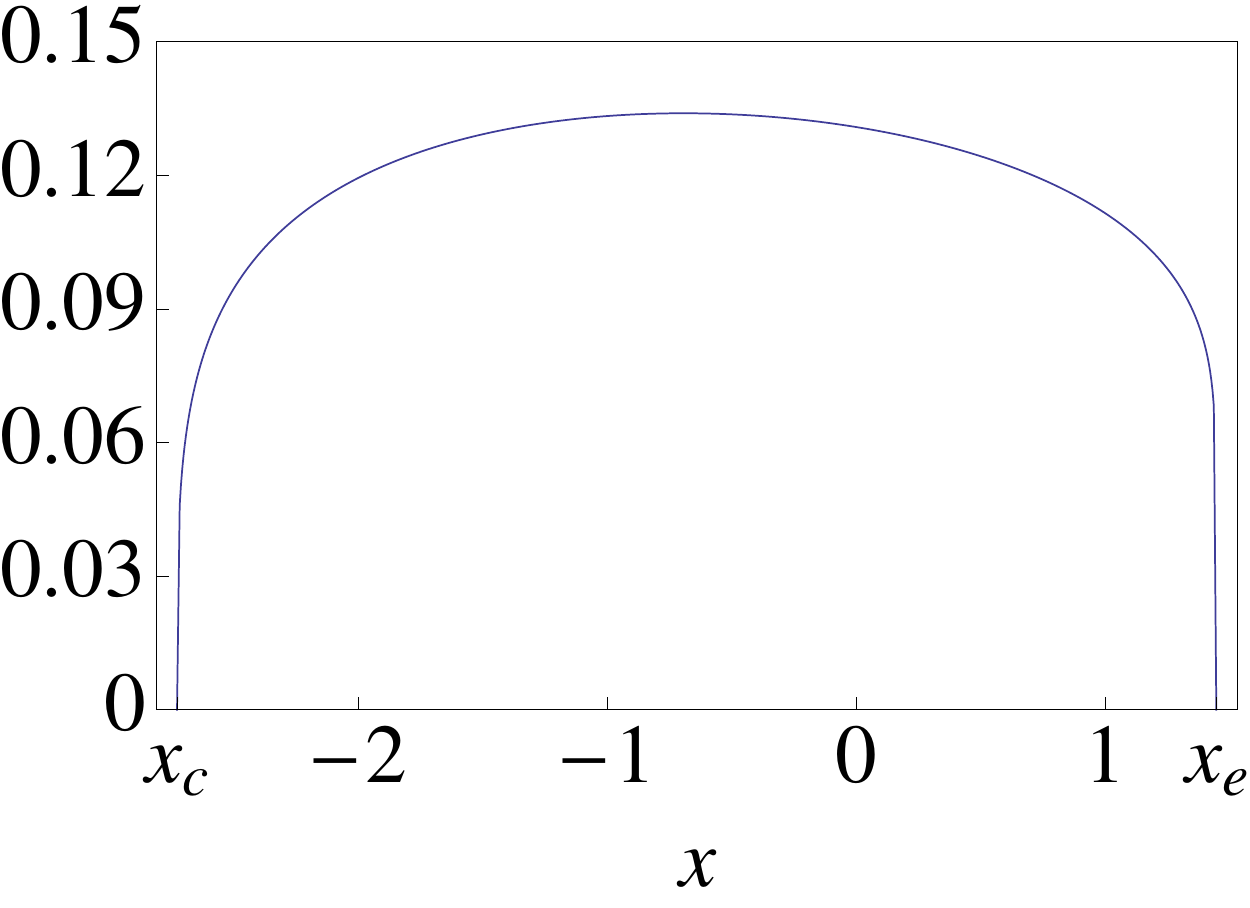}
\end{center}
\caption{\emph{Left:  a contour plot of the planar density $\sigma_\mathrm{P}$ over the region $T$.  Right:  a plot of the linear density $\sigma_\mathrm{L}$ over the interval $T\cap\mathbb{R}=(x_c,x_e)$. Both densities are nearly constant over most of their domains of definition but drop down rapidly to zero at the boundary $\partial T$.  The densities are defined in \eqref{g1-planar-density}--\eqref{g1-linear-density}.}}
\label{fig:g1-densities}
\end{figure}

Finally, we would like to draw attention to a few figures that appear later 
that may be of interest.  Plots of the exponential growth factor that 
dominates the asymptotic behavior of $\pu_m$ for large $m$ (i.e. $\Lambda(x)$ 
for $x\in T$ and $\lambda(x)$ for $x\in\mathbb{C}\setminus\overline{T}$)
are shown in Figure~\ref{fig:Lambdalambda}.  
In Figure~\ref{fig:ParallelogramPictures} we display tilings of $T$ into curvilinear parallelograms by level curves of two computable functions that appear as part of the proof of Theorem~\ref{theorem:g1-weak-limit}; this figure is of independent interest because it shows that the vertices of these curvilinear parallelograms evidently coincide nearly exactly with the pole locations of $\pu_m$ within the elliptic region $T$.  Figure~\ref{fig:g1-weak-limit} shows the real and imaginary parts of the ``macroscopic limit'' function $\dot{\pp}_\mathrm{macro}:\mathbb{C}\to\mathbb{C}$.

\subsection{Acknowledgements}
We benefited greatly from useful discussions with many people, including Peter Clarkson, Alexander Its, Andrei Kapaev, Erik Koelink, Andrei Martinez-Finkelshtein, Davide Masoero, and Boris Shapiro.

\section{The Riemann-Hilbert problem}
\label{Riemann-Hilbert-section}
The starting point of our analysis will be the following Riemann-Hilbert 
problem, from which it is 
possible to extract the functions $\pu_m$, $\pv_m$, $\pp_m$, and $\pq_m$ 
\cite{BuckinghamMcritical}.
We define for use here and later the Pauli spin matrices
\eq
\sigma_1:=\bbm 0 & 1 \\ 1 & 0 \ebm, \quad \sigma_2:=\bbm 0 & -i \\ i & 0 \ebm, \quad \sigma_3:=\bbm 1 & 0 \\ 0 & -1 \ebm.
\endeq
\begin{rhp}[]
Fix a real number $y$ and an integer $m$.  
Seek a $2\times 2$ matrix $\mathbf{Z}_m(\zeta;y)$ 
with the following properties:
\begin{itemize}
\item[]\textbf{Analyticity:}  $\mathbf{Z}_m(\zeta;y)$
is analytic in $\zeta$ except along the rays $\arg(\zeta)=k\pi/3$, 
$k=0,\dots,5$, may be continued from each sector of analyticity to
a slightly larger sector, and in each sector is H\"older continuous 
up to the boundary in a neighborhood of $\zeta=0$.
\item[]\textbf{Jump condition:}  The boundary values\footnote{We use a 
subscript $+$ ($-$) to indicate the boundary value
taken on a specified oriented arc from the left (right).} taken by 
$\mathbf{Z}_m(\zeta;y)$ on the six rays of discontinuity are related by the 
jump condition 
$\mathbf{Z}_{m+}(\zeta;y)=
\mathbf{Z}_{m-}(\zeta;y)
\mathbf{V}_{\mathbf{Z}}(\zeta;y)$, 
where the jump matrix $\mathbf{V}_{\mathbf{Z}}(\zeta;y)$ 
is as shown in 
Figure~\ref{fig:VlocR} 
\begin{figure}[h]
\setlength{\unitlength}{2pt}
\begin{center}
\begin{picture}(100,100)(-50,-50)
\thicklines
\put(-80,41){\framebox{$\zeta$}}
\put(-50,0){\line(1,0){100}}
\put(0,0){\vector(-1,0){25}}
\put(0,0){\vector(1,0){25}}
\put(0,0){\line(2,-3){30}}
\put(0,0){\vector(2,-3){15}}
\put(0,0){\line(-2,3){30}}
\put(0,0){\vector(-2,3){15}}
\put(0,0){\line(2,3){30}}
\put(0,0){\vector(2,3){15}}
\put(0,0){\line(-2,-3){30}}
\put(0,0){\vector(-2,-3){15}}
\put(0,0){\circle*{2}}
\put(-1.25,-6){$0$}
\put(-81,0){$\bbm 1 & 0 \\ ie^{\zeta^3+y\zeta} & 1 \ebm$}
\put(-63,39){$\bbm 1 & ie^{-\zeta^3-y\zeta} \\ 0 & 1 \ebm$}
\put(32,39){$\bbm 1 & 0 \\ ie^{\zeta^3+y\zeta} & 1 \ebm$}
\put(53,0){$\bbm -1 & -ie^{-\zeta^3-y\zeta} \\ 0 & -1 \ebm$}
\put(32,-41){$\bbm 1 & 0 \\ ie^{\zeta^3+y\zeta} & 1 \ebm$}
\put(-63,-41){$\bbm 1 & ie^{-\zeta^3-y\zeta} \\ 0 & 1 \ebm$}
\end{picture}
\end{center}
\caption{\emph{The jump matrix $\mathbf{V}^{(\mathbf{Z})}(\zeta;y)$.}}
\label{fig:VlocR}
\end{figure}
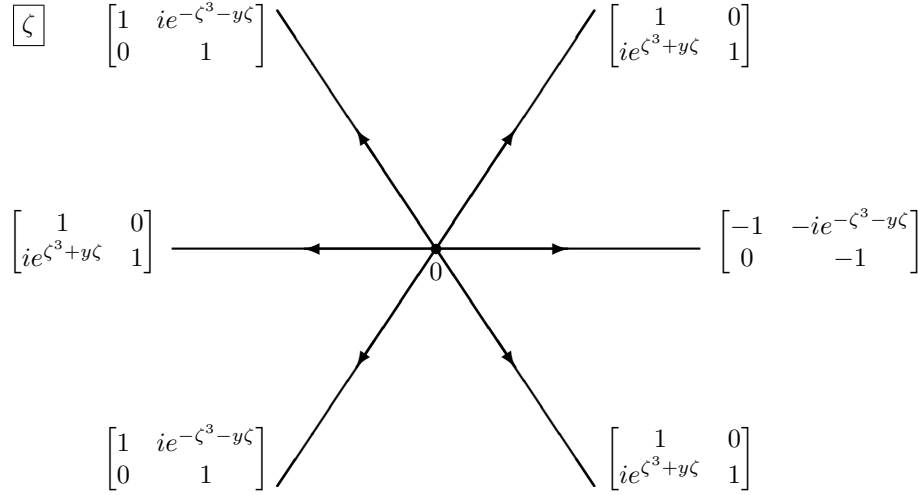
and all rays are oriented toward infinity.
\item[]\textbf{Normalization:}  The matrix 
$\mathbf{Z}_m(\zeta;y)$
satisfies the condition
\begin{equation}
\lim_{\zeta\to\infty}\mathbf{Z}_m(\zeta;y)
(-\zeta)^{(1-2m)\sigma_3/2}=\mathbb{I},
\label{eq:Zindnorm}
\end{equation}
with the limit being uniform with respect to direction in each of the
six sectors of analyticity.
\end{itemize}
\label{rhp:DSlocalII}
\end{rhp}
As shown in \cite[Section 5]{BuckinghamMcritical}, if we write 
\begin{equation}
\mathbf{Z}_m(\zeta;y)(-\zeta)^{(1-2m)\sigma_3/2}=
\mathbb{I}+\mathbf{A}_m(y)\zeta^{-1} + \mathbf{B}_m(y)\zeta^{-2} +\mathcal{O}(\zeta^{-3}),
\quad\zeta\to\infty,
\label{eq:Fexpansion}
\end{equation}
then
\eq
\label{eq:u}
\pu_m(y) = A_{m,12}(y),
\endeq
\eq
\label{eq:v}
\pv_m(y) = A_{m,21}(y),
\endeq
\eq
\label{eq:uprime}
\pu_m'(y) = -B_{m,12}(y) + A_{m,12}(y)A_{m,22}(y),
\endeq
and
\eq
\label{eq:vprime}
\pv_m'(y) = B_{m,21}(y) - A_{m,21}(y)A_{m,11}(y),
\endeq
where $A_{m,ij}(y)$ (respectively, $B_{m,ij}(y)$) denotes the entry in the 
$i^\text{th}$ row and $j^\text{th}$ column of $\mathbf{A}_m(y)$ 
(respectively, $\mathbf{B}_m(y)$).

Let $m\in\mathbb{Z}$ be fixed.  Riemann-Hilbert Problem~\ref{rhp:DSlocalII} has two elementary discrete
symmetries.  Firstly, it is easy to check that if we define the matrix $\widetilde{\mathbf{Z}}_m(\zeta;y):=\mathbf{Z}_m(\zeta^*;y^*)^*$, where the outer asterisk denotes elementwise complex conjugation (no transpose), then $\widetilde{\mathbf{Z}}_m(\zeta;y)$ satisfies Riemann-Hilbert Problem~\ref{rhp:DSlocalII} whenever $\mathbf{Z}_m(\zeta;y)$ does.  It follows by a uniqueness argument
based on Liouville's Theorem that in fact $\widetilde{\mathbf{Z}}_m(\zeta;y)=\mathbf{Z}_m(\zeta;y)$.
This implies in turn that the matrix coefficient $\mathbf{A}_m(y)$
satisfies $\mathbf{A}_m(y^*)=\mathbf{A}_m(y)^*$,
and therefore we have
\begin{equation}
\pu_m(y^*)=\pu_m(y)^*\quad\text{and}\quad\pv_m(y^*)=\pv_m(y)^*,
\end{equation}
and hence, by \eqref{log-derivative},
\begin{equation}
\pp_m(y^*)=\pp_m(y)^* \quad\text{and}\quad \pq_m(y^*)=\pq_m(y)^*.
\end{equation}
Next, it is also easy to see that the matrix $\widehat{\mathbf{Z}}_m(\zeta;y)$ defined by
\begin{equation}
\widehat{\mathbf{Z}}_m(\zeta;y):=\begin{cases}
e^{(2m-1)i\pi\sigma_3/3}\mathbf{Z}_m(e^{-2\pi i/3}\zeta;e^{2\pi i/3}y),&\quad \arg(\zeta)\in (2\pi/3,\pi)\cup (-\pi,0)\\
-e^{(2m-1)i\pi\sigma_3/2}\mathbf{Z}_m(e^{-2\pi i/3}\zeta;e^{2\pi i/3}y),&\quad\arg(\zeta)\in (0,2\pi/3)
\end{cases}
\end{equation}
again satisfies Riemann-Hilbert Problem~\ref{rhp:DSlocalII} whenever $\mathbf{Z}_m(\zeta;y)$ does.  By the same uniqueness argument we then have $\widehat{\mathbf{Z}}_m(\zeta;y)=\mathbf{Z}_m(\zeta;y)$.  It follows that $\mathbf{A}_m(e^{-2\pi i/3}y)=e^{2\pi i/3}e^{(2m-1)i\pi\sigma_3/3}\mathbf{A}_m(y)e^{-(2m-1)i\pi\sigma_3/3}$, which implies that
\begin{equation}
\label{pu-pv-symmetries}
\pu_m(e^{-2\pi i/3}y)=e^{-2m\pi i/3}\pu_m(y)\quad \text{and}\quad \pv_m(e^{-2\pi i/3}y)=e^{2(m-1)\pi i/3}\pv_m(y),
\end{equation}
and hence, by \eqref{log-derivative},
\begin{equation}
\label{pp-symmetry}
\pp_m(e^{-2\pi i/3}y)=e^{2\pi i/3}\pp_m(y) \quad\text{and}\quad \pq_m(e^{-2\pi i/3})=e^{2\pi i/3}\pq_m(y).
\end{equation}

Suppose that $m\geq 1$.  Then, rescaling the spectral parameter $\zeta$ by
\begin{equation}
\label{z-def}
z:=(m-\tfrac{1}{2})^{-1/3}\zeta,
\end{equation}
and also rescaling the Painlev\'e independent variable $y$ by
\begin{equation}
\label{x}
x:=(m-\tfrac{1}{2})^{-2/3}y,
\end{equation}
consider the transformation (which does not change any jump conditions)
\begin{equation}
\label{M-def}
\mathbf{M}(z;x,\epsilon):=\epsilon^{\sigma_3/(3\epsilon)}\mathbf{Z}_m(\epsilon^{-1/3}z;\epsilon^{-2/3}x).
\end{equation}
Here we have introduced 
\begin{equation}
\label{epsilon}
\epsilon:=(m-\tfrac{1}{2})^{-1}.
\end{equation}
Recalling the normalization condition \eqref{eq:Zindnorm} at infinity for 
$\mathbf{Z}_m(\zeta;y)$, we see that for each fixed $m\geq 1$,
\begin{equation}
\label{M-normalization}
\begin{split}
\lim_{z\to\infty}\mathbf{M}(z;x,\epsilon)(-z)^{-\sigma_3/\epsilon} &= 
\lim_{\zeta\to\infty}(m-\tfrac{1}{2})^{(1/2-m)\sigma_3/3}\mathbf{Z}_m(\zeta;y)(-(m-\tfrac{1}{2})^{-1/3}\zeta)^{(1/2-m)\sigma_3}\\
&=(m-\tfrac{1}{2})^{(1/2-m)\sigma_3/3}\left[\lim_{\zeta\to\infty}\mathbf{Z}_m(\zeta;y)(-\zeta)^{(1/2-m)\sigma_3}\right](m-\tfrac{1}{2})^{(m-1/2)\sigma_3/3}\\
&=\mathbb{I},
\end{split}
\end{equation}
so the normalization condition for $\mathbf{M}(z;x,\epsilon)$ is particularly simple in form.  
The matrix $\mathbf{M}(z;x,\epsilon)$ also is analytic in the six sectors $0<|\arg(z)|<\pi/3$,
$\pi/3<|\arg(z)|<2\pi/3$, and $0<|\arg(-z)|<\pi/3$.  The jump conditions are exactly of the
same form as those satisfied on the same rays by $\mathbf{Z}_m(\zeta;y)$ but in each case the
exponent $\zeta^3+y\zeta$ is replaced by $\epsilon^{-1}\theta(z;x)$, where
\begin{equation}
\theta(z;x):=z^3+xz.
\end{equation}

\section{Analysis for $x$ outside the elliptic region}
\label{section-gen0}
In this section we compute the large-$m$ (or small-$\epsilon$) asymptotic 
expansions of $\pu_m(\epsilon^{-2/3}x)$ and $\pp_m(\epsilon^{-2/3}x)$ for $x$ 
outside the elliptic region.  Our method involves the judicious use of deformations of the contours of the Riemann-Hilbert problem satisfied by $\mathbf{M}(z;x,\epsilon)$, and it will turn out that there is not a 
single consistent way to deform the contours that will be fruitful for all $x$ outside the 
elliptic region.  Rather, reflecting the three-fold symmetry of the problem, 
outside the elliptic region there will be one deformation strategy that works for 
$-\pi/3<\arg(x)<\pi$, one that works for $-\pi<\arg(x)<\pi/3$, 
and one that works for 
$|\arg(-x)|<2\pi/3$.  Due to the symmetries 
\eqref{pu-pv-symmetries} and \eqref{pp-symmetry}, analyzing the problem in 
any of these three sectors is sufficient.  Note that, except on the rays 
$\arg(x)=\pm\pi/3$ and 
$\arg(-x)=0$, two different deformations will work 
for a given $x$.  Along these rays the one deformation that does work is 
symmetric and natural.  For this reason, and because we are especially 
interested in real values of $x$ for applications, we will primarily work in 
the region 
$|\arg(-x)|<2\pi/3$ and illustrate most details of our methodology for 
$\arg(-x)=0$.  This has the added benefit that the 
Riemann-Hilbert analysis is well-adapted to analyzing $x$ near a corner of the 
elliptic region, which will be done in a subsequent work 
\cite{Buckingham-rational-crit}.  On the other hand, to study $x$ near an 
edge (but not a corner) of the elliptic region, it is most natural to carry 
out the analysis for $x$ on and near the positive real axis.  With this in 
mind we will briefly present in \S\ref{subsection-outside-positive-x} the setup for the analysis for these $x$ using 
the deformation valid for $-\pi/3<\arg(x)<\pi$.

\subsection{The genus-zero $g$-function}
\label{genus-zero-g}
Analysis of the Riemann-Hilbert problem will involve the use of a 
standard tool called the $g$-function, a scalar function used to reduce 
the jump matrices asymptotically to constant matrices in the 
small-$\epsilon$ (or large-$m$) limit.  Let $a$ and $b$ be distinct points in the complex $z$-plane,
and let $\Sigma$ (the \emph{band}) denote the oriented straight line segment $\overrightarrow{ab}$.
Given $\Sigma$, let $r(z)$ be the function analytic for $z\in\mathbb{C}\setminus\Sigma$ satisfying the conditions
\begin{equation}
r(z)^2=(z-a)(z-b)\quad\text{and}\quad r(z)=z+\mathcal{O}(z),\quad z\to\infty.
\label{eq:g0-rdefine}
\end{equation}
If we introduce the quantities
\begin{equation}
S:=a+b\quad\text{and}\quad\Delta:=b-a
\end{equation}
then $r(z)^2$ is the quadratic
\begin{equation}
r(z)^2=z^2-Sz+\frac{1}{4}(S^2-\Delta^2).
\end{equation}
The boundary values $r_\pm(z)$ taken from the left and right for $z\in\Sigma$ satisfy $r_+(z)+r_-(z)=0$.  Now for each $x\in\mathbb{C}$ we define a related function by setting
\begin{equation}
g'(z):=\frac{1}{2}\theta'(z;x)-\frac{3}{2}\left(z+\frac{1}{2}S\right)r(z)=\frac{3}{2}z^2+\frac{1}{2}x-\frac{3}{2}\left(z+\frac{1}{2}S\right)r(z),\quad z\in\mathbb{C}\setminus\Sigma.
\label{eq:g0-gprime}
\end{equation}
Suppose now that $a$, $b$, and $x$ are related by the two \emph{moment conditions}
\begin{equation}
\label{ab-conditions}
6S^2 + 3\Delta^2 = -8x \quad \text{and} \quad 3S\Delta^2 = 16.
\end{equation}
These conditions imply the following large-$z$ asymptotic behavior of $g'(z)$:
\begin{equation}
g'(z)=\frac{1}{z}+\mathcal{O}\left(\frac{1}{z^2}\right),\quad z\to\infty.
\end{equation}

Eliminating $\Delta^2$ from \eqref{ab-conditions} yields a cubic equation for $S$:
\eq
\label{cubic-equation}
3S^3 + 4xS + 8 = 0.
\endeq
Let $\Sigma_S$ be the contour in the complex $x$-plane illustrated in Figure~\ref{fig:Sigma-S}.
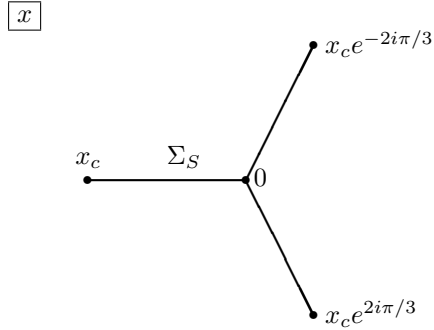
\begin{figure}[h]
\setlength{\unitlength}{1.5pt}
\begin{center}
\begin{picture}(80,80)(-40,-40)
\put(-20,4){$\Sigma_S$}
\thicklines
\put(-60,40){\framebox{$x$}}
\put(0,0){\line(-1,0){40}}
\put(0,0){\line(1,2){17}}
\put(0,0){\line(1,-2){17}}
\put(17,34){\circle*{2}}
\put(20,32){$x_c e^{-2i\pi/3}$}
\put(17,-34){\circle*{2}}
\put(20,-36){$x_c e^{2i\pi/3}$}
\put(0,0){\circle*{2}}
\put(2,-1.5){$0$}
\put(-40,0){\circle*{2}}
\put(-43,4){$x_c$}
\end{picture}
\end{center}
\caption{\emph{The contour $\Sigma_S$ is the union of three straight-line segments with nonzero endpoints coinciding with the three branch points of the cubic \eqref{cubic-equation}.  Here $x_c:=-(9/2)^{2/3}<0$.}}
\label{fig:Sigma-S}
\end{figure}
We claim that there exists a unique solution $S=S(x)$ of the cubic equation \eqref{cubic-equation}
that is defined and analytic for $x\in\mathbb{C}\setminus\Sigma_S$ and that has the asymptotic behavior
\eq
\label{S-large-x}
S(x) = -\frac{2}{x} + \mathcal{O}\left(\frac{1}{x^4}\right), \quad x\to\infty.
\endeq
Indeed, the three endpoints of $\Sigma_S$ are easily seen to be the only branch points of the cubic \eqref{cubic-equation} in the finite complex $x$-plane, and although the general solution is branched at infinity, the asymptotic condition \eqref{S-large-x} makes $S$ single-valued for large $x$.  Thus the unique existence of $S:\mathbb{C}\setminus\Sigma_S\to\mathbb{C}$ is a consequence of the Implicit Function Theorem.  Given the well-defined function $S(x)$, we then
define $\Delta(x)$ from the second equation in \eqref{ab-conditions} by taking an appropriate square root; from the asymptotic behavior \eqref{S-large-x} it is clear that $\Delta(x)$ will be branched at $x=\infty$.  We therefore define three  semi-infinite rays (each oriented toward $x=\infty$) by
\eq
\mathscr{R}_{\pi/3}:=\left(x_c e^{-2i\pi/3}, \infty e^{\pi/3}\right), \quad \mathscr{R}_{-}:=\left(x_c,-\infty \right), \quad \mathscr{R}_{-\pi/3}:=\left(x_c e^{2i\pi/3}, \infty e^{-i\pi/3}\right)
\endeq
(see Figure~\ref{fig:c-of-x-equals-zero}), and then define $\Delta(x)$ as the analytic function for $x\in\mathbb{C}\setminus(\Sigma_S\cup\mathscr{R}_{-\pi/3})$ satisfying $\Delta(x)^2=16/(3S(x))$ that is positive real for $x\in\mathscr{R}_-$.  We then have $a$ and $b$ defined in the same domain as analytic functions of $x$, and we note that these points are exchanged across the branch cut $\mathscr{R}_{-\pi/3}$ where $\Delta$ changes sign.  This definition implies that 
$\text{Re}(a(x))<\text{Re}(b(x))$ for $0<\arg(x)<5\pi/3$ and 
$\text{Im}(a(x))>\text{Im}(b(x))$ for $-\pi/3<\arg(x)<\pi$.  

Since $g'(z)$ has a residue at infinity, to define the $g$-function as an antiderivative of $g'(z)$ we must introduce a logarithmic branch cut.  Let $L$ denote an unbounded arc joining $z=b$ to 
$z=\infty$ without otherwise touching $\Sigma$, and suppose that $L$ agrees with the positive real $z$-axis for sufficiently large $|z|$.  Let $\log^{(L)}(b-z)$ denote the branch of $\log(b-z)$ with branch cut $L$ that agrees asymptotically with the principal branch of $\log(-z)$ for large negative $z$.  
Assuming that $a$ and $b$ are determined as functions of $x$ as above, we then define the $g$-function as the antiderivative of $g'(z)$ given by
\begin{equation}
g(z)=g(z;x):=\log^{(L)}(b-z)+\int_\infty^z\left[g'(\zeta)-\frac{1}{\zeta-b}\right]\,d\zeta,\quad z\in\mathbb{C}\setminus(\Sigma\cup L),\quad x\in\mathbb{C}\setminus(\Sigma_S\cup\mathscr{R}_{-\pi/3}).
\label{eq:g0-g-integral}
\end{equation}
(Note that the integral is independent of path as long as the path avoids $\Sigma$.)  We define a function related to $g$ by setting
\begin{equation}
h(z)=h(z;x):=\frac{1}{2}\theta(z;x)-g(z),\quad z\in\mathbb{C}\setminus(\Sigma\cup L),\quad x\in\mathbb{C}\setminus(\Sigma_S\cup\mathscr{R}_{-\pi/3}).
\label{eq:hgdef}
\end{equation}
The derivative (in $z$) of $h$ has no jump across $L$ and is given simply by
\begin{equation}
h'(z)=\frac{3}{2}\left(z+\frac{1}{2}S\right)r(z),\quad z\in\mathbb{C}\setminus\Sigma.
\label{eq:hprime}
\end{equation}
The basic properties of $g$ and $h$ are the following.
\begin{proposition}
Suppose that $a$ and $b$ are defined in terms of $x\not\in\Sigma_S\cup\mathscr{R}_{-\pi/3}$ as described above, guaranteeing that $g(z)$ and $h(z)$ are well-defined by \eqref{eq:g0-g-integral}
and \eqref{eq:hgdef} respectively.  Then the relation
\begin{equation}
\frac{d}{dz}\left[h_+(z)+h_-(z)\right]=\frac{d}{dz}\left[\theta(z;x)-g_+(z)-g_-(z)\right]=0
\end{equation}
holds as an identity for $z\in\Sigma$.  Also,
\begin{equation}
g_+(z)-g_-(z)=-2\pi i,\quad z\in L,
\end{equation}
and finally,
\begin{equation}
g(z)=\log(-z)+\mathcal{O}\left(\frac{1}{z}\right),\quad z\to\infty,
\end{equation}
assuming that $L$ agrees with the positive real axis for sufficiently large $|z|$.
\label{prop:g0-g-function}
\end{proposition}

According to Proposition~\ref{prop:g0-g-function}, there is a complex number $\lambda=\lambda(x)$ that is the constant value of $-(h_+(z)+h_-(z))$ for $z\in\Sigma$.  The function $g(z)$ and the corresponding constant $\lambda$ can be expressed in terms of elementary functions as
follows.  Let $\mathscr{L}(z)$ be given by
\begin{equation}
\mathscr{L}(z)=\mathscr{L}(z;x):=\log\left(\frac{S(x)-2z-2r(z;x)}{4}\right), \quad z\notin\Sigma\cup L, \quad x\notin \Sigma_S \cup \mathscr{R}_{-\pi/3}
\end{equation}
with the interpretation that $\mathscr{L}$ is single-valued and analytic in $z$ where defined, and
that $\mathscr{L}(z)=\log(-z)+\mathcal{O}(z^{-1})$ as $|z|\to\infty$.  Then $g(z)$ can be written 
in the explicit form
\begin{multline}
\label{g-formula}
g(z) = g(z;x) := \frac{1}{2}\theta(z;x) - \frac{1}{8}\left(4z^2+2S(x)z-2S(x)^2-\Delta(x)^2\right)r(z;x) + \mathscr{L}(z;x) + \frac{1}{8}S(x)^3,  \\
\quad z\notin\Sigma\cup L, \quad x\notin \Sigma_S \cup \mathscr{R}_{-\pi/3}.
\end{multline}
Furthermore, $\lambda=\lambda(x)$ may be expressed for $x\not\in\Sigma_S\cup\mathscr{R}_{-\pi/3}$ as
\begin{equation}
\begin{split}
\lambda(x)&=\mathscr{L}_+(z;x)+\mathscr{L}_-(z;x)+\frac{1}{4}S(x)^3,\quad z\in\Sigma\\
&=\frac{1}{4}S(x)^3 +\log\left(\frac{\Delta(x)^2}{16}\right)\\
&=\frac{1}{4}S(x)^3 -\log(3S(x))
\end{split}
\label{eq:lambdadef}
\end{equation}
for an appropriate choice of the logarithm (here not necessarily the principal branch).

\subsection{The genus-zero ansatz and formula for the boundary of the elliptic region}
\label{boundary-section}
The possibility of using the $g$-function defined in \S\ref{genus-zero-g} to asymptotically reduce
the jump matrices for the Riemann-Hilbert problem to a tractable form hinges on the qualitative nature of the zero level set of the function 
\begin{equation}
F(z)=F(z;x):=\mathrm{Re}(2h(z;x)+\lambda(x))
\label{eq:g0-Fdefine}
\end{equation}
in the $z$-plane.  This level set can undergo sudden topological changes (bifurcations) as $x$ varies; the bifurcations occur exactly when the critical point $z=-S(x)/2$ crosses the level set.  The zero level set of $F$ always includes the endpoints $z=a$ and $z=b$ of $\Sigma$.
A direct calculation using the cubic equation \eqref{cubic-equation} shows that $-S(x)/2$ lies on the band $\Sigma$ exactly when $x\in\mathscr{R}_-\cup\mathscr{R}_{\pi/3}\cup\mathscr{R}_{-\pi/3}$, and therefore the formula
\begin{equation}
\mathfrak{c}(x):=\int_{a(x)}^{-S(x)/2}h'(z;x)\,dz,\quad |\arg(x)|<\frac{\pi}{3},\;\text{or}\;
\frac{\pi}{3}<\arg(x)<\pi,\;\text{or}\;-\pi<\arg(x)<-\frac{\pi}{3},
\label{c-def-gen0}
\end{equation}
defines three different analytic functions of $x$ in abutting sectors of the complex $x$-plane, where the path of integration is taken to be a straight line.  Now, $\mathfrak{c}(x)$ cannot generally be identified via the Fundamental Theorem of Calculus with a difference of values of $h$, because it is possible for the straight-line path of integration to cross the logarithmic branch cut $L$ for $h$; however since $h_+(z;x)-h_-(z;x)=2\pi i$ for $z\in L$, it \emph{is} true that
\begin{equation}
2\mathrm{Re}(\mathfrak{c}(x))=2\mathrm{Re}\left(h\left(-\frac{1}{2}S(x);x\right)-h(a(x);x)\right) = 
F\left(-\frac{1}{2}S(x);x\right).
\end{equation}
While $\mathfrak{c}(x)$ has jump discontinuities across the three rays $\mathscr{R}_-$, $\mathscr{R}_{\pi/3}$, and $\mathscr{R}_{-\pi/3}$, $\mathrm{Re}(\mathfrak{c}(x))$ extends to these rays from either side as a continuous harmonic function.  To see this, one actually shows more by a direct calculation using \eqref{eq:hprime} and \eqref{c-def-gen0}; namely for $x\in\mathscr{R}_-\cup\mathscr{R}_{\pi/3}\cup\mathscr{R}_{-\pi/3}$, both boundary values taken by $\mathfrak{c}(x)$ are purely imaginary.  Hence $\mathrm{Re}(\mathfrak{c}(x))$ extends to these three rays with the value $\mathrm{Re}(\mathfrak{c}(x))=0$.  It can be shown that the only other points in the domain 
$\mathbb{C}\setminus\Sigma_S$ where the harmonic function $\mathrm{Re}(\mathfrak{c}(x))$ vanishes lie along three bounded arcs joining the endpoints of $\Sigma_S$.  These three arcs together with the three rays $\mathscr{R}_-$, $\mathscr{R}_{\pi/3}$, and $\mathscr{R}_{-\pi/3}$
define the locus of points in the domain $x\in\mathbb{C}\setminus\Sigma_S$ where the zero level set of $F$ undergoes a topological bifurcation.
Note that the condition 
$\mathrm{Re}(\mathfrak{c}(x))=0$ can  be written explicitly in terms of elementary functions as
\eq
\label{eq:EdgeCondition}
\text{Re}\left[\frac{x}{3}r\left(-\frac{S(x)}{2};x\right)\right]-\log\left\vert S(x)-r\left(-\frac{S(x)}{2};x\right)\right\vert-\frac{1}{2}\log\vert S(x)\vert+\log\left(\frac{2\sqrt{3}}{3}\right) = 0.
\endeq

\begin{definition}
\label{bulk-defn}
The 
\emph{elliptic region} $T$ is the domain of the complex $x$-plane bounded by the three arcs of the level curve $\mathrm{Re}(\mathfrak{c}(x))=0$ joining in pairs the three endpoints of $\Sigma_S$.  
The \emph{genus-zero region} is the unbounded domain $\mathbb{C}\backslash\overline{T}$ complementary to $T$.  
\end{definition}
The elliptic region $T$ is the open domain bounded by the curvilinear triangle $\partial T$ illustrated in Figure
\ref{fig:c-of-x-equals-zero}.  
\begin{figure}[h]
\includegraphics[width=2.5in]{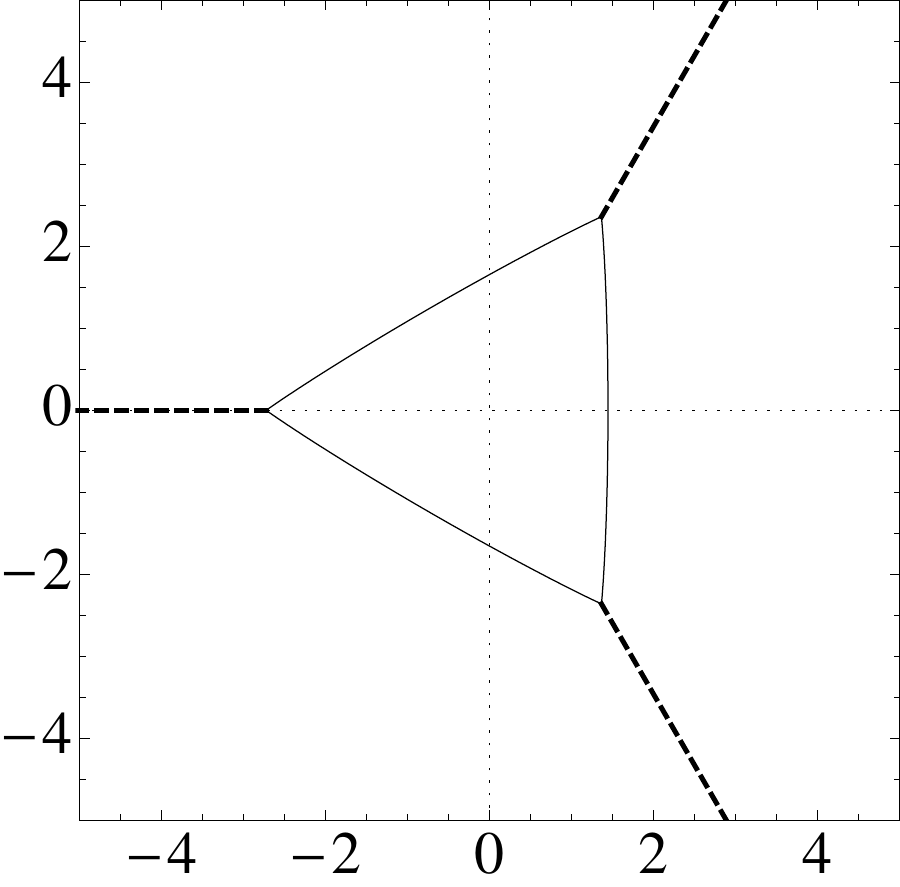}
\caption{\emph{The curves in the complex $x$-plane along which $\mathrm{Re}(\mathfrak{c}(x))=0$.  The solid lines are 
$\partial T$ (across which the rational Painlev\'e-II functions exhibit Stokes phenomena in the large-$m$ limit) and the dashed lines are the semi-infinite contours $\mathscr{R}_{\pm\pi/3}$ and $\mathscr{R}_-$.}}
\label{fig:c-of-x-equals-zero}
\end{figure}
\begin{remark}
Despite a striking resemblance from a distance, $\partial T$ is \emph{not} 
a Euclidean triangle; one can check that the interior angle at each corner is exactly $2\pi/5$.  
This turns out to be related to the fact that the pole sector of the 
tritronqu\'ee solutions to the Painlev\'e-I equation opens with the same angle.   
Details will be given in the sequel paper \cite{Buckingham-rational-crit}.
\end{remark}

We say that the genus-zero ansatz is valid if the topology of the level curves of $F(z;x)$ defined by \eqref{eq:g0-Fdefine} in the $z$-plane is well-suited to the asymptotic reduction of the Riemann-Hilbert problem for $\mathbf{M}(z)$ in the limit $\epsilon\downarrow 0$.  To be more precise, we offer the following definition.
\begin{definition}
The \emph{genus-zero ansatz} is valid for a given $x\in\mathbb{C}$ if 
\begin{itemize}
\item There are exactly three arcs of the zero-level set of 
$F(z;x)$ terminating at $z=a(x)$ and $z=b(x)$, and 
\item There exist six arcs in $\mathbb{C}\setminus\Sigma$, three from $a(x)$ to $\infty$ and three  from $b(x)$ 
to $\infty$, none of which crosses the zero level set of 
$F$, and such that the six arcs tend to infinity at distinct angles 
$0$, $\pm\pi/3$, $\pm 2\pi/3$, and $\pi$.
\end{itemize}
\end{definition}

\begin{proposition} \label{genus-zero-proposition} 
The genus-zero ansatz is valid exactly in the genus-zero region.  
\end{proposition}
\begin{proof}
We study the zero level set of $F(z;x)$ defined in terms of $h$ by \eqref{eq:g0-Fdefine}.  For each $x\not\in\Sigma_S$, this is the zero level set of a function harmonic in $z$ except on the band $\Sigma$, across which it generally has a jump discontinuity.  The zero level set therefore consists of a finite number of smooth arcs.  The points $z=a(x)$ and $z=b(x)$ are necessarily contained in the zero level set by the definition of $\lambda(x)$, but it can be checked that $F(z;x)$ vanishes at no other points of either edge of the branch cut $\Sigma$ \emph{unless} $x\in\mathscr{R}_-\cup\mathscr{R}_{\pi/3}\cup\mathscr{R}_{-\pi/3}$ in which case $F(z;x)$ vanishes identically for $z\in\Sigma$ (more properly, both boundary values taken on $\Sigma$ vanish identically).  
  
The only finite points where multiple arcs of the zero level set can intersect are zeros of 
$h'(z;x)$, so from \eqref{eq:hprime} the candidate points are $a(x)$, $b(x)$, 
and $-S(x)/2$, the first two of which are part of the zero level set for all $x$ under consideration.
The point $z=-S(x)/2$ lies on the zero level set of $F(z;x)$ exactly
when $\mathrm{Re}(\mathfrak{c}(x))=0$, namely for $x\in\mathscr{R}_-\cup\mathscr{R}_{\pi/3}\cup\mathscr{R}_{-\pi/3}$ and for $x\in\partial T$.  (See Figure~\ref{fig:c-of-x-equals-zero}.)  Note that $a(x)\neq b(x)$ since $\Delta(x)^2$ is zero free.
A local analysis then shows that if $a(x)\neq-S(x)/2$ then there 
are exactly three arcs of the level set of $F(z;x)$ emanating from $a(x)$ at angles separated 
by $2\pi/3$.  On the other hand, if $a(x)=-S(x)/2$ then there are exactly 
five level curves emanating from $a(x)$ at angles separated by $2\pi/5$ (and 
thus the genus-zero ansatz is not valid).  Analogous statements hold for 
$b(x)$.  By direct calculation, the only $x$-values for which either $a(x)=-S(x)/2$ 
or $b(x)=-S(x)/2$ are $x_c$, $x_c e^{-2i\pi/3}$, and $x_c e^{2i\pi/3}$ 
(the three corners of the elliptic region $T$, also the three endpoints of $\Sigma_S$).  

\begin{figure}
\includegraphics[width=1.5in]{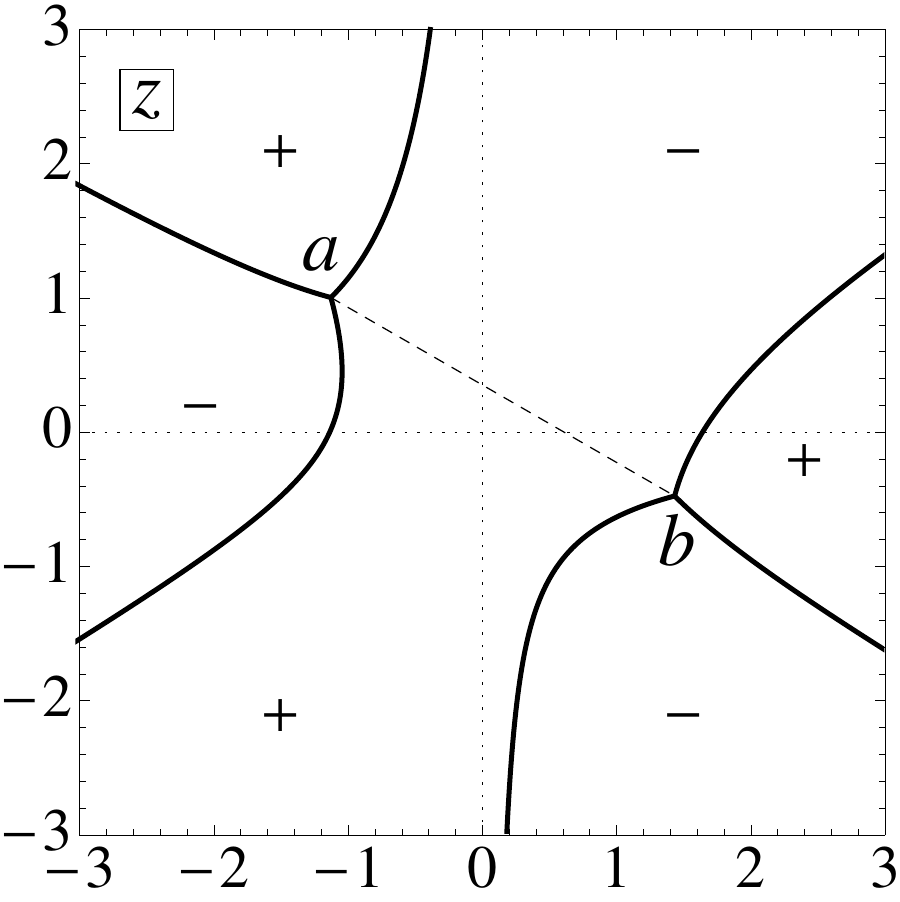}
\includegraphics[width=1.5in]{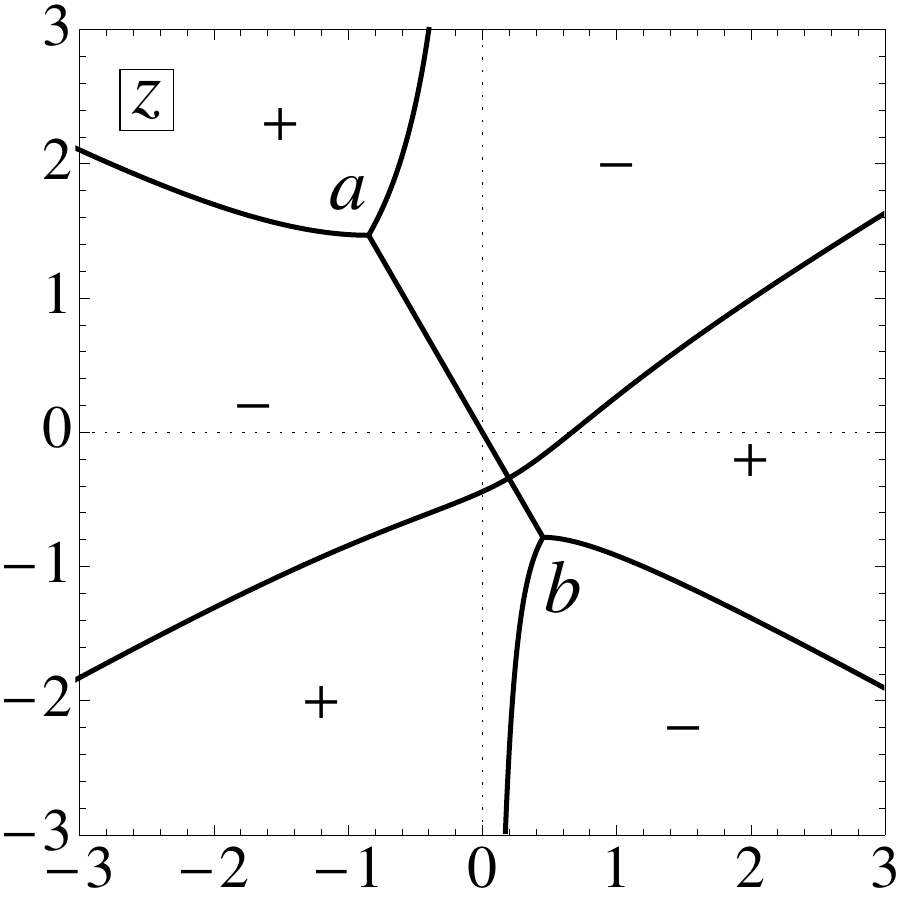}
\\
\includegraphics[width=1.5in]{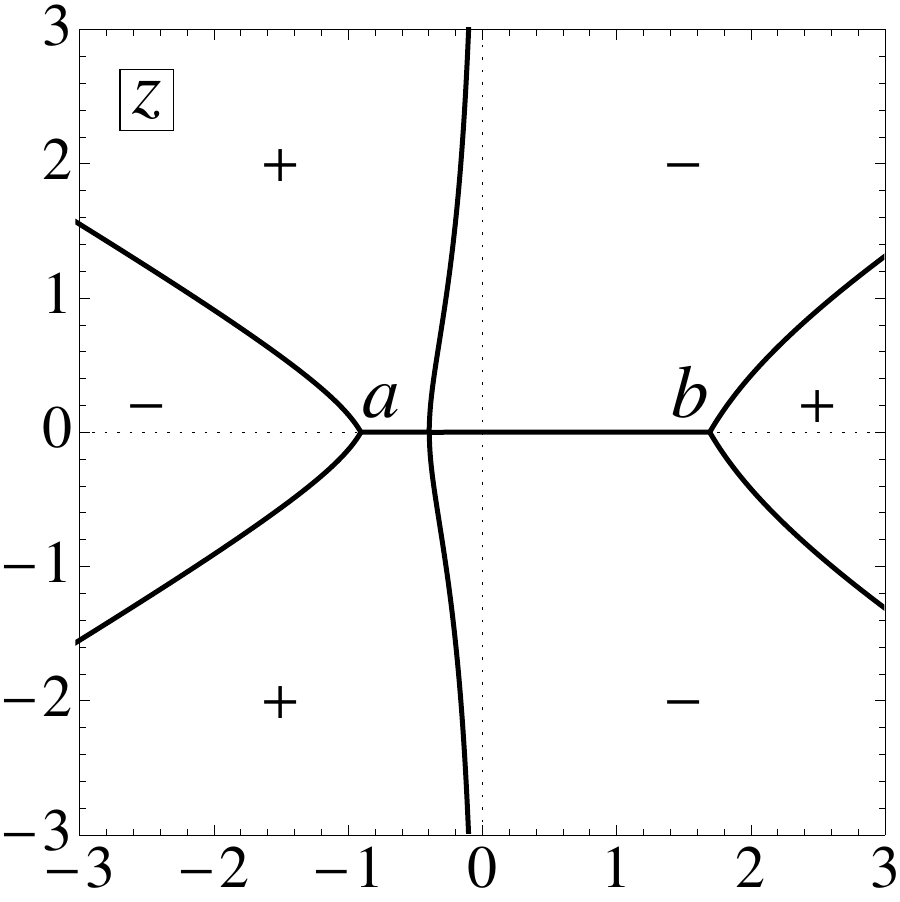}
\includegraphics[width=1.5in]{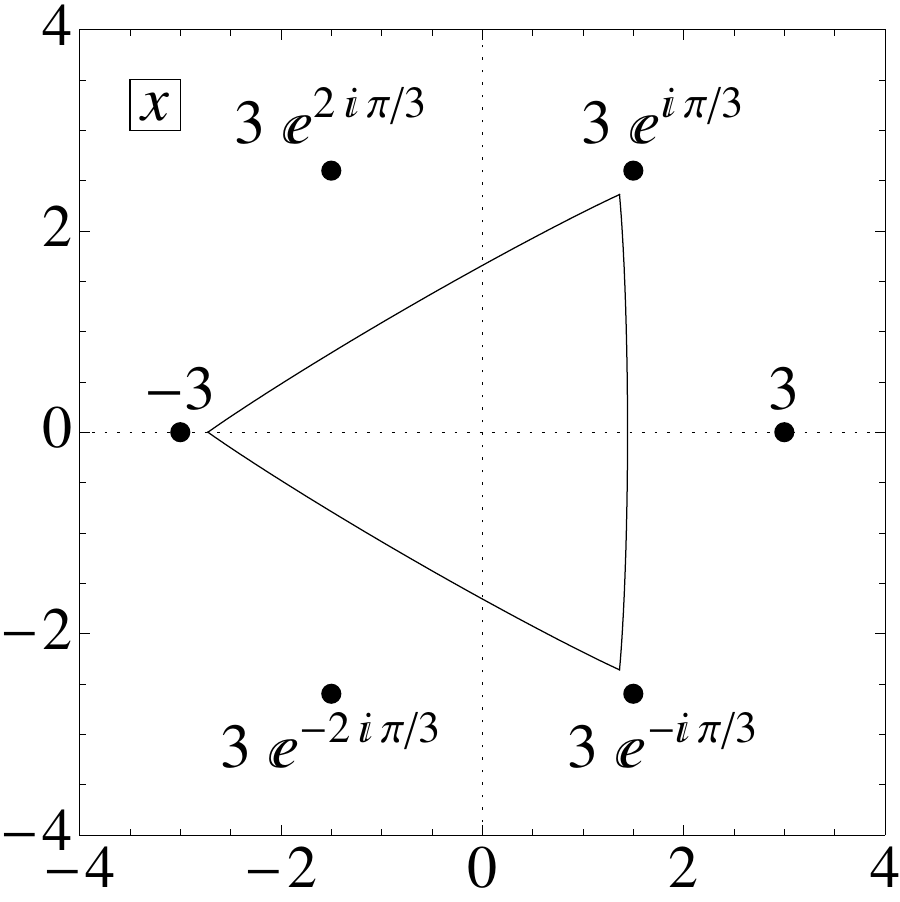}
\includegraphics[width=1.5in]{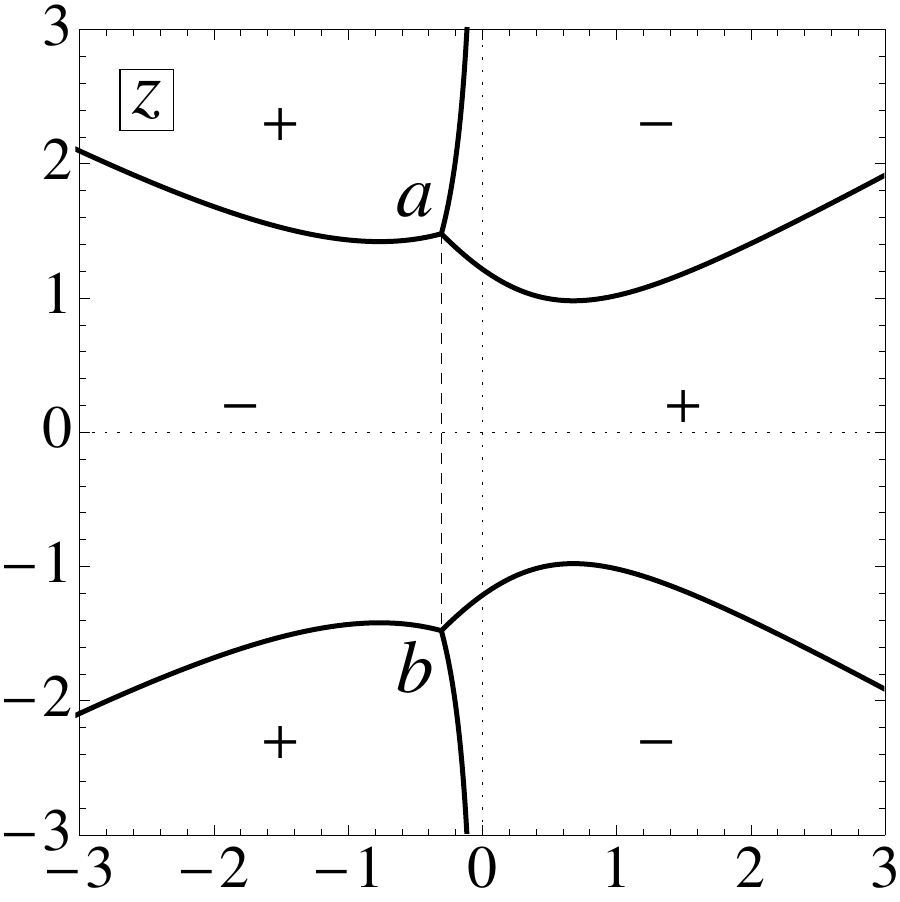}
\\
\includegraphics[width=1.5in]{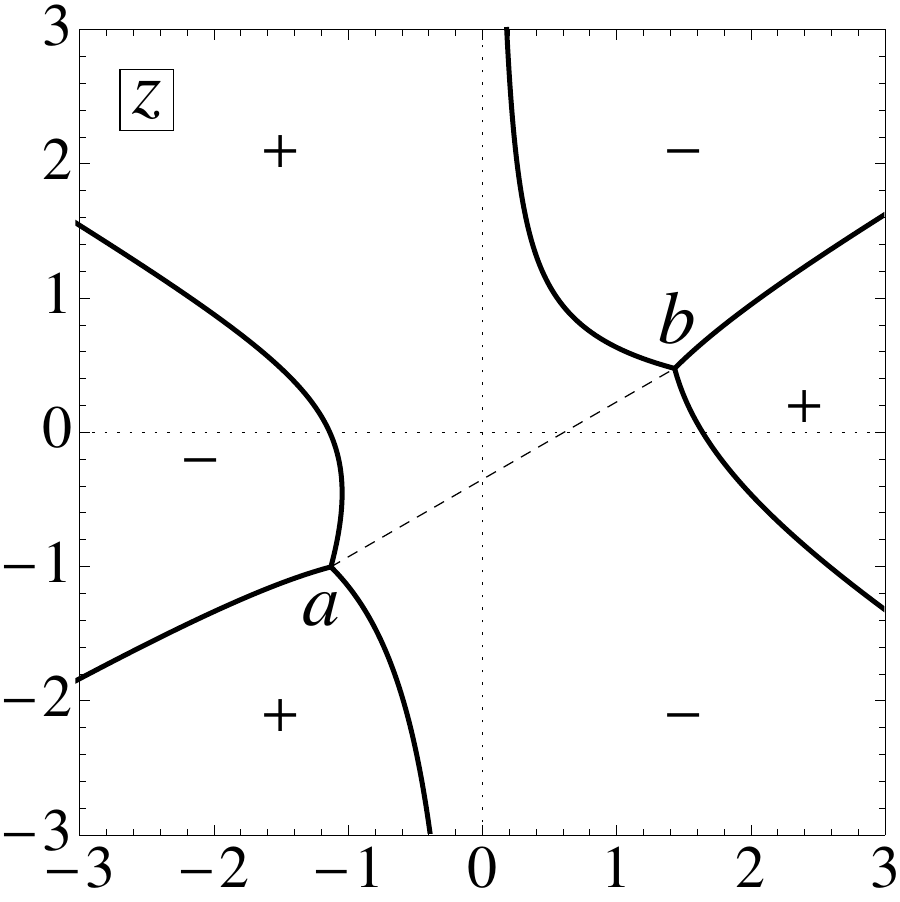}
\includegraphics[width=1.5in]{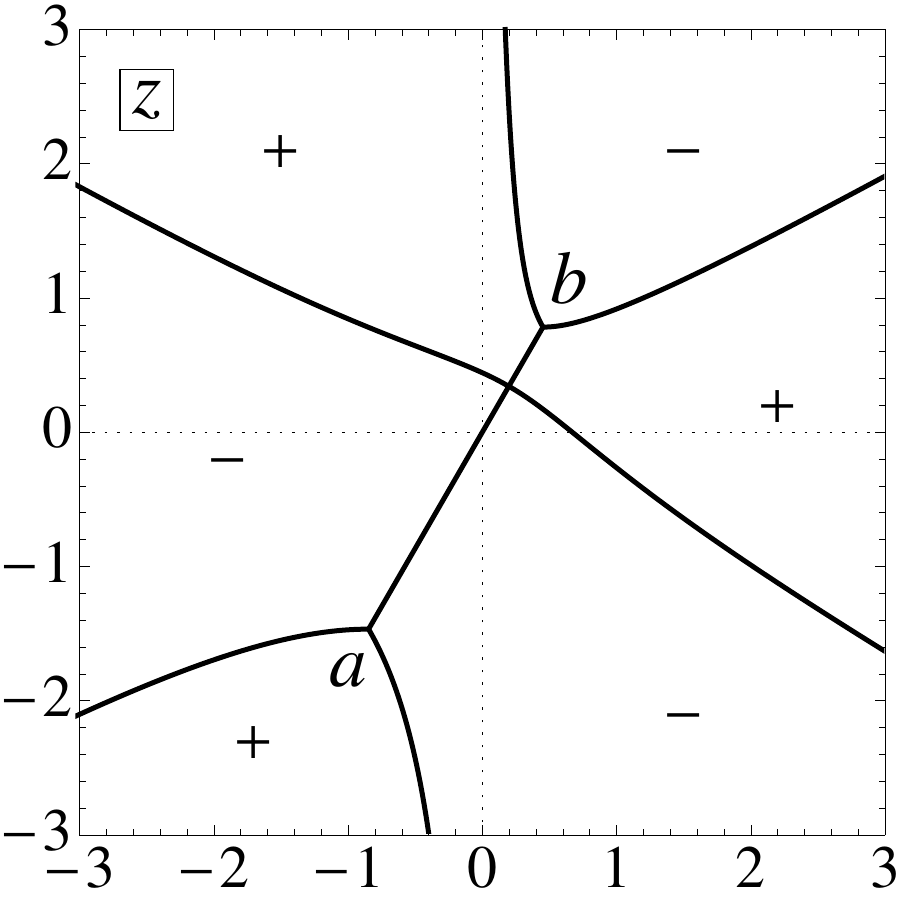}
\caption{\emph{Signature charts for $F(z;x):=\mathrm{Re}(2h(z;x)+\lambda(x))$ in the 
complex $z$-plane for various choices of $x$ outside the boundary curve.  
Counter-clockwise from the right-most chart:  
$x=3$, $x=3e^{i\pi/3}$, $x=3e^{2i\pi/3}$, 
$x=-3$, $x=3e^{-2i\pi/3}$, $x=3e^{-i\pi/3}$.  
The center plot illustrates the relation of the chosen $x$ values to the 
boundary curve.  The solid lines are zero level curves of 
$F$, while the dashed lines represent jump discontinuities 
across the contour $\Sigma$.  Recall that $a(x)$ and $b(x)$ are exchanged as 
$x$ crosses $\mathscr{R}_{-\pi/3}$.}}
\label{h-contours}
\end{figure}

For the remainder of the proof assume that $-S(x)/2$ does not equal either $a(x)$ 
or $b(x)$.  Again from \eqref{eq:hprime}, near $z=\infty$ there are exactly 
six arcs of the level curve that tend to infinity at angles $\pm\pi/6$, 
$\pm\pi/2$, and $\pm 5\pi/6$.  Furthermore, since 
$F(z;x)$ is harmonic as a function of $z$ and has no critical points for 
$z\in\mathbb{C}\backslash(\Sigma\cup\{-S(x)/2\})$, all arcs of the level set must terminate at either $z=a(x)$, 
$z=b(x)$, $z=-S(x)/2$ (only if $\mathrm{Re}(\mathfrak{c}(x))=0$), or $z=\infty$.  

Fix a point $x=x_0$ such that $F(z;x)$ is continuous (as a 
function of $x$) at $x_0$ (i.e. off the contour $\Sigma_S$ shown in Figure 
\ref{fig:Sigma-S}) and such that $z=-S(x_0)/2$ is not on the zero level set of $F(z;x_0)$.  There 
is a maximal open neighborhood $\mathscr{O}(x_0)$ of $x_0$ in the $x$-plane such that each point 
$x\in\mathscr{O}(x_0)$ can be connected to $x_0$ via a path along which 
$F(z;x)$ is continuous and $z=-S(x)/2$ is never on the zero level set of $F(z;x)$.  
The topology of the zero level set of $F$ cannot change without one of these two 
conditions failing, so if the genus-zero ansatz is valid for $x=x_0$, then it is 
also valid at every point in $\mathscr{O}(x_0)$, and vice-versa.  The $x$-plane can be 
written as 
$\mathscr{O}(3)\cup\mathscr{O}(3e^{2i\pi/3})\cup\mathscr{O}(3e^{-2i\pi/3})\cup\mathscr{R}_{\pi/3}\cup\mathscr{R}_{-\pi/3}\cup\mathscr{R}_-\cup\partial T\cup\mathscr{O}(1)\cup\mathscr{O}(e^{2i\pi/3})\cup\mathscr{O}(e^{-2i\pi/3})\cup\Sigma_S$.  
Now in 
$\mathscr{O}(3)\cup\mathscr{O}(3e^{2i\pi/3})\cup\mathscr{O}(3e^{-2i\pi/3})$ 
the genus-zero ansatz is valid, as can be seen from the signature charts shown in Figure \ref{h-contours}.  The same 
figure illustrates that the genus-zero ansatz is valid also on 
$\mathscr{R}_{\pi/3}\cup\mathscr{R}_{-\pi/3}\cup\mathscr{R}_-$ (moreover it is impossible for the 
topology of the zero level set of $F$ to change as $x$ varies along one of 
these rays).  
On the other hand, the genus-zero ansatz fails along 
$\partial T$,
as shown in Figure \ref{genus-zero-failure}, when two arcs of the zero level set of  $F$ pinch 
together at $z=-S(x)/2$, making it impossible to draw arcs from the band endpoints to one 
of the sectors at infinity.  Also, by checking the points 
$x=1$, $x=e^{2i\pi/3}$, and $x=e^{-2i\pi/3}$, it is seen that one of the arcs of the zero level 
set of $F$ emanating from $a(x)$ terminates at $b(x)$, while there is an arc of the zero level 
set of $F$ with both ends terminating at infinity.  As a result the genus-zero 
ansatz fails for 
$\mathscr{O}(1)\cup\mathscr{O}(e^{2i\pi/3})\cup\mathscr{O}(e^{-2i\pi/3})$.
It can be similarly checked that the genus-zero ansatz fails for 
$x\in\Sigma_S$.  
\begin{figure}
\includegraphics[width=1.5in]{h-contours-m3.pdf}
\includegraphics[width=1.5in]{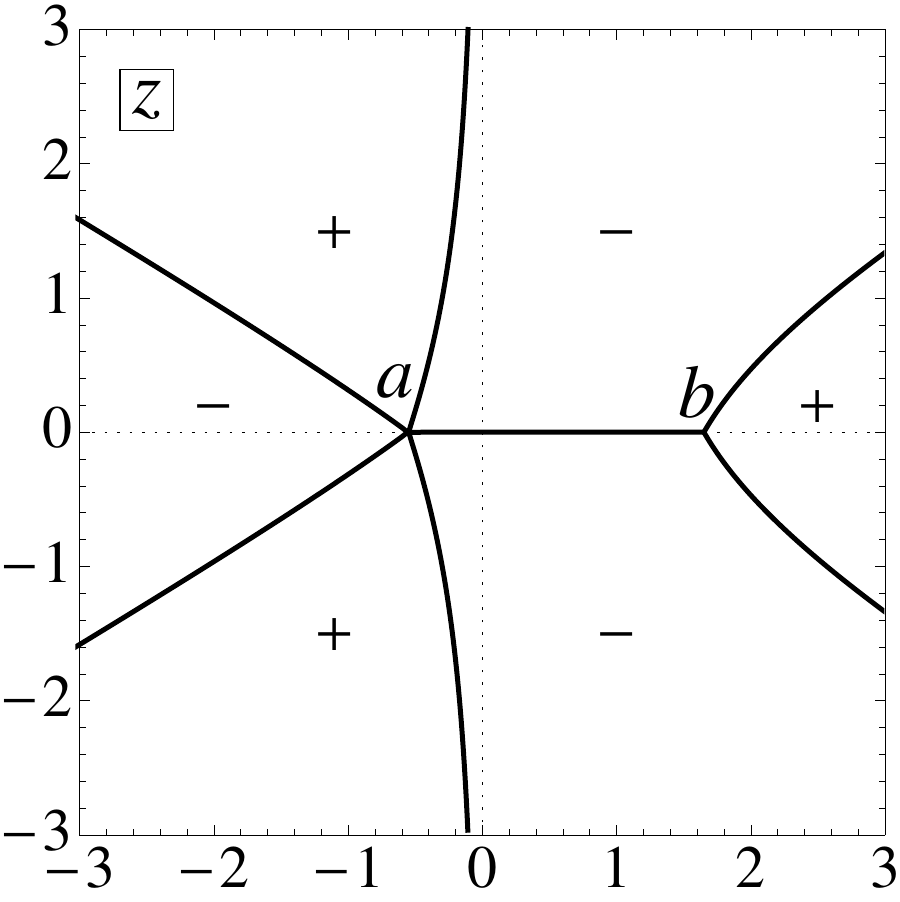}
\includegraphics[width=1.5in]{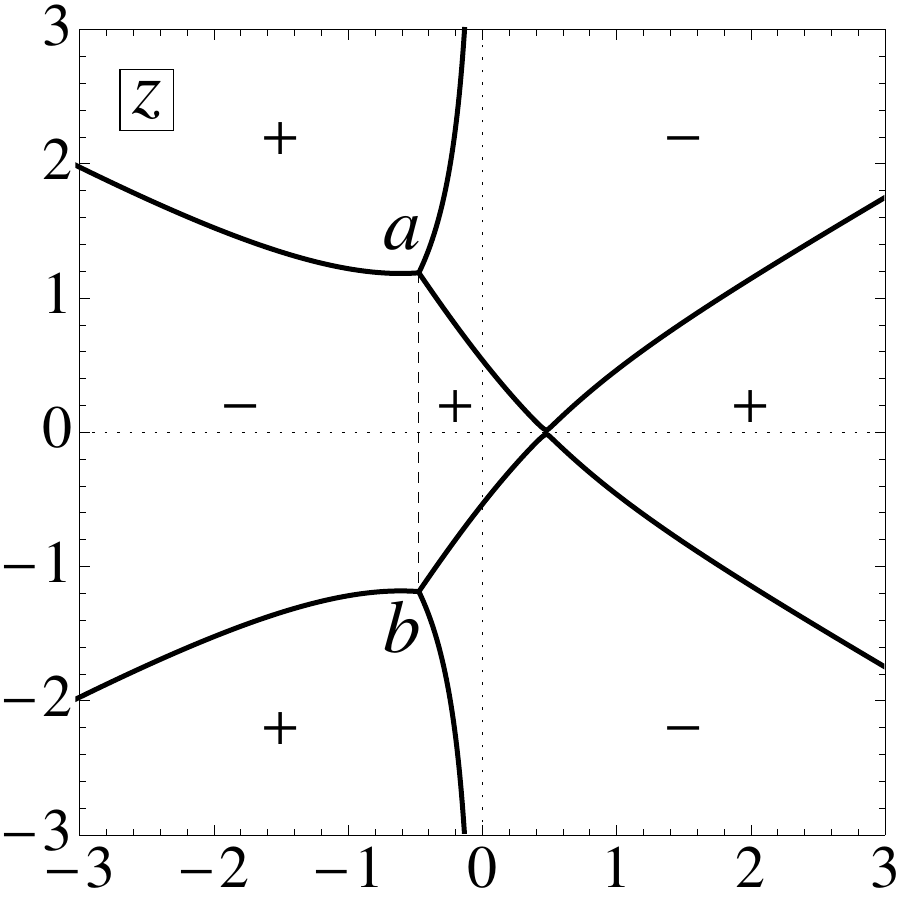}
\includegraphics[width=1.5in]{h-contours-3.pdf}
\\
\includegraphics[width=1.5in]{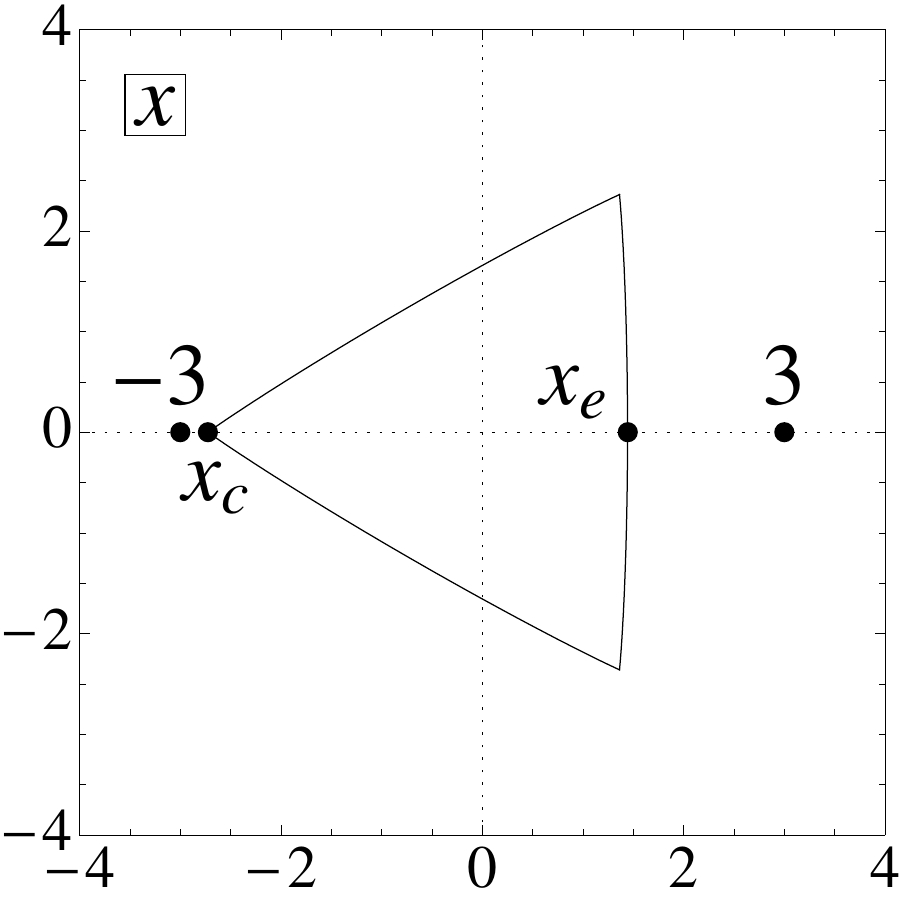}
\caption{\emph{Signature charts for $F(z;x):=\mathrm{Re}(2h(z;x)+\lambda(x))$ in the complex $z$-plane for various choices of $x$ illustrating the two ways the genus-zero ansatz can fail.  Top row, left to right:  $x=-3$, $x=x_c=-(9/2)^{2/3}$, $x=x_e\approx 1.445$ (the positive real point on the boundary of the elliptic region), $x=3$.  The solid lines are zero level curves of $F$, while the dashed lines represent jump discontinuities across the contour $\Sigma$. The bottom plot illustrates the relation of the chosen $x$ values to the boundary curve.}  }
\label{genus-zero-failure}
\end{figure}
\end{proof}

For all $x$ in the genus zero region, we now choose the contour $L$ connecting $z=b(x)$ to infinity (ultimately along the positive real axis) to lie entirely within the domain where the inequality
$\mathrm{Re}(h_+(z;x)+h_-(z;x)+\lambda(x))>0$ holds (except at the initial endpoint $z=b$).

\begin{figure}[h]
\includegraphics[width=2.5in]{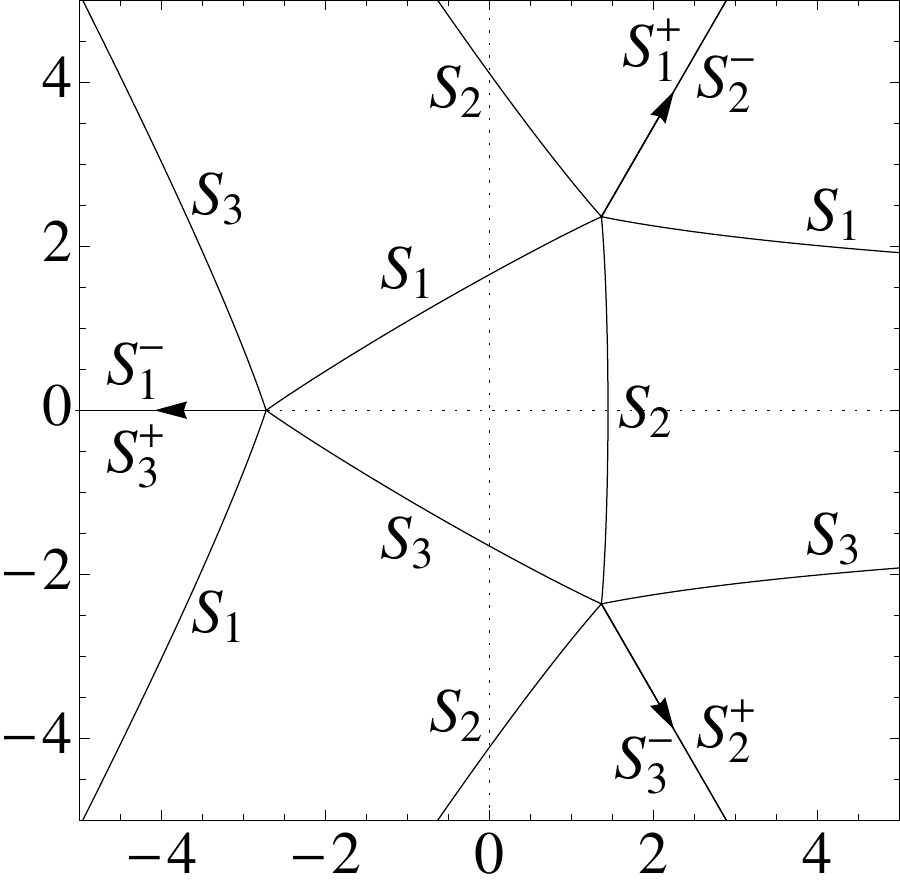}
\caption{\emph{The curves in the complex $x$-plane where general solutions of the Painlev\'e-II equation \eqref{PII} exhibit Stokes phenomenon in the large-$m$ limit.  The curves marked $S_j$ satisfy $\mathrm{Re}(h(-S_j(x)/2;x))=\mathrm{Re}(h(a(x);x))$ for $j=1,2,3$.  The curves marked $S_j^+$ (respectively, $S_j^-$) satisfy this condition when $S_j(x)$ is understood to be the limit as $x$ approaches the curve from the left (respectively, right) as indicated by the arrows.}
}
\label{fig:phantom-stokes}
\end{figure}

Now is an appropriate time to emphasize that the rational solutions exhibit 
fewer Stokes lines in the large-$m$ limit than the generic solution of the 
Painlev\'e-II equation.  The generic Stokes lines were found by Kapaev 
\cite{Kapaev:1997} and are illustrated in Figure \ref{fig:phantom-stokes}.  
The three curves that bound the elliptic region $T$ are bona-fide Stokes curves 
for the rational solutions of the Painlev\'e-II equation, and we have seen that $x$-values on the three 
semi-infinite rays $\mathscr{R}_{\pm\pi/3}$ and $\mathscr{R}_{-}$  
satisfy $\mathrm{Re}(\mathfrak{c}(x))=0$ (but the genus-zero ansatz is still valid).  The six 
remaining curves (two emanating from each of the corners) play a more subtle 
role in the analysis of the rational functions.  To see how they arise, we 
define three new functions $S_1(x)$, $S_2(x)$, and $S_3(x)$, each of which 
satisfies \eqref{cubic-equation} (the defining equation for $S(x)$).  
Specifically, choose these functions so 
\begin{itemize}
\item $S_1(x)$ is analytic off $\mathscr{R}_{\pi/3}\cup\mathscr{R}_{-}$ and $S_1(x)\sim(-\tfrac{4}{3}x)^{1/2}$ as $x\to\infty$ along $\mathscr{R}_{-\pi/3}$,
\item $S_2(x)$ is analytic off $\mathscr{R}_{\pi/3}\cup\mathscr{R}_{-\pi/3}$ and $S_2(x)\sim-(-\tfrac{4}{3}x)^{1/2}$ as $x\to\infty$ along $\mathscr{R}_{-}$,
\item $S_3(x)$ is analytic off $\mathscr{R}_{-}\cup\mathscr{R}_{-\pi/3}$ and $S_3(x)\sim(-\tfrac{4}{3}x)^{1/2}$ as $x\to\infty$ along $\mathscr{R}_{\pi/3}$.
\end{itemize}
(In each case we mean to take the principal branch of the square root.)  Then the six remaining curves in Figure \ref{fig:phantom-stokes} 
are curves on which 
\eq
\label{phantom-curve-condition}
\text{Re}\left(\int_{a(x)}^{-S_j(x)/2}h'(z;x)dz\right)=0
\endeq
for various choices of $j$ as indicated in the figure.  Since $S_j(x)$ never agrees 
with $S(x)$ on these six curves when \eqref{phantom-curve-condition} is 
satisfied, these curves play no role in the genus-zero analysis.  However, 
the fact that these curves lie outside the elliptic region will be used in 
\S\ref{bulk-section}.

\subsection{Reduction to a model Riemann-Hilbert problem.}
\label{model-RHP-genus1}
For the remainder of \S\ref{section-gen0}, we assume that $x$ lies in the genus zero region $\mathbb{C}\setminus\overline{T}$.
We are now ready to deform the jump contours in preparation for performing 
the nonlinear steepest-descent analysis of ${\bf M}(z;x)$, guided by the 
signature charts for $F(z;x):=\mathrm{Re}(2h(z;x)+\lambda(x))$ shown in Figure 
\ref{h-contours} for various values of $x$.  
Using standard sectionally analytic substitutions, 
we deform the jump contours for 
${\bf M}(z)$ so that the six rays intersect at $a$, 
and then collapse the three rays that tend to infinity in the directions $\arg(z)=0$, 
$\arg(z)=-\pi/3$, and $\arg(z)=-2\pi/3$ so that they coincide along 
$\Sigma$.  The resulting jump on $\Sigma$ is simply the product of three factors:
\eq
\label{collapsed-jump}
\bbm 1 & ie^{-\theta/\epsilon} \\ 0 & 1 \ebm \bbm 1 & 0 \\ ie^{\theta/\epsilon} & 1 \ebm \bbm -1 & -ie^{-\theta/\epsilon} \\ 0 & -1 \ebm = \bbm 0 & -ie^{-\theta/\epsilon} \\ -ie^{\theta/\epsilon} & 0 \ebm.
\endeq 
Finally, the six semi-infinite contours are deformed if necessary so that 
their behavior at infinity is unchanged but those with upper-diagonal jump matrices 
are confined to domains in which $F(z)=\mathrm{Re}(2h(z)+\lambda)>0$ and so that those with 
lower-diagonal jump matrices are confined to domains in which $F(z)=\mathrm{Re}(2h(z)+\lambda)<0$ (and $L$ has already been chosen so that $\mathrm{Re}(h_+(z)+h_-(z)+\lambda)>0$ holds along $L$).
These deformations result in a Riemann-Hilbert problem equivalent to that for the matrix $\mathbf{M}(z;x,\epsilon)$, but with a related unknown matrix $\mathbf{N}(z;x,\epsilon)$; we may assume that $\mathbf{N}(z;x,\epsilon)\equiv\mathbf{M}(z;x,\epsilon)$ for sufficiently large $|z|$,
and hence $\mathbf{N}(z;x,\epsilon)$ satisfies the normalization condition
\eq
\lim_{z\to\infty}{\bf N}(z;x,\epsilon)(-z)^{-\sigma_3/\epsilon} = \mathbb{I}.
\endeq
However, the jump contour for $\mathbf{N}(z;x,\epsilon)$ differs from that for $\mathbf{M}(z;x,\epsilon)$
in the finite part of the $z$-plane as illustrated in Figure~\ref{fig:N-jumps-pi}.  
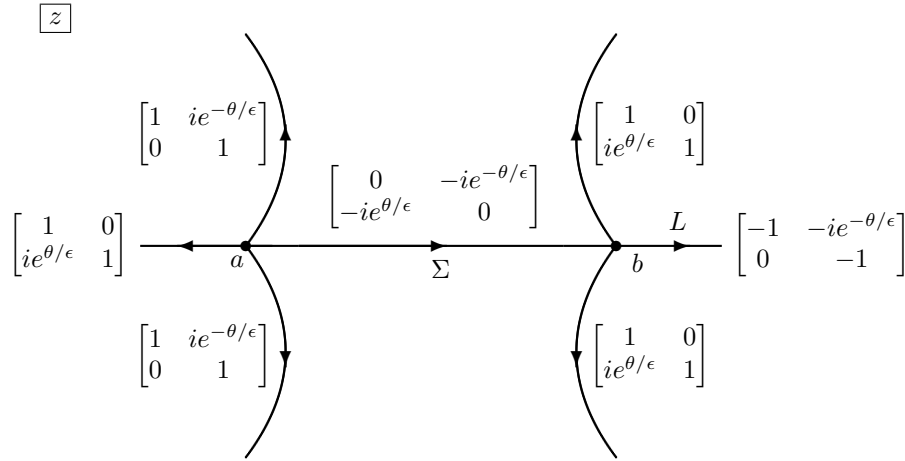
\begin{figure}[h]
\setlength{\unitlength}{2pt}
\begin{center}
\begin{picture}(100,100)(-50,-50)
\put(-74,42){\framebox{$z$}}
\put(-35,0){\circle*{2}}
\put(-38,-4){$a$}
\put(35,0){\circle*{2}}
\put(38,-5){$b$}
\thicklines
\put(25,0){\line(1,0){30}}
\put(48,0){\vector(1,0){1}}
\put(57,-1){$\bbm -1 & -ie^{-\theta/\epsilon} \\ 0 & -1 \ebm$}
\put(45,3){$L$}
\qbezier(35,0)(20,20)(35,40)
\put(27.5,22){\vector(0,1){1}}
\put(30,20){$\bbm 1 & 0 \\ ie^{\theta/\epsilon} & 1 \ebm$}
\put(-25,0){\line(1,0){50}}
\put(-25,0){\vector(1,0){28}}
\put(0,-6){$\Sigma$}
\put(-20,8){$\bbm 0 & -ie^{-\theta/\epsilon} \\ -ie^{\theta/\epsilon} & 0 \ebm$}
\qbezier(-35,0)(-20,20)(-35,40)
\put(-27.5,22){\vector(0,1){1}}
\put(-56,20){$\bbm 1 & ie^{-\theta/\epsilon} \\ 0 & 1 \ebm$}
\put(-25,0){\line(-1,0){30}}
\put(-30,0){\vector(-1,0){18}}
\put(-80,-1){$\bbm 1 & 0 \\ ie^{\theta/\epsilon} & 1 \ebm$}
\qbezier(-35,0)(-20,-20)(-35,-40)
\put(-27.5,-22){\vector(0,-1){1}}
\put(-56,-22){$\bbm 1 & ie^{-\theta/\epsilon} \\ 0 & 1 \ebm$}
\qbezier(35,0)(20,-20)(35,-40)
\put(27.5,-22){\vector(0,-1){1}}
\put(30,-22){$\bbm 1 & 0 \\ ie^{\theta/\epsilon} & 1 \ebm$}
\end{picture}
\end{center}
\caption{\emph{The jump matrices $\mathbf{V^{(N)}}(z;x,\epsilon)$ for 
$x$ on the negative real axis outside the elliptic region $T$.  A topologically 
equivalent deformation applies for any $x$ outside the elliptic region in the 
sector 
$|\arg(-x)|<2\pi/3$.}}
\label{fig:N-jumps-pi}
\end{figure}

We now introduce the $g$-function defined in \S\ref{genus-zero-g} into the deformed Riemann-Hilbert problem for $\mathbf{N}(z;x,\epsilon)$.  
With $g(z;x)$ and $\lambda(x)$ given by \eqref{g-formula} and 
\eqref{eq:lambdadef}, we define
\begin{equation}
\label{O-def-gen0}
\mathbf{O}(z;x,\epsilon):=e^{-\lambda(x)\sigma_3/(2\epsilon)}\mathbf{N}(z;x,\epsilon)e^{-g(z;x)\sigma_3/\epsilon}e^{\lambda(x)\sigma_3/(2\epsilon)}.
\end{equation}
The matrix-valued function $\mathbf{O}(z;x,\epsilon)$ has jump discontinuities across the 
same contours as $\mathbf{N}(z;x,\epsilon)$.  The jump matrices for 
$\mathbf{O}(z;x,\epsilon)$ are obtained from those of
$\mathbf{N}(z;x,\epsilon)$ by the recipe
\begin{equation}
\mathbf{V^{(N)}}(z;x,\epsilon) = \bbm v_{11}^{(\mathbf{N})} & v_{12}^{(\mathbf{N})} \\ v_{21}^{(\mathbf{N})} & v_{22}^{(\mathbf{N})} \ebm \implies \mathbf{V^{(O)}}(z;x,\epsilon) := \bbm v_{11}^{(\mathbf{N})}e^{-(g_+-g_-)/\epsilon} & v_{12}^{(\mathbf{N})}e^{(g_++g_--\lambda)/\epsilon} \\ v_{21}^{(\mathbf{N})}e^{-(g_++g_--\lambda)/\epsilon} & v_{22}^{(\mathbf{N})}e^{(g_+-g_-)/\epsilon} \ebm.
\end{equation}
Then, using the properties of $g$ described in Proposition~\ref{prop:g0-g-function}
along with $\epsilon^{-1}+\tfrac{1}{2}\in\mathbb{Z}$, we see that 
$\mathbf{O}(z;x,\epsilon)$ satisfies the Riemann-Hilbert problem with jump contour and jump matrices as illustrated in Figure~\ref{fig:O-jumps-pi} 
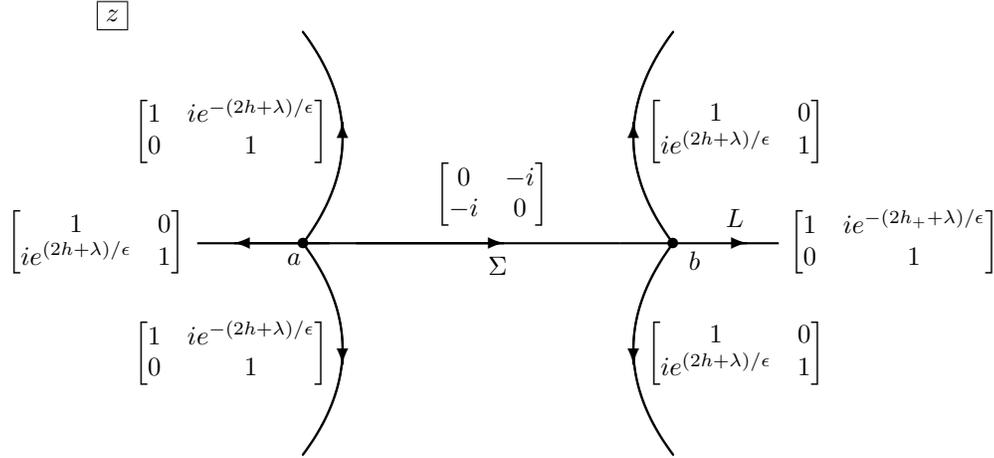
\begin{figure}[h]
\setlength{\unitlength}{2pt}
\begin{center}
\begin{picture}(100,100)(-50,-50)
\put(-74,42){\framebox{$z$}}
\put(-35,0){\circle*{2}}
\put(-38,-4){$a$}
\put(35,0){\circle*{2}}
\put(38,-5){$b$}
\thicklines
\put(25,0){\line(1,0){30}}
\put(48,0){\vector(1,0){1}}
\put(57,-1){$\bbm 1 & ie^{-(2h_++\lambda)/\epsilon} \\ 0 & 1 \ebm$}
\put(45,3){$L$}
\qbezier(35,0)(20,20)(35,40)
\put(27.5,22){\vector(0,1){1}}
\put(30,20){$\bbm 1 & 0 \\ ie^{(2h+\lambda)/\epsilon} & 1 \ebm$}
\put(-25,0){\line(1,0){50}}
\put(-25,0){\vector(1,0){28}}
\put(0,-6){$\Sigma$}
\put(-10,8){$\bbm 0 & -i \\ -i & 0 \ebm$}
\qbezier(-35,0)(-20,20)(-35,40)
\put(-27.5,22){\vector(0,1){1}}
\put(-67,20){$\bbm 1 & ie^{-(2h+\lambda)/\epsilon} \\ 0 & 1 \ebm$}
\put(-25,0){\line(-1,0){30}}
\put(-30,0){\vector(-1,0){18}}
\put(-91,-1){$\bbm 1 & 0 \\ ie^{(2h+\lambda)/\epsilon} & 1 \ebm$}
\qbezier(-35,0)(-20,-20)(-35,-40)
\put(-27.5,-22){\vector(0,-1){1}}
\put(-67,-22){$\bbm 1 & ie^{-(2h+\lambda)/\epsilon} \\ 0 & 1 \ebm$}
\qbezier(35,0)(20,-20)(35,-40)
\put(27.5,-22){\vector(0,-1){1}}
\put(30,-22){$\bbm 1 & 0 \\ ie^{(2h+\lambda)/\epsilon} & 1 \ebm$}
\end{picture}
\end{center}
\caption{\emph{The jump matrices $\mathbf{V^{(O)}}(z;x,\epsilon)$ for 
$x$ on the negative real axis outside the elliptic region $T$.  A topologically 
equivalent deformation works for any $x$ outside the elliptic region in the 
sector $|\arg(-x)|<2\pi/3$.  We say that such a topologically equivalent Riemann-Hilbert problem uses the 
Negative-$x$ Configuration.  }}
\label{fig:O-jumps-pi}
\end{figure}
and subject to the normalization condition 
\begin{equation}
\lim_{z\to\infty}\mathbf{O}(z;x,\epsilon)=\mathbb{I}.
\end{equation}
Examining the jump matrices for $\mathbf{O}(z;x,\epsilon)$ shown in Figure 
\ref{fig:O-jumps-pi} in light of the signature chart for $F=\text{Re}(2h+\lambda)$ 
shown in Figure \ref{h-contours}, we see that, \emph{because the genus zero ansatz is valid in the genus zero region according to Proposition~\ref{genus-zero-proposition}},  the jump matrix decays rapidly to the identity as $\epsilon\downarrow 0$ for all $z$ with the sole exception of $z\in\Sigma$.  
Therefore, we are led to attempt to construct an \emph{outer parametrix} $\Odot^{\rm(out)}(z;x)$ as a solution to the 
following ``one cut'' outer model problem formulated on the contour $\Sigma$:
\begin{rhp}[] 
\label{Odot-genus0-rhp}
Find a $2\times 2$ matrix-valued function $\Odot^{\rm(out)}(z;x)$ 
satisfying the following conditions:
\begin{itemize}
\item[]\textbf{Analyticity:}  $\Odot^{\rm(out)}(z;x)$ is analytic for 
$z\notin\Sigma$ with H\"older-continuous boundary values on $\Sigma$ with 
the exception of the endpoints $a$ and $b$, where negative one-fourth power 
singularities are allowed.
\item[]\textbf{Jump condition:}  $\Odot_+^{\rm(out)}(z;x) = \Odot_-^{\rm(out)}(z;x)\bbm 0 & -i \\ -i & 0 \ebm$ for $z\in\Sigma$.
\item[]\textbf{Normalization:}  $\Odot^{\rm(out)}(z;x)\to \mathbb{I}$ as $z\to\infty$. 
\end{itemize}
\label{rhp:outer-model}
\end{rhp}
We now solve for $\Odot^{\rm(out)}$.  Using the factorization 
\eq
\bbm 0 & -i \\ -i & 0 \ebm = \frac{1}{\sqrt{2}} \bbm 1 & -1 \\ 1 & 1 \ebm \cdot \bbm -i & 0 \\ 0 & i \ebm \cdot \frac{1}{\sqrt{2}} \bbm 1 & 1 \\ -1 & 1 \ebm,
\endeq
we see that 
\eq
{\bf \ddot{O}}(z;x):=\frac{1}{\sqrt{2}} \bbm 1 & 1 \\ -1 & 1 \ebm \cdot \Odot^{\rm(out)}(z;x) \cdot \frac{1}{\sqrt{2}}\bbm 1 & -1 \\ 1 & 1 \ebm
\endeq
satisfies the diagonal jump condition 
\eq
{\bf \ddot{O}}_+(z;x) = {\bf \ddot{O}}_-(z;x)\bbm -i & 0 \\ 0 & i \ebm \text{ for } z\in\Sigma
\endeq
and the normalization condition 
\eq
{\bf \ddot{O}}(z;x) \to \mathbb{I} \quad \text{as} \quad z\to\infty.
\endeq
Let $\beta(z;x)$ be the function analytic for $z\in\mathbb{C}\setminus\Sigma$ that is defined by the relation
\eq
\label{beta-gen0}
\beta(z;x)^4=\frac{z-a(x)}{z-b(x)}
\endeq
and with the branch chosen so that $\beta(z)\to 1$ as $z\to\infty$.  
Then 
\eq
{\bf \ddot{O}}(z;x) = \beta(z;x)^{\sigma_3},
\endeq
and thus
\eq
\label{Odot-gen0}
\Odot^{\rm(out)}(z;x) = \frac{1}{2}\bbm \beta(z;x)+\beta(z;x)^{-1} & \beta(z;x)-\beta(z;x)^{-1} \\ \beta(z;x)-\beta(z;x)^{-1} & \beta(z;x)+\beta(z;x)^{-1} \ebm.
\endeq

\subsection{Deformation valid outside the elliptic region near the positive real $x$-axis.}
\label{subsection-outside-positive-x}
\begin{figure}
\setlength{\unitlength}{2pt}
\begin{center}
\begin{picture}(100,100)(-50,-50)
\put(-84,42){\framebox{$z$}}
\put(-20,-10){\vector(1,1){8}}
\put(-10,15){\circle*{2}}
\put(-15,15){$a$}
\put(-10,-15){\circle*{2}}
\put(-11,-21){$b$}
\thicklines
\qbezier(-10,-15)(-5,0)(40,0)
\put(9,-2.75){\vector(4,1){4}}
\put(43,1){$\bbm 1 & ie^{-(2h_++\lambda)/\epsilon} \\ 0 & 1 \ebm$}
\put(20,1){$L$}
\put(5.25,30){\vector(1,2){1}}
\qbezier(-10,15)(0,15)(12,45)
\put(14,39){$\bbm 1 & 0 \\ ie^{(2h+\lambda)/\epsilon} & 1 \ebm$}
\put(-10,-15){\line(0,1){30}}
\put(-10,15){\vector(0,-1){17}}
\put(-8,-2){$\Sigma$}
\put(-42,-12){$\bbm 0 & -i \\ -i & 0 \ebm$}
\put(-10,15){\line(-1,2){15}}
\put(-10,15){\vector(-1,2){9}}
\put(-64,39){$\bbm 1 & ie^{-(2h+\lambda)/\epsilon} \\ 0 & 1 \ebm$}
\qbezier(-10,15)(-17,0)(-50,0)
\put(-30,2){\vector(-4,-1){4}}
\put(-85,-1){$\bbm 1 & 0 \\ ie^{(2h+\lambda)/\epsilon} & 1 \ebm$}
\qbezier(-10,-15)(-20,-15)(-32,-45)
\put(-24,-27.75){\vector(-1,-2){1}}
\put(-70,-41){$\bbm 1 & ie^{-(2h+\lambda)/\epsilon} \\ 0 & 1 \ebm$}
\put(-10,-14){\line(1,-2){15}}
\put(-10,-14){\vector(1,-2){9}}
\put(7,-41){$\bbm 1 & 0 \\ ie^{(2h+\lambda)/\epsilon} & 1 \ebm$}
\end{picture}
\end{center}
\caption{\emph{The jump matrices $\mathbf{V^{(O)}}(z;x,\epsilon)$ for 
$x$ on the positive real axis outside the elliptic region $T$.  A topologically 
equivalent deformation works for any $x$ outside the elliptic region in the 
sector $-\pi/3<\arg(x)<\pi$.  We say that such a topologically equivalent Riemann-Hilbert problem uses the Positive-$x$ Configuration.  }}
\label{fig:O-jumps}
\end{figure}
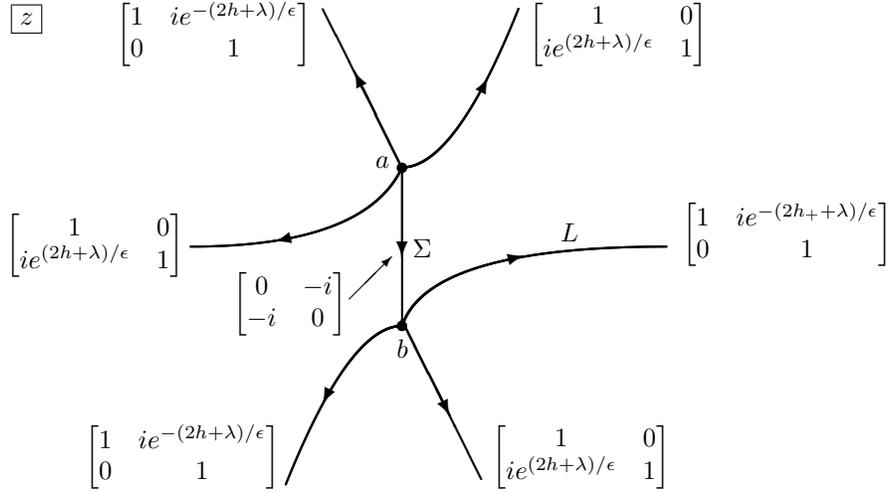
Here we briefly note a few details concerning the Riemann-Hilbert problem analysis 
when $-\pi/3<\arg(x)<\pi/3$.  The definition of ${\bf M}(z;x,\epsilon)$ 
in \eqref{M-def} remains the same.  To define ${\bf N}(z;x,\epsilon)$ via a suitable deformation of the jump contours for $\mathbf{M}(z;x,\epsilon)$, initial segments 
of the rays in Figure \ref{fig:VlocR} with angles $\arg(z)=\pi/3$, $\arg(z)=2\pi/3$, and 
$\arg(-z)=0$ are collapsed together to form part of $\Sigma$, while initial segments of the other 
three rays in Figure \ref{fig:VlocR} are collapsed together to form the remaining part of $\Sigma$.  
(Recall that for $|\arg(-x)|<2\pi/3$ we collapsed together the initial segments of the rays with 
angles $\arg(z)=\pm\pi/3$ and $\arg(z)=0$ to form part of $\Sigma$, and collapsed together the initial segments of the rays with angles $\arg(z)=\pm 2\pi/3$ and $\arg(-z)=0$ to form the remaining part of $\Sigma$.)  Once the contour arc $\Sigma$ is determined, the 
definition of ${\bf O}(z;x,\epsilon)$ in \eqref{O-def-gen0} takes exactly the same form.
The resulting jump matrices for ${\bf O}(z;x,\epsilon)$ are 
shown in Figure \ref{fig:O-jumps}.  

\subsection{Modification of the Riemann-Hilbert analysis for $x$ near $\partial T$}  
For $x$ in the elliptic region $T$ it will be necessary to use a different set of 
contour deformations and introduce additional cuts into the outer model 
problem (see \S\ref{bulk-section}).  On the other hand, for $x\in\mathbb{C}\setminus\overline{T}$ 
but sufficiently close to the boundary $\partial T$ of the elliptic region, the genus-zero contour 
ansatz \emph{nearly} works, but it becomes necessary to insert an additional local 
parametrix to recover uniform decay of the jump matrices to the identity.   
A different parametrix is required depending 
on whether $x$ is close to one of the three corners of $\partial T$ or not.  As shown in 
the second plot in the top row of Figure \ref{genus-zero-failure}, if $x$ is at a corner 
then five arcs of the zero level set of $F(z;x):=\mathrm{Re}(2h(z;x)+\lambda(x))$ meet at one of 
the band endpoints.  The necessary local parametrix is associated to a 
\emph{tritronqu\'ee} solution of the Painlev\'e-I equation, and the correction 
to the leading-order asymptotics is given in terms of the Hamiltonian of this 
function.  If $x$ is on the boundary but not at a corner, then the third 
plot in the top row of Figure \ref{genus-zero-failure} indicates it is necessary to 
insert a local parametrix where two arcs of the zero level set collide at the critical point $z=-S(x)/2$ of $h(z;x)$.  This required parametrix 
is associated to the Hermite orthogonal polynomials and the resulting 
correction to the leading-order asymptotics involves trigonometric functions.  
Once the appropriate parametrix is installed in each case, one obtains approximations for the rational solutions of the Painlev\'e-II equation that are also valid for $x$ slightly inside of the elliptic region $T$.  In other words, the parametrices describe the solution of the connection problem for the rational Painlev\'e-II functions across the Stokes curve $\partial T$.
Full details of the large-$m$ behavior of the rational Painlev\'e-II functions for $x$ near $\partial T$ in both the corner and edge cases 
will be given in a sequel to this paper \cite{Buckingham-rational-crit}.

\subsection{Error analysis}
Let $O_T$ be any open cover of $\overline{T}$ and fix $\delta>0$.  Define a sectorial domain of the exterior of $T$ by setting
\eq
\mathcal{S}_m:=\{x\in\mathbb{C}: |\arg(x)|\le \pi/3-\delta \text{ and } \, \mathrm{Re}(2\mathfrak{c}(x))>\log(m)/m\},
\endeq
where $\mathfrak{c}(x)$ is defined in \eqref{c-def-gen0}.  Define an $m$-dependent region of the complex $x$-plane as follows:
\eq
\label{genus-zero-region-Km}
K_m := (\mathbb{C}\setminus O_T) \cup \mathcal{S}_m \cup e^{2\pi
i/3}\mathcal{S}_m \cup e^{-2\pi i/3}\mathcal{S}_m.
\endeq
In words, $x\in K_m$ means that $x$ is outside of the closure $\overline{T}$ of the elliptic region, and while bounded away from the three corners of $\partial T$, $x$ may approach $\partial T$
elsewhere at a suitably slow rate (distance $\gtrsim\log(m)/m$).

For given $x\in K_m$, let $\mathbb{D}_a$ and $\mathbb{D}_b$ be closed  
disks of radius independent of $m$ containing the points $a(x)$ and $b(x)$, respectively, small enough to be disjoint and to exclude the point 
$z=-S(x)/2$.  The boundaries 
$\partial\mathbb{D}_a$ and $\partial\mathbb{D}_b$ are given a negative (clockwise) orientation. 
Within these disks we will use standard Airy 
parametrices that satisfy the same jump conditions as does ${\bf O}(z;x,\epsilon)$ and that match 
well onto the outer parametrix $\Odot^\text{(out)}(z;x)$.  The construction of 
these parametrices dates back to the work of Deift and Zhou on the Painlev\'e-II equation \cite{Deift:1995}.

\subsubsection{Inner (Airy) parametrix near $z=a$}  
The construction of the Airy parametrix near $z=a$ will depend on whether, given $x$, the 
deformed Riemann-Hilbert problem is 
topologically equivalent to that illustrated
in Figure \ref{fig:O-jumps-pi} (the \emph{Negative-$x$ Configuration}) or 
to that illustrated in Figure \ref{fig:O-jumps} (the \emph{Positive-$x$ 
Configuration}).  

For $x\in K_m$ (as defined in \eqref{genus-zero-region-Km}), we 
can choose $\mathbb{D}_a$ sufficiently small such that 
$(2h(z)+\lambda)^2$ is analytic for $z\in\mathbb{D}_a$, has exactly a 
third-order zero at $z=a$, and is non-zero for 
$z\in\mathbb{D}_a\setminus\{a\}$.  Taking into account the signature 
charts of $F(z)=\text{Re}(2h(z)+\lambda)$ shown in Figure \ref{h-contours}, we see that there is a univalent function $\tau_a:\mathbb{D}_a\to\mathbb{C}$ satisfying the equation
\begin{equation}
\tau_a(z)^3=(2h(z)+\lambda)^2,\quad z\in\mathbb{D}_a,
\end{equation}
and such that (assuming the arcs of the jump contour for $\mathbf{O}$ have been aligned correctly within $\mathbb{D}_a$):
\begin{itemize}
\item In the Negative-$x$ Configuration, 
$\tau_a(z)$ is positive real in $\mathbb{D}_a$ along the arc of the 
jump contour for ${\bf O}(z)$ where the jump matrix 
${\bf V}^{({\bf O})}(z)$ is lower-triangular.  
\item In the Positive-$x$ Configuration, 
$\tau_a(z)$ is positive real in $\mathbb{D}_a$ along the arc of the 
jump contour for ${\bf O}(z)$ where the jump matrix 
${\bf V}^{({\bf O})}(z)$ is upper-triangular.  
\end{itemize}
Note that $\tau_a(z)$ may depend both on $x$ and on the configuration 
of jump matrices.  The conformal mapping $\tau_a$ satisfies 
$\tau_a(a)=0$.  

Define a matrix function $\mathbf{H}_a:\mathbb{D}_a\to SL(2,\mathbb{C})$ 
for $z\in \mathbb{D}_a$ by 
\begin{equation}
\label{eq:g0-Ha}
\mathbf{H}_{a}(z):= \begin{cases} \dot{\mathbf{O}}^{(\mathrm{out})}(z)\bbm 0 & e^{i\pi/4} \\ -e^{-i\pi/4} & 0 \ebm \mathbf{V}^{-1}\tau_a(z)^{-\sigma_3/4} \text{ for the Negative-$x$ Configuration}, \vspace{.1in} ”\\
\dot{\mathbf{O}}^{(\mathrm{out})}(z) e^{i\pi\sigma_3/4} \mathbf{V}^{-1}\tau_a(z)^{-\sigma_3/4} \text{ for the Positive-$x$ Configuration},
\end{cases}
\end{equation}
where $\mathbf{V}$ is the unimodular and unitary matrix defined by 
\eqref{eq:AiryAppendix-EigenvectorMatrix}.  (Where  $\tau_a$ and 
$\dot{\mathbf{O}}^{(\mathrm{out})}$ both have jump discontinuities along 
$\Sigma\cap \mathbb{D}_a$, either boundary value suffices and gives the same 
value for $\mathbf{H}_a$.)  Since $\tau_a(z)^{\sigma_3/4}\mathbf{V}$ satisfies 
the same jump conditions as does $\dot{\mathbf{O}}^{(\mathrm{out})}(z)$ within 
the disk $\mathbb{D}_a$, and since $\mathbf{H}_a=\mathcal{O}((z-a)^{-1/2})$, 
the matrix defined by \eqref{eq:g0-Ha} is an analytic function within its 
disk of definition, and hence its norm is controlled by its size on 
$\partial\mathbb{D}_a$ via the maximum modulus principle.  Note that 
$\mathbf{H}_a(z)$ is independent of $\epsilon$, and also that 
$\det(\mathbf{H}_a(z))=1$ for $z\in \mathbb{D}_a$ by \eqref{Odot-gen0} and 
\eqref{eq:AiryAppendix-EigenvectorMatrix}.  

We define the Airy parametrix for $z\in\mathbb{D}_a$ as 
\begin{equation}
\dot{\mathbf{O}}^{(a)}(z):= \begin{cases} \mathbf{H}_a(z)\epsilon^{\sigma_3/6}\mathbf{A}(\epsilon^{-2/3}\tau_a(z)) \bbm 0 & -e^{i\pi/4} \\ e^{-i\pi/4} & 0 \ebm \text{ for the Negative-$x$ Configuration,} \\
\mathbf{H}_a(z)\epsilon^{\sigma_3/6}\mathbf{A}(\epsilon^{-2/3}\tau_a(z)) e^{-i\pi\sigma_3/4} \text{ for the Positive-$x$ Configuration,} 
\end{cases}
\end{equation}
where the matrix function $\mathbf{A}$ is defined by \eqref{eq:AiryAppendix-ParametrixDef-I}--\eqref{eq:AiryAppendix-ParametrixDef-IV}.  Then
\begin{equation}
\begin{split}
\dot{\mathbf{O}}^{(a)}(z)\dot{\mathbf{O}}^{(\mathrm{out})}(z)^{-1} &= \mathbf{H}_a(z) \epsilon^{\sigma_3/6}\mathbf{A}(\epsilon^{-2/3}\tau_a(z))\mathbf{V}^{-1}[\epsilon^{-2/3}\tau_a(z)]^{-\sigma_3/4}\epsilon^{-\sigma_3/6}\mathbf{H}_a(z)^{-1}\\
&=\mathbb{I}+\begin{bmatrix}\mathcal{O}(\epsilon^2) & \mathcal{O}(\epsilon)\\
\mathcal{O}(\epsilon) & \mathcal{O}(\epsilon^2)\end{bmatrix},\quad z\in\partial \mathbb{D}_a,
\end{split}
\end{equation}
where we have used \eqref{eq:AiryAppendix-ParametrixAsymp} and the fact that 
$\tau_a(z)$ is bounded away from zero on the boundary of $\mathbb{D}_a$.  
Also, from 
\eqref{eq:AiryAppendix-AiryJump-I}--\eqref{eq:AiryAppendix-AiryJump-III} it 
follows that $\dot{\mathbf{O}}^{(a)}(z)$ satisfies exactly the same 
jump conditions within $\mathbb{D}_a$ as does $\mathbf{O}(z)$.

\subsubsection{Inner (Airy) parametrix near $z=b$}  The construction of the 
Airy parametrix near $z=b$ differs from that near $z=a$ due to the presence of 
the branch cut of the function $h$ on the contour $L$.  
The arcs $(\Sigma\cup L)\cap\mathbb{D}_b$ of the jump contour for 
$\mathbf{O}(z)$ divide $\mathbb{D}_b$ into two complementary parts, 
$\mathbb{D}^+_b$ on the left and $\mathbb{D}^-_b$ on the right by orientation 
(see Figures \ref{fig:O-jumps-pi} and \ref{fig:O-jumps}).  Note that 
$\mathbb{D}_b^-$ contains an arc of one jump contour in the Negative-$x$ 
Configuration but arcs of two jump contours in the Positive-$x$ 
Configuration.  

We define a univalent function $\tau_b(z)$ for $z\in\mathbb{D}_b$ so that
\begin{equation}
\tau_b(z)^3=\begin{cases}(2h(z)+\lambda-2\pi i)^2,&\quad z\in\mathbb{D}_b^+,\\
(2h(z)+\lambda+2\pi i)^2,&\quad z\in\mathbb{D}_b^-, 
\end{cases}
\end{equation}
and such that (assuming the arcs of the jump contour for $\mathbf{O}$ have been aligned correctly within $\mathbb{D}_b$):
\begin{itemize}
\item In the Negative-$x$ Configuration, 
$\tau_b(z)$ is positive real in $\mathbb{D}_b$ along the arc of the 
jump contour for ${\bf O}(z)$ where the jump matrix 
${\bf V}^{({\bf O})}(z)$ is upper-triangular.  
\item In the Positive-$x$ Configuration, 
$\tau_b(z)$ is positive real in $\mathbb{D}_b$ along the arc of the 
jump contour for ${\bf O}(z)$ where the jump matrix 
${\bf V}^{({\bf O})}(z)$ is lower-triangular.  
\end{itemize}
Note that $\tau_b(b)=0$.  

Next, define an analytic matrix function 
$\mathbf{H}_b:\mathbb{D}_b\to SL(2,\mathbb{C})$ by ($\mathbf{V}$ is defined 
by \eqref{eq:AiryAppendix-EigenvectorMatrix})
\begin{equation}
\label{eq:g0-Hb}
\mathbf{H}_{b}(z):= \begin{cases} \dot{\mathbf{O}}^{(\mathrm{out})}(z) e^{i\pi\sigma_3/4} \mathbf{V}^{-1}\tau_b(z)^{-\sigma_3/4} \text{ for the Negative-$x$ Configuration}, \vspace{.1in} ”\\
\dot{\mathbf{O}}^{(\mathrm{out})}(z) \bbm 0 & e^{i\pi/4} \\ -e^{-i\pi/4} & 0 \ebm \mathbf{V}^{-1}\tau_b(z)^{-\sigma_3/4} \text{ for the Positive-$x$ Configuration}.
\end{cases}
\end{equation}
The function $\mathbf{H}_b(z)$ is analytic in $\mathbb{D}_b$ even though both 
$\dot{\mathbf{O}}^{(\mathrm{out})}(z)$ and $\tau_b(z)^{-\sigma_3/4}$ have 
jump discontinuities along $\Sigma\cap\mathbb{D}_D$, as seen by a direct 
calculation.  

Now for $z\in\mathbb{D}_b$, define $\Odot^{(b)}(z)$ by ($\mathbf{A}$ is defined by \eqref{eq:AiryAppendix-ParametrixDef-I}--\eqref{eq:AiryAppendix-ParametrixDef-IV})
\begin{equation}
\dot{\mathbf{O}}^{(b)}(z):= \begin{cases} \mathbf{H}_b(z)\epsilon^{\sigma_3/6}\mathbf{A}(\epsilon^{-2/3}\tau_b(z)) e^{-i\pi\sigma_3/4} \text{ for the Negative-$x$ Configuration,} \\
\mathbf{H}_b(z)\epsilon^{\sigma_3/6}\mathbf{A}(\epsilon^{-2/3}\tau_b(z)) \bbm 0 & -e^{i\pi/4} \\ e^{-i\pi/4} & 0 \ebm \text{ for the Positive-$x$ Configuration.} 
\end{cases}
\end{equation}
From \eqref{eq:AiryAppendix-AiryJump-I}--\eqref{eq:AiryAppendix-AiryJump-III}, 
this parametrix satisfies the same jump conditions in $\mathbb{D}_b$ as $\mathbf{O}(z)$, 
and from \eqref{eq:AiryAppendix-ParametrixAsymp}, 
\begin{equation}
\begin{split}
\dot{\mathbf{O}}^{(b)}(z)\dot{\mathbf{O}}^{(\mathrm{out})}(z)^{-1}&=\mathbf{H}_b(z)\epsilon^{\sigma_3/6}\mathbf{A}(\epsilon^{-2/3}\tau_b(z))\mathbf{V}^{-1}[\epsilon^{-2/3}\tau_b(z)]^{-\sigma_3/4}\epsilon^{-\sigma_3/6}\mathbf{H}_b(z)^{-1}\\
&=\mathbb{I}+\begin{bmatrix}\mathcal{O}(\epsilon^{2}) & \mathcal{O}(\epsilon)\\
\mathcal{O}(\epsilon) & \mathcal{O}(\epsilon^{2})\end{bmatrix},\quad z\in \partial \mathbb{D}_b,
\end{split}
\end{equation}
since $\tau_b(z)$ is bounded away from zero for $z\in\partial\mathbb{D}_b$.

\subsubsection{The global parametrix and its use}

We introduce the explicit global parametrix defined by
\eq
\Odot(z;x,\epsilon) := \begin{cases} \Odot^{\rm(out)}(z;x), & z\notin\mathbb{D}_a\cup\mathbb{D}_b, \\ \Odot^{(a)}(z;x,\epsilon), & z\in\mathbb{D}_a, \\ \Odot^{(b)}(z;x,\epsilon), & z\in\mathbb{D}_b \end{cases} 
\endeq
and the corresponding (unknown, because $\mathbf{O}$ is) error matrix 
\eq
\label{E-def-genus0}
{\bf E}(z;x,\epsilon):={\bf O}(z;x,\epsilon)\Odot(z;x,\epsilon)^{-1}.
\endeq
Define the ``mismatch'' jump matrices
\eq
\begin{split}
{\bf V}_a^{({\bf E})}(z;x,\epsilon)&:=\Odot^{(a)}(z;x,\epsilon)\Odot^\text{(out)}(z;x)^{-1}, \quad z\in\partial\mathbb{D}_a,\\
{\bf V}_b^{({\bf E})}(z;x,\epsilon)&:=\Odot^{(b)}(z;x,\epsilon)\Odot^\text{(out)}(z;x)^{-1},\quad z\in\partial\mathbb{D}_b.
\end{split}
\endeq
Now ${\bf E}(z;x,\epsilon)$ is the unique solution to the Riemann-Hilbert 
problem specified by the normalization 
${\bf E}(z;x,\epsilon)=\mathbb{I}+\mathcal{O}(z^{-1})$ 
as $z\to\infty$ and by the jump contour $\Sigma^{({\bf E})}$ and jump matrices 
${\bf V}^{({\bf E})}$ illustrated in Figure~\ref{fig:E-jumps-pi}.
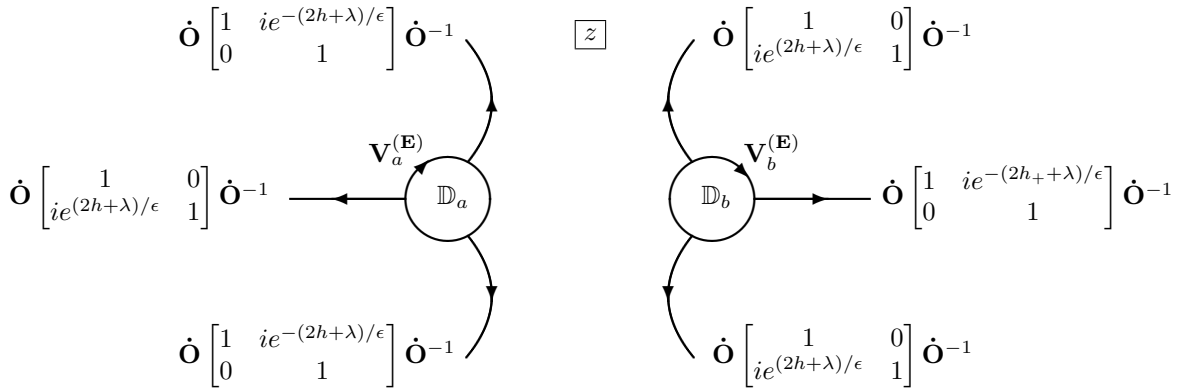
\begin{figure}[h]
\setlength{\unitlength}{2pt}
\begin{center}
\begin{picture}(100,100)(-50,-50)
\put(-1,30){\framebox{$z$}}
\put(-27,-1){$\mathbb{D}_a$}
\put(-40,7){${\bf V}^{({\bf E})}_a$}
\put(23,-1){$\mathbb{D}_b$}
\put(31,7){${\bf V}^{({\bf E})}_b$}
\thicklines
\put(25,0){\circle{16}}
\put(30,6){\vector(1,-1){2}}
\put(-25,0){\circle{16}}
\put(-30.5,5.5){\vector(1,1){2}}
\put(33,0){\line(1,0){22}}
\put(33,0){\vector(1,0){14}}
\put(57,-1){$\Odot\bbm 1 & ie^{-(2h_++\lambda)/\epsilon} \\ 0 & 1 \ebm\Odot^{-1}$}
\qbezier(21.2,7)(12,18.5)(21.5,30)
\put(16.7,18.5){\vector(0,1){1}}
\put(25,29){$\Odot\bbm 1 & 0 \\ ie^{(2h+\lambda)/\epsilon} & 1 \ebm\Odot^{-1}$}
\qbezier(-21.2,7)(-12,18.5)(-21.5,30)
\put(-16.7,18.5){\vector(0,1){1}}
\put(-76,29){$\Odot\bbm 1 & ie^{-(2h+\lambda)/\epsilon} \\ 0 & 1 \ebm\Odot^{-1}$}
\put(-33,0){\line(-1,0){22}}
\put(-33,0){\vector(-1,0){14}}
\put(-108,-1){$\Odot\bbm 1 & 0 \\ ie^{(2h+\lambda)/\epsilon} & 1 \ebm\Odot^{-1}$}
\qbezier(-21.2,-7)(-12,-18.5)(-21.5,-30)
\put(-16.7,-18.5){\vector(0,-1){1}}
\put(-76,-31){$\Odot\bbm 1 & ie^{-(2h+\lambda)/\epsilon} \\ 0 & 1 \ebm\Odot^{-1}$}
\qbezier(21.2,-7)(12,-18.5)(21.5,-30)
\put(16.7,-18.5){\vector(0,-1){1}}
\put(25,-31){$\Odot\bbm 1 & 0 \\ ie^{(2h+\lambda)/\epsilon} & 1 \ebm\Odot^{-1}$}
\end{picture}
\end{center}
\caption{\emph{The jump contour $\Sigma^{({\mathbf E})}=\Sigma^{(\mathbf{E})}(x)$ and jump matrices $\mathbf{V^{(E)}}(z;x,\epsilon)$ 
for $x$ on the negative real axis outside the elliptic region.  A topologically 
equivalent Riemann-Hilbert problem applies for any $x$ outside the elliptic region $T$ in the 
sector $|\arg(-x)|<2\pi/3$.}}
\label{fig:E-jumps-pi}
\end{figure}

Now for $|z|$ sufficiently large, we have
\eq
\begin{split}
\mathbf{Z}_m(\zeta;y)(-\zeta)^{-\sigma_3/\epsilon} & = \epsilon^{-\sigma_3/(3\epsilon)}\mathbf{M}(z;x,\epsilon)(-z)^{-\sigma_3/\epsilon}\epsilon^{\sigma_3/(3\epsilon)} \\
& = \epsilon^{-\sigma_3/(3\epsilon)}e^{\lambda\sigma_3/(2\epsilon)}\mathbf{O}(z;x,\epsilon)(-z)^{-\sigma_3/\epsilon}e^{(2g-\lambda)\sigma_3/(2\epsilon)}\epsilon^{\sigma_3/(3\epsilon)} \\
& = \epsilon^{-\sigma_3/(3\epsilon)}e^{\lambda\sigma_3/(2\epsilon)}{\bf E}(z;x,\epsilon)\Odot^{\rm(out)}(z;x)e^{g\sigma_3/\epsilon}(-z)^{-\sigma_3/\epsilon}e^{-\lambda\sigma_3/(2\epsilon)}\epsilon^{\sigma_3/(3\epsilon)}.
\end{split}
\endeq
Recall that we have $\pu_m(y)=A_{m,12}$, $\pv_m(y)=A_{m,21}$, 
$\mathcal{P}_m(y) = A_{m,22}-B_{m,12}/A_{m,12}$, and 
$\mathcal{Q}_m(y) = -A_{m,11}+B_{m,21}/A_{m,21}$, where 
${\bf Z}_m(-\zeta)^{-\sigma_3/\epsilon} = \mathbb{I}+{\bf A}_m\zeta^{-1} + {\bf B}_m\zeta^{-2}+\mathcal{O}(\zeta^{-3})$.  
If we expand 
\eq
\begin{split}
&{\bf E}(z;x,\epsilon) = \mathbb{I} + \frac{{\bf E}_1(x,\epsilon)}{z} + \frac{{\bf E}_2(x,\epsilon)}{z^2} + \mathcal{O}\left(\frac{1}{z^3}\right), \\
&\Odot^{\rm(out)}(z;x) = \mathbb{I} + \frac{\Odot_1(x)}{z} + \frac{\Odot_2(x)}{z^2} + \mathcal{O}\left(\frac{1}{z^3}\right), \\
&e^{-g(z;x)/\epsilon}(-z)^{1/\epsilon} =: \mathcal{F}(z;x,\epsilon) = 1 + \frac{\mathcal{F}_1(x,\epsilon)}{z} + \mathcal{O}\left(\frac{1}{z^2}\right),
\end{split}
\endeq
then we have (using $z=\epsilon^{1/3}\zeta$ from \eqref{z-def})
\eq
\begin{split}
A_{m,11} & = (E_{1,11}+\dot{O}_{1,11}-\mathcal{F}_1)\epsilon^{-1/3}, \\
A_{m,12} & = (E_{1,12}+\dot{O}_{1,12})\epsilon^{-(2+\epsilon)/(3\epsilon)}e^{\lambda/\epsilon}, \\
A_{m,21} & = (E_{1,21}+\dot{O}_{1,21})\epsilon^{(2-\epsilon)/(3\epsilon)}e^{-\lambda/\epsilon}, \\
A_{m,22} & = (E_{1,22}+\dot{O}_{1,22}+\mathcal{F}_1)\epsilon^{-1/3}, \\
B_{m,12} & = [E_{2,12}+\dot{O}_{2,12}+E_{1,11}\dot{O}_{1,12}+E_{1,12}\dot{O}_{1,22} + (E_{1,12}+\dot{O}_{1,12})\mathcal{F}_1]\epsilon^{-(2+2\epsilon)/(3\epsilon)}e^{\lambda/\epsilon}, \\
B_{m,21} & = [E_{2,21}+\dot{O}_{2,21}+E_{1,22}\dot{O}_{1,21}+E_{1,21}\dot{O}_{1,11} - (E_{1,21}+\dot{O}_{1,21})\mathcal{F}_1]\epsilon^{(2-2\epsilon)/(3\epsilon)}e^{-\lambda/\epsilon}.
\end{split}
\endeq
Here the left-hand sides are evaluated at $y=(m-\tfrac{1}{2})^{2/3}x$ and the 
right-hand sides are evaluated at $x$.
Therefore we see
\eq
\label{um-pm-ito-O-E-genus0}
\begin{split}
\epsilon^{(2+\epsilon)/(3\epsilon)}e^{-\lambda/\epsilon}\pu_m & = E_{1,12}+\dot{O}_{1,12}, \\
\epsilon^{-(2-\epsilon)/(3\epsilon)}e^{\lambda/\epsilon}\pv_m & = E_{1,21}+\dot{O}_{1,21}, \\
\epsilon^{1/3}\pp_m & = E_{1,22}+\dot{O}_{1,22} - \frac{E_{2,12}+\dot{O}_{2,12}+E_{1,11}\dot{O}_{1,12}+E_{1,12}\dot{O}_{1,22}}{E_{1,12}+\dot{O}_{1,12}}, \\
\epsilon^{1/3}\pq_m & = -E_{1,11}-\dot{O}_{1,11} + \frac{E_{2,21}+\dot{O}_{2,21}+E_{1,22}\dot{O}_{1,21}+E_{1,21}\dot{O}_{1,11}}{E_{1,21}+\dot{O}_{1,21}},
\end{split}
\endeq
where the functions $\pu_m$, $\pv_m$, $\pp_m$, and $\pq_m$ are evaluated at 
$y=(m-\tfrac{1}{2})^{2/3}x$ and all other functions are evaluated at $x$.

Recall the set $K_m\subset\mathbb{C}$ defined by \eqref{genus-zero-region-Km}.
\begin{proposition}
\label{error-prop-gen0}
The estimates
\eq
{\bf E}_1(x,\epsilon)=\mathcal{O}(\epsilon) \text{ and } {\bf E}_2(x,\epsilon)=\mathcal{O}(\epsilon)
\label{eq:E1E2estimates-g0}
\endeq
hold uniformly for $x\in K_m$ in the limit $\epsilon\downarrow 0$.
\end{proposition}
\begin{proof}
The proof relies on the $L^2$ theory of small-norm Riemann-Hilbert problems formulated relative to admissible contours $\Sigma^{(\mathbf{E})}$ as outlined in Appendix~\ref{small-norm-app}.  There are two key estimates required to apply this theory:  (i) an estimate of the difference between the jump matrix and the identity that decays to zero with $\epsilon$ in suitable norms and (ii) an estimate of the operator norm of the Cauchy projection operator $\mathcal{C}_-^{\Sigma^{(\mathbf{E})}}$.  Both of these estimates need to hold uniformly with respect to $x$ varying over the $m$-dependent set $K_m$.  An important freedom in establishing the latter estimate is the fact that the jump contour $\Sigma^{(\mathbf{E})}$ has not been completely determined; indeed, the only conditions so far placed upon its arcs are that they lie within certain domains in which strict inequalities hold for $F(z)=\mathrm{Re}(2h(z)+\lambda)$, that they tend to infinity in certain steepest-descent directions, and that the two disk boundaries $\partial\mathbb{D}_a$ and $\partial\mathbb{D}_b$ are bounded away from the corresponding band endpoints $a$ and $b$ respectively while the disks are small enough so as to be disjoint and not include the critical point $z=-S(x)/2$.

We begin by establishing the necessary estimates for the deviation of the jump matrix from the identity.
First, suppose $x\in\mathbb{C}\backslash O_T$.  
We divide $\Sigma^{({\bf E})}$ into a compact part 
$\Sigma^{({\bf E})}_C:=\partial\mathbb{D}_a\cup\partial\mathbb{D}_b$ and a 
non-compact part  
$\Sigma^{({\bf E})}_N:=\Sigma^{({\bf E})}\backslash\Sigma^{({\bf E})}_C$.  
By the construction of the Airy parametrices there is a positive constant $c$, 
independent of $x$ and $\epsilon$, such that 
\eq
\label{error-est1-gen0}
||{\bf V}^{({\bf E})}(z)-\mathbb{I}||_{L^\infty(\Sigma_C^{({\bf E})})} = \mathcal{O}(\epsilon).  
\endeq
By Proposition \ref{genus-zero-proposition} and the fact that the outer parametrix $\dot{\mathbf{O}}^{(\mathrm{out})}(z;x)$ and its inverse are uniformly bounded away from the points $a$ and $b$, there is a positive constant $c$ 
such that 
\eq
\label{error-est2-gen0} 
||{\bf V}^{({\bf E})}(z)-\mathbb{I}||_{L^\infty(\Sigma_N^{({\bf E})})} + ||z^2({\bf V}^{({\bf E})}(z)-\mathbb{I})||_{L^1(\Sigma_N^{({\bf E})})} = \mathcal{O}(e^{-c/\epsilon}).
\endeq

Now suppose 
$x\in \mathcal{S}_m \cup e^{2\pi i/3}\mathcal{S}_m \cup e^{-2\pi i/3}\mathcal{S}_m$, allowing $x$ to lie very close to (only) one of the three arcs of $\partial T$, separated by a distance proportional to $\log(m)/m$.  As $x$ tends to the edge of $T$ from outside, a singularity develops in the zero level set of $F(z;x)$ defined by \eqref{eq:g0-Fdefine} near the point $z=-S(x)/2$ in which one of the domains in which the inequality $F(z;x)>0$ holds pinches off at $z=-S(x)/2$, forcing an arc of the jump contour $\Sigma^{(\mathbf{E})}$ to pass over the saddle point where the necessary strict inequality fails.  To address this difficulty, we first suppose that for $x$ close to (but not exactly on) an arc of $\partial T$ the strict inequality $F(-S(x)/2;x)>0$ holds (that is, the saddle point lies in the correct ``basin'').  This inequality can only fail if $-S(x)/2$ is too close to the contour $\Sigma$ joining $a(x)$ and $b(x)$, and if it does fail it can be easily restored by an appropriate local deformation of  $\Sigma$ (which formerly was taken as a straight line segment purely as a matter of convenience).  Then, we take the arc of $\Sigma^{(\mathbf{E})}$ that has to lie in the nearly pinched-off domain where $F(z;x)>0$ holds to pass exactly over the saddle point $z=-S(x)/2$ locally as a straight-line segment of absolutely fixed length with angle coincident with the steepest ascent direction for $F(z;x)$ away from the saddle.  The near-singularity in the zero level set of $F(z;x)$ has no effect upon the estimate \eqref{error-est1-gen0} because the point $z=-S(x)/2$ has been excluded from the disks $\mathbb{D}_a$ and $\mathbb{D}_b$.  However, the estimate \eqref{error-est2-gen0} has to be modified because the decay of the jump matrix to the identity will be dominated by the behavior near the saddle point $z=-S(x)/2$.  On the segment of $\Sigma^{(\mathbf{E})}$ passing over the saddle point we have
the estimate
\begin{equation}
\left|e^{-(2h(z;x)+\lambda(x))/\epsilon}\right|=e^{-F(z;x)/\epsilon}\le
e^{-F(-S(x)/2;x)/\epsilon} = e^{-2\mathrm{Re}(\mathfrak{c}(x))/\epsilon}<2\epsilon
\end{equation}
since $2\mathrm{Re}(\mathfrak{c}(x))>\log(m)/m>\epsilon\log((2\epsilon)^{-1})$ for $\epsilon>0$ by hypothesis (on how close $x$ may approach $\partial T$).  It follows that in place of \eqref{error-est2-gen0} we have instead the estimate
\eq
\label{error-est3-gen0} 
||{\bf V}^{({\bf E})}(z)-\mathbb{I}||_{L^\infty(\Sigma_N^{({\bf E})})} + ||z^2({\bf V}^{({\bf E})}(z)-\mathbb{I})||_{L^1(\Sigma_N^{({\bf E})})} = \mathcal{O}(\epsilon).
\endeq
In other words, by allowing $x$ to approach within a distance proportional to $\log(m)/m$ from the boundary of $T$, the exponential decay of the jump matrices to the identity away from the Airy disks is compromised exactly to the point at which the discrepancy balances in magnitude the usually dominant error contribution from the disk boundaries themselves.  Any closer approach would lead to estimates worse than $\mathcal{O}(\epsilon)$ in the statement of the proposition (see \eqref{eq:E1E2estimates-g0}).

Now by Proposition \ref{prop:small-norm-moments-bound} in Appendix 
\ref{small-norm-app}, ${\bf E}_1=\mathcal{O}(K_{\Sigma^{({\bf E})}}'\epsilon)$ 
and ${\bf E}_2=\mathcal{O}(K_{\Sigma^{({\bf E})}}''\epsilon)$, where 
$K_{\Sigma^{({\bf E})}}'$ and $K_{\Sigma^{({\bf E})}}''$ are defined in 
\eqref{small-norm-K'} and \eqref{small-norm-K''}, respectively.  While this 
establishes the desired $\mathcal{O}(\epsilon)$ error bounds pointwise in 
$x$, we require uniformity for $x\in K_m$.  It remains to establish that 
$\|\mathcal{C}_-^{{\Sigma}^{({\bf E})}}\|_{L^2(\Sigma^{({\bf E})})\circlearrowleft}$
is uniformly bounded as $x$ varies in $K_m$.  For every $x\in K_m$, the contour $\Sigma^{(\mathbf{E})}$ may be assumed to be \emph{admissible} in the sense of the theory explained in Appendix~\ref{small-norm-app}, and while this fact implies that $\mathcal{C}_-^{\Sigma^{(\mathbf{E})}}$ is a bounded operator for each $x$, it is not sufficient to give a bound for the operator norm that is independent of $x\in K_m$.

We first claim that each point $x_0$ in the genus zero region $\mathbb{C}\setminus\overline{T}$ (which contains $K_m$ for each $m$) has an open neighborhood $\mathscr{O}_{x_0}$ such that for all $x\in\mathscr{O}_{x_0}$ the jump contour $\Sigma^{(\mathbf{E})}$ may be taken to be exactly the same as for $x=x_0$ with no change in the error estimates \eqref{error-est1-gen0}--\eqref{error-est2-gen0}.  Therefore, if $K$ is a $m$-independent compact subset of $K_m$, the open covering $K\subset \bigcup_{x\in K}\mathscr{O}_x$ has a finite sub-covering, and so a finite number of admissible contours $\Sigma^{(\mathbf{E})}$ suffice to analyze the error for values of $x\in K$.  But $K_m$ also contains points $x$ that approach the boundary of $T$ as $m$ increases, and furthermore $K_m$ is  unbounded, allowing $x$ to approach infinity.  Both of these  cases require special care to control the norm of $\mathcal{C}_-^{\Sigma^{(\mathbf{E})}}$.

As $x$ approaches $\partial T$, the issue is that one arc of 
$\Sigma^{({\bf E})}$ must pass through a narrow isthmus where 
$F(z;x)=\mathrm{Re}(2h(z;x)+\lambda(x))>0$.  This isthmus is pinched off at 
the saddle point $z=-S(x)/2$ for $x\in \partial T$ (see Figure \ref{genus-zero-failure}).  
Since the pinch-off point is moving as $x$ varies along $\partial T$ it is 
not possible to take the contour $\Sigma^{(\mathbf{E})}$ to be locally independent of $x$ near the boundary.  However, we can exploit the fact that the norm of $\mathcal{C}_-^{\Sigma^{(\mathbf{E})}}$ is unchanged if $\Sigma^{(\mathbf{E})}$ is subjected to a rigid motion in the plane (translation plus rotation).  For any $x\in\mathbb{C}\setminus T$ bounded away from the three corner points of $T$ (but otherwise $x\in\partial T$ is allowed), let $\psi(x)$ denote the angle of steepest ascent toward $z=\infty$ for $F(z;x)$ at the saddle point $z=-S(x)/2$.  We claim that each non-corner point $x_0$ of the boundary $\partial T$ has an open neighborhood $\mathscr{O}_{x_0}$ such that if for each $x\in\mathscr{O}_{x_0}$ the contour $\Sigma^{(\mathbf{E})}=\Sigma^{(\mathbf{E})}(x)$
is taken to be given by the rigid congruence
\begin{equation}
\Sigma^{(\mathbf{E})}(x):=\left(\Sigma^{(\mathbf{E})}(x_0)+\frac{1}{2}S(x_0)\right)e^{i(\psi(x)-\psi(x_0))}-\frac{1}{2}S(x),\quad x\in\mathscr{O}_{x_0},
\end{equation}
where $\Sigma^{(\mathbf{E})}(x_0)$ denotes an admissible contour\footnote{Note however that on the contour $\Sigma^{(\mathbf{E})}(x_0)$ no useful analog of the estimate \eqref{error-est2-gen0} holds for the jump matrix because there is no decay of $e^{-(2h(z)+\lambda)/\epsilon}$ at the saddle point on the jump contour when $x_0\in\partial T$.} for the error problem when $x=x_0$, then the estimates \eqref{error-est1-gen0} and \eqref{error-est3-gen0} hold for the jump matrix on
$\Sigma^{(\mathbf{E})}(x)$ if also $x\in K_m$.  Indeed, it is easy to see that $\Sigma^{(\mathbf{E})}(x)$ so-defined includes a small straight-line segment passing through the narrow isthmus over the saddle point at exactly the steepest ascent angle, and the parts of the contour away from the moving saddle point are only slightly deformed from $\Sigma^{(\mathbf{E})}(x_0)$ if $|x-x_0|$ is small enough, as can be enforced by choosing the diameter of $\mathscr{O}_{x_0}$ sufficiently small (independent of $m$).  Given any compact sub-arc $C$ of $\partial T$ we have the open covering $C\subset\bigcup_{x_0\in C}\mathscr{O}_{x_0}$ which has a finite sub-covering, and so in this way all points $x$ near the boundary $\partial T$ can be analyzed using rigid congruences of
a finite number of admissible contours $\Sigma^{(\mathbf{E})}(x_0)$; since the norm of the operator $\mathcal{C}_-^{\Sigma^{(\mathbf{E})}}$ depends only on the equivalence class of the contour under rigid congruence, essentially a finite number of admissible contours again suffices to study $x\in K_m$ approaching $\partial T$ as $m\to+\infty$.

Finally, we consider $x\in K_m$ tending to infinity.  In this case, the contour $\Sigma^{(\mathbf{E})}$ has to expand as $x\to\infty$ because both $a(x)$ and $b(x)$ are asymptotically proportional to $|x|^{1/2}$.  As $x\to\infty$ at a fixed angle $\varphi:=\arg(x)$, we have the following limit:
\begin{equation}
\lim_{x\to\infty}|x|^{-1}h'(|x|^{1/2}w;x)=\frac{3}{2}w\tilde{r}(w;\varphi),
\label{eq:hprime-limit}
\end{equation}
where $\tilde{r}(w;\varphi)$ is the function analytic for all $w$ not on the straight line connecting $\pm ie^{i\varphi/2}\sqrt{\tfrac{2}{3}}$ that satisfies $\tilde{r}(w;\varphi)^2=w^2+\tfrac{2}{3}e^{i\varphi}$ and $\tilde{r}(w;\varphi)=w+\mathcal{O}(1)$ as $w\to\infty$.  The existence of this limit shows that for each angle $\varphi$ there is a limiting rescaled version, denoted $\Sigma_\varphi$, of the contour $\Sigma^{(\mathbf{E})}$ consisting of six unbounded fixed arcs in the $w$-plane lying in basins defined by the limiting function defined by \eqref{eq:hprime-limit} together with two disjoint circles centered at the points $w=\pm ie^{i\varphi/2}\sqrt{\tfrac{2}{3}}$ and not containing the origin (the critical point of the rescaled $h$, and the limiting image in the $w$-plane of $-S(x)/2$).  We claim that for each angle $\varphi$ there exists some $\delta(\varphi)>0$ and $R_0(\varphi)<\infty$ such that for each $x$ in the domain $\mathscr{O}_\varphi:=\{x\in\mathbb{C}:|x|>R_0(\varphi),\;|\arg(x)-\varphi|<\delta(\varphi)\}$, the contour $\Sigma^{(\mathbf{E})}(x)$ may be defined as
\begin{equation}
\Sigma^{(\mathbf{E})}(x):=|x|^{1/2}\Sigma_\varphi,\quad x\in\mathscr{O}_\varphi,
\end{equation}
that is, a simple dilation of the fixed admissible contour $\Sigma_\varphi$, and the estimates \eqref{error-est1-gen0}--\eqref{error-est2-gen0} will hold on such a contour.  In particular it is no problem that the two Airy disks $\mathbb{D}_a$ and $\mathbb{D}_b$ in the contour $\Sigma^{(\mathbf{E})}$ are expanding as $x\to\infty$ as long as they remain disjoint and exclude the critical point $-S(x)/2$, which will be the case for large enough $|x|$.  Covering the neighborhood of $z=\infty$ via a union over the compact circle $S^1$ of angles $\varphi$ and extracting a finite sub-covering, we handle all unbounded values of $x$ with the use of dilates of a finite number of admissible contours $\Sigma_\varphi$.  Since the norm of the operator $\mathcal{C}_-^{\Sigma^{(\mathbf{E})}}$ depends only on the equivalence class of the contour under dilations, we obtain a uniform upper bound for the norm as a maximum over a finite number of values.

With the uniform estimates \eqref{error-est1-gen0} and  \eqref{error-est3-gen0} (which subsumes the sharper \eqref{error-est2-gen0} but is more general) for the difference of the jump matrix $\mathbf{V}^{(\mathbf{E})}$ from the identity, and the uniform control of the norm of the Cauchy projection $\mathcal{C}_-^{\Sigma^{(\mathbf{E})}}$ obtained as a maximum over a finite number of values, the proof is complete.
\end{proof}

With Proposition~\ref{error-prop-gen0} in hand, we can now prove our first main result describing how the Painlev\'e-II rational functions behave as $m\to +\infty$ for sufficiently large independent variable $y$ (so as to place $x$ in the region $K_m$).  Recall the region $K_m$ defined by \eqref{genus-zero-region-Km} and the analytic function $S:\mathbb{C}\setminus\Sigma_S\to\mathbb{C}$ defined as the solution of the cubic equation \eqref{cubic-equation} with asymptotic behavior for large $x$ given by \eqref{S-large-x}.
\begin{theorem}
\label{main-genus-zero-thm}
The rational Painlev\'e-II functions obey the asymptotic formulae
\eq
\begin{split}
m^{-2m/3}e^{-m\lambda(x)}\pu_m\left(\left(m-\frac{1}{2}\right)^{2/3}x\right) & = \dot{\pu}(x)  + \mathcal{O}\left(\frac{1}{m}\right), \\ 
m^{2(m-1)/3}e^{m\lambda(x)}\pv_m\left(\left(m-\frac{1}{2}\right)^{2/3}x\right) & = \dot{\pv}(x)  + \mathcal{O}\left(\frac{1}{m}\right), \\ 
m^{-1/3}\pp_m\left(\left(m-\frac{1}{2}\right)^{2/3}x\right) & = \dot{\pp}(x)+\mathcal{O}\left(\frac{1}{m}\right), \\
m^{-1/3}\pq_m\left(\left(m-\frac{1}{2}\right)^{2/3}x\right) & = \dot{\pq}(x)+\mathcal{O}\left(\frac{1}{m}\right),
\end{split}
\endeq
valid as $m\to +\infty$ uniformly for $x\in K_m$, 
where $\lambda:\mathbb{C}\setminus\overline{T}\to\mathbb{C}$ is given explicitly in terms of $S$ by \eqref{eq:lambdadef} and where
\eq
\label{g0-pudot-pvdot-ppdot-pqdot}
\dot{\pu}(x):=e^{xS(x)/6}, \quad \dot{\pv}(x):=\frac{1}{3S(x)}e^{-xS(x)/6}, \quad\dot{\pp}(x):=-\frac{1}{2}S(x), \quad \text{and} \quad \dot{\pq}(x):=\frac{1}{2}S(x).
\endeq
\end{theorem}

\begin{proof}
From Proposition \ref{error-prop-gen0} we see 
\eq
\label{pu-ito-Odot-genus0}
\begin{split}
\epsilon^{(2+\epsilon)/(3\epsilon)}e^{-\lambda/\epsilon}\pu_m = \dot{O}_{1,12} + \mathcal{O}(\epsilon), \quad \quad & \epsilon^{-(2-\epsilon)/(3\epsilon)}e^{\lambda/\epsilon}\pv_m = \dot{O}_{1,21} + \mathcal{O}(\epsilon), \\
\epsilon^{1/3}\pp_m = \dot{O}_{1,22} - \frac{\dot{O}_{2,12}}{\dot{O}_{1,12}} + \mathcal{O}(\epsilon), \quad \quad & \epsilon^{1/3}\pq_m = -\dot{O}_{1,11} + \frac{\dot{O}_{2,21}}{\dot{O}_{1,21}} + \mathcal{O}(\epsilon),
\end{split}
\endeq
with $\pu_m$, $\pv_m$, $\pp_m$, and $\pq_m$ evaluated at $y=(m-\tfrac{1}{2})^{2/3}x$ and all other functions evaluated at $x$.
A straightforward large-$z$ asymptotic expansion of the explicit formula \eqref{Odot-gen0} for the outer parametrix $\dot{\mathbf{O}}^{(\mathrm{out})}(z)$ yields
\eq
\label{O1,12-etc-genus0}
\dot{O}_{1,12}(x) = \dot{O}_{1,21}(x) = \frac{\Delta(x)}{4}, \quad \dot{O}_{1,11}(x) = \dot{O}_{1,22}(x) = 0, \quad \dot{O}_{2,12}(x) = \dot{O}_{2,21}(x) = \frac{S(x)\Delta(x)}{8},
\endeq
from which the desired asymptotic formulae for $m^{-1/3}\pp_m$ and 
$m^{-1/3}\pq_m$ follow immediately upon recalling $\epsilon=(m-\tfrac{1}{2})^{-1}$.
Furthermore, 
\begin{equation}
\epsilon^{(2+\epsilon)/(3\epsilon)}e^{-\lambda(x)/\epsilon}=(m-\tfrac{1}{2})^{-2m/3}e^{-m\lambda(x)}e^{\lambda(x)/2},
\end{equation}
and since $(m-\tfrac{1}{2})^{-2m/3}=m^{-2m/3}e^{1/3}(1+\mathcal{O}(m^{-1}))$ as $m\to\infty$, we obtain
\begin{equation}
m^{-2m/3}e^{-m\lambda(x)}\pu_m = e^{-1/3-\lambda(x)/2}\frac{\Delta(x)}{4} + \mathcal{O}(m^{-1})
\end{equation}
from which the desired asymptotic formula for $m^{-2m/3}e^{-m\lambda}\pu_m$ follows with the help of \eqref{cubic-equation}, 
\eqref{eq:lambdadef}, and the expression for $\Delta(x)$ in terms of $S(x)$ following from the second equation in \eqref{ab-conditions} upon taking the appropriate square root.  The formula for $m^{2(m-1)/3}e^{m\lambda}\pv_m$ 
follows from completely analogous calculations.
\end{proof}

Since $S(x)$ is real and negative only for $x$ on the positive 
real axis $\mathbb{R}_+$, it follows from the formula \eqref{eq:lambdadef} that $\lambda:\mathbb{C}\setminus\overline{T}\to\mathbb{C}$ has a branch cut along the positive real axis; however its real part is harmonic throughout its domain of definition and the imaginary part has a constant additive jump of $2\pi i$ across its branch cut, and hence $e^{m\lambda(x)}$ is well-defined and analytic for every $m\in\mathbb{Z}_+$.  One can easily check that for each fixed $m\in\mathbb{Z}_+$, 
\begin{equation}
e^{m\lambda(x)}=(-6x)^m(1+\mathcal{O}(x^{-1})),\quad x\to\infty.
\end{equation}
Contour plots of the real and imaginary parts (the latter modulo $2\pi$) of $\lambda:\mathbb{C}\setminus \overline{T}\to\mathbb{C}$ are illustrated (along with a closely related function to be defined for $x\in T$ in \S\ref{bulk-section}) in Figure~\ref{fig:Lambdalambda}.  
Likewise, contour plots of the real and imaginary parts of the analytic approximating function $\dot{\pp}:\mathbb{C}\setminus\overline{T}\to\mathbb{C}$ that gives the leading-order behavior of $m^{-1/3}\pp_m$ for $x$ outside of $T$ are shown (along with plots of a closely related function to be defined for $x\in T$ in \S\ref{bulk-section}) in Figure~\ref{fig:g1-weak-limit}.
We plot $m^{-2m/3}e^{-m\lambda(x)}\pu_m((m-\tfrac{1}{2})^{2/3}x)$ and 
$m^{-1/3}\pp_m((m-\tfrac{1}{2})^{2/3}x)$ for $m=3$, 5, and 10 and real $x$ outside the 
elliptic region $T$ along with their leading asymptotic approximations $\dot{\pu}(x)$ and $\dot{\pp}(x)$ respectively in Figures~\ref{exact-vs-asymp-um-genus0} and \ref{exact-vs-asymp-pm-genus0}.

\section{Analysis for $x$ inside the elliptic region $T$}
\label{bulk-section}
To analyze the rational Painlev\'e-II functions outside of the elliptic region $T$ it was convenient to introduce the new independent variable $x=(m-\tfrac{1}{2})^{-2/3}y$.  The same will be true now that we will consider the behavior inside of the elliptic region, however it will be useful to generalize a bit by writing $x\in T$ in the form
\eq
\label{x-ito-x0-gen1}
x=x_0+\epsilon w, 
\endeq
where $x_0$ is a point in $T$ and $w$ is bounded as 
$\epsilon\to 0$.  We think of the asymptotic formulae that will result as depending functionally on $w$ with 
$x_0\in T$ as an additional parameter.  For each $x$ and $\epsilon$ there are of course infinitely many 
ways to express $x$ in the form \eqref{x-ito-x0-gen1}, but from uniqueness 
of the solution to the Riemann-Hilbert Problem \ref{rhp:DSlocalII} and the validity of the error estimates that we will present below, the 
resulting formulae for the leading-order terms of $\mathcal{P}_m$ and 
$\mathcal{U}_m$ will provide valid approximations of the corresponding functions at a given value of $x\in T$, regardless  of the choice of $x_0$ and $w$ consistent with \eqref{x-ito-x0-gen1}; see Corollary~\ref{cor-kernel} for a more precise statement.  The reason 
for taking this approach is that the asymptotic formulae we will obtain will not depend on $x_0$ analytically due to the presence of nonzero $\overline{\partial}$-derivatives (see \eqref{eq:g1-dbarPi} below). On the other hand, for each fixed $x_0$, key ingredients in our approximating formulae will depend meromorphically on $w$ (so the only issues are with isolated poles).  This meromorphic dependence on $w$ therefore gives a better local picture of the asymptotic behavior of a family of functions, each of which is rational.

\subsection{Preliminary steps}
The contour of the Riemann-Hilbert problem for $\mathbf{M}(z)$ consists of six straight rays meeting at the origin, and by standard analytic substitutions the contour may be deformed into one that is topologically similar but with a different self-intersection point and with the rays deformed into non-intersecting curves that tend to infinity along the asymptotic directions $\arg(z)=n\pi/3$, $n=-2,-1,0,1,2,3$.  Given four distinct points in the complex plane $A$, $B$, $C$, and $D$, we assume that the contour is deformed in this way such that $C$ is the self-intersection point, the arc going to infinity with angle $\arg(z)=2\pi/3$ passes through the point $A$, the arc going to infinity with angle $\arg(z)=0$ passes through the point $D$, and the arc going to infinity with angle $\arg(z)=-2\pi/3$ passes through the point $B$.  See Figure~\ref{fig:deform-g1},
\begin{figure}[h]
\begin{center}
\includegraphics{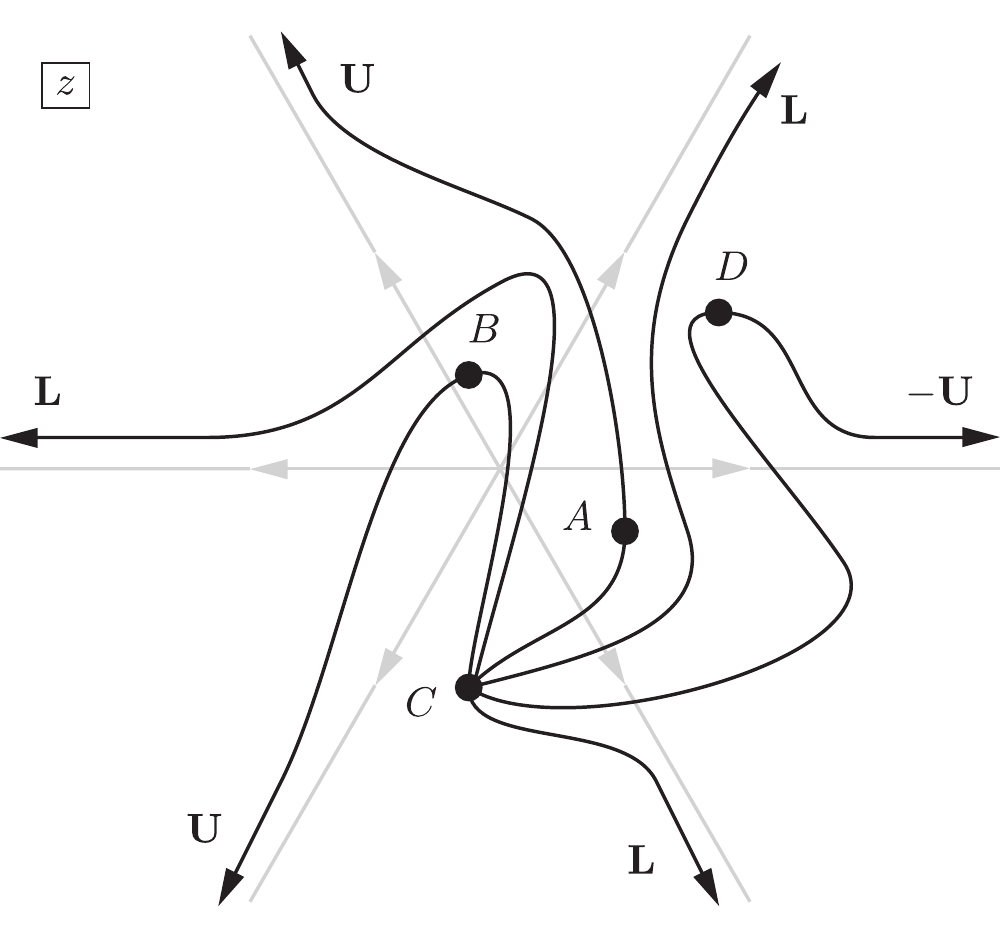}
\end{center}
\caption{\emph{The initial deformation of the jump contour for $\mathbf{M}$ given four distinct points $A$, $B$, $C$, and $D$.  The jump matrix is as indicated on each of the six ``rays''.  The reader should not be misled by the qualitative shape of the contour or the positions of the points $A$, $B$, $C$, and $D$; it will turn out that the actual picture is far more regular although this cannot be presumed at such an early stage of the argument.}}
\label{fig:deform-g1}
\end{figure}
where we use the following notation for the jump matrices:
\begin{equation}
\mathbf{L}:=\begin{bmatrix}1 & 0\\ie^{\theta/\epsilon} & 1\end{bmatrix}\quad\text{and}\quad
\mathbf{U}:=\begin{bmatrix}1 & ie^{-\theta/\epsilon}\\0 & 1\end{bmatrix}.
\end{equation}
Next, we choose a simple closed contour encircling $C$ and passing through the points $A$, $B$, and $D$ as illustrated with a dotted curve in Figure~\ref{fig:M-to-N-g1},
\begin{figure}[h]
\begin{center}
\includegraphics{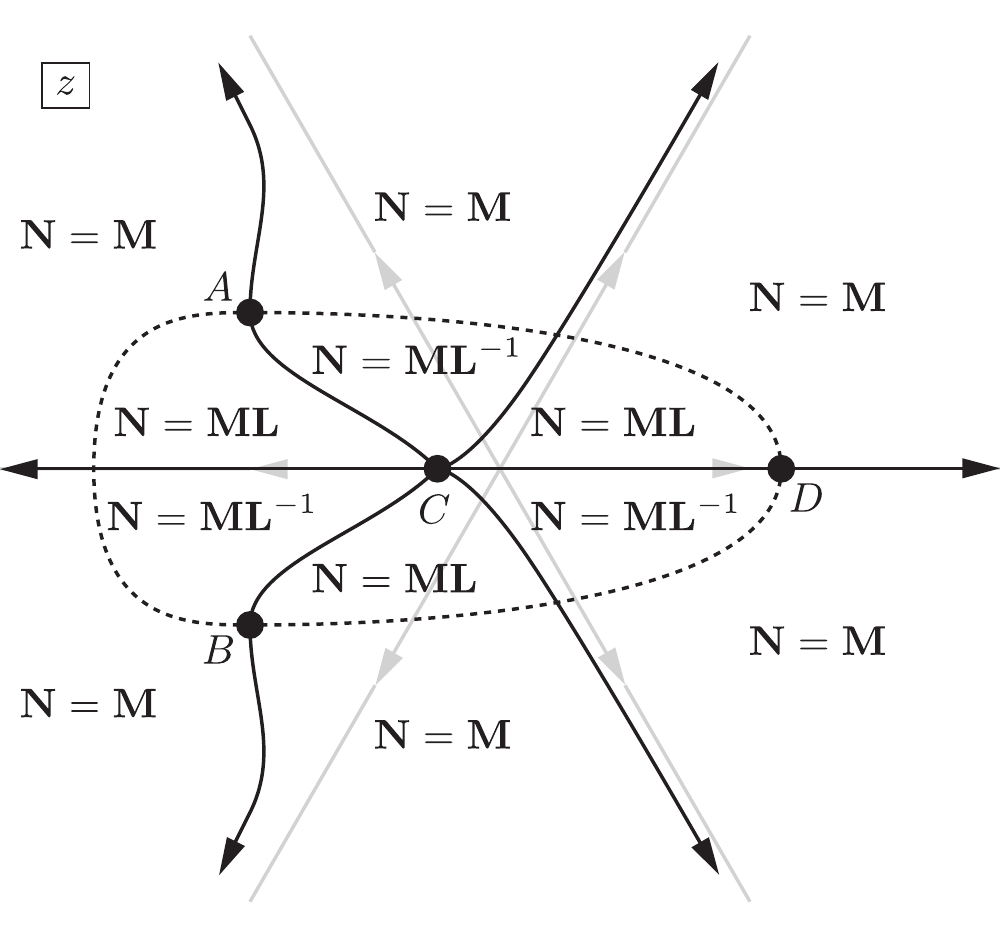}
\end{center}
\caption{\emph{The explicit and invertible relation between $\mathbf{M}$ and $\mathbf{N}$, illustrated in
a configuration of the points $A$, $B$, $C$, and $D$ resembling what will actually be the case for sufficiently small $x_0\in\mathbb{R}$.}}
\label{fig:M-to-N-g1}
\end{figure}
and define a new unknown matrix $\mathbf{N}(z)$ explicitly in terms of $\mathbf{M}(z)$ as shown.  The jump contour for the matrix $\mathbf{N}(z)$ and the corresponding jump matrix are illustrated in Figure~\ref{fig:N-jumps-g1}.  
\begin{figure}[h]
\begin{center}
\includegraphics{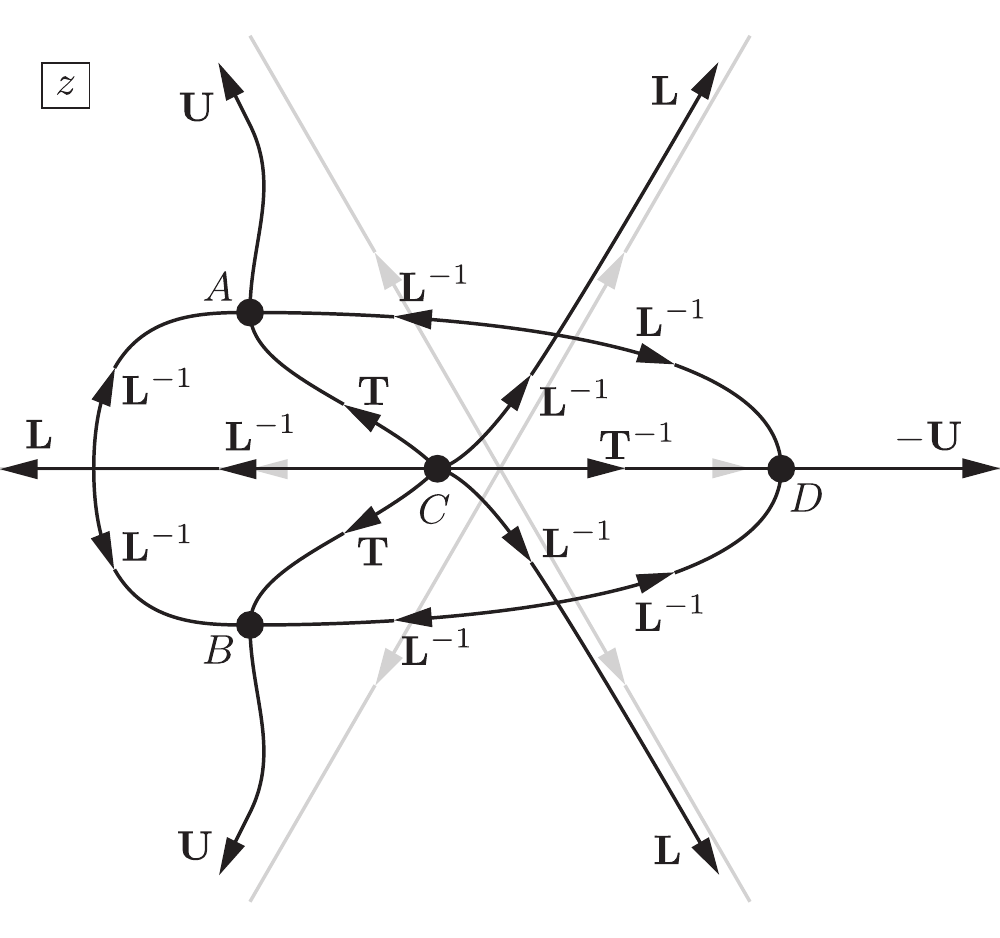}
\end{center}
\caption{\emph{The jump contour $\Sigma^{(\mathbf{N})}$ for the matrix $\mathbf{N}(z)$, and the jump matrix as defined on each oriented arc of the jump contour.  This is also the jump contour $\Sigma^{(\mathbf{O})}$ for the matrix $\mathbf{O}(z)$ defined in terms of $\mathbf{N}(z)$ below by \eqref{eq:g1-OfromN}.}}
\label{fig:N-jumps-g1}
\end{figure}
Here, the matrix $\mathbf{T}$ is defined as
\begin{equation}
\mathbf{T}:=\begin{bmatrix}0 & ie^{-\theta/\epsilon}\\ie^{\theta/\epsilon} & 0\end{bmatrix}.
\end{equation}
The next step is to introduce an appropriate $g$-function that will ultimately allow us to neglect the jump matrix on all but the three arcs on which the jump matrix is off-diagonal, $\mathbf{T}$ or $\mathbf{T}^{-1}$.

We emphasize that at this point $\mathbf{M}(z)=\mathbf{M}(z;x,\epsilon)$ and $\mathbf{N}(z)=\mathbf{N}(z;x,\epsilon)$ depend only on the combination $x=x_0+\epsilon w$.  Henceforth we will be making substitutions that separate the roles of $x_0$ and $w$.

\subsection{Definition and properties of the $g$-function}
\label{genus-one-g}
Let $\Sigma$ denote the contour that is the closure of the union of the oriented arcs $\overrightarrow{CA}$, $\overrightarrow{CB}$, and $\overrightarrow{CD}$.  Given $\Sigma$ (and hence the four points $A$, $B$, $C$, and $D$) let $R(z)$ denote the function analytic for $z\in\mathbb{C}\setminus\Sigma$ that satisfies the two conditions
\begin{equation}
R(z)^2=(z-A)(z-B)(z-C)(z-D)\quad\text{and}\quad R(z)=z^2+\mathcal{O}(z),\quad z\to\infty.
\end{equation}
At each point $z$ of the three open arcs whose union becomes $\Sigma$ under closure, $R(z)$ has two well-defined boundary values $R_\pm(z)$, and these satisfy $R_+(z)+R_-(z)=0$.  Given also a point $x_0\in\mathbb{C}$, define the related function
\begin{equation}
G'(z):=\frac{1}{2}\theta'(z;x_0)-\frac{3}{2}R(z)=\frac{3}{2}z^2+\frac{1}{2}x_0-\frac{3}{2}R(z),\quad z\in\mathbb{C}\setminus\Sigma.
\end{equation}
We assume now that $A$, $B$, $C$, $D$, and $x_0$ are related by the three \emph{moment conditions}
\begin{equation}
\begin{split}
M_1(A,B,C,D)&=0,\\
M_2(A,B,C,D)&=-\frac{4}{3}x_0,\\
M_3(A,B,C,D)&=4,
\end{split}
\label{eq:g1-moments}
\end{equation}
where
\begin{equation}
M_p(A,B,C,D):=A^p+B^p+C^p+D^p,\quad p=1,2,3.
\label{eq:g1-moments-define}
\end{equation}
These conditions imply that,
for large $z$, we have the expansion
\begin{equation}
G'(z)=\frac{1}{z}+\mathcal{O}\left(\frac{1}{z^2}\right),\quad z\to\infty.
\end{equation}
The presence of a residue at infinity implies that antiderivatives of $G'(z)$ will be logarithmically branched at infinity.  Let $L$ denote the unbounded arc of the jump contour for $\mathbf{N}(z)$ shown in Figure~\ref{fig:N-jumps-g1} that is attached to the point $D$ (the only arc along which the jump matrix is $-\mathbf{U}$), and let $\log^{(L)}(D-z)$ denote the branch of $\log(D-z)$ with branch cut $L$ and that agrees asymptotically with the principal branch of $\log(-z)$ for large negative $z$ (recall that $L$ tends to infinity along the positive real axis).  We then define an antiderivative of $G'(z)$ by the formula
\begin{equation}
G(z):=\log^{(L)}(D-z)+\int_\infty^z \left[G'(\zeta)-\frac{1}{\zeta-D}\right]\,d\zeta,\quad z\in\mathbb{C}\setminus (\Sigma\cup L).
\label{eq:g1-G}
\end{equation}
Finally, set 
\begin{equation}
H(z):=\frac{1}{2}\theta(z;x_0) -G(z),\quad z\in\mathbb{C}\setminus (\Sigma\cup L).
\label{eq:g1-H}
\end{equation}
Note that $H'(z)$ extends analytically to $L$ and that
\begin{equation}
H'(z)=\frac{3}{2}R(z),\quad z\in\mathbb{C}\setminus\Sigma.
\end{equation}
The functions $G$ and $H$ have a number of elementary properties that are immediate consequences of their definitions and that we summarize in the following proposition.
\begin{proposition}
Suppose that $A$, $B$, $C$, and $D$ are related with $x_0$ by the moment conditions \eqref{eq:g1-moments} (so that $G(z)$ and $H(z)$ are well-defined by \eqref{eq:g1-G} and \eqref{eq:g1-H} respectively).  Then the relation 
\begin{equation}
\frac{d}{dz}\left[H_+(z)+H_-(z)\right]=\frac{d}{dz}\left[\theta(z;x_0)-G_+(z)-G_-(z)\right]=0
\end{equation}
holds as an identity along each open arc of $\Sigma$.  Also, 
\begin{equation}
G_+(z)-G_-(z)=-2\pi i,\quad z\in L,
\label{eq:g1-G-jump}
\end{equation}
and finally,
\begin{equation}
G(z)=\log(-z) + \mathcal{O}\left(\frac{1}{z}\right),\quad z\to\infty,
\label{eq:g1-G-asymp}
\end{equation}
assuming that  $L$ agrees with the positive real axis for sufficiently large $|z|$.
\label{prop:g1-G-properties}
\end{proposition}
Under the conditions of Proposition~\ref{prop:g1-G-properties}, we may therefore identify with each of the three arcs of $\Sigma$ a finite complex constant (independent of $z$), namely the value of $-(H_+(z)+H_-(z))$.  Since $H$ is continuous at $z=A$ and $z=B$,  these three constants have the values $-2H(A)$, $-2H(B)$, and 
\begin{equation}
\Lambda:=-(H_+(D)+H_-(D))=-2H_\pm(D)\pm 2\pi i,
\label{eq:g1-Lambda}
\end{equation}
where we have used \eqref{eq:g1-G-jump}.
Let
\begin{equation}
\Phi_+:=-i(\Lambda+2H(A))\quad\text{and}\quad\Phi_-:=-i(\Lambda+2H(B)).
\end{equation}
In addition to the moment conditions \eqref{eq:g1-moments}, we also wish to impose the following additional conditions (sometimes known as \emph{Boutroux conditions}) relating $A$, $B$, $C$, and $D$ with $x_0$:
\begin{equation}
\begin{split}
\mathrm{Im}(\Phi_+)&=0,\\
\mathrm{Im}(\Phi_-)&=0.
\end{split}
\label{eq:g1-Boutroux}
\end{equation}

\subsection{Dependence of the branch points $A$, $B$, $C$, and $D$ on $x_0$}
We now consider the possibility of determining the branch points $A$, $B$, $C$, and $D$ from the moment equations \eqref{eq:g1-moments} and the Boutroux conditions \eqref{eq:g1-Boutroux} given $x_0\in T$.  
First, we observe that it is easy to find values of $A$, $B$, $C$, and $D$ that simultaneously satisfy the moment conditions \eqref{eq:g1-moments} and the Boutroux conditions \eqref{eq:g1-Boutroux} in the special case that $x_0=0$.  Indeed, 
the values
\eq
A=\sqrt[3]{\frac{4}{3}}e^{2\pi i/3}, \quad B=\sqrt[3]{\frac{4}{3}}e^{-2\pi i/3}, \quad C=0, \quad D=\sqrt[3]{\frac{4}{3}}
\label{eq:zeroconfig}
\endeq
obviously satisfy the moment conditions \eqref{eq:g1-moments} for $x_0=0$.  
We take the corresponding contour $\Sigma$ to be constructed from the union of three radial straight-line segments as illustrated in 
Figure \ref{fig:R-branches-x0}.
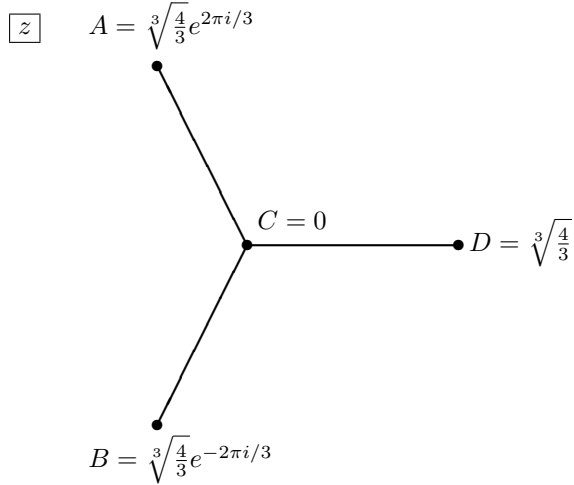
\begin{figure}[h]
\setlength{\unitlength}{2pt}
\begin{center}
\begin{picture}(100,100)(-50,-50)
\thicklines
\put(-45,40){\framebox{$z$}}
\put(0,0){\line(1,0){40}}
\put(0,0){\line(-1,2){17}}
\put(0,0){\line(-1,-2){17}}
\put(-17,34){\circle*{2}}
\put(-30,40){$A=\sqrt[3]{\frac{4}{3}}e^{2\pi i/3}$}
\put(-17,-34){\circle*{2}}
\put(-30,-42){$B=\sqrt[3]{\frac{4}{3}}e^{-2\pi i/3}$}
\put(0,0){\circle*{2}}
\put(2,3){$C=0$}
\put(40,0){\circle*{2}}
\put(42,-1){$D=\sqrt[3]{\frac{4}{3}}$}
\end{picture}
\end{center}
\caption{\emph{The branch cuts of $R(z)$ for $x_0=0$.}}
\label{fig:R-branches-x0}
\end{figure}
To confirm the Boutroux conditions \eqref{eq:g1-Boutroux}, 
first note that, with our choice of $\Sigma$, the function $R(z)$ has Schwarz symmetry:  $R(z^*)=R(z)^*$.  This implies also that $H(z^*)=H(z)^*$, so that $\Lambda\in\mathbb{R}$ and
$H(B)=H(A)^*$.  Therefore, $\Phi_-=-\Phi_+^*$, so $\mathrm{Im}(\Phi_+)=\mathrm{Im}(\Phi_-)$.   Hence it suffices to calculate $\Phi_+$:
\begin{equation}
\Phi_+=2\pi-2i(H(A)-H_+(D))=2\pi-2i\int_D^AH'(z)\,dz = 2\pi-3i\int_D^AR(z)\,dz.
\end{equation}
But, in addition to Schwarz symmetry, $R(z)$ also satisfies the rotational symmetry
$R(e^{2\pi i/3}z)=e^{-2\pi i/3}R(z)$ in this configuration, so making the substitution $z=e^{-2\pi i/3}\zeta$ yields
\begin{equation}
\Phi_+=2\pi-3ie^{-2\pi i/3}\int_A^BR(e^{-2\pi i/3}\zeta)\,d\zeta =2\pi -3i\int_A^BR(\zeta)\,d\zeta
\end{equation}
which is easily seen to be purely real by virtue of Schwarz symmetry of $R$.  It follows
that both Boutroux conditions \eqref{eq:g1-Boutroux} are indeed satisfied.

\begin{proposition}
There exist unique functions 
$A:T\to\mathbb{C}$, $B:T\to\mathbb{C}$, $C:T\to\mathbb{C}$, and $D:T\to\mathbb{C}$ whose real and imaginary parts are real differentiable functions of $(\mathrm{Re}(x_0),\mathrm{Im}(x_0))$, $x_0\in T$,  that yield the special values \eqref{eq:zeroconfig} for $x_0=0$, and that simultaneously satisfy the moment conditions
 \eqref{eq:g1-moments} and the Boutroux conditions \eqref{eq:g1-Boutroux}.  
\end{proposition}
\begin{proof}
The proof is a consequence of the Implicit Function Theorem, which we shall use two consecutive times.  We begin by \emph{assuming} that the moment
conditions \eqref{eq:g1-moments} hold, and this assumption implies that the square of $R(z)$ can be written in the form
\begin{equation}
R^2 = z^4 + \frac{2}{3}x_0 z^2 -\frac{4}{3}z + u+iv,
\label{eq:R-rewrite-gen1}
\end{equation}
where the only undetermined coefficient is the product of the roots, which we write as $u+iv=ABCD$.
Given $x_0\in T$,  we first attempt to solve the Boutroux equations \eqref{eq:g1-Boutroux} 
for real functions $u=u(x_0)$ and $v=v(x_0)$.  
Using the representation $H'(z)=\tfrac{3}{2}R(z)$ and the representation \eqref{eq:R-rewrite-gen1} of $R(z)$, the 
Boutroux conditions may be equivalently formulated as\footnote{We are now thinking of $R$ as a function on the elliptic curve $\Gamma$ given by the equation \eqref{eq:R-rewrite-gen1}, which consists of two sheets.  On one of these sheets $z$ is a suitable coordinate and $R$ may be identified with the previously defined function $R(z)$.  On the other sheet $R$ takes the opposite sign.  A more suitable notation distinguishing the two interpretations of $R$ will be given in \S4.5.2.}
\begin{equation}
I_\mathfrak{a}(u,v;x_0):=\mathrm{Re}\left(\oint_\mathfrak{a} R\,dz\right)=0,\quad I_\mathfrak{b}(u,v;x_0):=\mathrm{Re}\left(\oint_\mathfrak{b} R\,dz\right)=0,
\label{eq:BoutrouxConditions}
\end{equation}
where $\mathfrak{a}=\mathfrak{a}(x_0)$ and $\mathfrak{b}=\mathfrak{b}(x_0)$ form a basis of homology 
cycles on the elliptic curve $\Gamma$ given by the equation \eqref{eq:R-rewrite-gen1} and compactified at infinity.
Note that these functions only depend on the 
homology \emph{classes} of the paths $\mathfrak{a}$ and $\mathfrak{b}$, because as a consequence of
 \eqref{eq:R-rewrite-gen1},
the residues of the meromorphic differential $R\,dz$ at its only poles (the two points of $\Gamma$ over $z=\infty$) are both real ($\pm\tfrac{2}{3}$ in fact).
Noting that
\begin{equation}
\omega_0:=\frac{dz}{R}
\label{eq:omega0define-gen1}
\end{equation}
spans the (one-dimensional) vector space of holomorphic differentials on $\Gamma$, we see
that the Jacobian matrix for the Boutroux conditions takes the form
\begin{equation}
J(u,v;x_0):=\begin{bmatrix}\displaystyle \frac{\partial I_\mathfrak{a}}{\partial u} &\displaystyle \frac{\partial I_\mathfrak{a}}{\partial v}\\\\
\displaystyle \frac{\partial I_\mathfrak{b}}{\partial u} &\displaystyle \frac{\partial I_\mathfrak{b}}{\partial v}\end{bmatrix}=
\frac{1}{2}\begin{bmatrix}\displaystyle \mathrm{Re}(\Omega_\mathfrak{a})&\displaystyle 
\mathrm{Re}(i\Omega_\mathfrak{a})\\
\displaystyle \mathrm{Re}(\Omega_\mathfrak{b}) &\displaystyle \mathrm{Re}(i\Omega_\mathfrak{b})
\end{bmatrix},\quad \Omega_{\mathfrak{a},\mathfrak{b}}=\Omega_{\mathfrak{a},\mathfrak{b}}(x_0):=\oint_{\mathfrak{a},\mathfrak{b}}\omega_0.
\end{equation}
Hence the Jacobian determinant is 
\begin{equation}
\begin{split}
\det J(u,v;x_0)&=\frac{1}{4}\left[\mathrm{Re}(\Omega_\mathfrak{a})
\mathrm{Re}(i\Omega_\mathfrak{b}) -\mathrm{Re}(\Omega_\mathfrak{b})
\mathrm{Re}(i\Omega_\mathfrak{a})\right] \\
&=\frac{1}{4}\left[\mathrm{Re}(\Omega_\mathfrak{b})\mathrm{Im}(\Omega_\mathfrak{a})
-\mathrm{Re}(\Omega_\mathfrak{a})\mathrm{Im}(\Omega_\mathfrak{b})\right]\\
&= \frac{1}{4}\mathrm{Im}(\Omega_\mathfrak{a}\Omega_\mathfrak{b}^*).
\end{split}
\end{equation}
Assuming that the roots of \eqref{eq:R-rewrite-gen1} are distinct (as they 
are when $x_0$ and $u+iv$ are both zero), $\Gamma$ is a smooth compact 
elliptic curve, and it is a basic result (see, for example,  \cite[Ch. II, Corollary 1]{Dubrovin81}) that the strict inequality 
$\mathrm{Im}(\Omega_\mathfrak{a}\Omega_\mathfrak{b}^*)<0$ holds.   Consequently,
$\det J(u,v;x_0)\neq 0$ and the Implicit Function Theorem implies that the solution $u+iv=0$ for $x_0=0$ of the 
Boutroux equations may be continued uniquely to nonzero $x_0$  so long as the 
curve $\Gamma$ does not degenerate (that is, as long as the branch points $A$, $B$, $C$, and $D$ remain distinct).  

Given $u$ and $v$,
we may then try to recover the roots $A$, $B$, $C$, and $D$ via the moment equations 
\eqref{eq:g1-moments} by adjoining to them the relation
\begin{equation}
\Pi(A,B,C,D):=ABCD=u+iv.
\end{equation}
Thus both $x_0$ and $u+iv$ are known, and we seek to solve four equations for the four 
unknowns $A$, $B$, $C$, and $D$.  This amounts to a second invocation of the Implicit Function Theorem.
Of course when $x_0=0$ (and then $u(0)+iv(0)=0$) we know that
\eqref{eq:zeroconfig} furnishes a solution of this system.  The Jacobian 
determinant of the system is
\begin{equation}
\label{Jacobian-gen1}
\begin{split}
\det\frac{\partial(M_1,M_2,M_3,\Pi)}{\partial(A,B,C,D)}&=
\det\begin{bmatrix} 1 & 1 & 1 & 1\\2A & 2B & 2C & 2D\\
3A^2 & 3B^2 & 3C^2 & 3D^2\\
BCD & ACD & ABD & ABC
\end{bmatrix}\\
&=-6(A-B)(A-C)(B-C)(A-D)(B-D)(C-D).
\end{split}
\end{equation}
This is nonzero under exactly the same condition that guarantees that $u$ and 
$v$ may be obtained from the Boutroux conditions as functions of $x_0$: the 
Riemann surface $\Gamma$ must be nonsingular.  

It remains only to show that the continuation of the solution from $x_0=0$ exists throughout the domain $T$, that is, that $\Gamma$ is nonsingular for all $x_0\in T$.  The elliptic curve $\Gamma$ becomes singular exactly when the quartic $R^2$ has fewer than 
four distinct roots.  We now consider all of the possible singular cases.
\subsubsection*{One fourth order root}
If all four roots coincide ($A=B=C=D$), then the condition $M_1(A,A,A,A)=0$ implies that $A=B=C=D=0$.  However, this configuration is then inconsistent with the condition $M_3=4$.
Hence there can be no degenerate configurations of this type.
\subsubsection*{One triple root and one simple root}
Three roots can coincide in four ways, depending on which of $A,B,C,D$ remains distinct from the other three coalescing roots; however all of the equations are symmetric under permutations of the roots, so without loss of generality we suppose that $A\neq B=C=D$.  Then the condition $M_1(A,B,B,B)=0$ implies that $A=-3B$, and then $M_2(-3B,B,B,B)=12B^2=-4x_0/3$ and $M_3(-3B,B,B,B)=-24B^3=4$.  From the latter equation we obtain three distinct solutions:  $B=C=D=-6^{-1/3}e^{2\pi in/3}$, $n=0,\pm 1$.  However the equation $12B^2=-4x_0/3$ is then only consistent if $x_0=-(81/4)^{1/3}e^{-2\pi in/3}$, that is, if $x_0$ coincides with one of the three corner points
of the domain $T$.  

This means that these three points in the $x_0$-plane are the only points for which the moment conditions \eqref{eq:g1-moments}
are satisfied by a configuration of roots of $R(z)^2$ in which exactly three roots coincide.
Moreover, in all of these cases, it is easy to see (by choosing homology cycle representatives in \eqref{eq:BoutrouxConditions} that shrink with the three coalescing branch points) that the Boutroux conditions \eqref{eq:g1-Boutroux} are satisfied exactly by such a  degenerate configuration.  Therefore the three corners of the domain $T$ indeed correspond to simultaneous solutions of \eqref{eq:g1-moments} and \eqref{eq:g1-Boutroux} having
exactly three coincident roots, and these are the only values of $x_0\in\mathbb{C}$ where this degeneracy can occur.

\subsubsection*{Two pairs of double roots}
If the quartic $R(z)^2$ has two distinct double roots, say $A=B\neq C=D$, then $M_1(A,A,C,C)=2A+2C=0$ so $A=-C$, but then the equation $M_3=4$ is obviously inconsistent.
Hence there can be no degenerate configurations of this type.

\subsubsection*{One double root and two simple roots}
If the quartic $R(z)^2$ has one double root, say $C=D$, and two additional distinct simple roots, $A$ and $B$, then the condition $M_1(A,B,C,C)=A+B+2C=0$ implies that $C=-(A+B)/2$.  With this information, it is easy to confirm that
\begin{equation}
M_2\left(A,B,-\frac{1}{2}(A+B),-\frac{1}{2}(A+B)\right)=-\frac{4}{3}x_0
\quad \Leftrightarrow\quad 6S^2+3\Delta^2=-8x_0,
\label{eq:g1-cubic}
\end{equation}
where $S=S(x_0):=A(x_0)+B(x_0)$ and $\Delta=\Delta(x_0):=B(x_0)-A(x_0)$, and 
also
\begin{equation}
M_3\left(A,B,-\frac{1}{2}(A+B),-\frac{1}{2}(A+B)\right)=4\quad
\Leftrightarrow\quad 3S\Delta^2=16.
\end{equation}
We immediately recognize these as the two equations \eqref{ab-conditions} that 
determine the sum $S$ and difference $\Delta$ of the roots $a$ and $b$ in the genus-zero case with only the replacement of $x$ by $x_0$.  The cubic equation for $S$ that results upon elimination of $\Delta$ has three distinct solutions for all $x_0\in\mathbb{C}$ different from the
three corner points of the domain $T$ (and these three points are not under consideration as they
correspond to a coalescence of three of the four points $A$, $B$, $C$, and $D$ as discussed above).  For $x_0\in T$, each of these three solutions corresponds to the analytic continuation of $S(x_0)$ as defined in \S\ref{section-gen0} from the exterior domain $x_0\in \mathbb{C}\setminus\overline{T}$ through exactly one of the three smooth arcs of $\partial T$ to all of $T$.  For each of these analytic functions of $x_0\in T$ we obtain corresponding candidate solution pairs $(A(x_0), B(x_0))$ up to permutation, and we may then attempt to apply the Boutroux conditions which we take in the
form \eqref{eq:BoutrouxConditions}.  
Taking one of the homology 
representatives, say $\beta$, to be a concrete loop surrounding the coalescing 
points $C$ and $D$, we easily see that $I_\beta=0$.  The remaining Boutroux 
condition $I_\alpha=0$ can be simplified as follows:
\begin{equation}
I_\alpha(u,v;x_0)=0\quad\Leftrightarrow\quad
\mathrm{Re}\left(\int_{A(x_0)}^{-S(x_0)/2}h'(z;x_0)\,dz\right)=0.
\end{equation}
In the case that $S$ agrees with the function defined concretely in \S\ref{section-gen0}, we have
$A(x_0)=a(x_0)$ and $B(x_0)=b(x_0)$, in which case we see that this condition is \emph{identical} to the condition $\mathrm{Re}(\mathfrak{c}(x))=0$ with $\mathfrak{c}$ defined by  \eqref{c-def-gen0} that defines the failure of the genus-zero ansatz from outside the domain $T$.  When $S$ coincides with one of the other two roots of the cubic equation \eqref{eq:g1-cubic}, one has an analogue of the condition $\mathrm{Re}(\mathfrak{c}(x))=0$ in which $S$ is replaced by one of the other roots in the definition \eqref{c-def-gen0}; this gives a different system of curves in the $x_0$-plane with three arcs emanating from each of the three corner points of $T$ into the \emph{exterior} region $\mathbb{C}\setminus \overline{T}$.  
These curves evidently coincide exactly with the unbounded Stokes curves 
discovered by Kapaev \cite{Kapaev:1997} and illustrated in Figure 
\ref{fig:phantom-stokes}.  
This analysis shows that there cannot be a value $x_0\in T$ for which exactly two of the four points $A$, $B$, $C$, and $D$ coincide.  On the other hand, this occurs exactly when $x_0$ lies on the three open arcs whose closure is $\partial T$, and for such $x_0$ the distinct pair of roots, say $\{A(x_0),B(x_0)\}$, coincides with the analytic continuation through the arc of $\partial T$ of the pair $\{a(x_0),b(x_0)\}$ of 
points defined in the exterior region $\mathbb{C}\setminus\overline{T}$ as explained in \S\ref{section-gen0}.

\subsubsection*{Conclusion}
It follows that the points $A$, $B$, $C$, and $D$ remain distinct and hence the simultaneous solution of the 
moment conditions \eqref{eq:g1-moments} and Boutroux conditions \eqref{eq:g1-Boutroux} can be uniquely and smoothly continued from $x_0=0$, where it is given by \eqref{eq:zeroconfig}, right
up to the boundary $\partial T$ of the domain $T$.
\end{proof}

\begin{remark}
Unlike the endpoints $a(x)$ and $b(x)$ of the band $\Sigma$ that appeared in the analysis for $x$ outside of the elliptic region $T$, which undergo an exchange permutation upon monodromy of $x$ once about the hole $T$ in the multiply connected domain $\mathbb{C}\setminus\overline{T}$, the four points $A$, $B$, $C$, and $D$ are determined once the roots of $R$ are labeled at any one point in $T$ as we have done at $x=0$.  A numerically useful condition to determine $D$ is that it is the root of $R(z)^2$ with the most positive real part.  Similar extreme conditions can be used to determine $A$ and $B$, and then $C$ is the remaining root.
\end{remark}

The Boutroux conditions \eqref{eq:g1-Boutroux} may be viewed as a statement about the relation
between the points $A$, $B$, $C$, and $D$ and the level curves of the
function $\mathrm{Re}(2H(z)+\Lambda)$:  all four of the points lie on its zero level curve.  The function
$\mathrm{Re}(2H(z)+\Lambda)$ is harmonic for $z\in\mathbb{C}\setminus\Sigma$, and it has no critical points in this domain of definition.  This means that the zero level curve consists solely of (possibly unbounded) arcs connected to the points $A$, $B$, $C$, and $D$, and the arcs may not intersect except at these points.  Local analysis near each of these points shows that there are exactly three level curves emanating from each at angles separated by $2\pi/3$.  Similar local analysis near $z=\infty$ shows that there are exactly six arcs of the level curve that tend to infinity at angles $\pm\pi/6$, $\pm\pi/2$, and $\pm 5\pi/6$.  From this information it can be shown that exactly one of the four points $A$, $B$, $C$, and $D$
is directly connected only to the other three points, each of which is additionally connected to exactly two unbounded arcs.  Since the level curve itself deforms continuously as $x_0$ varies throughout $T$, we let $C$ denote the central point connected to each of the others, while the unbounded arcs from $A$ have asymptotic angles $5\pi/6$ and $\pi/2$, those from $B$ have asymptotic angles $-5\pi/6$ and $-\pi/2$, and those from $D$ have asymptotic angles $\pm\pi/6$.
The last choice we make is to take the arcs of $\Sigma$ to coincide with the connected union of bounded arcs of the zero level curve of $\mathrm{Re}(2H(z)+\Lambda)$, which implies (because of the Boutroux conditions \eqref{eq:g1-Boutroux}, really) that the latter function actually extends continuously to $\Sigma$. 
In this situation the special contour $\Sigma$ is sometimes called the \emph{Stokes graph} of
the radical $R(z)$.  The Stokes graph $\Sigma$ and the full level curve $\mathrm{Re}(2H(z)+\Lambda)=0$
are easily constructed numerically by implementing a root-finding scheme to find the points $A$, $B$, $C$, and $D$ from the moment equations \eqref{eq:g1-moments} and Boutroux conditions \eqref{eq:g1-Boutroux} and then computing curves emanating from each of the four points along which $\mathrm{Re}(H'(z)\,dz)=0$.  The results of several such computations are given in Figure~\ref{H-contours-inside}.
\begin{figure}[h]
\includegraphics[width=1.5in]{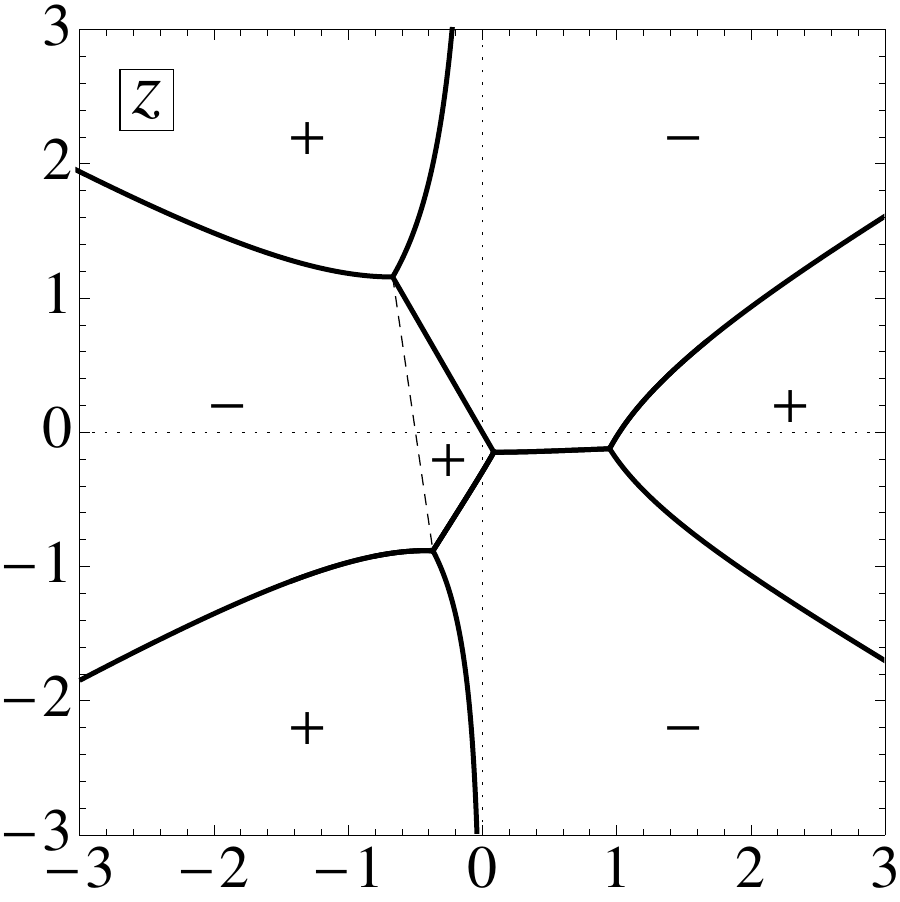}\\
\includegraphics[width=1.5in]{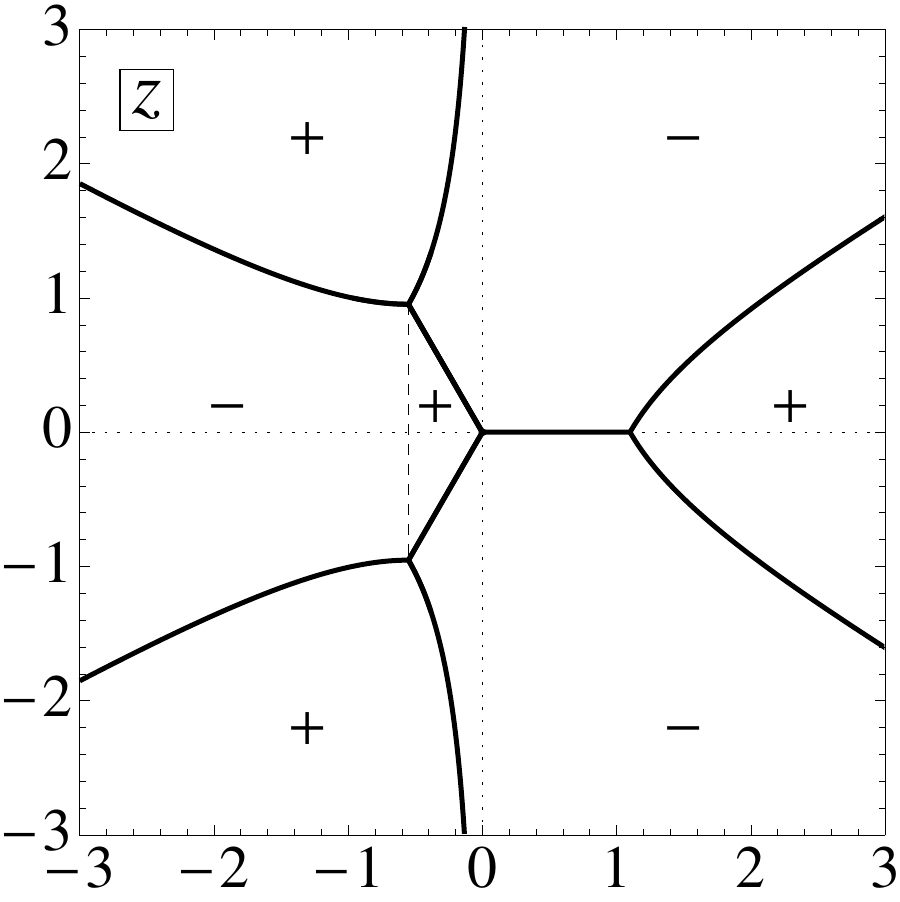}
\includegraphics[width=1.5in]{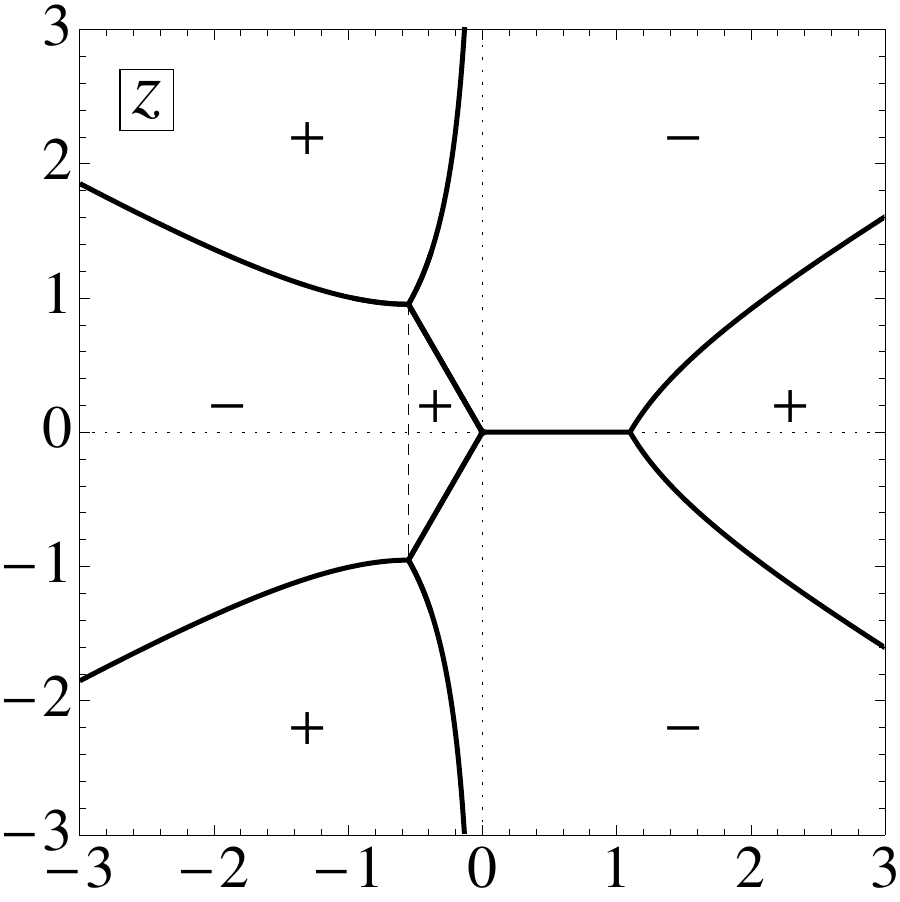}
\includegraphics[width=1.5in]{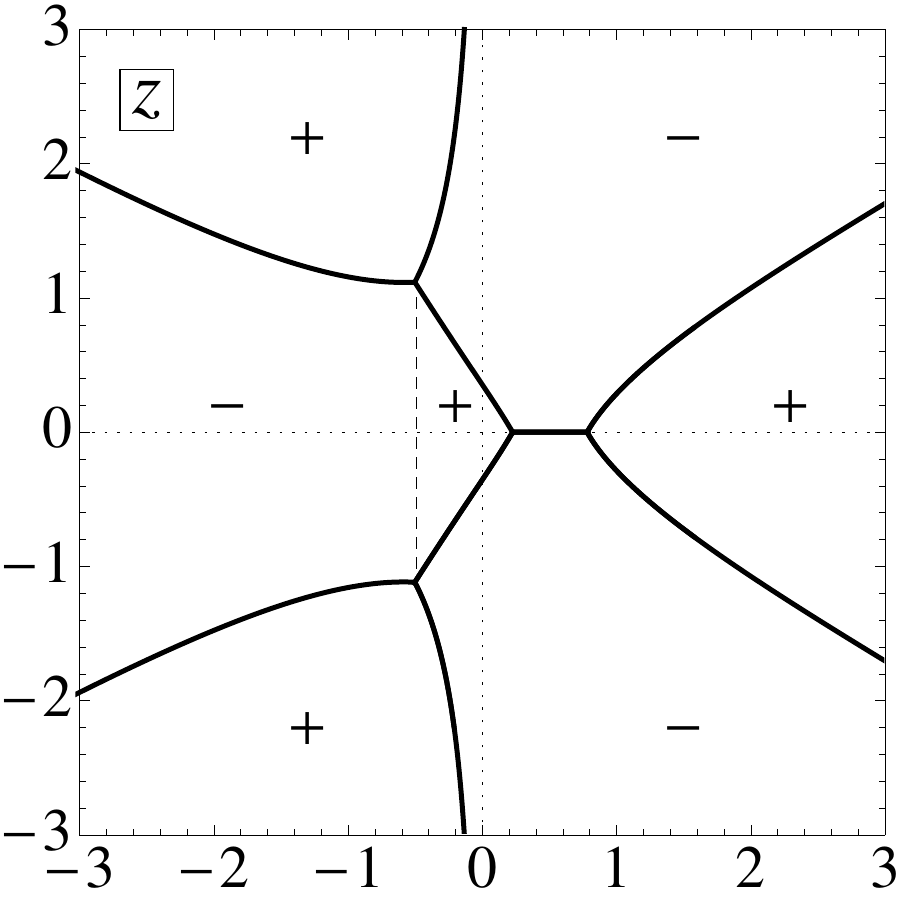}\\
\includegraphics[width=1.5in]{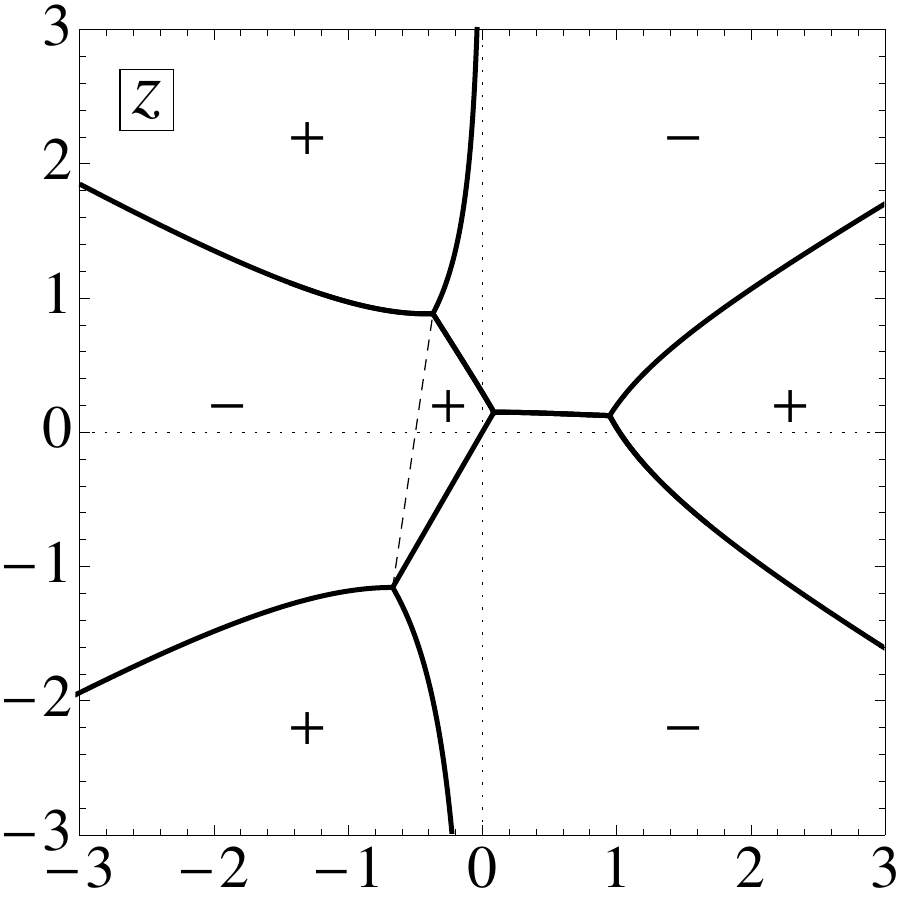}
\caption{\emph{
Signature charts for $\mathrm{Re}(2H(z)+\Lambda)$ in the 
complex $z$ plane for various choices of $x_0$ inside the boundary curve with 
$-\pi/3\leq\arg(x_0)\leq\pi/3$.  
Clockwise from the right-most chart:  
$x_0=1$, $x_0=e^{-i\pi/3}$, $x_0=0$, $x_0=e^{i\pi/3}$.
The center plot illustrates the relation of the chosen $x_0$ values to the 
boundary curve.  In (numerical) practice, among the four points $\{A,B,C,D\}$ determined up to permutation by applying Newton iteration to the moment conditions \eqref{eq:g1-moments} and Boutroux conditions \eqref{eq:g1-Boutroux}, the point $D(x_0)$ maximizes $\mathrm{Re}(\cdot)$, the point $A(x_0)$ maximizes $\mathrm{Re}(e^{-2\pi i/3}\cdot)$, and the point $B(x_0)$ maximizes $\mathrm{Re}(e^{2\pi i/3}\cdot)$.
}
}
\label{H-contours-inside}
\end{figure}%

Note that for $x_0\in T\cap\mathbb{R}$, the Stokes graph $\Sigma$ is Schwarz-symmetric, as illustrated qualitatively in Figure~\ref{fig:N-jumps-g1}, and that in this case $R(z^*)=R(z)^*$, $G(z^*)=G(z)^*$, and $H(z^*)=H(z)^*$, implying further that $\Phi_-=-\Phi_+$ (both real by the Boutroux conditions).

Note also that as a result of the moment conditions \eqref{eq:g1-moments} and Boutroux conditions \eqref{eq:g1-Boutroux}, all of the quantities defined in this section depend parametrically on $x_0\in T$.  The basic quantities are $(A,B,C,D)=(A(x_0),B(x_0),C(x_0),D(x_0))$, and then we also have $G(z)=G(z;x_0)$, $H(z)=H(z;x_0)$, $\Lambda=\Lambda(x_0)$, $\Phi_\pm=\Phi_\pm(x_0)$, $\Sigma=\Sigma(x_0)$, and $R(z)=R(z;x_0)$.  While the parametric dependence on $(\mathrm{Re}(x_0),\mathrm{Im}(x_0))$ is infinitely smooth, there is no analyticity with respect to $x_0$ as a complex variable; indeed, by implicit differentiation of the Boutroux conditions in the form \eqref{eq:BoutrouxConditions} one can prove that
\begin{equation}
\overline{\partial}_{x_0}\Pi:=\frac{1}{2}\left[\frac{\partial\Pi}{\partial\mathrm{Re}(x_0)} +i\frac{\partial\Pi}{\partial\mathrm{Im}(x_0)}\right]=-\frac{4\pi}{3\mathrm{Im}(\Omega_\mathfrak{a}\Omega_\mathfrak{b}^*)}>0
\label{eq:g1-dbarPi}
\end{equation}
holds for all $x_0\in T$, where $\Pi:=u+iv=ABCD$.  In other words, the smooth functions $u$ and $v$ \emph{do not} satisfy the Cauchy-Riemann equations.

Finally, as a consequence of the fact that we can identify (after appropriate contour deformation and analytic continuation) $h(z;x)$ with $H(z;x)$ for any $x\in\partial T$, we have the following result.
\begin{proposition}
If $x_0$ belongs to either of the two maximal smooth arcs of $\partial T$ that meet at the point $x_c<0$, then $\Lambda(x_0)=\lambda(x_0)$.  If $x_0$ belongs to the remaining arc of $\partial T$ that crosses the positive real axis at $x_e$, then $\mathrm{Re}(\Lambda(x_0))=\mathrm{Re}(\lambda(x_0))$ but the imaginary parts do not agree, even modulo $2\pi$.
\label{prop:g1-Lambdalambda}
\end{proposition}
\begin{proof}
Recall that $\Lambda(x_0)$ may be expressed in terms of the boundary values of $H(z;x_0)$ from above and below at the point $D\in\Sigma\cap L$.  If $x_0\in T$ tends to one of the two edges of $\partial T$ that meet at $x_c<0$, then either $A$ or $B$ collides with $C$, and the collision point coincides with $-S(x_0)/2$.  It is easy to see in this degeneration that $\Lambda(x_0)=\lambda(x_0)$ because the latter function is similarly defined in terms of the boundary values of $h(z;x_0)$ on the surviving branch cut of $R$ (which becomes the branch cut for $r$), and $D$ survives as one of the endpoints of this cut.  If instead $x_0\in T$ tends to the third edge of $\partial T$, then it is $D$ that coalesces with $C$ in the limit, so $\Lambda(x_0)$ tends to the sum of boundary values of $H$ at the collision point, which again coincides with $-S(x_0)/2$.  This value can no longer be generally identified with $\lambda(x_0)$, because it is the branch points $A$ and $B$ that become the roots of $r^2$ in the limit, and $\lambda(x_0)$ is defined in terms of the values taken by $h(z;x_0)$ on a cut connecting these two points.  However, the condition that defines $\partial T$ is that the point $-S(x_0)/2$ lies on the same level of $\mathrm{Re}(H(z;x_0))$ as do the branch points $A$ and $B$, and this guarantees that on the third edge we nonetheless have $\mathrm{Re}(\Lambda(x_0))=\mathrm{Re}(\lambda(x_0))$.
\end{proof}

\subsection{Introduction of the $g$-function}
Given $x_0\in T$ and corresponding Stokes graph $\Sigma$ with points $A$, $B$, $C$, and $D$,
we consider the matrix $\mathbf{N}(z)$ analytic on the complement of the jump contour $\Sigma^{(\mathbf{N})}$ illustrated in Figure~\ref{fig:N-jumps-g1}, satisfying the jump condition $\mathbf{N}_+(z)=\mathbf{N}_-(z)\mathbf{V}^{(\mathbf{N})}(z)$ where the jump matrix $\mathbf{V}^{(\mathbf{N})}(z)$ is defined as indicated in that figure on each arc, and normalized at infinity by the condition $\mathbf{N}(z)(-z)^{-\sigma_3/\epsilon}=\mathbb{I}+\mathcal{O}(z^{-1})$ as $z\to\infty$.  Using the function $G(z)$ defined for each $x_0\in T$ as described in \S\ref{genus-one-g}, we introduce the substitution
\begin{equation}
\mathbf{O}(z)=\mathbf{O}(z;x_0,w,\epsilon):=e^{-\Lambda\sigma_3/(2\epsilon)}\mathbf{N}(z)e^{-G(z)\sigma_3/\epsilon}e^{\Lambda\sigma_3/(2\epsilon)}.
\label{eq:g1-OfromN}
\end{equation}
Then, $\mathbf{O}(z)$ is analytic exactly where $\mathbf{N}(z)$ is (that is, for $z\in\mathbb{C}\setminus\Sigma^{(\mathbf{O})}$ where $\Sigma^{(\mathbf{O})}=\Sigma^{(\mathbf{N})}$), and with the help of the condition 
\eqref{eq:g1-G-asymp} one checks that $\mathbf{O}(z)=\mathbb{I}+\mathcal{O}(z^{-1})$
as $z\to\infty$.  On the arcs of the jump contour $\Sigma^{(\mathbf{O})}$ illustrated in Figure~\ref{fig:N-jumps-g1}, 
$\mathbf{O}(z)$ satisfies the jump condition $\mathbf{O}_+(z)=\mathbf{O}_-(z)\mathbf{V}^{(\mathbf{O})}(z)$, where the jump matrix $\mathbf{V}^{(\mathbf{O})}(z)$ is a systematic modification of $\mathbf{V}^{(\mathbf{N})}(z)$ as follows:
\begin{equation}
\mathbf{V}^{(\mathbf{N})}=\mathbf{L}\quad\implies\quad\mathbf{V}^{(\mathbf{O})}(z)=\begin{bmatrix}
1 & 0\\ie^{(2H(z)+\Lambda)/\epsilon}e^{wz} & 1\end{bmatrix},
\label{eq:g1-Ojump-1}
\end{equation}
\begin{equation}
\mathbf{V}^{(\mathbf{N})}=\mathbf{L}^{-1}\quad\implies\quad\mathbf{V}^{(\mathbf{O})}(z)=\begin{bmatrix}
1 & 0\\-ie^{(2H(z)+\Lambda)/\epsilon}e^{wz} & 1\end{bmatrix},
\label{eq:g1-Ojump-2}
\end{equation}
\begin{equation}
\mathbf{V}^{(\mathbf{N})}=\mathbf{U}\quad\implies\quad\mathbf{V}^{(\mathbf{O})}(z)=
\begin{bmatrix} 1 & ie^{-(2H(z)+\Lambda)/\epsilon}e^{-wz}\\ 0 & 1\end{bmatrix},
\label{eq:g1-Ojump-3}
\end{equation}
and
\begin{equation}
\mathbf{V}^{(\mathbf{N})}=-\mathbf{U}\quad\implies\quad\mathbf{V}^{(\mathbf{O})}(z)=
\begin{bmatrix} 1 & ie^{-(2H_\pm(z)+\Lambda)/\epsilon}e^{-wz}\\ 0 & 1\end{bmatrix}.
\label{eq:g1-jump-L}
\end{equation}
To check \eqref{eq:g1-jump-L} one has to recall  the condition \eqref{eq:g1-G-jump} and the fact that $\epsilon^{-1}=m-\tfrac{1}{2}$, $m\in\mathbb{Z}$.  This also shows that it does not matter
which boundary value of $H$ is used in \eqref{eq:g1-jump-L}.  Finally, on the three remaining arcs along which $\mathbf{V}^{(\mathbf{N})}=\mathbf{T}$ or $\mathbf{T}^{-1}$, we have
\begin{equation}
\mathbf{V}^{(\mathbf{O})}(z)=\begin{bmatrix}0 & -ie^{-wz}\\-ie^{wz}& 0\end{bmatrix},\quad z\in \overrightarrow{CD},
\label{eq:g1-Ojump-5}
\end{equation}
\begin{equation}
\mathbf{V}^{(\mathbf{O})}(z)=\begin{bmatrix}0 & ie^{-i\Phi_+/\epsilon}e^{-wz}\\ie^{i\Phi_+/\epsilon}e^{wz} & 0\end{bmatrix},\quad z\in\overrightarrow{CA},
\label{eq:g1-Ojump-6}
\end{equation}
and
\begin{equation}
\mathbf{V}^{(\mathbf{O})}(z)=\begin{bmatrix}0 & ie^{-i\Phi_-/\epsilon}e^{-wz}\\
ie^{i\Phi_-/\epsilon}e^{wz} & 0\end{bmatrix},\quad z\in\overrightarrow{CB}.
\label{eq:g1-Ojump-7}
\end{equation}
To derive all of these formulae, we have used the fact that $x=x_0+\epsilon w$ implies that 
$\theta=\theta(z;x)=\theta(z;x_0) + \epsilon wz$.  

\subsection{The global parametrix for $\mathbf{O}(z)$} 
The free arcs of the jump contour $\Sigma^{(\mathbf{O})}$ for $\mathbf{O}(z)$ (those not part of the Stokes graph $\Sigma$)
are now chosen to lie in two complementary domains as follows:
\begin{itemize}
\item Those arcs along which either $\mathbf{V}^{(\mathbf{N})}=\mathbf{L}$ or $\mathbf{L}^{-1}$ are
placed so that (except at the points $A$, $B$, and $D$) the inequality $\mathrm{Re}(2H(z)+\Lambda)<0$ holds.
\item Those arcs along which either $\mathbf{V}^{(\mathbf{N})}=\mathbf{U}$ or $-\mathbf{U}$ are
placed so that (except at the points $A$, $B$, and $D$) the inequality $\mathrm{Re}(2H(z)+\Lambda)>0$ holds.
\end{itemize}
In both cases, we will in fact choose the arcs precisely so that, along each, $\mathrm{Im}(2H(z)+\Lambda)$ is constant near any of the four points $\{A,B,C,D\}$ at which the arc terminates.  This makes each arc a local steepest descent/ascent path for $\mathrm{Re}(2H(z)+\Lambda)$ near these four points.
Assuming that $w\in\mathbb{C}$ is bounded, it then follows easily that on all of these arcs the jump matrix $\mathbf{V}^{(\mathbf{O})}(z)$ is exponentially close to the identity matrix uniformly for $z$ bounded away from $A$, $B$, and $D$.  By ignoring these discontinuities, we are left with 
jump discontinuities only along the Stokes graph $\Sigma$, and assuming that $w$ avoids certain exceptional values (a discrete set depending on $x_0\in T$ and $\epsilon>0$, see \eqref{eq:g1-Lambda-p} below), we will be able to construct an exact solution of the corresponding jump conditions that is bounded in the complex $z$-plane except near the four points $A$, $B$, $C$, and $D$.  By combining this exact solution with suitable local parametrices defined near the four special points, in \S\ref{sec:g1-globalpar-def} we will construct an explicit global parametrix for $\mathbf{O}(z)$, that is, a globally-defined matrix function of $z$ that we will prove in \S\ref{sec:g1-error-analysis} is a close approximation to $\mathbf{O}(z)$ when $\epsilon$ is small.

\subsubsection{The outer parametrix.  Definition and basic properties}
The outer parametrix is a matrix function $\dot{\mathbf{O}}^\mathrm{(out)}(z)=\dot{\mathbf{O}}^{(\mathrm{out})}(z;x_0,w,\epsilon)$ that is required to be analytic for $z\in\mathbb{C}\setminus\Sigma$ with continuous boundary values except at the four points $z=A,B,C,D$ at which negative one-fourth power singularities are admissible.  The boundary values taken on each of the three smooth arcs of $\Sigma$ are related by exactly the same jump condition as for $\mathbf{O}(z)$, namely $\dot{\mathbf{O}}^\mathrm{(out)}_+(z)=\dot{\mathbf{O}}^\mathrm{(out)}_-(z)\mathbf{V}^{(\mathbf{O})}(z)$ for $z\in\Sigma\setminus\{A, B, C,D\}$.  Like $\mathbf{O}(z)$, we require that $\dot{\mathbf{O}}^\mathrm{(out)}(z)$ be normalized to the identity matrix in the limit $z\to\infty$.  It follows that the outer parametrix has a convergent Laurent series for $|z|$ sufficiently large, and in particular satisfies
\begin{equation}
\dot{\mathbf{O}}^\mathrm{(out)}(z)=\mathbb{I} +\dot{\mathbf{A}}z^{-1} + \dot{\mathbf{B}}z^{-2} +
\dot{\mathbf{C}}z^{-3}+\mathcal{O}(z^{-4}),\quad z\to\infty
\label{eq:dotOexpansion}
\end{equation}
for some matrix coefficients $\dot{\mathbf{A}}=\dot{\mathbf{A}}(w)$, $\dot{\mathbf{B}}=\dot{\mathbf{B}}(w)$, and $\dot{\mathbf{C}}=\dot{\mathbf{C}}(w)$ (also depending parametrically on $x_0\in T$ and $\epsilon>0$, or equivalently, $m\in \mathbb{Z}_+$).  This series is also differentiable term-by-term with respect to $w$.  

Before solving this problem, let us note some properties of its solution which we will show exists and is analytic for all $w\in\mathbb{C}$ with the exception of a certain $x_0$ and $\epsilon$-dependent doubly-periodic lattice of points at which the solution has poles.  The matrix function $\mathbf{F}(z;w):=\dot{\mathbf{O}}^\mathrm{(out)}(z)e^{-wz\sigma_3/2}$ necessarily has determinant equal to $1$, and clearly satisfies jump conditions independent of $w$; it follows that the matrix product $\mathbf{F}_w\mathbf{F}^{-1}$ is analytic with the possible exception of the four points $\{A,B,C,D\}$.  Moreover, the mild nature of the singularities admitted in $\dot{\mathbf{O}}^\mathrm{(out)}(z)$ at these points shows that all four points are removable singularities for $\mathbf{F}_w\mathbf{F}^{-1}$, which is therefore entire.  
From the expansion \eqref{eq:dotOexpansion} it follows that $\mathbf{F}_w\mathbf{F}^{-1}$ is a linear matrix function of $z$, and multiplication on the right by $\mathbf{F}$ and writing $\mathbf{F}$ in terms of $\dot{\mathbf{O}}^\mathrm{(out)}(z)$ yields the differential equation
\begin{equation}
\frac{\partial\dot{\mathbf{O}}^\mathrm{(out)}}{\partial w}=-\frac{1}{2}[\sigma_3,\dot{\mathbf{O}}^\mathrm{(out)}] z +\frac{1}{2}[\sigma_3,\dot{\mathbf{A}}]\dot{\mathbf{O}}^\mathrm{(out)}.
\label{eq:g1-Lax-w}
\end{equation}
Similar arguments show that the matrix $\mathbf{G}(z):=R(z)\dot{\mathbf{O}}^\mathrm{(out)}(z)\sigma_3\dot{\mathbf{O}}^\mathrm{(out)}(z)^{-1}$ is also entire, and with the use of \eqref{eq:dotOexpansion} and the Laurent expansion
\begin{equation}
R(z)=z^2 +\frac{1}{3}x_0 -\frac{2}{3}z^{-1}+\mathcal{O}(z^{-2}),\quad z\to\infty
\label{eq:g1-R-expansion}
\end{equation}
(following from \eqref{eq:R-rewrite-gen1} and the condition that $R(z)z^{-2}\to 1$ as $z\to\infty$) $\mathbf{G}(z)$ is identified with a quadratic matrix polynomial in $z$ with coefficients involving $\dot{\mathbf{A}}$ and $\dot{\mathbf{B}}$:
\begin{equation}
\mathbf{G}(z)=
\sigma_3z^2 -[\sigma_3,\dot{\mathbf{A}}]z +\left([\sigma_3,\dot{\mathbf{A}}]\dot{\mathbf{A}}-[\sigma_3,\dot{\mathbf{B}}]+\frac{1}{3}x_0\sigma_3\right).
\label{eq:GmatrixDefine}
\end{equation}
With $\mathbf{G}(z)$ determined in this way, we reinterpret the definition of $\mathbf{G}(z)$ in terms of $\dot{\mathbf{O}}^\mathrm{(out)}(z)$ as the relation\footnote{We owe special thanks to Alexander Its for pointing out the utility of this equation in the context of the present calculation.}
\begin{equation}
\mathbf{G}(z)\dot{\mathbf{O}}^\mathrm{(out)}(z)=
\dot{\mathbf{O}}^\mathrm{(out)}(z)\cdot R(z)\sigma_3
\label{eq:g1-Lax-ev}
\end{equation}
asserting that $\dot{\mathbf{O}}^\mathrm{(out)}(z)$ is a matrix of eigenvectors of $\mathbf{G}(z)$ with diagonal eigenvalue matrix $R(z)\sigma_3$.   Further equations can be deduced from the conditions characterizing the outer parametrix $\dot{\mathbf{O}}^\mathrm{(out)}(z)$, including a linear differential equation with respect to $z$ having Fuchsian singularities at $z=A,B,C,D$, but we will only need to use \eqref{eq:g1-Lax-w} and 
\eqref{eq:g1-Lax-ev}.

Using \eqref{eq:dotOexpansion}, we expand both sides of the differential equation \eqref{eq:g1-Lax-w} in Laurent series for large $z$.  The terms proportional to $z$ and $1$ give no information, but from the coefficient of $z^{-1}$ we obtain the equation
\begin{equation}
\frac{d\dot{\mathbf{A}}}{dw}=\frac{1}{2}[\sigma_3,\dot{\mathbf{A}}]\dot{\mathbf{A}}-\frac{1}{2}[\sigma_3,\dot{\mathbf{B}}]
\label{eq:dotAODE}
\end{equation}
and from the coefficient of $z^{-2}$ we obtain the equation
\begin{equation}
\frac{d\dot{\mathbf{B}}}{dw}=\frac{1}{2}[\sigma_3,\dot{\mathbf{A}}]\dot{\mathbf{B}}-\frac{1}{2}[\sigma_3,\dot{\mathbf{C}}]
\label{eq:dotBODE}
\end{equation}
Similarly, expanding both sides of \eqref{eq:g1-Lax-ev} using \eqref{eq:dotOexpansion}, \eqref{eq:g1-R-expansion}, and \eqref{eq:GmatrixDefine}, the terms proportional to $z^2$, $z$, and $1$ give no information, but from the terms proportional to $z^{-1}$ we obtain the identity
\begin{equation}
[\sigma_3,\dot{\mathbf{C}}]=[\sigma_3,\dot{\mathbf{A}}]\dot{\mathbf{B}}+[\sigma_3,\dot{\mathbf{B}}]\dot{\mathbf{A}}-[\sigma_3,\dot{\mathbf{A}}]\dot{\mathbf{A}}^2-\frac{1}{3}x_0[\sigma_3,\dot{\mathbf{A}}]-\frac{2}{3}\sigma_3
\end{equation}
which allows \eqref{eq:dotAODE} and \eqref{eq:dotBODE} to be recast as a closed autonomous system of first-order differential equations governing the $w$-dependence of the matrices $\dot{\mathbf{A}}$ and $\dot{\mathbf{B}}$:
\begin{equation}
\begin{split}
\frac{d\dot{\mathbf{A}}}{dw}&=\frac{1}{2}[\sigma_3,\dot{\mathbf{A}}]\dot{\mathbf{A}}-\frac{1}{2}[\sigma_3,\dot{\mathbf{B}}],\\
\frac{d\dot{\mathbf{B}}}{dw}&=\frac{1}{2}[\sigma_3,\dot{\mathbf{A}}]\dot{\mathbf{A}}^2-\frac{1}{2}[\sigma_3,\dot{\mathbf{B}}]\dot{\mathbf{A}}+\frac{1}{6}x_0[\sigma_3,\dot{\mathbf{A}}]+\frac{1}{3}\sigma_3.
\end{split}
\label{eq:ABsystem}
\end{equation}
Multiplying \eqref{eq:g1-Lax-ev} on the left by $\mathbf{G}(z)$ and using invertibility of $\dot{\mathbf{O}}^\mathrm{(out)}(z)$, one obtains the matrix identity
\begin{equation}
\mathbf{G}(z)^2=R(z)^2\mathbb{I} = \left(z^4 + \frac{2}{3}x_0z^2 -\frac{4}{3}z +\Pi\right)\mathbb{I}.
\end{equation}
Expanding out $\mathbf{G}(z)^2$ using \eqref{eq:GmatrixDefine} gives
\begin{equation}
\begin{split}
\mathbf{G}(z)^2&=\mathbb{I} z^4+\frac{2}{3}x_0\mathbb{I}z^2-4\left(\dot{\mathbf{A}}_\mathrm{OD}\dot{\mathbf{B}}_\mathrm{OD}+\dot{\mathbf{B}}_\mathrm{OD}\dot{\mathbf{A}}_\mathrm{OD}-
\dot{\mathbf{A}}_\mathrm{OD}^2\dot{\mathbf{A}}_\mathrm{D}-\dot{\mathbf{A}}_\mathrm{OD}
\dot{\mathbf{A}}_\mathrm{D}\dot{\mathbf{A}}_\mathrm{OD}\right)z \\
&\quad\quad{}+\left(2\sigma_3\dot{\mathbf{A}}_\mathrm{OD}^2+2\sigma_3\dot{\mathbf{A}}_\mathrm{OD}
\dot{\mathbf{A}}_\mathrm{D}-2\sigma_3\dot{\mathbf{B}}_\mathrm{OD}+\frac{1}{3}x_0\sigma_3\right)^2,
\end{split}
\label{eq:Gsquared}
\end{equation}
where the subscripts OD and D denote respectively the off-diagonal and diagonal parts of the matrix.
It is easy to see that the matrix coefficients on the right-hand side are in fact all multiples of the identity.

It will be convenient to introduce the following particular combinations of matrix elements of $\dot{\mathbf{A}}(w)$ and $\dot{\mathbf{B}}(w)$:
\begin{equation}
\dot{\mathcal{U}}^0_m(w)=\dot{\mathcal{U}}^0_m(w;x_0):=e^{-\Lambda/2-1/3}\dot{A}_{12}(w)\quad\text{and}\quad
\dot{\mathcal{V}}^0_m(w)=\dot{\mathcal{V}}^0_m(w;x_0):=e^{\Lambda/2+1/3}\dot{A}_{21}(w),
\label{eq:g1-dotU-dotV-def}
\end{equation}
\begin{equation}
\dot{\mathcal{P}}_m(w)=\dot{\mathcal{P}}_m(w;x_0):=\dot{A}_{22}(w)-\frac{\dot{B}_{12}(w)}{\dot{A}_{12}(w)}\quad\text{and}\quad\dot{\mathcal{Q}}_m(w)=\dot{\mathcal{Q}}_m(w;x_0):=-\dot{A}_{11}(w)+\frac{\dot{B}_{21}(w)}{\dot{A}_{21}(w)}.
\label{eq:g1-dotP-dotQ-def}
\end{equation}
From the off-diagonal terms in \eqref{eq:dotAODE} it follows that
\begin{equation}
\dot{\mathcal{P}}_m(w)=\frac{d}{dw}\log(\dot{\mathcal{U}}^0_m(w))=\frac{1}{\dot{\mathcal{U}}^0_m(w)}\frac{d\dot{\mathcal{U}}^0_m}{dw}(w)\quad\text{and}\quad
\dot{\mathcal{Q}}_m(w)=\frac{d}{dw}\log(\dot{\mathcal{V}}^0_m(w))=\frac{1}{\dot{\mathcal{V}}^0_m(w)}\frac{d\dot{\mathcal{V}}^0_m}{dw}(w).
\label{eq:g1-dotPQlogderivs}
\end{equation}
These alternate representations of $\dot{\mathcal{P}}_m(w)$ and $\dot{\mathcal{Q}}_m(w)$ will be useful later (see the derivation of \eqref{eq:g1-dotP-formula} from \eqref{eq:g1-dotU-formula}).  We combine the latter two functions into a diagonal matrix as follows:
\begin{equation}
\mathbf{D}(w):=\begin{bmatrix}\dot{\mathcal{P}}_m(w) & 0\\0 & -\dot{\mathcal{Q}}_m(w)\end{bmatrix}=\left(\dot{\mathbf{A}}_\mathrm{OD}
\dot{\mathbf{A}}_\mathrm{D}-\dot{\mathbf{B}}_\mathrm{OD}\right)\dot{\mathbf{A}}_\mathrm{OD}^{-1}.
\end{equation}
Using the equations \eqref{eq:ABsystem}, one may easily calculate $d\mathbf{D}/dw$, and then it is straightforward to confirm the matrix identity
\begin{equation}
\left(\frac{d\mathbf{D}}{dw}\right)^2=\mathbf{G}(\mathbf{D})^2, 
\end{equation}
where the right-hand side is computed by taking into account that the coefficients of $\mathbf{G}(z)^2$ are in fact scalars.  It follows that both functions $\dot{\mathcal{P}}_m$ and $-\dot{\mathcal{Q}}_m$ satisfy the same differential equation:
\begin{equation}
\left(\frac{d\dot{\mathcal{P}}_m}{dw}\right)^2=R(\dot{\mathcal{P}}_m)^2 = \dot{\mathcal{P}}_m^4 +\frac{2}{3}x_0\dot{\mathcal{P}}_m^2-\frac{4}{3}\dot{\mathcal{P}}_m+\Pi.
\label{eq:g1-dotP-ellipticeqn}
\end{equation}
Note that for $x_0\in T$ the quartic polynomial on the right-hand side has distinct roots.  We have therefore proved the following.  
\begin{proposition}
Suppose that $x_0\in T$ and let $w\in \mathbb{C}$ be a value for which the outer parametrix $\dot{\mathbf{O}}^{(\mathrm{out})}(z)$ exists and is differentiable with respect to $w$.  Then, the quantities $\dot{\mathcal{U}}^0_m(w;x_0)$, $\dot{\mathcal{V}}^0_m(w;x_0)$, $\dot{\mathcal{P}}_m(w;x_0)$, and $\dot{\mathcal{Q}}_m(w;x_0)$ defined therefrom by \eqref{eq:g1-dotU-dotV-def}--\eqref{eq:g1-dotP-dotQ-def} are related by the differential equations \eqref{eq:g1-dotPQlogderivs}, and the functions $\dot{\mathcal{P}}_m(w;x_0)$ and $-\dot{\mathcal{Q}}_m(w;x_0)$ are both elliptic functions of $w$ that satisfy the Boutroux ansatz differential equation \eqref{eq:BoutrouxElliptic} where the integration constant $\Pi=u+iv$ is determined as a (non-analytic) function of $x_0$ by the Boutroux conditions \eqref{eq:g1-Boutroux}.
\label{prop:g1-diffeq}
\end{proposition}

\subsubsection{The outer parametrix.  Explicit construction}
We begin by introducing a scalar function $k(z)$ defined as follows:
\begin{equation}
k(z):=\frac{1}{2}wz +\frac{\Phi_+R(z)}{2\pi\epsilon}\int_C^A\frac{du}{(u-z)R_+(u)} + \frac{\Phi_-R(z)}{2\pi\epsilon}\int_C^B\frac{du}{(u-z)R_+(u)},\quad z\in\mathbb{C}\setminus\Sigma,
\end{equation}
where in each case the integrals are taken over appropriate oriented arcs of $\Sigma$ (and $R_+(u)$ indicates the boundary value of $R$ taken on these arcs from the left side).  Note that $k$ depends parametrically on $x_0\in T$ and $\epsilon>0$ or equivalently $m\in\mathbb{Z}_+$.  For sufficiently large $|z|$, $k(z)$ has a convergent Laurent expansion that yields an asymptotic representation:
\begin{equation}
k(z)=\left(\frac{1}{2}w + \frac{\kappa_1}{\epsilon}\right)z + \frac{\kappa_0}{\epsilon} + \mathcal{O}\left(\frac{1}{z}\right),\quad z\to\infty,
\label{eq:g1-k-expansion}
\end{equation}
where
\begin{equation}
\kappa_1=\kappa_1(x_0):=-\frac{\Phi_+}{2\pi}\int_C^A\frac{du}{R_+(u)} -\frac{\Phi_-}{2\pi}\int_C^B\frac{du}{R_+(u)}
\label{eq:g1-k1}
\end{equation}
and
\begin{equation}
\kappa_0=\kappa_0(x_0):=-\frac{\Phi_+}{2\pi}\int_C^A\frac{u\,du}{R_+(u)} -\frac{\Phi_-}{2\pi}\int_C^B\frac{u\,du}{R_+(u)}.
\label{eq:g1-k0}
\end{equation}
To obtain \eqref{eq:g1-k0} we have taken into account the first of the moment conditions \eqref{eq:g1-moments}.  The dependence of the coefficients $\kappa_1$ and $\kappa_0$  on $x_0$ enters through the moment conditions \eqref{eq:g1-moments} and the Boutroux conditions \eqref{eq:g1-Boutroux}.  In the case that $x_0\in T\cap\mathbb{R}$, Schwarz symmetry of $\Sigma$ and $R$ and the coincident relation $\Phi_-=-\Phi_+\in\mathbb{R}$ imply that $\kappa_0$ and $\kappa_1$ are both real.
It is easy to check that for each $\epsilon\neq 0$ and for each $x_0\in T$, $k(z)$ is bounded for bounded $z$ (including at the four points $z=A,B,C,D$), and that it satisfies
\begin{equation}
k_+(z)+k_-(z)=\begin{cases} wz,&\quad z\in\overrightarrow{CD},\\
wz+i\epsilon^{-1}\Phi_+,&\quad z\in\overrightarrow{CA},\\
wz+i\epsilon^{-1}\Phi_-,&\quad z\in\overrightarrow{CB}.
\end{cases}
\end{equation}
Hence, the matrix $\mathbf{P}(z)$ related explicitly to $\dot{\mathbf{O}}^\mathrm{(out)}(z)$ by
\begin{equation}
\mathbf{P}(z):=e^{\kappa_0\sigma_3/\epsilon}\dot{\mathbf{O}}^\mathrm{(out)}(z)e^{-k(z)\sigma_3},\quad z\in\mathbb{C}\setminus\Sigma
\end{equation}
is easily checked to be analytic where defined and to satisfy the normalization condition
\begin{equation}
\mathbf{P}(z)e^{wz\sigma_3/2}e^{\kappa_1z\sigma_3/\epsilon}=\mathbb{I}+\mathcal{O}\left(\frac{1}{z}\right),\quad z\to\infty
\end{equation}
and modified jump conditions on the arcs of $\Sigma$ of the form $\mathbf{P}_+(z)=\mathbf{P}_-(z)\mathbf{V}^{(\mathbf{P})}(z)$, where the jump matrix is now piecewise constant and of a universal form:
\begin{equation}
\mathbf{V}^{(\mathbf{P})}(z)=\begin{cases}i\sigma_1,&\quad z\in\overrightarrow{CA}\cup\overrightarrow{CB},\\
-i\sigma_1,&\quad z\in\overrightarrow{CD}.
\end{cases}.
\end{equation}
Negative one-fourth power singularities are admissible for $\mathbf{P}(z)$ at the points $z=A,B,C,D$ but otherwise the boundary values are required to be continuous.

Next, we define a second scalar function $\beta(z)$ to be analytic for $z\in\mathbb{C}\setminus\Sigma$, to satisfy the equation
\begin{equation}
\beta(z)^4 = \frac{(z-C)(z-D)}{(z-A)(z-B)},\quad z\in\mathbb{C}\setminus\Sigma,
\end{equation}
and with the particular branch chosen so that $\beta(\infty)=1$.  Note that on the three oriented arcs of the Stokes graph $\Sigma$ we have
\begin{equation}
\frac{\beta_+(z)}{\beta_-(z)}=\begin{cases}-i,&\quad z\in\overrightarrow{CA}\cup\overrightarrow{CB},\\
i,&\quad z\in\overrightarrow{CD}.
\end{cases}
\end{equation}
The matrix function $\mathbf{P}(z)$ is then transformed into another matrix function $\mathbf{Q}(z)$ by the invertible relation
\begin{equation}
\mathbf{Q}(z)=\beta(z)\mathbf{P}(z),\quad z\in\mathbb{C}\setminus\Sigma.
\end{equation}
The matrix function $\mathbf{Q}(z)$ is analytic where defined, takes continuous boundary values on $\Sigma$ except at the points $z=A$ and $z=B$ at which negative one-half power singularities are admissible, and satisfies the normalization condition
\begin{equation}
\mathbf{Q}(z)e^{wz\sigma_3/2}e^{\kappa_1z\sigma_3/\epsilon}=\mathbb{I}+\mathcal{O}\left(\frac{1}{z}\right),\quad z\to\infty.
\end{equation}
The other effect of the factor $\beta$ is that all jump conditions have the same involutive form:
$\mathbf{Q}_+(z)=\mathbf{Q}_-(z)\sigma_1$ holds on each arc of $\Sigma$.

The fact that the jump conditions for $\mathbf{Q}(z)$ correspond to a simple exchange of the columns across each arc of $\Sigma$ suggests that the two columns should be viewed as different branches of the same vector-valued function defined on a double covering of the complex $z$-plane.  To make this notion precise, recall the elliptic curve $\Gamma=\Gamma(x_0)$ consisting of pairs $P=(z,\mathscr{R})\in\mathbb{C}^2$ satisfying the relation $\mathscr{R}^2=(z-A)(z-B)(z-C)(z-D)$, and compactified with two points denoted $\infty^\pm$ corresponding to $z=\infty$.  The curve $\Gamma$ may be cut into two sheets, each of which may be identified with $\mathbb{C}\setminus\Sigma$ and on which $z$ is a suitable coordinate.  On the sheet $\Gamma^\pm$ containing the point $\infty^\pm$, $\mathscr{R}$ and $z$ are explicitly related by $\mathscr{R}=\pm R(z)$.  We now proceed to define several functions on $\Gamma$ from which we shall construct $\mathbf{Q}(z)$ later.

We first define meromorphic functions $f^\pm:\Gamma\to\overline{\mathbb{C}}$ by the formulae
\begin{equation}
f^\pm(P):=\frac{1}{2}\left(1\pm\frac{(z-C)(z-D)}{\mathscr{R}}\right).
\label{eq:g1-fpm-define}
\end{equation}
Because $A$, $B$, $C$, and $D$ are all distinct, these functions both have simple poles at the branching points $P=(A,0)$ and $P=(B,0)$ of $\Gamma$ and no other singularities (the simple pole nature of the singularities becomes obvious upon introducing appropriate holomorphic local coordinates in the neighborhood of these branch points).  Also,
\begin{equation}
f^\pm((C,0))=f^\pm((D,0))=\frac{1}{2},\quad f^\pm(\infty^\pm)=1,\quad\text{and}\quad f^\pm(\infty^\mp)=0.
\end{equation}
The zero of $f^\pm$ at $\infty^\mp$ is simple, and $f^\pm$ has exactly one other simple zero on $\Gamma$, denoted $Q^\pm$.  Note that the product $f^+(P)f^-(P)$ is a function of $z$ alone:
\begin{equation}
f^+(P)f^-(P)=\frac{1}{4}\left(1-\frac{(z-C)(z-D)}{(z-A)(z-B)}\right).
\end{equation}
Moreover, the function $4(z-A)(z-B)f^+(P)f^-(P)$ is analytic and non-vanishing at the branch points $z=A$ and $z=B$ and is given by
\begin{equation}
4(z-A)(z-B)f^+(P)f^-(P)=(z-A)(z-B)-(z-C)(z-D) = (C+D-A-B)z+(AB-CD).
\end{equation}
The points $Q^\pm$ therefore correspond to the same $z$-value:
\begin{equation}
z_Q:=\frac{AB-CD}{A+B-C-D}.
\end{equation}
This is a finite value, since the equation $A+B=C+D$ taken together with the moment condition $M_1(A,B,C,D)=0$ implies that $A+B=C+D=0$, and this is then inconsistent with the moment condition $M_3(A,B,C,D)=4$ (see \eqref{eq:g1-moments} and \eqref{eq:g1-moments-define}).
Also, $z_Q$ cannot coincide with any of the four values $A$, $B$, $C$, or $D$.  Since the functions $f^\pm$ are hyperelliptic involutes of each other, it follows that $Q^+$ and $Q^-$ are distinct points of $\Gamma$ related by hyperelliptic involution (corresponding to opposite finite and nonzero values of $\mathscr{R}$).  In the special case $x_0=0$, we have $z_Q=-\frac{1}{2}(\frac{4}{3})^{1/3}$, which lies to the left of the interval $[C,D]$.  Generalizing to nonzero real values of $x_0$ one sees easily that $z_Q$ is a real and continuous function of $x_0\in\mathbb{R}$ that necessarily satisfies 
$z_Q<C<D$ since $z_Q=C$  is impossible.  For real $x_0$, the function $R(z)$ is by definition positive real for $z<C$ (and for $z>D$), so it follows that $Q^\pm$ lies on the sheet $\Gamma^\mp$, that is, $Q^\pm=(z_Q,\mp R(z_Q))$.  More generally, from \eqref{eq:g1-fpm-define}, the condition $f^\pm(Q^\pm)=0$ implies that 
\begin{equation}
Q^\pm=(z_Q,\mp(z_Q-C)(z_Q-D)).
\end{equation}

Next we construct two \emph{Baker-Akhiezer functions} on $\Gamma$.  We begin by choosing a homology basis $\{\mathfrak{a},\mathfrak{b}\}$ of cycles on $\Gamma$ with the property that $\mathfrak{b}$ intersects $\mathfrak{a}$ exactly once from the right as $\mathfrak{a}$ is traversed according to its orientation (so $\mathfrak{a}\circ\mathfrak{b}=1$).  Define a normalization constant $c_1=c_1(x_0)$  by
\begin{equation}
c_1:=2\pi i\left(\oint_\mathfrak{a}\omega_0\right)^{-1}
\label{eq:g1-c1}
\end{equation}
so that the holomorphic differential 
\begin{equation}
\omega:=c_1\omega_0\quad\text{satisfies} \quad\oint_\mathfrak{a}\omega=2\pi i.
\label{eq:g1-omega-norm}
\end{equation}
The other period of $\omega$ is denoted $\mathcal{H}=\mathcal{H}(x_0)$:
\begin{equation}
\mathcal{H}:=\oint_\mathfrak{b}\omega.
\label{eq:g1-H-def}
\end{equation}
It is a basic result that, regardless of how the homology basis is chosen consistent with the condition $\mathfrak{a}\circ\mathfrak{b}=1$, $\mathrm{Re}(\mathcal{H})<0$ holds strictly.  We make a concrete choice of cycles as illustrated for the case of $x_0\in T\cap\mathbb{R}$ in Figure~\ref{fig:HomologyCycles}.  The homology basis deforms continuously with the branch points of $\Gamma$ as $x_0$ leaves the real axis (meaning that the integrals over $\mathfrak{a}$ and $\mathfrak{b}$ of $\omega_0$ defined by \eqref{eq:omega0define-gen1} vary continuously with $x_0$ in $T$).  
\begin{figure}[h]
\begin{center}
\includegraphics{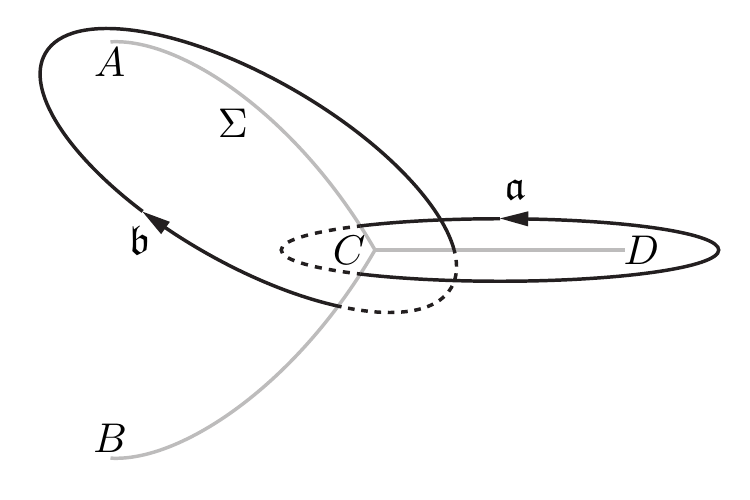}
\end{center}
\caption{\emph{The homology cycles $\mathfrak{a}$ and $\mathfrak{b}$ for a configuration of points $A$, $B$, $C$, and $D$ corresponding to $x_0\in T\cap\mathbb{R}$.   Paths on $\Gamma^+$ are shown with solid curves and paths on $\Gamma^-$ are shown with broken curves.  The Stokes graph $\Sigma$ is shown with shaded curves.}}
\label{fig:HomologyCycles}
\end{figure}
With this choice of homology, $c_1$ is real and strictly positive for all $x_0\in T\cap\mathbb{R}$.  Also, by deforming the cycle $\mathfrak{b}$ to pass through the branch points $A$, $B$, and $D$, it is not difficult to see that for general $x_0\in T$ we have
\begin{equation}
\mathcal{H}=i\pi +\frac{1}{2}\mathcal{H}_0,\quad\mathcal{H}_0=\mathcal{H}_0(x_0):=2c_1\int_A^D\frac{dz}{R(z)} + 2c_1\int_B^D\frac{dz}{R(z)}
\label{eq:g1-H-H0}
\end{equation}
where the paths of integration in the $z$-plane lie in the domain $\mathbb{C}\setminus\Sigma$.  It follows that for $x_0\in T\cap\mathbb{R}$, $\mathcal{H}_0$ is real and strictly negative (for more general $x_0\in T$ we simply have $\mathrm{Re}(\mathcal{H}_0)<0$).  An antiderivative of $\omega$ is given by the \emph{Abel map} $\mathscr{A}:\Gamma\to\mathbb{C}$ defined by
\begin{equation}
\mathscr{A}(P):=\int_{(D,0)}^P\omega.
\end{equation}
Actually, since integrals of $\omega$ have periods $2\pi i$ and $\mathcal{H}$, the Abel map is well-defined if the image is taken to be the Jacobian torus $\mathbb{C}/\{2\pi in_1+\mathcal{H}n_2 | n_j\in\mathbb{Z}\}$.  It is well-known that $\mathscr{A}$ is an injective map in this setting (it follows from Abel's Theorem that if $\mathscr{A}(P)=\mathscr{A}(Q)$ up to periods, then there is a nontrivial meromorphic function on $\Gamma$ with a simple pole at $P$ and a simple zero at $Q$, but such a function cannot exist unless $P=Q$ because the sum of residues of any meromorphic function on $\Gamma$ must vanish).  We choose to define a concrete value of $\mathscr{A}(P)$ for each $P=(z,\mathscr{R})$ for which $z\in\overline{\mathbb{C}}\setminus\Sigma$ by taking the path of integration to lie completely (after the initial point) in the sheet of $\Gamma$ containing the point $P$.  Next, consider the differential $\Omega_0$ defined on $\Gamma$ by the formula
\begin{equation}
\Omega_0:=\frac{z^2}{\mathscr{R}}\,dz.
\end{equation}
Making use of holomorphic local coordinates for $\Gamma$ at its four branch points, one sees that the only singularities of $\Omega_0$ on $\Gamma$ lie at the points $P=\infty^\pm$, and that
\begin{equation}
\Omega_0 = \pm\left(1+\mathcal{O}\left(\frac{1}{z^2}\right)\right)\,dz,\quad P\to\infty^\pm,
\end{equation}
showing that the singularities are double poles (in the local coordinate $\zeta:=1/z$) and that, as a consequence of the moment condition $M_1(A,B,C,D)=0$ (see \eqref{eq:g1-moments}), there are no residues; hence $\Omega_0$ is an abelian differential of the second kind on $\Gamma$.  Now let $c_2=c_2(x_0)$ be the normalization constant 
given by
\begin{equation}
c_2:=-\frac{1}{2\pi i}\oint_\mathfrak{a}\Omega_0,
\label{eq:g1-c2}
\end{equation} 
so that
\begin{equation}
\Omega:=\Omega_0+c_2 \omega\quad\text{satisfies}\quad\oint_\mathfrak{a}\Omega=0.
\label{eq:g1-Omega-norm}
\end{equation}
The other period of $\Omega$ is denoted $U=U(x_0)$:
\begin{equation}
U:=\oint_\mathfrak{b}\Omega.
\label{eq:frequencydef}
\end{equation}
By integrating the differential $\mathscr{A}\Omega$ around the boundary of the canonical dissection of $\Gamma$ obtained by cutting $\Gamma$ along the cycles of the homology basis (see \cite[Lemma B.1]{BuckinghamMMemoir}), one obtains the identity
\begin{equation}
U=2c_1.
\label{eq:Uc1-identity}
\end{equation}
Hence, $U(x_0)$ is real and strictly positive for all $x_0\in T\cap\mathbb{R}$.  Note that the limits
\begin{equation}
E^\pm:=\lim_{P\to\infty^\pm}\left[\int_{(D,0)}^P\Omega\mp z\right]
\end{equation}
both exist modulo integer multiples of $U$.  Choosing the path of integration to lie (except for the initial point) in the sheet $\Gamma^\pm$ gives unambiguous values to these limits:  $E^-=-E^+$, and
\begin{equation}
E^+=E^+(x_0):=\int_D^{z_0}\frac{z^2+c_1c_2}{R(z)}\,dz+\int_{z_0}^\infty\left[\frac{z^2+c_1c_2}{R(z)}-1\right]\,dz-z_0,
\label{eq:g1-Eplus-rewrite}
\end{equation}
where $z_0\in\mathbb{C}\setminus\Sigma$ is arbitrary.  It is obvious that $E^+$ is real and continuous for $x_0\in T\cap\mathbb{R}$, and by taking $z_0\ge D$ one sees that $E^+$ has a finite limit as $x_0$ decreases to the left-most real boundary point of $T$ (where $A$, $B$, and $C$ coalesce).  In fact, $E^+$ also has a finite limit as $x_0$ increases to the right-most real boundary point of $T$, or more generally when $x_0$ tends to a point on $\partial T$ in the sector $|\arg(x_0)|<\pi/3$ (where $C$ and $D$ coalesce); to see this one should add and subtract $\tfrac{1}{2}U$ with $U$ in the form \eqref{eq:frequencydef} to replace the  integral above that terminates at the branch point $D$ by one that is obviously convergent in this degeneration of $\Gamma$,  and then use the identity \eqref{eq:Uc1-identity} and observe that $c_1$ clearly remains finite.  In other words, we have the following two equivalent formulae for $E^+$:
\begin{equation}
\begin{split}
E^+&=\int_A^{z_0}\frac{z^2+c_1c_2}{R(z)}\,dz +\int_{z_0}^\infty\left[\frac{z^2+c_1c_2}{R(z)}-1\right]\,dz-z_0-c_1\\
&=\int_B^{z_0}\frac{z^2+c_1c_2}{R(z)}\,dz +\int_{z_0}^\infty\left[\frac{z^2+c_1c_2}{R(z)}-1\right]\,dz-z_0-c_1,
\end{split}
\label{eq:g1-Eplus-rewrite-again}
\end{equation}
where again $z_0\in\mathbb{C}\setminus\Sigma$ is arbitrary.  It is convenient to use \eqref{eq:g1-Eplus-rewrite} for $x_0$ bounded away from the edge of $T$ along which $C$ and $D$ coalesce, and the first (respectively second) of the formulae \eqref{eq:g1-Eplus-rewrite-again} for $x_0$ bounded away from the edge of $T$ along which $A$ (respectively $B$) and $C$ coalesce.  Thus, regardless of where $x_0$ lies within $T$, we have available a computationally robust formula for $E^+$.

The \emph{Riemann theta function} is an entire function of $z$ defined by the Fourier-type series\footnote{In the notation of \cite{DLMF}, $\Theta(z;\mathcal{H})=\theta_3(w|\tau)$, where $z=2iw$ and $\mathcal{H}=2\pi i\tau$.}:
\begin{equation}
\Theta(z)=\Theta(z;\mathcal{H}):=\sum_{n=-\infty}^\infty e^{\mathcal{H}n^2/2}e^{nz}.
\end{equation}
It satisfies the identities
\begin{equation}
\Theta(z+2\pi i)=\Theta(z),\quad\Theta(z+\mathcal{H})=e^{-\mathcal{H}/2}e^{-z}\Theta(z),\quad
\Theta(-z)=\Theta(z).
\label{eq:g1-theta-identities}
\end{equation}
Defining the \emph{Riemann constant} by
\begin{equation}
\mathcal{K}=\mathcal{K}(x_0):=i\pi +\frac{1}{2}\mathcal{H}=\frac{3}{2}i\pi +\frac{1}{4}\mathcal{H}_0,
\label{eq:g1-RiemannConstant}
\end{equation}
the function $\Theta(z)$ has simple zeros only at the points $z=\mathcal{K}+2\pi i n_1+\mathcal{H}n_2$ for $n_j\in\mathbb{Z}$.  

The Baker-Akhiezer functions $\varphi_0^\pm:\Gamma\setminus\{\infty^+,\infty^-\}\to\overline{\mathbb{C}}$ are explicitly written in terms of these ingredients by Krichever's formula:
\begin{equation}
\varphi_0^\pm(P):=\frac{\Theta(\mathscr{A}(P)-\mathscr{A}(Q^\pm)-\mathcal{K}-(\tfrac{1}{2}w+\epsilon^{-1}\kappa_1)U)}{\Theta(\mathscr{A}(P)-\mathscr{A}(Q^\pm)-\mathcal{K})}
\exp\left(-\left(\frac{1}{2}w+\epsilon^{-1}\kappa_1\right)\int_{(D,0)}^P\Omega\right),
\label{eq:g1-Krichever}
\end{equation}
in which it is understood that the path of integration in the exponent is arbitrary but is the same as the path used to define the Abel map $\mathscr{A}(P)$.  The values of $\mathscr{A}(Q^\pm)$ are determined by the concrete choice of the Abel map described above.  This formula defines $\varphi_0^\pm(P)$ unambiguously even if (once the value of $\mathscr{A}(Q^\pm)$ is fixed) the paths in the Abel map $\mathscr{A}(P)$ and the exponent are augmented by adding cycles.  Indeed, the three factors in the formula are individually invariant under adding  the $\mathfrak{a}$-cycle to each path, as this contributes nothing to the exponent by \eqref{eq:g1-Omega-norm} and increments $\mathscr{A}(P)$ by $2\pi i$ according to \eqref{eq:g1-omega-norm}, leaving each theta function invariant by \eqref{eq:g1-theta-identities}.  On the other hand, if the paths of integration are augmented with the $\mathfrak{b}$-cycle, the exponential acquires a factor of $\exp(-(\tfrac{1}{2}w+\epsilon^{-1}\kappa_1)U)$ according to \eqref{eq:frequencydef} while $\mathscr{A}(P)$ is incremented by $\mathcal{H}$ according to \eqref{eq:g1-H-def} and the ratio of theta functions acquires a compensating factor of $\exp((\tfrac{1}{2}w+\epsilon^{-1}\kappa_1)U)$ by \eqref{eq:g1-theta-identities}.  

Although the formula \eqref{eq:g1-Krichever} makes sense with an arbitrary choice of the path of integration in the Abel map $\mathscr{A}(P)$ (and in the exponent), for convenience when evaluating $\varphi_0^\pm(P)$ for points $P$ whose projections on the $z$-plane are disjoint from the Stokes graph $\Sigma$, \emph{we will always use the concrete definition of $\mathscr{A}(P)$ in which the path lies on the same sheet as does $P$ but is otherwise arbitrary}.  

The only singularities of the function $\varphi_0^\pm(P)$ on $\Gamma$ are at the points $P=Q^\pm$ (a simple pole) and $P=\infty^\pm$ (essential singularities).  The asymptotic behavior near the essential singularities is easy to determine:
\begin{equation}
\begin{split}
\varphi_0^\pm(P)\exp\left(\left(\frac{1}{2}w+\epsilon^{-1}\kappa_1\right)z\right)&=N_+^\pm+\mathcal{O}(z^{-1}),
\quad P\to\infty^+,\\
\varphi_0^\pm(P)\exp\left(-\left(\frac{1}{2}w+\epsilon^{-1}\kappa_1\right)z\right)&=N_-^\pm+\mathcal{O}(z^{-1}),
\quad P\to\infty^-,
\end{split}
\label{eq:g1-varphi0-essential}
\end{equation}
where
\begin{equation}
\begin{split}
N_+^\pm&:=\frac{\Theta(\mathscr{A}(\infty^+)-\mathscr{A}(Q^\pm)-\mathcal{K}-(\tfrac{1}{2}w+\epsilon^{-1}\kappa_1)U)}{\Theta(\mathscr{A}(\infty^+)-\mathscr{A}(Q^\pm)-\mathcal{K})}\exp\left(-\left(\frac{1}{2}w+\epsilon^{-1}\kappa_1\right)E^+\right),\\
N_-^\pm&:=\frac{\Theta(\mathscr{A}(\infty^-)-\mathscr{A}(Q^\pm)-\mathcal{K}-(\tfrac{1}{2}w+\epsilon^{-1}\kappa_1)U)}{\Theta(\mathscr{A}(\infty^-)-\mathscr{A}(Q^\pm)-\mathcal{K})}\exp\left(-\left(\frac{1}{2}w+\epsilon^{-1}\kappa_1\right)E^-\right).
\end{split}
\end{equation}
The denominators cannot vanish because $z(Q^\pm)=z_Q$ is necessarily a finite value.
The outer parametrix will exist as long as both $N^+_+$ and $N^-_-$ are nonzero, which is a condition on $w\in\mathbb{C}$.  The exceptional condition on $w$ that $N^+_+=0$ is exactly the same condition that $N^-_-=0$; this follows from the fact that
\begin{equation}
\left(\mathscr{A}(\infty^+)-\mathscr{A}(Q^+)\right)-\left(\mathscr{A}(\infty^-)-\mathscr{A}(Q^-)\right)=\mathscr{A}(\infty^+-Q^+-\infty^-+Q^-)=\mathscr{A}\left(\left(\frac{f^-}{f^+}\right)\right),
\label{eq:g1-AbelsTheorem-application}
\end{equation}
that is, the difference in arguments of the theta functions in the numerators of $N^+_+$ and $N^-_-$
is the Abel map evaluated on the (principal) divisor of a meromorphic function globally defined on $\Gamma$, namely $f^-(P)/f^+(P)$,
and this is necessarily an integer linear combination of the periods $2\pi i$ and $\mathcal{H}$ by Abel's Theorem.  The exceptional values of $w$ therefore form a regular lattice $\mathscr{P}_m=\mathscr{P}_m(x_0)$ in the $w$-plane defined by the condition
\begin{equation}
w\in\mathscr{P}_m(x_0)\quad\text{if and only if}\quad \frac{1}{2}Uw=\mathscr{A}(\infty^+)-\mathscr{A}(Q^+) -\frac{\kappa_1U}{\epsilon}\pmod{\mathbb{L}},\quad\epsilon = (m-\tfrac{1}{2})^{-1},
\label{eq:g1-Lambda-p}
\end{equation}
where $\mathbb{L}=\mathbb{L}(x_0)$ is the (algebraic) lattice defined as
\begin{equation}
\mathbb{L}(x_0):=\{2\pi i n_1+\mathcal{H}(x_0)n_2: (n_1,n_2)\in\mathbb{Z}^2\}.
\label{eq:g1-LatticeL-def}
\end{equation}

Assuming that $w\not\in\mathscr{P}_m$, we normalize the Baker-Akhiezer functions by setting
\begin{equation}
\varphi^\pm(P):=\frac{1}{N_\pm^\pm}\varphi_0^\pm(P),\quad P\in\Gamma\setminus\{\infty^+,\infty^-\}.
\end{equation}
We will now construct the matrix $\mathbf{Q}(z)$, and hence the outer parametrix $\dot{\mathbf{O}}^\mathrm{(out)}(z)$, out of the two scalar functions $f^+(z)\varphi^+(z)$ and $f^-(z)\varphi^-(z)$.  We simply set
\begin{equation}
\mathbf{Q}(z):=\begin{bmatrix}f^+((z,R(z)))\varphi^+((z,R(z))) & f^+((z,-R(z)))\varphi^+((z,-R(z)))\\
f^-((z,R(z)))\varphi^-((z,R(z))) & f^-((z,-R(z)))\varphi^-((z,-R(z)))\end{bmatrix},\quad z\in\mathbb{C}\setminus\Sigma.
\end{equation}
Therefore
\begin{equation}
\dot{\mathbf{O}}^\mathrm{(out)}(z)=\frac{1}{\beta(z)}e^{-\kappa_0\sigma_3/\epsilon}
\begin{bmatrix}f^+((z,R(z)))\varphi^+((z,R(z))) & f^+((z,-R(z)))\varphi^+((z,-R(z)))\\
f^-((z,R(z)))\varphi^-((z,R(z))) & f^-((z,-R(z)))\varphi^-((z,-R(z)))\end{bmatrix}e^{k(z)\sigma_3},\quad
z\in\mathbb{C}\setminus\Sigma.
\label{eq:g1-dot-Oout-formula}
\end{equation}

This construction, along with the basic conditions defining the outer parametrix,  yields the following important fact.  For each $x_0\in T$, $\epsilon>0$ (or alternately, $m\in\mathbb{Z}_+$), and $\delta>0$, let $\mathscr{S}_m(x_0,\delta)$ denote the subset of the complex $w$-plane defined by
\begin{equation}
\mathscr{S}_m(x_0,\delta):=\{w\in\mathbb{C}: |w|\le\delta^{-1}\text{ and }|w-\mathscr{P}_m(x_0)|\ge \delta\},
\label{eq:g1-cheese-def}
\end{equation}
where $|w-\mathscr{P}_m(x_0)|$ denotes the distance from $w$ to the closest point of the lattice $\mathscr{P}_m(x_0)$.  For $\delta$ sufficiently small, $\mathscr{S}_m(x_0,\delta)$ resembles a slice perpendicular to the axis from a wheel of Swiss cheese of radius $\delta^{-1}$ with omitted holes of radius $\delta$ centered at the points of the lattice $\mathscr{P}_m(x_0)$.
\begin{proposition}
Let $K\subset T$ be a compact subset of the open region $T$.  Fix a small constant $\delta>0$, and define the set $\Upsilon:=\{(z;x_0,m,w):  x_0\in K, m\in\mathbb{Z}_+,w\in\mathscr{S}_m(x_0,\delta), |z-\{A(x_0),B(x_0),C(x_0),D(x_0)\}|\ge \delta\}$, where $|z-\{A,B,C,D\}|$ denotes the distance from $z$ to the nearest of $\{A,B,C,D\}$.  Then $\dot{\mathbf{O}}^\mathrm{(out)}(z)$ is uniformly bounded on $\Upsilon$.
\label{prop:g1-dot-Oout-bound}
\end{proposition}
Note that the set $\Upsilon$ simultaneously bounds $z$ away from the $x_0$-dependent branch points $\{A,B,C,D\}$ and  restricts $w$ to lie within the bounded region $\mathscr{S}_m(x_0,\delta)$.  The key part of this result is the fact that, under the conditions defining $\Upsilon$, 
$\dot{\mathbf{O}}^{(\mathrm{out})}$ remains bounded as $m\to +\infty$, or equivalently as $\epsilon\downarrow 0$, a fact that does not seem obvious from the explicit formula \eqref{eq:g1-dot-Oout-formula}, which contains both exponential factors with large complex exponents and theta functions with large complex arguments.
\begin{proof}
The main idea is to note that  the parameter $\epsilon=(m-\tfrac{1}{2})^{-1}$ enters into the properties characterizing $\dot{\mathbf{O}}^{(\mathrm{out})}$ only through the two phase factors $e^{i\Phi_\pm(x_0)/\epsilon}$.  We may define real (due to the Boutroux conditions \eqref{eq:g1-Boutroux} satisfied by the elliptic curve $\Gamma=\Gamma(x_0)$) angles $\gamma_\pm:=\Phi_\pm(x_0)/\epsilon$ and generalize the formulation of the problem by allowing $\gamma_\pm$ to be independent (both of each other, and of $x_0$) real parameters.  In fact,  since the angles appear only in exponents, we may consider the pair $(\gamma_-,\gamma_+)$ to be an element of $\mathbb{T}^2$, the real two-dimensional torus, a compact manifold.  It is easy to check that all steps of the construction of $\dot{\mathbf{O}}^{(\mathrm{out})}$ in terms of Baker-Akhiezer functions go through in this generalization, the only effect being that the factors $\Phi_\pm(x_0)/\epsilon$ appearing in the definition of $k(z)$ and the ratios $\kappa_0/\epsilon$ and $\kappa_1/\epsilon$ are replaced by the angles $\gamma_\pm$.

It is then clear from the construction that $\dot{\mathbf{O}}^{(\mathrm{out})}$ is a continuous function of $x_0\in T$, $w\in\mathbb{C}$, and $(\gamma_-,\gamma_+)\in\mathbb{T}^2$, as long as $z$ avoids the branch points $\{A(x_0),B(x_0), C(x_0), D(x_0)\}$ and $w$ avoids the lattice of poles $\mathscr{P}_m$ defined by \eqref{eq:g1-Lambda-p}, which in the generalized context depends (only) on $x_0$, $\gamma_-$, and $\gamma_+$.  The branch points depend smoothly on $x_0\in T$, so given the compact set $K\subset T$ there exists a constant $M_K>0$ such that all four branch points lie within the disk $|z|<M_K$ whenever $x_0\in K\subset T$.  If we impose the conditions $x_0\in K\subset T$, $w\in\mathscr{S}_m(x_0,\delta)$ (note that in the generalized context, $\mathscr{S}_m(x_0,\delta)$ is a set in the complex $w$-plane depending only on $x_0$, $\delta$, $\gamma_-$, and $\gamma_+$), $|z-\{A(x_0),B(x_0),C(x_0),D(x_0)\}|\ge \delta$, and $|z|\le 2M_K$, we get a compact subset $\Upsilon_0'$ of
$\mathbb{C}^3\times\mathbb{T}^2\ni (z,x_0,w,(\gamma_-,\gamma_+))$ on which $\dot{\mathbf{O}}^{(\mathrm{out})}$ is continuous, and hence bounded.  Now since $\dot{\mathbf{O}}^{(\mathrm{out})}$ is analytic in $z$ for $|z|>M_K$ and by normalization $\dot{\mathbf{O}}^{(\mathrm{out})}=\mathbb{I}$ for $z=\infty$, an application of the Maximum Modulus Principle in the variable $1/z$ shows that the same bound holds on the noncompact set $\Upsilon'$ containing $\Upsilon'_0$ and obtained simply by dropping the condition $|z|\le 2M_K$.

With this bound in hand, we may restore the dependence on $\epsilon>0$ by substituting $\gamma_\pm=\Phi_\pm(x_0)/\epsilon$.  Thus, for a given value of $x_0$ we no longer have available the whole torus $\mathbb{T}^2$ but rather a linear orbit that is a proper subset of $\mathbb{T}^2$ (the orbit is dense if and only if the ratio $\Phi_+/\Phi_-$ is irrational and otherwise is periodic).  Thus the set $\Upsilon$ in the statement of the proposition is realized as a proper subset of $\Upsilon'$ on which boundedness of $\dot{\mathbf{O}}^{(\mathrm{out})}$ has been established.
\end{proof}

Another elementary consequence of the conditions defining $\dot{\mathbf{O}}^{(\mathrm{out})}(z)$ that is not completely obvious from the explicit formula \eqref{eq:g1-dot-Oout-formula} is the following.
\begin{proposition}
The outer parametrix satisfies $\det(\dot{\mathbf{O}}^\mathrm{(out)}(z))=1$ where defined.  
\label{prop:g1-dot-Oout-det}
\end{proposition}

\subsubsection{Properties of functions of $w$ derived from the outer parametrix}
Recall that, by definition (see \eqref{eq:g1-dotU-dotV-def}), we have
\begin{equation}
\dot{\mathcal{U}}^0_m(w;x_0):=e^{-\Lambda/2-1/3}\lim_{z\to\infty}z\dot{O}_{12}^\mathrm{(out)}(z)\quad\text{and}\quad
\dot{\mathcal{V}}^0_m(w;x_0):=e^{\Lambda/2+1/3}\lim_{z\to\infty}z\dot{O}_{21}^\mathrm{(out)}(z).
\end{equation}
Therefore, combining \eqref{eq:g1-k-expansion}, \eqref{eq:g1-fpm-define}, \eqref{eq:g1-varphi0-essential}, \eqref{eq:g1-dot-Oout-formula}, and using the facts that $R(z)=z^2+\mathcal{O}(1)$ and $\beta(z)=1+\mathcal{O}(z^{-1})$ as $z\to\infty$,
\begin{equation}
\begin{split}
\dot{\mathcal{U}}^0_m(w;x_0)&=e^{-\Lambda/2-1/3}e^{-2\kappa_0/\epsilon}\frac{N^+_-}{N^+_+}\lim_{P\to\infty^-}zf^+(P)\\
&=
e^{-\Lambda/2-1/3}\frac{C+D}{2}e^{-2\kappa_0/\epsilon}\frac{N^+_-}{N^+_+}\\
&=e^{-\Lambda/2-1/3}\frac{C+D}{2}\frac{\Theta(\mathscr{A}(\infty^+)-\mathscr{A}(Q^+)-\mathcal{K})}{\Theta(\mathscr{A}(\infty^-)-\mathscr{A}(Q^+)-\mathcal{K})}
\frac{\Theta(\mathscr{A}(\infty^-)-\mathscr{A}(Q^+)-\mathcal{K}-(\tfrac{1}{2}w+\epsilon^{-1}\kappa_1)U)}
{\Theta(\mathscr{A}(\infty^+)-\mathscr{A}(Q^+)-\mathcal{K}-(\tfrac{1}{2}w+\epsilon^{-1}\kappa_1)U)}\\
&\quad\quad\quad\quad{}\times
\exp\left(wE^++\frac{2}{\epsilon}\left(\kappa_1E^+-\kappa_0\right)\right)
\end{split}
\label{eq:g1-dotU-formula}
\end{equation}
and
\begin{equation}
\begin{split}
\dot{\mathcal{V}}^0_m(w;x_0)&=e^{\Lambda/2+1/3}e^{2\kappa_0/\epsilon}\frac{N^-_+}{N_-^-}\lim_{P\to\infty^+}zf^-(P)\\ &=
e^{\Lambda/2+1/3}\frac{C+D}{2}e^{2\kappa_0/\epsilon}\frac{N^-_+}{N_-^-}\\
&=e^{\Lambda/2+1/3}\frac{C+D}{2}\frac{\Theta(\mathscr{A}(\infty^-)-\mathscr{A}(Q^-)-\mathcal{K})}{\Theta(\mathscr{A}(\infty^+)-\mathscr{A}(Q^-)-\mathcal{K})}
\frac{\Theta(\mathscr{A}(\infty^+)-\mathscr{A}(Q^-)-\mathcal{K}-(\tfrac{1}{2}w+\epsilon^{-1}\kappa_1)U)}
{\Theta(\mathscr{A}(\infty^-)-\mathscr{A}(Q^-)-\mathcal{K}-(\tfrac{1}{2}w+\epsilon^{-1}\kappa_1)U)}\\
&\quad\quad\quad\quad{}\times\exp\left(-wE^+-\frac{2}{\epsilon}\left(\kappa_1E^+-\kappa_0\right)\right).
\end{split}
\label{eq:g1-dotV-formula}
\end{equation}
Using these formulae and recalling the relations \eqref{eq:g1-dotPQlogderivs} we obtain
\begin{multline}
\dot{\mathcal{P}}_m(w;x_0)=\\
E^+ +\frac{U}{2}\left[\frac{\Theta'(\mathscr{A}(\infty^+)-\mathscr{A}(Q^+)-\mathcal{K}-(\tfrac{1}{2}w+\epsilon^{-1}\kappa_1)U)}{\Theta(\mathscr{A}(\infty^+)-\mathscr{A}(Q^+)-\mathcal{K}-(\tfrac{1}{2}w+\epsilon^{-1}\kappa_1)U)}-
\frac{\Theta'(\mathscr{A}(\infty^-)-\mathscr{A}(Q^+)-\mathcal{K}-(\tfrac{1}{2}w+\epsilon^{-1}\kappa_1)U)}
{\Theta(\mathscr{A}(\infty^-)-\mathscr{A}(Q^+)-\mathcal{K}-(\tfrac{1}{2}w+\epsilon^{-1}\kappa_1)U)}\right]
\label{eq:g1-dotP-formula}
\end{multline}
and 
\begin{multline}
\dot{\mathcal{Q}}_m(w;x_0)=\\-E^+-\frac{U}{2}\left[\frac{\Theta'(\mathscr{A}(\infty^+)-\mathscr{A}(Q^-)-\mathcal{K}-(\tfrac{1}{2}w+\epsilon^{-1}\kappa_1)U)}{\Theta(\mathscr{A}(\infty^+)-\mathscr{A}(Q^-)-\mathcal{K}-(\tfrac{1}{2}w+\epsilon^{-1}\kappa_1)U)}-\frac{\Theta'(\mathscr{A}(\infty^-)-\mathscr{A}(Q^-)-\mathcal{K}-(\tfrac{1}{2}w+\epsilon^{-1}\kappa_1)U)}{\Theta(\mathscr{A}(\infty^-)-\mathscr{A}(Q^-)-\mathcal{K}-(\tfrac{1}{2}w+\epsilon^{-1}\kappa_1)U)}\right].
\label{eq:g1-dotQ-formula}
\end{multline}
All four of the functions $\dot{\mathcal{U}}_m^0(w;x_0)$, $\dot{\mathcal{V}}_m^0(w;x_0)$,
$\dot{\mathcal{P}}_m(w;x_0)$, and $\dot{\mathcal{Q}}_m^0(w;x_0)$ are meromorphic functions of $w\in\mathbb{C}$ for each $x_0\in T$ and $m\in\mathbb{Z}_+$ (recall that $\epsilon:=(m-\tfrac{1}{2})^{-1}$).
Note that from \eqref{eq:g1-dotP-formula} and \eqref{eq:g1-dotQ-formula} it follows that
\begin{equation}
\dot{\mathcal{Q}}_m(w;x_0)=-\dot{\mathcal{P}}_m(w + 2U^{-1}(\mathscr{A}(Q^-)-\mathscr{A}(Q^+));x_0),
\end{equation}
a fact that is consistent with Proposition~\ref{prop:g1-diffeq}.
The functions $\dot{\mathcal{U}}^0_m(w;x_0)$ and $\dot{\mathcal{V}}^0_m(w;x_0)$ have simple poles (only) at the points $w$ in the lattice $\mathscr{P}_m(x_0)$, and $\dot{\mathcal{U}}_m^0(w;x_0)$ has simple zeros (only) at the points of the lattice $\mathscr{Z}_m[\dot{\mathcal{U}}]=\mathscr{Z}_m[\dot{\mathcal{U}}](x_0)$ given by the
conditions
\begin{equation}
w\in\mathscr{Z}_m[\dot{\mathcal{U}}](x_0)\quad\text{if and only if}\quad
\frac{1}{2}Uw=\mathscr{A}(\infty^-)-\mathscr{A}(Q^+) -\frac{\kappa_1U}{\epsilon} \pmod{\mathbb{L}},\quad\epsilon = (m-\tfrac{1}{2})^{-1},
\label{eq:g1-Lambda-z}
\end{equation}
while $\dot{\mathcal{V}}^0_m(w;x_0)$ has simple zeros (only) at the points of the lattice $\mathscr{Z}_m[\dot{\mathcal{V}}]=\mathscr{Z}_m[\dot{\mathcal{V}}](x_0)$ given by the conditions
\begin{equation}
w\in\mathscr{Z}_m[\dot{\mathcal{V}}](x_0)\quad\text{if and only if}\quad \frac{1}{2}Uw=\mathscr{A}(\infty^+)-\mathscr{A}(Q^-) -\frac{\kappa_1U}{\epsilon} \pmod{\mathbb{L}},\quad
\epsilon=(m-\tfrac{1}{2})^{-1},
\label{eq:g1-Lambda-z-V}
\end{equation}
where we recall the definition \eqref{eq:g1-LatticeL-def} of $\mathbb{L}$.
By injectivity of the Abel map $\mathscr{A}$, neither $\mathscr{Z}_m[\dot{\mathcal{U}}](x_0)$ nor $\mathscr{Z}_m[\dot{\mathcal{V}}](x_0)$ can agree with $\mathscr{P}_m(x_0)$ for any $x_0\in T$.

In the special case that $x_0\in T\cap\mathbb{R}$, we recall that $\kappa_0$ and $\kappa_1$ are both real;  by our choice of homology basis $\mathcal{H}_0$ and $U$ are also both real; and by our specific choice of path for the integrals in the Abel map and the exponent of the Baker-Akhiezer functions, we also have $\mathscr{A}(\infty^\pm)$, $\mathscr{A}(Q^\pm)$, and $E^+$ all real.  Recalling the relation \eqref{eq:g1-H-H0}, the definition \eqref{eq:g1-RiemannConstant} of $\mathcal{K}$, and the fact that
\begin{equation}
\Theta(y-\tfrac{3}{2}i\pi;\mathcal{H})=\Theta(y-\tfrac{3}{2}i\pi;  i\pi +\tfrac{1}{2}\mathcal{H}_0)=\sum_{n=-\infty}^\infty i^{n^2+n}e^{\mathcal{H}_0n^2/4}e^{ny},
\end{equation}
we see that since $n^2+n$ is always even for integer $n$,  the above formulae for $\dot{\mathcal{U}}^0_m(w;x_0)$, $\dot{\mathcal{V}}^0_m(w;x_0)$,  $\dot{\mathcal{P}}_m(w;x_0)$, and 
$\dot{\mathcal{Q}}_m(w;x_0)$ are all real for $w\in\mathbb{R}$ and hence are meromorphic Schwarz-symmetric functions of $w$.

We already know from Proposition~\ref{prop:g1-diffeq} 
that $\dot{\mathcal{P}}_m(w)=\dot{\mathcal{P}}_m(w;x_0)$ is an elliptic function of $w$ with independent fundamental periods $2\pi i/c_1$ and $\mathcal{H}/c_1$, and this fact is also evident from the formula \eqref{eq:g1-dotP-formula}.  Indeed, the formula \eqref{eq:g1-dotP-formula} is explicitly $2\pi i/c_1$-periodic by \eqref{eq:Uc1-identity} and \eqref{eq:g1-theta-identities}.  
Next, observe from \eqref{eq:g1-theta-identities} that
\begin{equation}
p(z):=\frac{\Theta'(z)}{\Theta(z)}+\frac{z}{\mathcal{H}}
\end{equation}
is a function with period $\mathcal{H}$.  We can rewrite \eqref{eq:g1-dotP-formula} in the form
\begin{multline}
\dot{\mathcal{P}}_m(w)=E^+-\frac{\mathscr{A}(\infty^+)}{\mathcal{H}}U  \\{}+
\frac{U}{2}\left[p\left(\mathscr{A}(\infty^+)-\mathscr{A}(Q^+)-\mathcal{K}-(\tfrac{1}{2}w+\epsilon^{-1}\kappa_1)U\right)-
p\left(\mathscr{A}(\infty^-)-\mathscr{A}(Q^+)-\mathcal{K}-(\tfrac{1}{2}w+\epsilon^{-1}\kappa_1)U\right)\right],
\label{eq:g1-dotP-rewrite}
\end{multline}
where we have used the relation $\mathscr{A}(\infty^-)=-\mathscr{A}(\infty^+)$ holding for the choice of path we have described. 
It is then clear from this formula and \eqref{eq:Uc1-identity} that $\dot{\mathcal{P}}_m(w)$ is periodic with period $\mathcal{H}/c_1$.
From \eqref{eq:g1-H-H0} we see that
$\mathcal{H}_0/c_1$ is also a (non-fundamental) period of $\dot{\mathcal{P}}_m(w)$.  It follows that in the real case $x_0\in T\cap\mathbb{R}$, $\dot{\mathcal{P}}_m(w)$ is a real-valued, $\mathcal{H}_0/c_1$-periodic meromorphic function of $w\in\mathbb{R}$.  Of course similar statements hold for $\dot{\mathcal{Q}}_m(w)$.

The mean value  of the elliptic function $\dot{\mathcal{P}}_m(w;x_0)$ is defined as follows.  Fix $w_0\in\mathbb{C}$, and let $\pgram\subset\mathbb{C}$ denote the period parallelogram with vertices $w_0$, $w_0+2\pi i/c_1$, $w_0+\mathcal{H}/c_1$, and $w_0+2\pi i/c_1+\mathcal{H}/c_1$.
Then, for each fixed $x_0\in T$, we set
\begin{equation}
\langle\dot{\mathcal{P}}\rangle:=\frac{\displaystyle\iint_\pgram\dot{\mathcal{P}}_m(w;x_0)\,dA}
{\displaystyle\iint_\pgram\,dA}.
\label{eq:g1-dotP-average-define}
\end{equation}
Here, $dA=d\mathrm{Re}(w)\,d\mathrm{Im}(w)$ is the real area element in the plane.  The integral in the numerator is well-defined because as a function of $w$, $\dot{\mathcal{P}}_m$ has only simple poles and hence is absolutely integrable in two dimensions.  Also, by double-periodicity of $\dot{\mathcal{P}}_m$, $\langle\dot{\mathcal{P}}\rangle$ is independent of both $w_0\in\mathbb{C}$ and $m\in\mathbb{Z}_+$.

\begin{proposition}
The mean value of $\dot{\mathcal{P}}_m$ is a function of $x_0\in T$ alone,  given explicitly by
\begin{equation}
\langle\dot{\mathcal{P}}\rangle = E^+-2c_1\frac{\mathscr{A}(\infty^+)+\mathscr{A}(\infty^+)^*}{\mathcal{H}+\mathcal{H}^*}.
\label{eq:g1-dotPaverage-final}
\end{equation}
\label{prop:g1-dotPaverage}
\end{proposition}
\begin{proof}
Fix $x_0\in T$.  Without loss of generality, we choose $w_0$ so that $\dot{\mathcal{P}}_m$ is bounded along the boundary $\partial\pgram$ of the parallelogram $\pgram$.
The parallelogram $\pgram$ then contains in its interior exactly one point each of the
lattices $\mathscr{Z}_m[\dot{\mathcal{U}}]$ and $\mathscr{P}_m$, which we denote by $w_\mathrm{z}$ and $w_\mathrm{p}$ respectively.  Let $B_\delta(w)$ denote the open disk of radius $\delta>0$ centered at $w$.  Then we have
\begin{equation}
\langle\dot{\mathcal{P}}\rangle = \frac{\displaystyle \lim_{\delta\downarrow 0}\iint_{\pgram\setminus (B_\delta(w_\mathrm{z})\cup B_\delta(w_\mathrm{p}))}\dot{\mathcal{P}}_m(w)\,dA}
{\displaystyle \iint_\pgram\,dA}=\frac{\displaystyle\lim_{\delta\downarrow 0}\iint_{\pgram\setminus (B_\delta(w_\mathrm{z})\cup B_\delta(w_\mathrm{p}))}\overline{\partial}\left(w^*\dot{\mathcal{P}}_m(w)\right)\,dA}{\displaystyle \iint_\pgram \overline{\partial}(w^*)\,dA},
\end{equation}
where $\overline{\partial}$ denotes the partial derivative with respect to $w^*$ (holding $w$ fixed).  By Stokes' Theorem, we then obtain
\begin{equation}
\langle\dot{\mathcal{P}}\rangle=\frac{\displaystyle \oint_{\partial\pgram}w^*\dot{\mathcal{P}}_m(w)\,dw -
\lim_{\delta\downarrow 0}\left[\oint_{\partial B_\delta(w_\mathrm{z})}
w^*\dot{\mathcal{P}}_m(w)\,dw
+\oint_{\partial B_\delta(w_\mathrm{p})}w^*\dot{\mathcal{P}}_m(w)\,dw\right]}{\displaystyle \oint_{\partial\pgram}w^*\,dw}.
\end{equation}
But $\dot{\mathcal{P}}_m(w)$ has residue $1$ at the point $w_\mathrm{z}$ and has residue $-1$ at the point $w_\mathrm{p}$, and as both poles are simple it follows that
\begin{equation}
\lim_{\delta\downarrow 0}\oint_{\partial B_\delta(w_\mathrm{z})}w^*\dot{\mathcal{P}}_m(w)\,dw=
2\pi iw_\mathrm{z}^*\quad\text{and}\quad
\lim_{\delta\downarrow 0}\oint_{\partial B_\delta(w_\mathrm{p})}w^*\dot{\mathcal{P}}_m(w)\,dw=
-2\pi iw_\mathrm{p}^*.
\end{equation}
Therefore, since $w_\mathrm{z}\in\mathscr{Z}_m[\dot{\mathcal{U}}]$ while $w_\mathrm{p}\in\mathscr{P}_m$, we see from \eqref{eq:g1-Lambda-p} and \eqref{eq:g1-Lambda-z} that
\begin{equation}
\begin{split}
\lim_{\delta\downarrow 0}\left[\oint_{\partial B_\delta(w_\mathrm{z})}w^*\dot{\mathcal{P}}_m(w)\,dw +\oint_{\partial B_\delta(w_\mathrm{p})}w^*\dot{\mathcal{P}}_m(w)\,dw\right]&=-2\pi i\cdot(w_\mathrm{p}-w_\mathrm{z})^*\\&=-2\pi i\frac{2}{U^*}\left[(\mathscr{A}(\infty^+)-\mathscr{A}(\infty^-))^*-2\pi in_1 +\mathcal{H}^*n_2\right]\\&=
-\frac{2\pi i}{c_1^*}\left[(\mathscr{A}(\infty^+)-\mathscr{A}(\infty^-))^*-2\pi i n_1 +\mathcal{H}^*n_2\right],
\end{split}
\label{eq:g1-average-deltaloops}
\end{equation}
for some integers $n_1$ and $n_2$, where we have used \eqref{eq:Uc1-identity} to express $U$ in terms of $c_1$.
Next, 
\begin{equation}
\oint_{\partial\pgram}w^*\,dw=\int_{w_0}^{w_0+2\pi i/c_1}\left(w^*-(w+\mathcal{H}/c_1)^*\right)\,dw +\int_{w_0+\mathcal{H}/c_1}^{w_0}\left(w^*-(w+2\pi i/c_1)^*\right)\,dw
= -\frac{2\pi i (\mathcal{H}+\mathcal{H}^*)}{|c_1|^2}.
\label{eq:g1-dotP-average-denominator}
\end{equation}
It remains to calculate the integral around $\partial\pgram$ of $w^*\dot{\mathcal{P}}_m(w)$.  By double-periodicity of $\dot{\mathcal{P}}_m(w)$, 
\begin{equation}\begin{split}
\oint_{\partial\pgram}w^*\dot{\mathcal{P}}_m(w)\,dw &= \int_{w_0}^{w_0+2\pi i/c_1}\left(w^*-(w+\mathcal{H}/c_1)^*\right)\dot{\mathcal{P}}_m(w)\,dw +\int_{w_0+\mathcal{H}/c_1}^{w_0}\left(w^*-(w+2\pi i/c_1)^*\right)\dot{\mathcal{P}}_m(w)\,dw\\
&=-\frac{\mathcal{H}^*}{c_1^*}\int_{w_0}^{w_0+2\pi i/c_1}\dot{\mathcal{P}}_m(w)\,dw
+\frac{2\pi i}{c_1^*}\int_{w_0+\mathcal{H}/c_1}^{w_0}\dot{\mathcal{P}}_m(w)\,dw.
\end{split}
\label{eq:g1-wbar-dotP-boundary}
\end{equation}
We evaluate this expression with the help of the formula \eqref{eq:g1-dotP-formula}.  The constant term $E^+$ contributes to the integral on the left-hand side of \eqref{eq:g1-wbar-dotP-boundary} the product of $E^+$ and the right-hand side of \eqref{eq:g1-dotP-average-denominator}.  The remaining terms in $\dot{\mathcal{P}}_m(w)$ constitute the logarithmic derivative with respect to $w$ of a ratio of theta functions, and hence their contribution to the integrals on the right-hand side of \eqref{eq:g1-wbar-dotP-boundary} can be evaluated by the Fundamental Theorem of Calculus (and a careful accounting of branches of the logarithm) together with the automorphic identities \eqref{eq:g1-theta-identities}.  The result is that
\begin{equation}
\oint_{\partial\pgram}w^*\dot{\mathcal{P}}_m(w)\,dw = -\frac{2\pi i(\mathcal{H}+\mathcal{H}^*)E^+}{|c_1|^2} -\frac{\mathcal{H}^*}{c_1^*}\cdot 2\pi i n_1'
+\frac{2\pi i}{c_1^*}\left( \mathscr{A}(\infty^+)-\mathscr{A}(\infty^-)+ 2\pi i n_2'\right),
\end{equation}
where $n_1'$ and $n_2'$ are two specific integers.  
Putting the results together and using the fact that $\mathscr{A}(\infty^-)=-\mathscr{A}(\infty^+)$ for the specific branch of the Abel mapping in force yields
\begin{equation}
\langle\dot{\mathcal{P}}\rangle =E^+-
\frac{c_1}{\mathcal{H}+\mathcal{H}^*}\left[2\mathscr{A}(\infty^+) +2\mathscr{A}(\infty^+)^*-2\pi i(n_1-n_2') +\mathcal{H}^*(n_2-n_1')\right].
\label{eq:g1-dotP-average-penultimate}
\end{equation}

The fact that only the combinations of integers $n_1-n_2'$ and $n_2-n_1'$ enter is related to the arbitrary nature of the base point $w_0$ for the parallelogram $\pgram$.  Indeed, varying $w_0$ can effectively drive one of the points $w_\mathrm{z}$ and $w_\mathrm{p}$ to $\partial\pgram$, and as a pole of $\dot{\mathcal{P}}_m$ passes through $\partial\pgram$ a new point of the corresponding lattice $\mathscr{Z}_m[\dot{\mathcal{U}}]$ or $\mathscr{P}_m$ enters through the opposite edge of $\partial\pgram$.  This simultaneously amounts to an increment/decrement in either $n_1$ or $n_2$ (to re-calculate the difference $w_\mathrm{p}-w_\mathrm{z}$) and a corresponding change in either $n_2'$ or $n_1'$ (a residue contribution to one of the two integrals on the right-hand side of \eqref{eq:g1-wbar-dotP-boundary}).  Therefore, although the integers $n_1$, $n_2$, $n_1'$, and $n_2'$ can jump as $w_0$ varies, the differences $n_1-n_2'$ and $n_2-n_1'$ remain fixed.  A similar argument holds for variations of $x_0$, which result in corresponding variations of $c_1$, $\mathcal{H}$, $E^+$, and $\mathscr{A}(\infty^\pm)$, and/or variations of $\epsilon> 0$ taken as a continuous parameter (which also affect the lattices $\mathscr{P}_m$ and $\mathscr{Z}_m[\dot{\mathcal{U}}]$ through $m=\tfrac{1}{2}+\epsilon^{-1}$); the differences $n_1-n_2'$ and $n_2-n_1'$ are also independent of $x_0$ and $\epsilon>0$ (hence also $m$).
To determine the way that $\langle\dot{\mathcal{P}}\rangle$ depends on $x_0\in T$, it therefore suffices to 
calculate the differences $n_1-n_2'$ and $n_2-n_1'$ for any fixed $x_0\in T$, say $x_0=0$, and arbitrary $\epsilon>0$.

When $x_0=0$, we can use the symmetry $R(e^{2\pi i/3}z)=e^{-2\pi i/3}R(z)$ to evaluate $\mathcal{H}_0\in\mathbb{R}$ and $\mathscr{A}(\infty^+)\in\mathbb{R}$.  Indeed, this symmetry yields
\begin{equation}
\mathscr{A}(\infty^+):=c_1\int_D^\infty\frac{dz}{R(z)} = c_1e^{2\pi i/3}\int_A^\infty\frac{dz}{R(z)}.
\end{equation}
Therefore, by Cauchy's Theorem,
\begin{equation}
(1-e^{-2\pi i/3})\mathscr{A}(\infty^+)=c_1\int_D^A\frac{dz}{R(z)},
\end{equation}
and by our choice of homology cycles and the definitions of $c_1$ and $\mathcal{H}$ the right-hand side is equal to $i\pi-\tfrac12\mathcal{H}$.  Writing $\mathcal{H}$ in terms of $\mathcal{H}_0\in\mathbb{R}$ yields the identity
\begin{equation}
(1-e^{-2\pi i/3})\mathscr{A}(\infty^+)=\frac{1}{2}\pi i -\frac{1}{4}\mathcal{H}_0.
\label{eq:g1-x0zero-Ainftyplus}
\end{equation}
We wish to represent $\mathscr{A}(\infty^+)$ uniquely in the form $\mathscr{A}(\infty^+)=2\pi i\alpha +\mathcal{H}\beta$ for real $\alpha$ and $\beta$.  Taking the imaginary part of this equation using the fact that $\mathscr{A}(\infty^+)\in\mathbb{R}$ and the representation $\mathcal{H}=i\pi +\tfrac12\mathcal{H}_0$ with $\mathcal{H}_0\in\mathbb{R}$ yields $\beta=-2\alpha$, and then using \eqref{eq:g1-x0zero-Ainftyplus} allows us to solve for both $\mathcal{H}_0$ and $\alpha$.  Thus we obtain
\begin{equation}
\mathcal{H}_0=-2\pi\sqrt{3}\quad\text{and}\quad \mathscr{A}(\infty^+)=\frac{1}{3}\pi\sqrt{3}.
\label{eq:g1-x0zero-Adiff}
\end{equation}
This information allows us to calculate the fundamental lattice vectors $2\pi i/c_1$ and $\mathcal{H}/c_1$ and the relative shift $(\mathscr{A}(\infty^+)-\mathscr{A}(\infty^-))/c_1=2\mathscr{A}(\infty^+)/c_1$ of the zero and pole lattices $\mathscr{Z}_m[\dot{\mathcal{U}}](0)$ and $\mathscr{P}_m(0)$, with the result being as illustrated in Figure~\ref{fig:g1-x0zero-lattices}.
\begin{figure}[h]
\begin{center}
\includegraphics{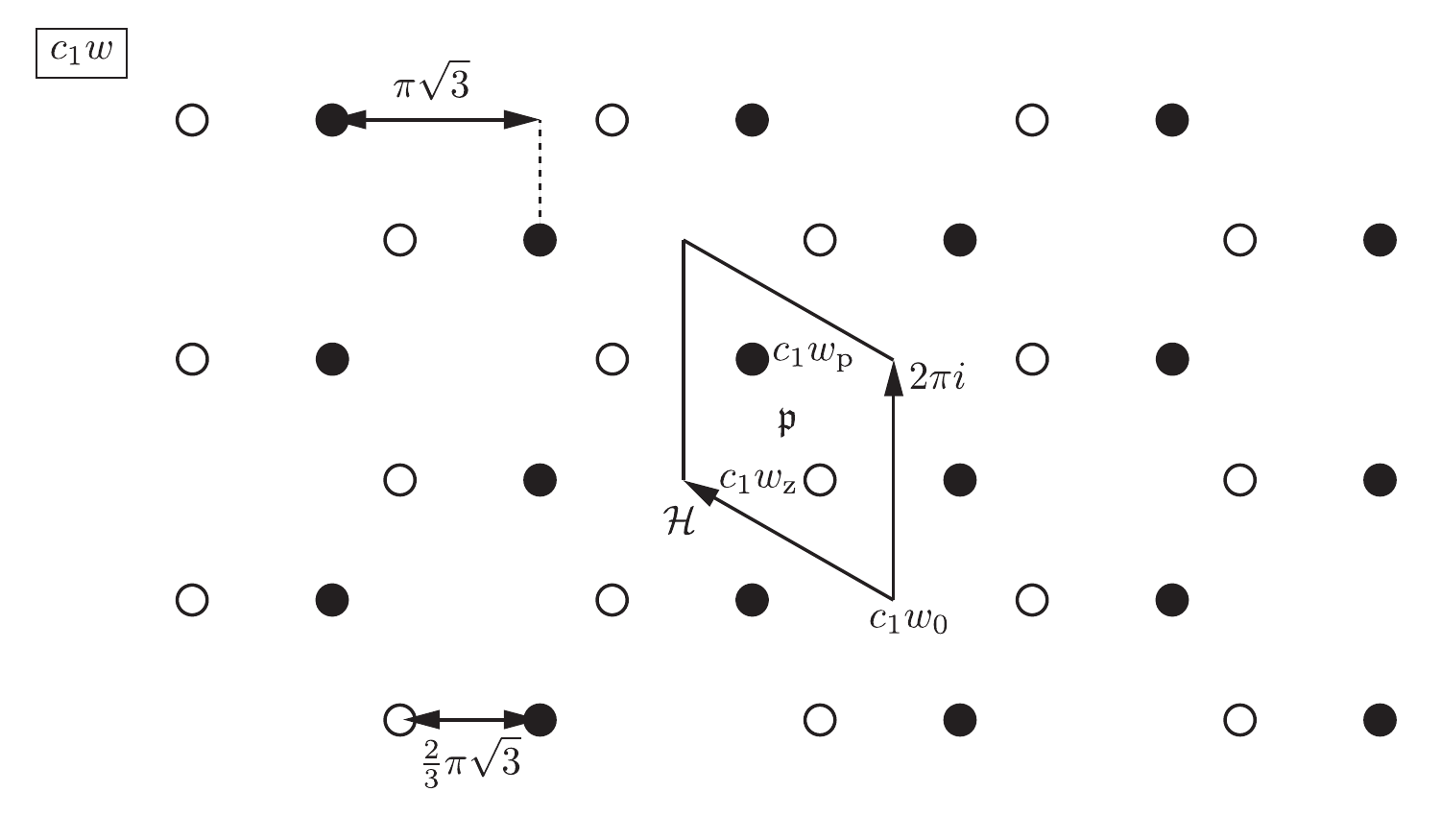}
\end{center}
\caption{\emph{The lattices $\mathscr{Z}_m[\dot{\mathcal{U}}](x_0)$ (open circles) and $\mathscr{P}_m(x_0)$ (filled circles) are regular hexagonal lattices in the $c_1w$-plane for $x_0=0$.  Also shown is the period parallelogram $\pgram$ and the distinguished points $w_0$, $w_\mathrm{z}\in\mathscr{Z}_m[\dot{\mathcal{U}}](0)$, and $w_\mathrm{p}\in\mathscr{P}_m(0)$.}}
\label{fig:g1-x0zero-lattices}
\end{figure}
Consider a period parallelogram $\pgram$ with base point $w_0$ set in the $c_1w$-plane as indicated in Figure~\ref{fig:g1-x0zero-lattices}.  To calculate the integers $n_1$ and $n_2$ we need to compute the difference between the specific lattice representatives $w_\mathrm{z}\in\mathscr{Z}_m[\dot{\mathcal{U}}](0)\cap\pgram$ and $w_\mathrm{p}\in\mathscr{P}_m(0)\cap\pgram$ that lie within $\pgram$.  From the diagram and \eqref{eq:g1-x0zero-Adiff} it is clear that
\begin{equation}
w_\mathrm{p}-w_\mathrm{z}=\frac{1}{c_1}\left(i\pi-\frac{1}{3}\pi\sqrt{3}\right)=\frac{1}{c_1}\left(2\mathscr{A}(\infty^+)+\mathcal{H}\right)=\frac{1}{c_1}\left(\mathscr{A}(\infty^+)-\mathscr{A}(\infty^-)+\mathcal{H}\right)
\end{equation}
so, comparing with \eqref{eq:g1-average-deltaloops}, we see that $n_1=0$ but $n_2=1$.  Now, to calculate $n_1'$ and $n_2'$ we need to consider how the (multivalued) logarithm of the ratio 
\begin{equation}
\rho(w):=\frac{\Theta(c_1 w;\mathcal{H})}{\Theta(c_1 w-2\mathscr{A}(\infty^+);\mathcal{H})},\quad
\text{for}\quad\mathcal{H}=i\pi-\pi\sqrt{3}\quad\text{and}\quad 2\mathscr{A}(\infty^+)=\frac{2}{3}\pi\sqrt{3}
\label{eq:g1-rhodef}
\end{equation}
varies as $w$ varies from $w_0$ to $w_0+2\pi i/c_1$ (to calculate $n_1'$) or from $w_0+\mathcal{H}/c_1$ to $w_0$ (to calculate $n_2'$).  
It is easy to deduce the correct (integer-valued) winding numbers from a sufficiently resolved numerical calculation of the theta series to track the way that $\rho(w)$ varies; the result of these numerical calculations is that $n_1'=1$ while $n_2'=0$.  Therefore $n_1-n_2'=0-0=0$ while $n_2-n_1'=1-1=0$.  Using this information in \eqref{eq:g1-dotP-average-penultimate} completes the proof. 
\end{proof}

Recall that, for fixed $x_0\in T\cap\mathbb{R}$, $\dot{\mathcal{P}}_m$ is a real periodic function of $w\in\mathbb{R}$ with fundamental period $-\mathcal{H}_0/c_1>0$.  We may therefore try also to define an average of this function over real $w$.  Since $\dot{\mathcal{P}}_m$ has real simple poles (exactly two per fundamental period, with opposite residues), the integrals involved in the average have to be regularized, and we choose the Cauchy principal value to preserve reality.  Thus we define
\begin{equation}
\langle\dot{\mathcal{P}}\rangle_\mathbb{R}:=-\frac{c_1}{\mathcal{H}_0}\dashint_{w_0}^{w_0-\mathcal{H}_0/c_1}\dot{\mathcal{P}}_m(w;x_0)\,dw,\quad x_0\in T\cap\mathbb{R},
\label{eq:g1-dotP-real-average}
\end{equation}
where $w_0$ is an arbitrary real point disjoint from the lattices $\mathscr{Z}_m[\dot{\mathcal{U}}](x_0)$ and $\mathscr{P}_m(x_0)$.
Since averaging over a curve in the complex $w$-plane is quite a different thing from averaging over a two-dimensional region, one generally cannot expect any relation between $\langle\dot{\mathcal{P}}\rangle_\mathbb{R}$ and $\langle\dot{\mathcal{P}}\rangle$.  Nonetheless, the following shows that the two are the same for those $x_0$ where they can be compared:
\begin{proposition}
\begin{equation}
\langle\dot{\mathcal{P}}\rangle_\mathbb{R}=\langle\dot{\mathcal{P}}\rangle \text{ for all } x_0\in T\cap\mathbb{R}.
\end{equation}
\label{prop:g1-equivalent-averages}
\end{proposition}
\begin{proof}
From \eqref{eq:g1-dotP-formula} and \eqref{eq:g1-rhodef}, we can write
\begin{equation}
\langle\dot{\mathcal{P}}\rangle_\mathbb{R}=-\frac{c_1}{\mathcal{H}_0}\dashint_{w_0}^{w_0-\mathcal{H}_0/c_1}\left[E^++\frac{d}{dw}\log|\rho(w)|\right]\,dw = 
E^+-\frac{c_1}{\mathcal{H}_0}\dashint_{w_0}^{w_0-\mathcal{H}_0/c_1}\frac{d}{dw}\log|\rho(w)|\,dw,
\end{equation}
where $\log|\rho(w)|$ denotes the real-valued principal branch, which is differentiable away from the pole at $w=w_\mathrm{p}$ and the zero at $w=w_\mathrm{z}$ of $\rho(w)$, the only singularities of the integrand in the interval of integration.  Assuming without loss of generality that $w_0<w_\mathrm{z}<w_\mathrm{p}<w_0-\mathcal{H}_0/c_1$, we then obtain, by the Fundamental Theorem of Calculus applied separately on the intervals $[w_0,w_\mathrm{z}-\delta]$, $[w_\mathrm{z}+\delta,w_\mathrm{p}-\delta]$, and $[w_\mathrm{p}+\delta,w_0-\mathcal{H}_0/c_1]$ for $\delta>0$ small,
\begin{equation}
\begin{split}
\langle\dot{\mathcal{P}}\rangle_\mathbb{R}&=E^+-\frac{c_1}{\mathcal{H}_0}\lim_{\delta\downarrow 0}\bigg(\log|\rho(w_\mathrm{z}-\delta)|-\log|\rho(w_0)| \\
&\qquad\qquad\qquad{}+\log|\rho(w_\mathrm{p}-\delta)|-\log|\rho(w_\mathrm{z}+\delta)| \\
&\qquad\qquad\qquad{}+\log|\rho(w_0-\mathcal{H}_0/c_1)|-\log|\rho(w_\mathrm{p}+\delta)|\bigg)\\
&=E^+-\frac{c_1}{\mathcal{H}_0}\bigg(\log|\rho(w_0-\mathcal{H}_0/c_1)|-\log|\rho(w_0)|\bigg)\\
&=E^+-\frac{c_1}{\mathcal{H}_0}\log\left|\frac{\rho(w_0-\mathcal{H}_0/c_1)}{\rho(w_0)}\right|.
\end{split}
\end{equation}
Now using the fact that $\mathcal{H}_0=2\mathcal{H}-2\pi i$ and applying the automorphic identities \eqref{eq:g1-theta-identities} to the formula \eqref{eq:g1-rhodef}, one sees that
\begin{equation}
\frac{\rho(w_0-\mathcal{H}_0/c_1)}{\rho(w_0)}=e^{4\mathscr{A}(\infty^+)}.
\end{equation}
Of course $\mathscr{A}(\infty^+)\in\mathbb{R}$ for $x_0\in T\cap \mathbb{R}$, so we have shown that
\begin{equation}
\langle\dot{\mathcal{P}}\rangle_\mathbb{R}=E^+-4c_1\frac{\mathscr{A}(\infty^+)}{\mathcal{H}_0}.
\end{equation}
The result then follows from the observation that, for $x_0\in T\cap \mathbb{R}$, $\mathcal{H}_0\in\mathbb{R}$, so $\mathcal{H}_0=\mathcal{H}+\mathcal{H}^*$ and $2\mathscr{A}(\infty^+)=\mathscr{A}(\infty^+)+\mathscr{A}(\infty^+)^*$.
\end{proof}

\begin{proposition}
Let $x\in \partial T$ be a point on the boundary of the elliptic region.  Then
\begin{equation}
\mathop{\lim_{x_0\to x}}_{x_0\in T}\langle\dot{\mathcal{P}}\rangle(x_0) = -\frac{1}{2}S(x),
\label{eq:g1-average-T-boundary}
\end{equation}
where $S(x)$ denotes the solution of the cubic equation \eqref{cubic-equation} with large $x$ asymptotic behavior \eqref{S-large-x} continued analytically to $\partial T$.  Also,
\begin{equation}
\mathop{\lim_{x_0\to x_c}}_{x_0\in T\cap\mathbb{R}}\langle\dot{\mathcal{P}}\rangle_\mathbb{R}(x_0)=-\frac{1}{2}S(x_c)\quad\text{and}\quad
\mathop{\lim_{x_0\to x_e}}_{x_0\in T\cap\mathbb{R}}\langle\dot{\mathcal{P}}\rangle_\mathbb{R}(x_0)=-\frac{1}{2}S(x_e),
\label{eq:g1-average-real-boundary}
\end{equation}
where we recall that $x_c:=\inf (T\cap\mathbb{R})$ and $x_e:=\sup (T\cap\mathbb{R})$.
\label{prop:g1-average-boundary}
\end{proposition}
\begin{proof}
According to Proposition~\ref{prop:g1-equivalent-averages}, the second statement \eqref{eq:g1-average-real-boundary} follows from the first, so it suffices to prove \eqref{eq:g1-average-T-boundary}.
We first suppose that $|\arg(x)|<\pi/3$,
so that the limit $x_0\to x$ corresponds to $A\to a(x)$ and $B\to b(x)$, while $C\to -S(x)/2$ and $D\to -S(x)/2$ and $a(x)+b(x)=S(x)$.
The proof relies on the observation that, for the values of $z$ that will be needed below, $R(z)$ degenerates to the product $(z+S/2)r(z)$ in the limit $x_0\to x$, where $r$ is the simpler quadratic radical with branch points $a$ and $b$ defined in \S\ref{section-gen0}.
Directly from \eqref{eq:g1-c1} and \eqref{eq:g1-c2}, passing to the limit under the integral sign and evaluating the resulting integrals by residues at the pole arising from the coalescence of $C$ and $D$ at $-S/2$, we have
\begin{equation}
\mathop{\lim_{x_0\to x}}_{x_0\in T} c_1= r(-S/2)\quad\text{and}\quad\mathop{\lim_{x_0\to x}}_{x_0\in T}c_2=-\frac{S^2}{4r(-S/2)}.
\end{equation}
Also, the real part of $\mathcal{H}$ tends to negative infinity as $x_0\to x$.  Since by simple contour deformations we can write $\mathscr{A}(\infty^+)$ as a sum of $-\mathcal{H}/2$ and an integral that remains bounded as $x_0\to x$, we therefore obtain
\begin{equation}
\mathop{\lim_{x_0\to x}}_{x_0\in T}\langle\dot{\mathcal{P}}\rangle=r(-S/2)+\mathop{\lim_{x_0\to x}}_{x_0\in T}E^+.
\end{equation}
To compute the limit of $E^+$, we average the formulae in \eqref{eq:g1-Eplus-rewrite-again} suitable for use near the arc of $\partial T$ with $|\arg(x)|<\pi/3$ and pass to the limit to obtain
\begin{equation}
\mathop{\lim_{x_0\to x}}_{x_0\in T}E^+=\frac{1}{2}\int_{a}^{z_0}\frac{z-S/2}{r(z)}\,dz +\frac{1}{2}\int_b^{z_0}\frac{z-S/2}{r(z)}\,dz +\int_{z_0}^\infty\left[\frac{z-S/2}{r(z)}-1\right]\,dz-z_0-r(-S/2),
\end{equation}
where we have used the fact that $c_1c_2$ tends to $-S^2/4$.  Now it only remains to observe that
$r'(z)=(z-S/2)/r(z)$ to integrate exactly; since $r(a)=r(b)=0$ one finds that
\begin{equation}
\mathop{\lim_{x_0\to x}}_{x_0\in T}E^+=-\frac{S}{2}-r(-S/2)
\end{equation}
(the term $-S/2$ arises from the limit of $r(z)-z$ as $z\to\infty$ that is needed to compute the convergent improper integral) and the proof is therefore complete for $x\in \partial T$ with $|\arg(x)|<\pi/3$.

If instead $x$ lies on one of the other two arcs of $\partial T$, then it is either $A$ or $B$ coalescing with $C$ at $-S(x)/2$, and we denote the limiting value of $D$ as $b$ and the limiting value of the other non-coalescing root as $a$; we again have $a+b=S$.  In this situation, one can easily show that the second term in the formula \eqref{eq:g1-dotPaverage-final} tends to zero as $x_0\to x$, so the limit we seek to compute is just that of $E^+$.  To compute this limit, we first write $E^+$ in the form \eqref{eq:g1-Eplus-rewrite}, and observe that
although $c_1$ tends to zero while $c_2$ blows up, the product $c_1c_2$ tends to a finite limit of $-S^2/4$ as $x_0\to x$.  Therefore, we may pass to the limit to obtain
\begin{equation}
\mathop{\lim_{x_0\to x}}_{x_0\in T} E^+=\int_b^{z_0}\frac{z-S/2}{r(z)}\,dz + \int_{z_0}^\infty\left[\frac{z-S/2}{r(z)}-1\right]\,dz-z_0.
\end{equation}
Evaluating the integrals in closed form using $r'(z)=(z-S/2)/r(z)$ gives the desired result.
\end{proof}

Next, given $x_0\in T$, we may define the planar density $\dot{\sigma}_\mathrm{P}$ of poles of $\dot{\mathcal{U}}_m^0$  as follows.  
We use the representation of $x\in T$ as $x=x_0+\epsilon w$ with $x_0\in T$ fixed to define
\begin{equation}
\dot{\sigma}_\mathrm{P}(x_0):=\lim_{M\uparrow\infty}\frac{\#\{\text{poles $w$ of $\dot{\mathcal{U}}_m^0(w;x_0)$ with $|w|<M$}\}}{\pi M^2},\quad x_0\in T.
\label{eq:dotsigmaPdefine}
\end{equation}
Similarly, given $x_0\in T\cap\mathbb{R}$, we may define the linear density of $\dot{\sigma}_\mathrm{L}$ of real poles of $\dot{\mathcal{U}}_m^0$ as
\begin{equation}
\dot{\sigma}_\mathrm{L}(x_0):=\lim_{M\uparrow\infty}\frac{\#\{\text{real poles $w$ of $\dot{\mathcal{U}}^0_m(w;x_0)$ in $(-M,M)$}\}}{2M},\quad x_0\in T\cap \mathbb{R}.
\label{eq:dotsigmaLdefine}
\end{equation}
Since the poles of $\dot{\mathcal{U}}^0_m$ in the $w$-plane for fixed $x_0$ correspond to points of the regular lattice $\mathscr{P}_m(x_0)$, or more specifically for $x_0\in\mathbb{R}$ to equally spaced real points with spacing $-\mathcal{H}_0/c_1$, it is easy to see that $\dot{\sigma}_\mathrm{P}$ is simply the reciprocal of the area of a fundamental period parallelogram $\pgram$ in the $w$-plane, while $\dot{\sigma}_\mathrm{L}$ is the reciprocal of the length of the period interval.  Therefore,
\begin{equation}
\dot{\sigma}_\mathrm{P}(x_0)=-\frac{|c_1|^2}{2\pi\mathrm{Re}(\mathcal{H})}=\left[\mathrm{Im}\left(\left(\oint_\mathfrak{a}\omega_0\right)^*\oint_\mathfrak{b}\omega_0\right)\right]^{-1}>0\quad x_0\in T
\label{eq:dotsigmaPformula}
\end{equation}
(see \cite[Ch. II, Corollary 1]{Dubrovin81} for the final inequality) and
\begin{equation}
\dot{\sigma}_\mathrm{L}(x_0)=-\frac{c_1}{\mathcal{H}_0}=\left[\oint_\mathfrak{a}\omega_0-2\oint_\mathfrak{b}\omega_0\right]^{-1}=\left[2\int_D^A\frac{dz}{R(z)} +2\int_D^B\frac{dz}{R(z)}\right]^{-1}>0,\quad x_0\in T\cap\mathbb{R}.
\label{eq:dotsigmaLformula}
\end{equation}

\subsubsection{Inner (Airy) parametrices near $z=A$ and $z=B$}  Given $x_0\in T$, let $\mathbb{D}_A$ and $\mathbb{D}_B$ be $\epsilon$-independent open disks of sufficiently small radius containing the points $z=A$ and $z=B$, respectively.  Consider the equations
\begin{equation}
\tau_A^3 = (2H(z)+\Lambda-i\Phi_+)^2,\quad z\in \mathbb{D}_A\quad\text{and}\quad
\tau_B^3 = (2H(z)+\Lambda-i\Phi_-)^2,\quad z\in \mathbb{D}_B.
\label{eq:AiryAB-conformalmaps}
\end{equation}
The functions on the right-hand side are analytic and vanish to precisely third order at the the points $z=A$ and $z=B$, respectively.  This implies the existence of two $\epsilon$-independent univalent functions $\tau_A:\mathbb{D}_A\to\mathbb{C}$ and $\tau_B:\mathbb{D}_B\to\mathbb{C}$ that are determined by \eqref{eq:AiryAB-conformalmaps} up to factors of the cube roots of unity.  We select these factors
so that $\tau_A(z)$ and $\tau_B(z)$ are positive real in $\mathbb{D}_A$ and $\mathbb{D}_B$, respectively, along the arcs of the jump contour $\Sigma^{(\mathbf{O})}$ for $\mathbf{O}(z)$ where $\mathbf{V}^{(\mathbf{N})}(z)=\mathbf{U}$.  The conformal mappings $\tau_A$ and $\tau_B$ satisfy $\tau_A(A)=\tau_B(B)=0$.

We now define matrix functions $\mathbf{H}_A:\mathbb{D}_A\to SL(2,\mathbb{C})$ and $\mathbf{H}_B:\mathbb{D}_B\to SL(2,\mathbb{C})$
by the formulae
\begin{equation}
\mathbf{H}_{A}(z):=\dot{\mathbf{O}}^{(\mathrm{out})}(z)e^{-wz\sigma_3/2}e^{i\pi\sigma_3/4}
e^{-i\Phi_+\sigma_3/(2\epsilon)}\mathbf{V}^{-1}\tau_A(z)^{-\sigma_3/4},\quad z\in \mathbb{D}_A
\label{eq:g1-HA}
\end{equation}
and
\begin{equation}
\mathbf{H}_{B}(z):=\dot{\mathbf{O}}^{(\mathrm{out})}(z)e^{-wz\sigma_3/2}e^{i\pi\sigma_3/4}
e^{-i\Phi_-\sigma_3/(2\epsilon)}\mathbf{V}^{-1}\tau_B(z)^{-\sigma_3/4},\quad z\in \mathbb{D}_B,
\label{eq:g1-HB}
\end{equation}
where $\mathbf{V}$ is the unimodular and unitary matrix defined by \eqref{eq:AiryAppendix-EigenvectorMatrix}.  (Where  $\tau_{A,B}^{-\sigma_3/4}$ and $\dot{\mathbf{O}}^{(\mathrm{out})}$ both have jump discontinuities along $\Sigma\cap \mathbb{D}_{A,B}$, either boundary value suffices and gives the same value for $\mathbf{H}_{A,B}$.)  Since $\tau_{A,B}(z)^{\sigma_3/4}\mathbf{V}$ satisfies the same jump conditions as does $\dot{\mathbf{O}}^{(\mathrm{out})}(z)$ within the disk $\mathbb{D}_{A,B}$, and since $\mathbf{H}_A=\mathcal{O}((z-A)^{-1/2})$ and 
$\mathbf{H}_B=\mathcal{O}((z-B)^{-1/2})$, the matrices defined by \eqref{eq:g1-HA} and \eqref{eq:g1-HB} are analytic functions within their respective disks of definition, and hence they are controlled by their size on $\partial\mathbb{D}_A$ and $\partial\mathbb{D}_B$, respectively, via the Maximum Modulus Principle.  It follows from the Boutroux conditions \eqref{eq:g1-Boutroux} and
from Proposition~\ref{prop:g1-dot-Oout-bound} that the functions $\mathbf{H}_{A,B}(z)$ are bounded independent of $\epsilon$ uniformly for $x_0$ within a compact subset $K\subset T$ and $w\in\mathscr{S}_m(x_0,\delta)$ for some $\delta>0$ (see \eqref{eq:g1-cheese-def}).  From Proposition~\ref{prop:g1-dot-Oout-det} and \eqref{eq:AiryAppendix-EigenvectorMatrix} it follows that $\det(\mathbf{H}_{A,B}(z))=1$ for $z\in \mathbb{D}_{A,B}$, and hence similar uniform bounds hold for $\mathbf{H}_{A,B}(z)^{-1}$.

We use these matrix functions to define parametrices near $A$ and $B$ as follows.  Set
\begin{equation}
\dot{\mathbf{O}}^{(A)}(z)=\dot{\mathbf{O}}^{(A)}(z;x_0,w,\epsilon):=\mathbf{H}_A(z)\epsilon^{\sigma_3/6}\mathbf{A}(\epsilon^{-2/3}\tau_A(z))e^{wz\sigma_3/2}e^{-i\pi\sigma_3/4}e^{i\Phi_+\sigma_3/(2\epsilon)},\quad
z\in \mathbb{D}_A
\end{equation}
and
\begin{equation}
\dot{\mathbf{O}}^{(B)}(z)=\dot{\mathbf{O}}^{(B)}(z;x_0,w,\epsilon):=\mathbf{H}_B(z)\epsilon^{\sigma_3/6}\mathbf{A}(\epsilon^{-2/3}\tau_B(z))e^{wz\sigma_3/2}e^{-i\pi\sigma_3/4}e^{i\Phi_+\sigma_3/(2\epsilon)},\quad
z\in \mathbb{D}_B,
\end{equation}
where the matrix function $\mathbf{A}$ is defined by \eqref{eq:AiryAppendix-ParametrixDef-I}--\eqref{eq:AiryAppendix-ParametrixDef-IV}.  One then has
\begin{equation}
\begin{split}
\dot{\mathbf{O}}^{(A)}(z)\dot{\mathbf{O}}^{(\mathrm{out})}(z)^{-1} &= \mathbf{H}_A(z)
\epsilon^{\sigma_3/6}\mathbf{A}(\epsilon^{-2/3}\tau_A(z))\mathbf{V}^{-1}[\epsilon^{-2/3}\tau_A(z)]^{-\sigma_3/4}\epsilon^{-\sigma_3/6}\mathbf{H}_A(z)^{-1}\\
&=\mathbb{I}+\begin{bmatrix}\mathcal{O}(\epsilon^2) & \mathcal{O}(\epsilon)\\
\mathcal{O}(\epsilon) & \mathcal{O}(\epsilon^2)\end{bmatrix},\quad z\in\partial \mathbb{D}_A,
\end{split}
\label{eq:g1-Airy-A-estimate}
\end{equation}
and similarly
\begin{equation}
\begin{split}
\dot{\mathbf{O}}^{(B)}(z)\dot{\mathbf{O}}^{(\mathrm{out})}(z)^{-1} &= \mathbf{H}_B(z)
\epsilon^{\sigma_3/6}\mathbf{A}(\epsilon^{-2/3}\tau_B(z))\mathbf{V}^{-1}[\epsilon^{-2/3}\tau_B(z)]^{-\sigma_3/4}\epsilon^{-\sigma_3/6}\mathbf{H}_B(z)^{-1}\\
&=\mathbb{I}+\begin{bmatrix}\mathcal{O}(\epsilon^2) & \mathcal{O}(\epsilon)\\
\mathcal{O}(\epsilon) & \mathcal{O}(\epsilon^2)\end{bmatrix},\quad z\in\partial \mathbb{D}_B,
\end{split}
\label{eq:g1-Airy-B-estimate}
\end{equation}
where we have used \eqref{eq:AiryAppendix-ParametrixAsymp} and the fact that $\tau_{A,B}(z)$ is bounded away from zero on  $\partial\mathbb{D}_{A,B}$.  Also, from \eqref{eq:AiryAppendix-AiryJump-I}--\eqref{eq:AiryAppendix-AiryJump-III} it follows that $\dot{\mathbf{O}}^{(A,B)}(z)$ satisfies exactly the same jump conditions within $\mathbb{D}_{A,B}$ as does $\mathbf{O}(z)$.

\subsubsection{Inner (Airy) parametrix near $z=D$}  Although it is also an endpoint of the Stokes graph $\Sigma$, the point $z=D$ is unlike $A$ and $B$ because it is the origin of an additional branch cut of the function $H$, namely the ray $L$.  Given $x_0\in T$, let $\mathbb{D}_D$ be an $\epsilon$-independent open disk of sufficiently small radius containing the point $z=D$.  The arcs of the jump contour $\Sigma^{(\mathbf{O})}$ for $\mathbf{O}(z)$ near $z=D$ along which either $\mathbf{V}^{(\mathbf{N})}(z)=\mathbf{T}^{-1}$ or $\mathbf{V}^{(\mathbf{N})}(z)=-\mathbf{U}$ divide $\mathbb{D}_D$ into two complementary parts, $\mathbb{D}^+_D$ on the left and $\mathbb{D}^-_D$ on the right by orientation (see Figure~\ref{fig:N-jumps-g1}).  We define an $\epsilon$-independent univalent function $\tau_D(z)$ in $\mathbb{D}_D$ to satisfy the equation
\begin{equation}
\tau_D^3 = \begin{cases}(2H(z)+\Lambda-2\pi i)^2,&\quad z\in \mathbb{D}^+_D,\\
(2H(z)+\Lambda+2\pi i)^2,&\quad z\in \mathbb{D}^-_D.
\end{cases}
\end{equation}
It is a consequence of \eqref{eq:g1-H}, \eqref{eq:g1-G-jump}, and \eqref{eq:g1-Lambda}  that the right-hand side defines an analytic function on the whole disk $\mathbb{D}_D$ that vanishes exactly to third order at $z=D$; we choose for $\tau_D:\mathbb{D}_D\to\mathbb{C}$ the analytic branch that is positive real for $z\in L\cap\mathbb{D}_D$.  Note that $\tau_D(D)=0$.

Next, define an analytic matrix function $\mathbf{H}_D:\mathbb{D}_D\to SL(2,\mathbb{C})$ by ($\mathbf{V}$ is defined by \eqref{eq:AiryAppendix-EigenvectorMatrix})
\begin{equation}
\mathbf{H}_D(z):=\dot{\mathbf{O}}^{(\mathrm{out})}(z)e^{-wz\sigma_3/2}e^{-i\pi\sigma_3/4}\mathbf{V}^{-1}\tau_D(z)^{-\sigma_3/4},\quad z\in \mathbb{D}_D.
\end{equation}
The fact that $\mathbf{H}_D$ is analytic in $\mathbb{D}_D$ although both $\dot{\mathbf{O}}^{(\mathrm{out})}(z)$ and $\tau_D(z)^{-\sigma_3/4}$ have jump discontinuities along the arc of the jump contour for $\mathbf{O}(z)$ within $\mathbb{D}_D$ that coincides with the Stokes graph $\Sigma$ follows by a direct calculation.  Due to \eqref{eq:g1-Boutroux}, Proposition~\ref{prop:g1-dot-Oout-bound}, and the Maximum Modulus Principle, the elements of $\mathbf{H}_D$ are bounded independently of $\epsilon$ uniformly for  $x_0$ within a compact subset $K\subset T$ and $w\in\mathscr{S}_m(x_0,\delta)$ for some $\delta>0$.  By Proposition~\ref{prop:g1-dot-Oout-det} and \eqref{eq:AiryAppendix-EigenvectorMatrix} the same is true for the elements of $\mathbf{H}_D(z)^{-1}$.

Now we define the parametrix near $z=D$ by setting ($\mathbf{A}$ is defined by \eqref{eq:AiryAppendix-ParametrixDef-I}--\eqref{eq:AiryAppendix-ParametrixDef-IV})
\begin{equation}
\dot{\mathbf{O}}^{(D)}(z)=\dot{\mathbf{O}}^{(D)}(z;x_0,w,\epsilon):=\mathbf{H}_D(z)\epsilon^{\sigma_3/6}\mathbf{A}(\epsilon^{-2/3}\tau_D(z))e^{wz\sigma_3/2}e^{i\pi\sigma_3/4},\quad z\in \mathbb{D}_D.
\end{equation}
Due to \eqref{eq:AiryAppendix-AiryJump-I}--\eqref{eq:AiryAppendix-AiryJump-III}, this matrix satisfies exactly the same jump conditions within $\mathbb{D}_D$ as does $\mathbf{O}(z)$, and 
as a consequence of \eqref{eq:AiryAppendix-ParametrixAsymp}, 
\begin{equation}
\begin{split}
\dot{\mathbf{O}}^{(D)}(z)\dot{\mathbf{O}}^{(\mathrm{out})}(z)^{-1}&=\mathbf{H}_D(z)\epsilon^{\sigma_3/6}\mathbf{A}(\epsilon^{-2/3}\tau_D(z))\mathbf{V}^{-1}[\epsilon^{-2/3}\tau_D(z)]^{-\sigma_3/4}\epsilon^{-\sigma_3/6}\mathbf{H}_D(z)^{-1}\\
&=\mathbb{I}+\begin{bmatrix}\mathcal{O}(\epsilon^{2}) & \mathcal{O}(\epsilon)\\
\mathcal{O}(\epsilon) & \mathcal{O}(\epsilon^{2})\end{bmatrix},\quad z\in \partial \mathbb{D}_D,
\end{split}
\label{eq:g1-Airy-D-estimate}
\end{equation}
because $\tau_D(z)$ is bounded away from zero for $z\in\partial\mathbb{D}_D$.

\subsubsection{Inner (Airy) parametrix near $z=C$}
The point $z=C$ is distinguished as the self-intersection point of the Stokes graph $\Sigma$.
Let $\mathbb{D}_C$ be an $\epsilon$-independent open disk of sufficiently small radius containing $z=C$, and note that the jump contour $\Sigma^{(\mathbf{O})}$ for $\mathbf{O}$ divides $\mathbb{D}_C$ into six complementary regions.  Beginning with the region to the left of the oriented arc $\overrightarrow{CD}$, we label these regions in counterclockwise order as $\mathbb{D}^\mathrm{I}_C$, $\mathbb{D}^\mathrm{II}_C$,
$\mathbb{D}^\mathrm{III}_C$, $\mathbb{D}^\mathrm{IV}_C$, $\mathbb{D}^\mathrm{V}_C$,
and $\mathbb{D}^\mathrm{VI}_C$.  Furthermore, the Stokes graph $\Sigma$ separates $\mathbb{D}_C$ into three disjoint regions:  $\mathbb{D}^0_C$ bounded by $\overrightarrow{CA}$ and $\overrightarrow{CB}$, $\mathbb{D}^+_C$ bounded by $\overrightarrow{CA}$ and $\overrightarrow{CD}$, and finally $\mathbb{D}^-_C$ bounded by $\overrightarrow{CB}$ and $\overrightarrow{CD}$.  Upon restriction to $\mathbb{D}_C$, the function $H$ therefore becomes three different analytic functions, $H_0:\mathbb{D}_C^0\to\mathbb{C}$, $H_+:\mathbb{D}^+_C\to\mathbb{C}$, and $H_-:\mathbb{D}^-_C\to\mathbb{C}$.  The function $H_0$ has an analytic continuation to $\mathbb{D}_C\setminus\overrightarrow{CD}$ (we denote this continuation also by $H_0$), and because the sum of the boundary values taken by $H$ along each arc of the Stokes graph $\Sigma$ is constant, 
it is easy to obtain the relations
\begin{equation}
H_\pm(z)=i\Phi_\pm-\Lambda-H_0(z),\quad z\in\mathbb{D}^\pm_C.
\end{equation}
Note that the boundary value $H_0(C)$ is finite and unambiguous.  We define an $\epsilon$-independent univalent analytic function $\tau_C:\mathbb{D}_C\to\mathbb{C}$ by solving the equation
\begin{equation}
\tau_C^3 = (2H_0(C)-2H_0(z))^2,\quad z\in\mathbb{D}_C.
\end{equation}
The right-hand side extends to $\overrightarrow{CD}$ as an analytic function in all of $\mathbb{D}_C$ that vanishes to precisely third order at $z=C$.  We choose for $\tau_C(z)$ the analytic solution of this equation that is positive real on the common boundary of $\mathbb{D}_C^\mathrm{III}$ and $\mathbb{D}_C^\mathrm{IV}$.

We define a matrix function $\mathbf{H}_C$ as follows:
\begin{equation}
\mathbf{H}_C(z):=\dot{\mathbf{O}}^{(\mathrm{out})}(z)e^{-wz\sigma_3/2}e^{-i\pi\sigma_3/4}
e^{(\Lambda+2H_0(C)-2i\Phi_+)\sigma_3/(2\epsilon)}\mathbf{V}^{-1}\tau_C(z)^{-\sigma_3/4},\quad z\in
\mathbb{D}_C^\mathrm{I}\cup\mathbb{D}_C^\mathrm{II},
\end{equation}
\begin{multline}
\mathbf{H}_C(z):=-\dot{\mathbf{O}}^{(\mathrm{out})}(z)\begin{bmatrix}0 & ie^{-i\Phi_+/\epsilon}e^{-wz}\\ie^{i\Phi_+/\epsilon}e^{wz} & 0\end{bmatrix}\\{}\times e^{-wz\sigma_3/2}e^{-i\pi\sigma_3/4}
e^{(\Lambda+2H_0(C)-2i\Phi_+)\sigma_3/(2\epsilon)}\mathbf{V}^{-1}\tau_C(z)^{-\sigma_3/4},\quad z\in
\mathbb{D}_C^\mathrm{III},
\end{multline}
\begin{multline}
\mathbf{H}_C(z):=-\dot{\mathbf{O}}^{(\mathrm{out})}(z)\begin{bmatrix}0 & ie^{-i\Phi_-/\epsilon}e^{-wz}\\ie^{i\Phi_-/\epsilon}e^{wz} & 0\end{bmatrix}\\{}\times e^{-wz\sigma_3/2}e^{-i\pi\sigma_3/4}
e^{(\Lambda+2H_0(C)-2i\Phi_-)\sigma_3/(2\epsilon)}\mathbf{V}^{-1}\tau_C(z)^{-\sigma_3/4},\quad z\in
\mathbb{D}_C^\mathrm{IV},
\end{multline}
\begin{equation}
\mathbf{H}_C(z):=-\dot{\mathbf{O}}^{(\mathrm{out})}(z)e^{-wz\sigma_3/2}e^{-i\pi\sigma_3/4}
e^{(\Lambda+2H_0(C)-2i\Phi_-)\sigma_3/(2\epsilon)}\mathbf{V}^{-1}\tau_C(z)^{-\sigma_3/4},\quad z\in
\mathbb{D}_C^\mathrm{V}\cup\mathbb{D}_C^\mathrm{VI}.
\end{equation}
With the help of the identity (obtained by considering the constant values of the sum of boundary values of $H$ along each of the three arcs of $\Sigma$)
\begin{equation}
i(\Phi_++\Phi_-)=2H_0(C)+\Lambda,
\label{eq:HconstsIdentity}
\end{equation}
it is straightforward to check that $\mathbf{H}_C$ extends to the whole disk $\mathbb{D}_C$ as an analytic function $\mathbf{H}_C:\mathbb{D}_C\to SL(2,\mathbb{C})$ that (by arguments parallel to those given in the construction of Airy parametrices near $z=A,B,D$) is bounded independent of $\epsilon$ uniformly for $(x_0,w)$ for which $x_0$ lies within a compact subset $K$ of $T$ and $w\in\mathscr{S}_m(x_0,\delta)$, and similar estimates hold for $\mathbf{H}_C(z)^{-1}$.

Now we define the parametrix $\dot{\mathbf{O}}^{(C)}(z)=\dot{\mathbf{O}}^{(C)}(z;x_0,w,\epsilon)$ near $C$ as follows.  Set
\begin{equation}
\dot{\mathbf{O}}^{(C)}(z):=\mathbf{H}_C(z)\epsilon^{\sigma_3/6}\mathbf{A}(\epsilon^{-2/3}\tau_C(z))e^{-(\Lambda+2H_0(C)-2i\Phi_+)\sigma_3/(2\epsilon)}e^{wz\sigma_3/2}e^{i\pi\sigma_3/4},\quad z\in \mathbb{D}_C^\mathrm{I}\cup\mathbb{D}_C^\mathrm{II},
\end{equation}
\begin{multline}
\dot{\mathbf{O}}^{(C)}(z):=\mathbf{H}_C(z)\epsilon^{\sigma_3/6}\mathbf{A}(\epsilon^{-2/3}\tau_C(z))\\
{}\times e^{-(\Lambda+2H_0(C)-2i\Phi_+)\sigma_3/(2\epsilon)}e^{wz\sigma_3/2}e^{i\pi\sigma_3/4}\begin{bmatrix}
0 & ie^{-i\Phi_+/\epsilon}e^{-wz}\\ie^{i\Phi_+/\epsilon}e^{wz} & 0\end{bmatrix},\quad z\in \mathbb{D}_C^\mathrm{III},
\end{multline}
\begin{multline}
\dot{\mathbf{O}}^{(C)}(z):=\mathbf{H}_C(z)\epsilon^{\sigma_3/6}\mathbf{A}(\epsilon^{-2/3}\tau_C(z))\\
{}\times e^{-(\Lambda+2H_0(C)-2i\Phi_-)\sigma_3/(2\epsilon)}e^{wz\sigma_3/2}e^{i\pi\sigma_3/4}\begin{bmatrix}
0 & ie^{-i\Phi_-/\epsilon}e^{-wz}\\ie^{i\Phi_-/\epsilon}e^{wz} & 0\end{bmatrix},\quad z\in \mathbb{D}_C^\mathrm{IV},
\end{multline}
\begin{equation}
\dot{\mathbf{O}}^{(C)}(z):=-\mathbf{H}_C(z)\epsilon^{\sigma_3/6}\mathbf{A}(\epsilon^{-2/3}\tau_C(z))e^{-(\Lambda+2H_0(C)-2i\Phi_-)\sigma_3/(2\epsilon)}e^{wz\sigma_3/2}e^{i\pi\sigma_3/4},\quad z\in \mathbb{D}_C^\mathrm{V}\cup\mathbb{D}_C^\mathrm{VI}.
\end{equation}
With the help of \eqref{eq:HconstsIdentity} and the Airy jump conditions \eqref{eq:AiryAppendix-AiryJump-I}--\eqref{eq:AiryAppendix-AiryJump-III}, one proves easily that $\dot{\mathbf{O}}^{(C)}(z)$ satisfies exactly the same jump conditions within $\mathbb{D}_C$ as does $\mathbf{O}(z)$.  Also, by direct calculation using \eqref{eq:AiryAppendix-ParametrixAsymp} one sees that
\begin{equation}
\begin{split}
\dot{\mathbf{O}}^{(C)}(z)\dot{\mathbf{O}}^{(\mathrm{out})}(z)^{-1}&=\mathbf{H}_C(z)
\epsilon^{\sigma_3/6}\mathbf{A}(\epsilon^{-2/3}\tau_C(z))\mathbf{V}^{-1}[\epsilon^{-2/3}\tau_C(z)]^{-\sigma_3/4}\epsilon^{-\sigma_3/6}\mathbf{H}_C(z)^{-1}\\
&=\mathbb{I}+\begin{bmatrix}\mathcal{O}(\epsilon^2) & \mathcal{O}(\epsilon)\\\mathcal{O}(\epsilon) & \mathcal{O}(\epsilon^2)\end{bmatrix},\quad z\in \partial\mathbb{D}_C
\end{split}
\label{eq:g1-Airy-C-estimate}
\end{equation}
because $\tau_C$ is bounded away from zero for $z\in\partial\mathbb{D}_C$.

\subsubsection{Definition of the global parametrix}
\label{sec:g1-globalpar-def}
The global parametrix $\dot{\mathbf{O}}(z)$ is a model for the unknown matrix function $\mathbf{O}(z)$ that is defined throughout the complex plane as follows:
\begin{equation}
\dot{\mathbf{O}}(z)=\dot{\mathbf{O}}(z;x_0,w,\epsilon):=\begin{cases}
\dot{\mathbf{O}}^{(A)}(z),&\quad z\in\mathbb{D}_A,\\
\dot{\mathbf{O}}^{(B)}(z),&\quad z\in\mathbb{D}_B,\\
\dot{\mathbf{O}}^{(C)}(z),&\quad z\in\mathbb{D}_C,\\
\dot{\mathbf{O}}^{(D)}(z),&\quad z\in\mathbb{D}_D,\\
\dot{\mathbf{O}}^{(\mathrm{out})}(z),&\quad z\in\mathbb{C}\setminus \overline{\mathbb{D}_A\cup\mathbb{D}_B\cup\mathbb{D}_C\cup\mathbb{D}_D}.
\end{cases}
\end{equation}

\subsection{Error analysis}
\label{sec:g1-error-analysis}
The accuracy of the explicit global parametrix as a model for the unknown matrix $\mathbf{O}(z)$ can be gauged by consideration of the \emph{error}, namely the matrix $\mathbf{E}(z)$ defined by the formula
\begin{equation}
\mathbf{E}(z)=\mathbf{E}(z;x_0,w,\epsilon):=\mathbf{O}(z)\dot{\mathbf{O}}(z)^{-1}
\end{equation}
for all $z$ for which both factors on the right-hand side are unambiguously defined and analytic.  Thus, $\mathbf{E}(z)$ is well-defined and analytic for $z\in\Sigma^{(\mathbf{O})}\cup\partial\mathbb{D}_A\cup\partial\mathbb{D}_B\cup\partial\mathbb{D}_C\cup\partial\mathbb{D}_D$.
However, it is easy to check that, since $\mathbf{O}(z)$ and $\dot{\mathbf{O}}(z)$ satisfy the
same jump conditions on the three arcs of the Stokes graph $\Sigma$ as well as on the remaining arcs of $\Sigma^{(\mathbf{O})}$ that lie within the disks $\mathbb{D}_A$, $\mathbb{D}_B$, $\mathbb{D}_C$, and $\mathbb{D}_D$, the matrix $\mathbf{E}(z)$ can be analytically continued to these contour arcs, and hence $\mathbf{E}(z)$ is analytic for $z\in\Sigma^{(\mathbf{E})}$, where 
$\Sigma^{(\mathbf{E})}$ is the arc-wise oriented contour shown in Figure~\ref{fig:E-jumps-g1}.
\begin{figure}[h]
\begin{center}
\includegraphics{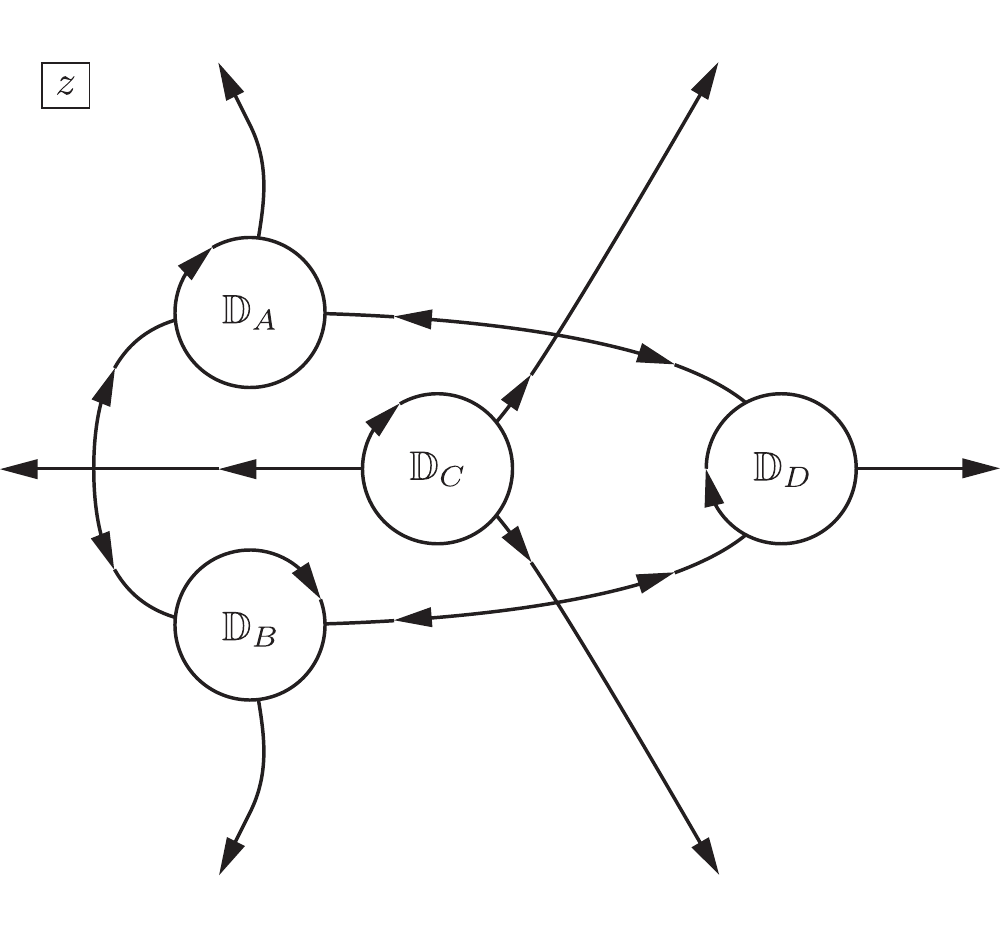}
\end{center}
\caption{\emph{The jump contour $\Sigma^{(\mathbf{E})}$ of the error matrix $\mathbf{E}(z)$.  All arcs of the four disk boundaries are taken to be oriented in the negative (clockwise) direction.}}
\label{fig:E-jumps-g1}
\end{figure}
It is important that all arcs of $\Sigma^{(\mathbf{E})}$ are sufficiently smooth (of class $C^1$), that the six unbounded arcs of $\Sigma^{(\mathbf{E})}$ are chosen to coincide for large enough $|z|$ with the rays $\arg(z)\in \pi\mathbb{Z}/3$, and that at all self-intersection points no intersecting arcs meet tangentially.

Although $\mathbf{E}(z)$ is unknown, it can be characterized completely in terms of known quantities.  First note that since by definition we have $\mathbf{O}(z)\to\mathbb{I}$ as $z\to\infty$
and for sufficiently large $|z|$ we have $\dot{\mathbf{O}}(z)=\dot{\mathbf{O}}^{(\mathrm{out})}(z)$
with the latter being a matrix equal to $\mathbb{I}$ for $z=\infty$, from the definition of $\mathbf{E}(z)$ we deduce that
\begin{equation}
\lim_{z\to\infty}\mathbf{E}(z)=\mathbb{I}
\end{equation}
where (as it will be seen) the limit may be taken in any direction, including along either side of the unbounded arcs of $\Sigma^{(\mathbf{E})}$.  Also, the jump matrix $\mathbf{V}^{(\mathbf{E})}(z):=
\mathbf{E}_-(z)^{-1}\mathbf{E}_+(z)$ can be calculated explicitly for $z$ lying in each arc of $\Sigma^{(\mathbf{E})}$.  Indeed, if $*$ stands for any of the symbols $A,B,C,D$, along each arc
of the disk boundary $\partial\mathbb{D}_*$ we have $\mathbf{O}_+(z)=\mathbf{O}_-(z)$, so
\begin{equation}
\mathbf{V}^{(\mathbf{E})}(z)=[\mathbf{O}_-(z)\dot{\mathbf{O}}_-(z)^{-1}]^{-1}\mathbf{O}_+(z)\dot{\mathbf{O}}_+(z)^{-1}=\dot{\mathbf{O}}^{(*)}(z)\dot{\mathbf{O}}^{(\mathrm{out})}(z)^{-1},\quad z\in\partial\mathbb{D}_*\subset\Sigma^{(\mathbf{E})},
\label{eq:g1-VE-circles}
\end{equation}
while on all other arcs of $\Sigma^{(\mathbf{E})}$ we have $\dot{\mathbf{O}}(z)=\dot{\mathbf{O}}^{(\mathrm{out})}(z)$ having no jump discontinuity while $\mathbf{O}_+(z)=\mathbf{O}_-(z)\mathbf{V}^{(\mathbf{O})}(z)$ with $\mathbf{V}^{(\mathbf{O})}(z)$ being explicitly given by \eqref{eq:g1-Ojump-1}--\eqref{eq:g1-Ojump-7}, so
\begin{equation}
\begin{split}
\mathbf{V}^{(\mathbf{E})}(z)&=[\mathbf{O}_-(z)\dot{\mathbf{O}}_-(z)^{-1}]^{-1}
\mathbf{O}_+(z)\dot{\mathbf{O}}_+(z)^{-1}\\ &{}=\dot{\mathbf{O}}^{(\mathrm{out})}(z)\mathbf{V}^{(\mathbf{O})}(z)\dot{\mathbf{O}}^{(\mathrm{out})}(z)^{-1},\quad z\in\Sigma^{(\mathbf{E})}\setminus (\partial\mathbb{D}_A\cup\partial\mathbb{D}_B\cup\partial\mathbb{D}_C\cup\partial\mathbb{D}_D).
\end{split}
\label{eq:g1-VE-not-circles}
\end{equation}
We may therefore characterize $\mathbf{E}(z)$ as the solution of a matrix Riemann-Hilbert problem on the contour $\Sigma^{(\mathbf{E})}$ with identity normalization at $z=\infty$ and 
given, explicitly known jump matrices.

Not only is the jump matrix $\mathbf{V}^{(\mathbf{E})}:\Sigma^{(\mathbf{E})}\setminus\{\text{self-intersection points}\}\to SL(2,\mathbb{C})$ known explicitly, but also it is easily estimated in terms of $\epsilon>0$ under certain conditions.  Let $K\subset T$ be compact, and suppose that $x_0\in K$.  Then in particular the points $\{A,B,C,D\}$ are bounded away from each other, and hence the four disks $\mathbb{D}_{A,B,C,D}$ can be given a common small radius $\delta_1>0$ depending only on $K$ so that they are pairwise disjoint.  Invoking continuity of $H(z)$ with respect to the parameter $x_0$ then shows that on all of the arcs of $\Sigma^{(\mathbf{E})}$, except for the four circles $\partial\mathbb{D}_{A,B,C,D}$, we have $|\mathbf{V}^{(\mathbf{O})}(z)-\mathbb{I}|\le e^{-C/\epsilon}$ for all $\epsilon>0$, where $C>0$ depends only on $K$.  Since $H(z)\sim \tfrac{1}{2}z^3$ as $z\to\infty$, similar considerations show that on each of the six unbounded arcs of $\Sigma^{(\mathbf{E})}$ a pointwise estimate of the form $|\mathbf{V}^{(\mathbf{O})}(z)-\mathbb{I}|\le e^{-\delta_2|z|^3/\epsilon}$ holds for some $\delta_2>0$ depending only on $K$.
Now we invoke the assumption that, for some $\delta>0$ depending only on $K$, $w\in\mathscr{S}_m(x_0,\delta)$.  It then follows from Proposition~\ref{prop:g1-dot-Oout-bound} and Proposition~\ref{prop:g1-dot-Oout-det} and the formula \eqref{eq:g1-VE-not-circles} that the above estimates for $\mathbf{V}^{(\mathbf{O})}(z)-\mathbb{I}$ carry over to $\mathbf{V}^{(\mathbf{E})}(z)-\mathbb{I}$ up to constants depending only on $K$.  It then remains to estimate $\mathbf{V}^{(\mathbf{E})}(z)-\mathbb{I}$ for $z$ in the disk boundaries $\partial\mathbb{D}_{A,B,C,D}$.
However, combining the formula \eqref{eq:g1-VE-circles} with \eqref{eq:g1-Airy-A-estimate},
\eqref{eq:g1-Airy-B-estimate}, \eqref{eq:g1-Airy-D-estimate}, and \eqref{eq:g1-Airy-C-estimate},
one sees that uniformly for $x_0\in K\subset T$, $w\in \mathscr{S}_m(x_0,\delta)$, and $\epsilon>0$, 
\begin{equation}
\sup_{z\in\partial\mathbb{D}_{A,B,C,D}}|\mathbf{V}^{(\mathbf{E})}(z)-\mathbb{I}|=\mathcal{O}(\epsilon).
\end{equation}
These estimates clearly imply that 
\begin{equation}
\|\mathbf{V}^{(\mathbf{E})}-\mathbb{I}\|_{L^\infty(\Sigma^{(\mathbf{E})})}=\mathcal{O}(\epsilon),
\label{eq:g1-VE-Linfty}
\end{equation}
with the dominant contribution coming from the off-diagonal matrix elements for $z$ on the four disk boundaries.
Moreover,  if $\epsilon\le 1$, say,
\begin{equation}
\begin{split}
\int_{\Sigma^{(\mathbf{E})}} |\mathbf{V}^{(\mathbf{E})}(z)-\mathbb{I}|\,|z|^2\,|dz| &\lesssim
 \|\mathbf{V}^{(\mathbf{E})}-\mathbb{I}\|_{L^\infty(\Sigma^{(\mathbf{E})})} +\int_{\Sigma^{(\mathbf{E})}\cap\{|z|>R\}}|\mathbf{V}^{(\mathbf{E})}(z)-\mathbb{I}|\,|z|^2\,|dz|\\
 &\lesssim  \|\mathbf{V}^{(\mathbf{E})}-\mathbb{I}\|_{L^\infty(\Sigma^{(\mathbf{E})})} +
 \int_{\Sigma^{(\mathbf{E})}\cap\{|z|>R\}} e^{-\delta_2|z|^3/\epsilon}|z|^2\,|dz|\\
 &\le \|\mathbf{V}^{(\mathbf{E})}-\mathbb{I}\|_{L^\infty(\Sigma^{(\mathbf{E})})} +
 e^{-\delta_2 R^3/(2\epsilon)}\int_{\Sigma^{(\mathbf{E})}\cap\{|z|>R\}}
 e^{-\delta_2|z|^3/2}|z|^2\,|dz|\\
 &=\mathcal{O}(\epsilon).
 \end{split}
 \label{eq:g1-VE-weighted-L1}
\end{equation}
The order estimates \eqref{eq:g1-VE-Linfty} and \eqref{eq:g1-VE-weighted-L1} hold uniformly
on the set $0<\epsilon<1$, $x_0\in K\subset T$, and $w\in \mathscr{S}_m(x_0,\delta)$ for some small $\delta>0$.

Now we wish to make an important point regarding the above estimates and the contour $\Sigma^{(\mathbf{E})}$, namely that for each point $x_0'\in K$ there exists a corresponding $\eta(x_0')>0$ such that, whenever $x_0$ satisfies $|x_0-x_0'|<\eta(x_0')$, the contour $\Sigma^{(\mathbf{E})}$ can be taken to be exactly the same, and exactly the same estimates \eqref{eq:g1-VE-Linfty}--\eqref{eq:g1-VE-weighted-L1} will hold for $\mathbf{V}^{(\mathbf{E})}-\mathbb{I}$ (even with the same implicit constants).  This is a consequence of (i) the fact that the open disks $\mathbb{D}_{A,B,C,D}$ simply have to be independent of $\epsilon$ and have to contain the corresponding points $\{A(x_0), B(x_0),C(x_0),D(x_0)\}$ which vary continuously with $x_0$, and (ii) the fact that the function $\mathbf{H}(z)=\mathbf{H}(z;x_0)$ depends continuously on $x_0$, while the related inequalities that hold on the arcs of $\Sigma^{(\mathbf{E})}\setminus(\partial\mathbb{D}_A\cup\partial\mathbb{D}_B\cup\partial\mathbb{D}_C\cup\partial\mathbb{D}_D)$ are strict.  An additional technical point is that the arcs of $\Sigma^{(\mathbf{O})}$ lying inside of the four disks have been chosen to depend on $x_0$ in a precise manner (agreeing with steepest descent curves of $\mathrm{Re}(2H(z;x_0)+\Lambda(x_0))$) and so for these moving arcs to join onto fixed arcs outside of the disks a small ``jog'' is required in the arcs of $\Sigma^{(\mathbf{O})}$ that cross the boundaries of the disks, giving rise to an additional factor in $\mathbf{V}^{(\mathbf{E})}$ where $\Sigma^{(\mathbf{O})}$ overlaps with the disk boundaries; however the same exponential estimate holds for the difference of this factor from the identity as holds for $\mathbf{V}^{(\mathbf{O})}-\mathbb{I}$ on the arcs of  $\Sigma^{(\mathbf{O})}$ exterior to the disks.  Now $K\subset T$ is compact, and hence the open covering of $K$ consisting of the union of disks $|x_0-x_0'|<\eta(x_0')$ over $x_0'\in K$ has a finite sub-covering.
This implies that, as $x_0$ varies over $K$, only a finite number of different contours $\Sigma^{(\mathbf{E})}$ are required in order that the estimates \eqref{eq:g1-VE-Linfty}--\eqref{eq:g1-VE-weighted-L1} hold for the corresponding jump matrix $\mathbf{V}^{(\mathbf{E})}$.

Therefore, as $x_0$ varies over the compact set $K\subset T$, and assuming that $0<\epsilon<1$ and that $w\in \mathscr{S}_m(x_0,\delta)$ for some $\delta>0$, the error matrix $\mathbf{E}(z)$ satisfies the conditions of a Riemann-Hilbert problem of small-norm type as described in Appendix~\ref{small-norm-app}, and moreover, this small-norm problem is formulated relative to a finite number of admissible contours $\Sigma^{(\mathbf{E})}$.  The latter fact means that the constants $K_\Sigma'$ and $K_\Sigma''$ in the statement of Proposition~\ref{prop:moments-control} in Appendix~\ref{small-norm-app} can be assumed to depend only on the compact set $K$ and not on the specific point $x_0\in K$.  It then follows from this proposition
and the estimates \eqref{eq:g1-VE-Linfty}--\eqref{eq:g1-VE-weighted-L1} that the error matrix $\mathbf{E}(z)$ is uniquely characterized by its jump contour $\Sigma^{(\mathbf{E})}$, jump matrix $\mathbf{V}^{(\mathbf{E})}$, and identity normalization condition at $z=\infty$, and the 
the first two moments $\mathbf{E}_1$ and $\mathbf{E}_2$ defined by the expansion
\begin{equation}
\mathbf{E}(z)=\mathbb{I}+\frac{\mathbf{E}_1}{z} +\frac{\mathbf{E}_2}{z^2}+o\left(\frac{1}{z^2}\right),\quad z\to\infty
\label{eq:g1-E-expansion}
\end{equation}
(valid for $z$ along any ray distinct from the six unbounded arcs of $\Sigma^{(\mathbf{E})}$) both exist and satisfy
\begin{equation}
|\mathbf{E}_1|=\mathcal{O}(\epsilon)\quad\text{and}\quad|\mathbf{E}_2|=\mathcal{O}(\epsilon)
\end{equation}
uniformly on the set $0<\epsilon<1$ and $x_0\in K\subset T$ with $w\in \mathscr{S}_m(x_0,\delta)$ for some $\delta>0$.  In fact, it can easily be shown that $\mathbf{E}(z)$ is a far more regular object than Proposition~\ref{prop:moments-control} guarantees.  Indeed, $\mathbf{E}(z)$ arises from the original unknown matrix $\mathbf{M}(z)$ by piecewise analytic substitutions, and when $\mathbf{M}(z)$ exists it can be obtained 
explicitly for each $m=\epsilon^{-1}+\tfrac{1}{2}\in\mathbb{Z}$ via iterated Darboux-type transformations \cite{BuckinghamMcritical}.  This means that $\mathbf{E}(z)$ assumes its boundary values in a completely classical sense, and it implies also that the expansion \eqref{eq:g1-E-expansion} is in fact valid as $z\to\infty$ uniformly in any direction whatsoever.

\subsection{Approximate formulae for the rational Painlev\'e-II functions for $x_0\in T$}
\label{sec:g1-approximate-formulae}
From \eqref{eq:Fexpansion}--\eqref{eq:vprime} and \eqref{M-def}--\eqref{M-normalization},
we may begin with the following exact formulae:
\begin{equation}
\mathcal{U}_m((m-\tfrac{1}{2})^{2/3}x)=(m-\tfrac{1}{2})^{2m/3}M_{1,12}(x,(m-\tfrac{1}{2})^{-1}),
\end{equation}
\begin{equation}
\mathcal{V}_m((m-\tfrac{1}{2})^{2/3}x)=(m-\tfrac{1}{2})^{-2(m-1)/3}M_{1,21}(x,(m-\tfrac{1}{2})^{-1}),
\end{equation} 
\begin{equation}
\mathcal{P}_m((m-\tfrac{1}{2})^{2/3}x)=\frac{\mathcal{U}_m'((m-\tfrac{1}{2})^{2/3}x)}{\mathcal{U}_m((m-\tfrac{1}{2})^{2/3}x)}=(m-\tfrac{1}{2})^{1/3}\left[M_{1,22}(x,(m-\tfrac{1}{2})^{-1})-
\frac{M_{2,12}(x,(m-\tfrac{1}{2})^{-1})}{M_{1,12}(x,(m-\tfrac{1}{2})^{-1})}\right],
\end{equation}
and
\begin{equation}
\mathcal{Q}_m((m-\tfrac{1}{2})^{2/3}x)=\frac{\mathcal{V}_m'((m-\tfrac{1}{2})^{2/3}x)}{\mathcal{V}_m((m-\tfrac{1}{2})^{2/3}x)}=(m-\tfrac{1}{2})^{1/3}\left[\frac{M_{2,21}(x,(m-\tfrac{1}{2})^{-1})}{M_{1,21}(x,(m-\tfrac{1}{2})^{-1})}-M_{1,11}(x,(m-\tfrac{1}{2})^{-1})\right],
\end{equation}
where $M_{1,ij}(x,\epsilon)$ and $M_{2,ij}(x,\epsilon)$ are the matrix elements of the moments
$\mathbf{M}_1(x,\epsilon)$ and $\mathbf{M}_2(x,\epsilon)$ defined by the asymptotic expansion
\begin{equation}
\mathbf{M}(z;x,\epsilon)=\mathbb{I} + \mathbf{M}_1(x,\epsilon)z^{-1} +
\mathbf{M}_2(x,\epsilon)z^{-2}+o(z^{-2}),\quad z\to\infty.
\end{equation}
For sufficiently large $|z|$, 
\begin{equation}
\begin{split}
\mathbf{M}(z;x_0+\epsilon w,\epsilon)&=\mathbf{N}(z;x_0+\epsilon w,\epsilon)\\&=
e^{\Lambda(x_0)\sigma_3/(2\epsilon)}\mathbf{O}(z;x_0,w,\epsilon)
e^{-\Lambda(x_0)\sigma_3/(2\epsilon)}e^{G(z;x_0)\sigma_3/\epsilon} \\
&= e^{\Lambda(x_0)\sigma_3/(2\epsilon)}\mathbf{E}(z;x_0,w,\epsilon)\dot{\mathbf{O}}(z;x_0,w,\epsilon)e^{-\Lambda(x_0)\sigma_3/(2\epsilon)}e^{G(z;x_0)\sigma_3/\epsilon}\\
&=e^{\Lambda(x_0)\sigma_3/(2\epsilon)}\mathbf{E}(z;x_0,w,\epsilon)\dot{\mathbf{O}}^{(\mathrm{out})}(z;x_0,w,\epsilon)e^{-\Lambda(x_0)\sigma_3/(2\epsilon)}e^{G(z;x_0)\sigma_3/\epsilon}.
\end{split}
\end{equation}
Recalling the Laurent expansion \eqref{eq:dotOexpansion} of $\dot{\mathbf{O}}^{(\mathrm{out})}(z)$ and the asymptotic expansion \eqref{eq:g1-E-expansion} of $\mathbf{E}(z)$, as well as the definitions \eqref{eq:g1-dotU-dotV-def} and \eqref{eq:g1-dotP-dotQ-def}, we arrive at the exact formulae
\begin{equation}
\mathcal{U}_m((m-\tfrac{1}{2})^{2/3}x)=e^{1/3}(m-\tfrac{1}{2})^{2m/3}e^{m\Lambda(x_0)}
\left[\dot{\mathcal{U}}^0_m(w;x_0)+e^{-\Lambda(x_0)/2-1/3}E_{1,12}\right],
\label{eq:g1-UmExact}
\end{equation}
\begin{equation}
\mathcal{V}_m((m-\tfrac{1}{2})^{2/3}x)=e^{-1/3}(m-\tfrac{1}{2})^{-2(m-1)/3}e^{-m\Lambda(x_0)}
\left[\dot{\mathcal{V}}^0_m(w;x_0)+e^{\Lambda(x_0)/2+1/3}E_{1,21}\right],
\label{eq:g1-VmExact}
\end{equation}
\begin{multline}
\mathcal{P}_m((m-\tfrac{1}{2})^{2/3}x)=\\
(m-\tfrac{1}{2})^{1/3}\left[
\frac{\dot{\mathcal{U}}^0_m(w;x_0)\dot{\mathcal{P}}_m(w;x_0) +(E_{1,22}-E_{1,11})\dot{\mathcal{U}}^0_m(w;x_0)+e^{-\Lambda(x_0)/2-1/3}(E_{1,12}E_{1,22}-E_{2,12})}{\dot{\mathcal{U}}^0_m(w;x_0)+e^{-\Lambda(x_0)/2-1/3}E_{1,12}}\right],
\end{multline}
and 
\begin{multline}
\mathcal{Q}_m((m-\tfrac{1}{2})^{2/3}x)=\\(m-\tfrac{1}{2})^{1/3}
\left[\frac{\dot{\mathcal{V}}^0_m(w;x_0)\dot{\mathcal{Q}}_m(w;x_0)+(E_{1,22}-E_{1,11})\dot{\mathcal{V}}^0_m(w;x_0)+e^{\Lambda(x_0)/2+1/3}(E_{2,21}-E_{1,11}E_{1,21})}{\dot{\mathcal{V}}^0_m(w;x_0)+e^{\Lambda(x_0)/2+1/3}E_{1,21}}\right],
\end{multline}
where on the right-hand side $x_0$ and $w$ are chosen so that $x=x_0+(m-\tfrac{1}{2})^{-1}w$.
Being as
\begin{equation}
e^{\pm m\Lambda(x_0)}=e^{\pm m\Lambda(x-\epsilon w)}=e^{\pm m\Lambda(x)}
e^{\pm(w\partial\Lambda(x_0)+w^*\overline{\partial}\Lambda(x_0))}(1+\mathcal{O}(m^{-1}))
\end{equation}
holds uniformly for $x_0$ in compact subsets of $T$ and for bounded $w$, where $\partial$ and $\overline{\partial}$ denote the partial derivatives (applied to the non-analytic function $\Lambda$) with respect to $x_0$ and $x_0^*$, respectively, it makes sense from \eqref{eq:g1-UmExact} and \eqref{eq:g1-VmExact} to introduce the functions (no longer analytic with respect to $w$ due to the explicit presence of $w^*$)
\begin{equation}
\dot{\mathcal{U}}_m(w;x_0):=\dot{\mathcal{U}}_m^0(w;x_0)e^{-(w\partial\Lambda(x_0)+w^*\overline{\partial}\Lambda(x_0))}\quad\text{and}\quad
\dot{\mathcal{V}}_m(w;x_0):=\dot{\mathcal{V}}_m^0(w;x_0)e^{w\partial\Lambda(x_0)+w^*\overline{\partial}\Lambda(x_0)}.
\label{eq:g1-UVdot-redefine}
\end{equation}
Note also that, according to \eqref{eq:g1-dotPQlogderivs}, the products $\dot{\mathcal{U}}^0_m\dot{\mathcal{P}}_m$ and $\dot{\mathcal{V}}^0_m\dot{\mathcal{Q}}_m$ can be rewritten as $\dot{\mathcal{U}}^{0\prime}_m$ and $\dot{\mathcal{V}}^{0\prime}_m$ respectively, where prime denotes the partial derivative with respect to $w$ ($x_0$ held fixed).  Since $e^{1/3}(m-\tfrac{1}{2})^{2m/3}=m^{2m/3}(1+\mathcal{O}(m^{-1}))$ and $(m-\tfrac{1}{2})^{1/3}=m^{1/3}(1+\mathcal{O}(m^{-1}))$ as $m\to\infty$, and 
since the matrix elements of $\mathbf{E}_1$ and $\mathbf{E}_2$ are $\mathcal{O}(m^{-1})$ as $m\to\infty$ uniformly for $x_0\in K\subset T$ and $w\in\mathscr{S}_m(x_0,\delta)$,
we have proven the following result.
\begin{proposition}
Suppose that $K\subset T$ is compact, and a small number $\delta>0$ is given.  Then
\begin{equation}
m^{-2m/3}e^{-m\Lambda(x)}\mathcal{U}_m((m-\tfrac{1}{2})^{2/3}x)=\dot{\mathcal{U}}_m(w;x_0) + \mathcal{O}(m^{-1}),
\label{eq:g1-Um-cheese}
\end{equation}
\begin{equation}
m^{2(m-1)/3}e^{m\Lambda(x)}\mathcal{V}_m((m-\tfrac{1}{2})^{2/3}x)=\dot{\mathcal{V}}_m(w;x_0) +
\mathcal{O}(m^{-1}),
\end{equation}
\begin{equation}
m^{-1/3}\mathcal{P}_m((m-\tfrac{1}{2})^{2/3}x)=\frac{\dot{\mathcal{U}}^0_m(w;x_0)\dot{\mathcal{P}}_m(w;x_0) + \mathcal{O}(m^{-1})}{\dot{\mathcal{U}}^0_m(w;x_0)+\mathcal{O}(m^{-1})} = 
\frac{\dot{\mathcal{U}}^{0\prime}_m(w;x_0) + \mathcal{O}(m^{-1})}{\dot{\mathcal{U}}^0_m(w;x_0)+\mathcal{O}(m^{-1})},
\end{equation}
and
\begin{equation}
m^{-1/3}\mathcal{Q}_m((m-\tfrac{1}{2})^{2/3}x)=\frac{\dot{\mathcal{V}}^0_m(w;x_0)\dot{\mathcal{Q}}_m(w;x_0) + \mathcal{O}(m^{-1})}{\dot{\mathcal{V}}^0_m(w;x_0)+\mathcal{O}(m^{-1})} =
\frac{\dot{\mathcal{V}}^{0\prime}_m(w;x_0)+\mathcal{O}(m^{-1})}{\dot{\mathcal{V}}^0_m(w;x_0)+\mathcal{O}(m^{-1})},
\end{equation}
where $x=x_0+(m-\tfrac{1}{2})^{-1}w$, all hold in the limit $m\to +\infty$ uniformly for $x_0\in K$
and $w\in\mathscr{S}_m(x_0,\delta)$.  If also $w$ is bounded away from 
$\mathscr{Z}_m[\dot{\mathcal{U}}](x_0)$ or $\mathscr{Z}_m[\dot{\mathcal{V}}](x_0)$, then
\begin{equation}
m^{-1/3}\mathcal{P}_m((m-\tfrac{1}{2})^{2/3}x)=\dot{\mathcal{P}}_m(w;x_0) + \mathcal{O}(m^{-1}),
\end{equation}
or 
\begin{equation}
m^{-1/3}\mathcal{Q}_m((m-\tfrac{1}{2})^{2/3}x)=\dot{\mathcal{Q}}_m(w;x_0)+\mathcal{O}(m^{-1})
\end{equation}
hold, respectively.
\label{prop:g1-cheese-approx}
\end{proposition}

Now we discuss how the above asymptotic formulae can be extended to any bounded set in the $w$-plane, that is, how the holes in the ``Swiss cheese'' $\mathscr{S}_m(x_0,\delta)$, which we recall are centered at the points of the pole lattice $\mathscr{P}_m(x_0)$ of the approximating functions $\dot{\mathcal{U}}^0_m(\cdot;x_0)$ and $\dot{\mathcal{V}}^0_m(\cdot;x_0)$, can be filled in.  The main idea here is to make use of the B\"acklund transformations \eqref{backlund-positive}--\eqref{backlund-negative} to exchange $m$ for $m\pm 1$.  For example, according to
\eqref{backlund-positive}, if $n:=m+1$, then
\begin{equation}
\begin{split}
\mathcal{U}_m((m-\tfrac{1}{2})^{2/3}(x_0+(m-\tfrac{1}{2})^{-1}w))&=
\mathcal{V}_n((m-\tfrac{1}{2})^{2/3}(x_0+(m-\tfrac{1}{2})^{-1}w))^{-1}\\ &=
\mathcal{V}_n((n-\tfrac{1}{2})^{2/3}(x_0+(n-\tfrac{1}{2})^{-1}(w-\tfrac{2}{3}x_0 +\mathcal{O}(m^{-1}))))^{-1}.
\end{split}
\label{eq:g1-Uflip}
\end{equation}
Now Proposition~\ref{prop:g1-cheese-approx} provides a valid approximation for the right-hand side, even if $w$ is near the lattice $\mathscr{P}_m(x_0)$, as long as (neglecting the unimportant $\mathcal{O}(m^{-1})$ term) $w-\tfrac{2}{3}x_0\in\mathscr{S}_n(x_0,\delta)$ for some $\delta>0$.  So, suppose that $w\in\mathscr{P}_m(x_0)$ so that $w\not\in\mathscr{S}_m(x_0,\delta)$ as $w$ is a pole of $\dot{\mathcal{U}}_m(\cdot;x_0)$. We wish to show that, for some $\delta>0$ sufficiently small, 
$|w-(\tfrac{2}{3}x_0+\mathscr{P}_{m+1}(x_0))|>\delta$ so that $w\in\mathscr{S}_{m+1}(x_0,\delta)$ and hence Proposition~\ref{prop:g1-cheese-approx} can be applied to the right-hand side of \eqref{eq:g1-Uflip} even as it cannot be applied to the left-hand side.  It suffices to show that the lattices $\mathscr{P}_m(x_0)$ and $\tfrac{2}{3}x_0+\mathscr{P}_{m+1}(x_0)$ are disjoint.
\begin{proposition}
For all $x_0\in T$ and all $m\in \mathbb{Z}_+$, $\frac{2}{3}x_0+\mathscr{P}_{m+1}(x_0)=\mathscr{Z}_m[\dot{\mathcal{U}}](x_0)$ and $\tfrac{2}{3}x_0+\mathscr{Z}_{m+1}[\dot{\mathcal{V}}](x_0)=\mathscr{P}_m(x_0)$.
\label{prop:g1-lattice-equivalence}
\end{proposition}
\begin{proof}
The lattices $\tfrac{2}{3}x_0+\mathscr{P}_{m+1}(x_0)$ and $\mathscr{Z}_m[\dot{\mathcal{U}}](x_0)$ agree if and only if the identity $\tfrac{1}{3}x_0U+\mathscr{A}(\infty^+)-\tfrac{1}{2}U\kappa_1=\mathscr{A}(\infty^-)+\tfrac{1}{2}U\kappa_1$ holds modulo integer multiples of $2\pi i$ and $\mathcal{H}$.  Similarly, the lattices $\tfrac{2}{3}x_0+\mathscr{Z}_m[\dot{\mathcal{V}}](x_0)$ and $\mathscr{P}_{m+1}(x_0)$ agree if and only if $\tfrac{1}{3}x_0U+\mathscr{A}(Q^+)-\tfrac{1}{2}U\kappa_1=\mathscr{A}(Q^-) +\tfrac{1}{2}U\kappa_1$ holds under the same modular condition.  Using \eqref{eq:g1-AbelsTheorem-application} and applying Abel's Theorem to the divisor of the meromorphic function $f^-/f^+$ shows that these two equations are equivalent, so it is enough to prove that $\tfrac{2}{3}x_0+\mathscr{P}_{m+1}(x_0)=\mathscr{Z}_m[\dot{\mathcal{U}}](x_0)$.  Since regardless of how the path of integration is chosen we have $\mathscr{A}(\infty^+)+\mathscr{A}(\infty^-)=0$ modulo integer multiples of $2\pi i$ and $\mathcal{H}$, using
\eqref{eq:Uc1-identity} shows that it is sufficient to establish the identity
\begin{equation}
\frac{2}{3}x_0c_1+2\mathscr{A}(\infty^+)-2c_1\kappa_1=2\pi in_1+\mathcal{H}n_2 \quad \text{for all } x_0\in T,
\quad n_j\in\mathbb{Z}.
\label{eq:g1-lattice-equivalence-modulo}
\end{equation}

To prove this, consider the meromorphic differential $\Omega_H:=\tfrac{3}{2}\mathscr{R}\,dz$ which has poles of order $4$ at each of the two points $\infty^\pm\in\Gamma$ and no other singularities.  The function $H(z)$ is the restriction of an integral of $\Omega_H$ to the sheet $\Gamma^+$ with an additional cut along the contour $L$ connecting $z=D$ to $z=\infty$.  Using $\mathscr{R}^2=z^4+\tfrac{2}{3}x_0z^2-\tfrac{4}{3}z+\Pi$ we see that
\begin{equation}
\Omega_H=\mp\left(\frac{3}{2}\frac{1}{v^{4}} +\frac{1}{2}x_0\frac{1}{v^2}-\frac{1}{v}+\mathcal{O}(1)\right)\,dv,\quad P\to\infty^\pm,
\label{eq:g1-OmegaH-expansion}
\end{equation}
where $v=1/z$ is a local holomorphic coordinate near each of the points $P=\infty^\pm$.  Therefore, $\Omega_H$ has residues $\pm 1$ at $P=\infty^\pm$.   Similarly, the Abel map $\mathscr{A}$ can be expanded near the singularities of $\Omega_H$ as follows:
\begin{equation}
\mathscr{A}(P)=\mathscr{A}(\infty^\pm)\mp c_1\left[v-\frac{1}{9}x_0v^3 + \mathcal{O}(v^4)\right],\quad P\to\infty^\pm.
\label{eq:g1-AbelMap-expansion}
\end{equation}
Here the precise values of $\mathscr{A}(\infty^\pm)$ will be determined only modulo integer multiples of $2\pi i$ and $\mathcal{H}$, as the Abel map depends on the path of integration in the integral that defines it.
We now follow a line of reasoning similar to that used to prove the identity \eqref{eq:Uc1-identity}, for which we have given \cite[Lemma B.1]{BuckinghamMMemoir} as a reference.  However, some details are different in this case so we give some more of the steps.  Let $\widetilde{\Gamma}$ denote the canonical dissection of $\Gamma$ obtained by cutting $\Gamma$ along the cycles $\mathfrak{a}$ and $\mathfrak{b}$.  Thus $\widetilde{\Gamma}$ is a parallelogram in the complex plane
whose positively-oriented boundary $\partial\widetilde{\Gamma}$ is the ordered sequence of paths $\mathfrak{a}$, $\mathfrak{b}$, $-\mathfrak{a}$, $-\mathfrak{b}$, and all four vertices correspond to the same point of $\Gamma$.  We obtain a single-valued branch of the Abel map on $\widetilde{\Gamma}$ by insisting that the path of integration from the point of $\widetilde{\Gamma}$ corresponding to $(D,0)\in\Gamma$ to that corresponding to $P\in\Gamma$ to remain in $\widetilde{\Gamma}$, and we will in this proof (re-)use the symbol $\mathscr{A}(P)$ for this function (generally this is a different branch from the precise one we have used so far).  Now consider the product $\mathscr{A}\Omega_H$ as a meromorphic differential on $\widetilde{\Gamma}$ and consider the integral $I$ defined by the formula
\begin{equation}
I:=\frac{1}{2\pi i}\oint_{\partial\widetilde{\Gamma}} \mathscr{A}\Omega_H,
\end{equation}
where the boundary is positively oriented.  We evaluate $I$ two different ways.

First, we can evaluate $I$ by residues.  Using the expansions \eqref{eq:g1-OmegaH-expansion} and \eqref{eq:g1-AbelMap-expansion}, we easily obtain
\begin{equation}
I=\frac{2}{3}x_0c_1+\mathscr{A}(\infty^+) -\mathscr{A}(\infty^-).
\label{eq:g1-I-residues}
\end{equation}
On the other hand, we can also evaluate $I$ directly, using the fact that the increment of the Abel map along any of the edges of $\partial\widetilde{\Gamma}$ is easy to calculate.  The resulting formula is 
\begin{equation}
\begin{split}
I&=-\frac{\mathcal{H}}{2\pi i}\oint_\mathfrak{a}\Omega_H + \oint_\mathfrak{b}\Omega_H\\
&=-\frac{\mathcal{H}}{2\pi i}\left(\frac{3}{2}\int_B^DR(z)\,dz +\frac{3}{2}\int_D^AR(z)\,dz -\frac{3}{2}\int_A^BR(z)\,dz\right) \\
&\quad\quad{}+
\left(\frac{3}{2}\int_B^DR(z)\,dz-\frac{3}{2}\int_D^AR(z)\,dz -\frac{3}{2}\int_A^BR(z)\,dz\right).
\end{split}
\label{eq:g1-I-direct}
\end{equation}
Therefore, comparing \eqref{eq:g1-I-residues} and \eqref{eq:g1-I-direct} we obtain the identity
\begin{multline}
\frac{2}{3}x_0c_1+\mathscr{A}(\infty^+)-\mathscr{A}(\infty^-) +\frac{\mathcal{H}}{2\pi i}\left(\frac{3}{2}\int_B^DR(z)\,dz +\frac{3}{2}\int_D^AR(z)\,dz -\frac{3}{2}\int_A^BR(z)\,dz\right)\\
{}-
\left(\frac{3}{2}\int_B^DR(z)\,dz -\frac{3}{2}\int_D^AR(z)\,dz-\frac{3}{2}\int_A^BR(z)\,dz\right)=0.
\label{eq:g1-lattice-identity-raw}
\end{multline}
Next, recall that, by definition,
\begin{equation}
-2c_1\kappa_1=\frac{\Phi_+}{\pi}c_1\int_C^A\frac{dz}{R_+(z)} +\frac{\Phi_-}{\pi}c_1\int_C^B\frac{dz}{R_+(z)},
\end{equation}
where $R_+$ indicates the boundary value taken on the left as the indicated arc of $\Sigma$ is traversed.  These integrals can be evaluated explicitly in terms of the fundamental periods $2\pi i$ and $\mathcal{H}$ of the differential $\omega$:
\begin{equation}
-2c_1\kappa_1=i\Phi_+\frac{\mathcal{H}}{2\pi i} + i\Phi_-\left(1-\frac{\mathcal{H}}{2\pi i}\right).
\end{equation}
Next, we rewrite the definitions of $\Phi_+$ and $\Phi_-$ in terms of integrals of derivatives of $H(z)$ as follows:
\begin{equation}
\begin{split}
\Phi_+&=-i\left(\frac{3}{2}\int_D^AR(z)\,dz-\frac{3}{2}\int_A^BR(z)\,dz-\frac{3}{2}\int_B^DR(z)\,dz\right),\\
\Phi_-&=-i\left(\frac{3}{2}\int_D^AR(z)\,dz+\frac{3}{2}\int_A^BR(z)\,dz-\frac{3}{2}\int_B^DR(z)\,dz\right).
\end{split}
\end{equation}
Therefore,
\begin{equation}
-2c_1\kappa_1=\frac{\mathcal{H}}{2\pi i}\left(-2\cdot\frac{3}{2}\int_A^BR(z)\,dz\right) +
\left(\frac{3}{2}\int_D^AR(z)\,dz+\frac{3}{2}\int_A^BR(z)\,dz-\frac{3}{2}\int_B^DR(z)\,dz\right).
\label{eq:g1-minus2c1kappa1}
\end{equation}
But since the residue of $\tfrac{3}{2}R(z)$ at $z=\infty$ is $-1$, we have the identity
\begin{equation}
\frac{3}{2}\int_D^AR(z)\,dz +\frac{3}{2}\int_A^BR(z)\,dz +\frac{3}{2}\int_B^DR(z)\,dz=-2\pi i,
\end{equation}
and using this in the first term on the right-hand side of \eqref{eq:g1-minus2c1kappa1} gives
\begin{multline}
-2c_1\kappa_1=\mathcal{H} +\frac{\mathcal{H}}{2\pi i}\left(\frac{3}{2}\int_D^AR(z)\,dz -\frac{3}{2}\int_A^BR(z)\,dz +\frac{3}{2}\int_B^DR(z)\,dz\right)\\
{}+\left(\frac{3}{2}\int_D^AR(z)\,dz+\frac{3}{2}\int_A^BR(z)\,dz-\frac{3}{2}\int_B^DR(z)\,dz\right).
\end{multline}
Comparing with \eqref{eq:g1-lattice-identity-raw} we obtain
\begin{equation}
\frac{2}{3}x_0c_1+\mathscr{A}(\infty^+)-\mathscr{A}(\infty^-)-2c_1\kappa_1=\mathcal{H}.
\end{equation}
Noting that  $\mathscr{A}(\infty^+)+\mathscr{A}(\infty^-)=2\mathscr{A}(\infty^+)$ holds modulo integer multiples of $2\pi i$ and $\mathcal{H}$, the proof is complete.  (Numerical calculations show that with $\mathscr{A}(P)$ defined as it is elsewhere in this paper aside from this proof, the identity \eqref{eq:g1-lattice-equivalence-modulo} actually holds with $n_1=n_2=0$ for all $x_0\in T$.)
\end{proof}

\begin{proposition}
For all $x_0\in T$ and $m\in\mathbb{Z}_+$, the lattices $\mathscr{P}_m(x_0)$ and $\tfrac{2}{3}x_0+\mathscr{P}_{m+1}(x_0)$ are disjoint.
\label{prop:g1-lattices-disjoint}
\end{proposition}
\begin{proof}
Due to Proposition~\ref{prop:g1-lattice-equivalence} it suffices to show that $\mathscr{P}_m(x_0)$
and $\mathscr{Z}_m[\dot{\mathcal{U}}](x_0)$ are disjoint.  
From \eqref{eq:g1-Lambda-p} and \eqref{eq:g1-Lambda-z}, the condition that these lattices coincide is that $\mathscr{A}(\infty^+)=\mathscr{A}(\infty^-)$ modulo integer multiples of $2\pi i$ and $\mathcal{H}$.  But the Abel map is injective in this setting, and hence the lattices cannot coincide as $\infty^+$ and $\infty^-$ are distinct points on the elliptic curve $\Gamma(x_0)$.
\end{proof}

According to Proposition~\ref{prop:g1-lattices-disjoint}, to approximate $\mathcal{U}_m((m-\tfrac{1}{2})^{2/3}x)$ for $x=x_0+(m-\tfrac{1}{2})^{-1}w$ and $w$ near a point of $\mathscr{P}_m(x_0)$, one can apply Proposition~\ref{prop:g1-cheese-approx} to the right-hand side of \eqref{eq:g1-Uflip} to obtain
\begin{equation}
m^{-2m/3}e^{-m\Lambda(x)}\mathcal{U}_m((m-\tfrac{1}{2})^{2/3}x)=
\frac{e^{2/3+\Lambda(x_0)}e^{-2x_0\partial\Lambda(x_0)/3-2x_0^*\overline{\partial}\Lambda(x_0)/3}(1+\mathcal{O}(m^{-1}))}{\dot{\mathcal{V}}_{m+1}(w-\tfrac{2}{3}x_0+\mathcal{O}(m^{-1});x_0)+\mathcal{O}(m^{-1})},
\label{eq:g1-Um-holes-1}
\end{equation}
where, on the left-hand side, $x=x_0+(m-\tfrac{1}{2})^{-1}w$.
\begin{proposition}
For all $w\in\mathbb{C}$, $x_0\in T$, and $m\in\mathbb{Z}_+$, the following identity holds:
\begin{equation}
\dot{\mathcal{U}}_m(w;x_0)=\frac{e^{2/3+\Lambda(x_0)}e^{-2x_0\partial\Lambda(x_0)/3-2x_0^*\overline{\partial}\Lambda(x_0)/3}}{\dot{\mathcal{V}}_{m+1}(w-\tfrac{2}{3}x_0;x_0)}.
\label{eq:g1-exact-identity}
\end{equation}  
\label{prop:exact-identity}
\end{proposition}
\begin{proof}
According to \eqref{eq:g1-UVdot-redefine}, 
\begin{equation}
\dot{\mathcal{U}}_m(w;x_0)\dot{\mathcal{V}}_{m+1}(w-\tfrac{2}{3}x_0;x_0)e^{2x_0\partial\Lambda(x_0)/3+2x_0^*\overline{\partial}\Lambda(x_0)/3}=
\dot{\mathcal{U}}_m^0(w;x_0)\dot{\mathcal{V}}^0_{m+1}(w-\tfrac{2}{3}x_0;x_0),
\label{eq:g1-exact-identity-withzeros}
\end{equation}
and the right-hand side is a meromorphic function of $w\in\mathbb{C}$.
To prove \eqref{eq:g1-exact-identity}, we first show that 
the product $\dot{\mathcal{U}}^0_m(w;x_0)\dot{\mathcal{V}}^0_{m+1}(w-\tfrac{2}{3}x_0;x_0)$ is independent of both $m\in\mathbb{Z}_+$ and $w\in\mathbb{C}$.
Indeed, it follows from Proposition~\ref{prop:g1-lattice-equivalence} that the product $\dot{\mathcal{U}}^0_m(w;x_0)\dot{\mathcal{V}}^0_{m+1}(w-\tfrac{2}{3}x_0;x_0)$ is an entire function of $w$ for each $x_0\in T$ and each $m\in\mathbb{Z}_+$.  From the exact formulae \eqref{eq:g1-dotU-formula}
and \eqref{eq:g1-dotV-formula} and the first two identities in \eqref{eq:g1-theta-identities} it follows that this entire function is also doubly-periodic with independent periods $2\pi i/c_1$ and $\mathcal{H}/c_1$.  It therefore follows from Liouville's Theorem that it is a constant function of $w$.  Moreover, from the explicit formulae \eqref{eq:g1-dotU-formula}--\eqref{eq:g1-dotV-formula} it is easy to see that $m$ enters into the product $\dot{\mathcal{U}}^0_m(w;x_0)\dot{\mathcal{V}}^0_{m+1}(w-\tfrac{2}{3}x_0;x_0)$ only via the combination $\tfrac{1}{2}w+m\kappa_1$, and hence there can be no dependence on $m\in \mathbb{Z}_+$ either.

The product $\dot{\mathcal{U}}^0_m(w;x_0)\dot{\mathcal{V}}^0_{m+1}(w-\tfrac{2}{3}x_0;x_0)$ therefore depends on $x_0\in T$ only.  To determine its value as a function of $x_0$, let $x_0\in T$ be fixed.  Then, there exists a bounded sequence $\{w_m\}_{m=0}^\infty\subset\mathbb{C}$ such that $w_m\in \mathscr{S}_m(x_0,\delta)$ and $w_m-\tfrac{2}{3}x_0\in\mathscr{S}_{m+1}(x_0,\delta)$ for some fixed $\delta>0$.  Since both \eqref{eq:g1-Um-cheese} and \eqref{eq:g1-Um-holes-1} are valid for $w=w_m$, we can pass to the limit $m\to\infty$ to see that $\dot{\mathcal{U}}^0_m(w;x_0)\dot{\mathcal{V}}_{m+1}^0(w-\tfrac{2}{3}x_0;x_0)=e^{2/3+\Lambda(x_0)}$, and therefore the exact identity \eqref{eq:g1-exact-identity} holds\footnote{The identity \eqref{eq:g1-exact-identity} can be rewritten exactly as 
\begin{equation}
\left(\frac{C+D}{2}\right)^2\frac{\Theta(\mathcal{H}/2;\mathcal{H})^2}{\Theta(2\mathscr{A}(\infty^+)+\mathcal{H}/2;\mathcal{H})^2}=e^{2\mathscr{A}(\infty^+)(1+E^+/c_1)-2\kappa_0+2/3+\Lambda},
\label{eq:g1-exact-identity-rewrite}
\end{equation}
a fact which provides a second, direct proof that the product $\dot{\mathcal{U}}_m(w;x_0)\dot{\mathcal{V}}_{m+1}(2+\tfrac{2}{3}x_0;x_0)$ depends only on $x_0$.  While a direct proof that \eqref{eq:g1-exact-identity-rewrite} holds for all $x_0\in T$ eludes us, we have confirmed it numerically to high accuracy on a dense grid of sample points $x_0\in T$.  However the indirect proof we have given also suffices to establish this identity.} for each $x_0\in T$.
\end{proof}
It follows from Proposition~\ref{prop:exact-identity} that \eqref{eq:g1-Um-holes-1} can be written in the equivalent form
\begin{equation}
m^{-2m/3}e^{-m\Lambda(x)}\mathcal{U}_m((m-\tfrac{1}{2})^{2/3}x)=\frac{1}{\dot{\mathcal{U}}_m(w;x_0)^{-1}+\mathcal{O}(m^{-1})},
\end{equation}
which is valid as $m\to\infty$ with $|w-\mathscr{P}_m(x_0)|<\delta$.  Similar results can be obtained for the functions $\mathcal{V}_m$, $\mathcal{P}_m$, and $\mathcal{Q}_m$.  Combining these with Proposition~\ref{prop:g1-cheese-approx} yields the following result.
\begin{theorem}
Let $x=x_0+(m-\tfrac{1}{2})^{-1}w$.  The following asymptotic formulae hold in the limit $m\to+\infty$ uniformly for $x_0$ in compact subsets of $T$ and bounded $w$:
\begin{equation}
m^{-2m/3}e^{-m\Lambda(x)}\mathcal{U}_m((m-\tfrac{1}{2})^{2/3}x)=\frac{\dot{\mathcal{U}}_m(w;x_0)}{1+\mathcal{O}(m^{-1}\dot{\mathcal{U}}_m(w;x_0))},
\end{equation}
\begin{equation}
m^{2(m-1)/3}e^{m\Lambda(x)}\mathcal{V}_m((m-\tfrac{1}{2})^{2/3}x)=\frac{\dot{\mathcal{V}}_m(w;x_0)}{1+\mathcal{O}(m^{-1}\dot{\mathcal{V}}_m(w;x_0))},
\end{equation}
\begin{equation}
m^{-1/3}\mathcal{P}_m((m-\tfrac{1}{2})^{2/3}x)=\frac{\dot{\mathcal{P}}_m(w;x_0)}{1+\mathcal{O}(m^{-1}\dot{\mathcal{P}}_m(w;x_0))},
\end{equation}
and
\begin{equation}
m^{-1/3}\mathcal{Q}_m((m-\tfrac{1}{2})^{2/3}x)=\frac{\dot{\mathcal{Q}}_m(w;x_0)}{1+\mathcal{O}(m^{-1}\dot{\mathcal{Q}}_m(w;x_0))}.
\end{equation}
Here $\Lambda:T\to\mathbb{C}$ is defined in \eqref{eq:g1-Lambda} and 
$\dot\pp_m$, $\dot\pq_m$, $\dot\pu_m$, and $\dot\pv_m$ are defined in 
\eqref{eq:g1-dotU-formula}--\eqref{eq:g1-dotQ-formula} and 
\eqref{eq:g1-UVdot-redefine}.  
On subsets where $\dot{\mathcal{U}}_m(w;x_0)$, $\dot{\mathcal{V}}_m(w;x_0)$, $\dot{\mathcal{P}}_m(w;x_0)$, or $\dot{\mathcal{Q}}_m(w;x_0)$, respectively, is bounded, the above formulae can be written in the simpler form
\begin{equation}
m^{-2m/3}e^{-m\Lambda(x)}\mathcal{U}_m((m-\tfrac{1}{2})^{2/3}x)=\dot{\mathcal{U}}_m(w;x_0)+\mathcal{O}(m^{-1}),
\end{equation}
\begin{equation}
m^{2(m-1)/3}e^{m\Lambda(x)}\mathcal{V}_m((m-\tfrac{1}{2})^{2/3}x)=\dot{\mathcal{V}}_m(w;x_0)+\mathcal{O}(m^{-1}),
\end{equation}
\begin{equation}
m^{-1/3}\mathcal{P}_m((m-\tfrac{1}{2})^{2/3}x)=\dot{\mathcal{P}}_m(w;x_0)+\mathcal{O}(m^{-1}),
\end{equation}
and
\begin{equation}
m^{-1/3}\mathcal{Q}_m((m-\tfrac{1}{2})^{2/3}x)=\dot{\mathcal{Q}}_m(w;x_0)+\mathcal{O}(m^{-1}),
\end{equation}
respectively.
\label{theorem-g1-basic}
\end{theorem}
One of the dominant features of the asymptotic description of $\mathcal{U}_m$ that is evident both from Theorem~\ref{main-genus-zero-thm} (for $x$ outside $T$) and from Theorem~\ref{theorem-g1-basic} (for $x$ inside $T$) is that $\mathcal{U}_m$ is proportional to an exponentially varying factor in each case, namely $e^{m\lambda(x)}$ for $x$ outside of $T$ and $e^{m\Lambda(x)}$ for $x$ inside of $T$.  We may study this factor by defining a single function for $x\in\mathbb{C}\setminus\partial T$ that equals $\lambda(x)$ for $x\in\mathbb{C}\setminus\overline{T}$ and equals $\Lambda(x)$ for $x\in T$.  Plots of the real part and the sine and cosine of the imaginary part (a phase, really) of this function are shown in Figure~\ref{fig:Lambdalambda}.
\begin{figure}[h]
\begin{center}
\includegraphics[width=2 in]{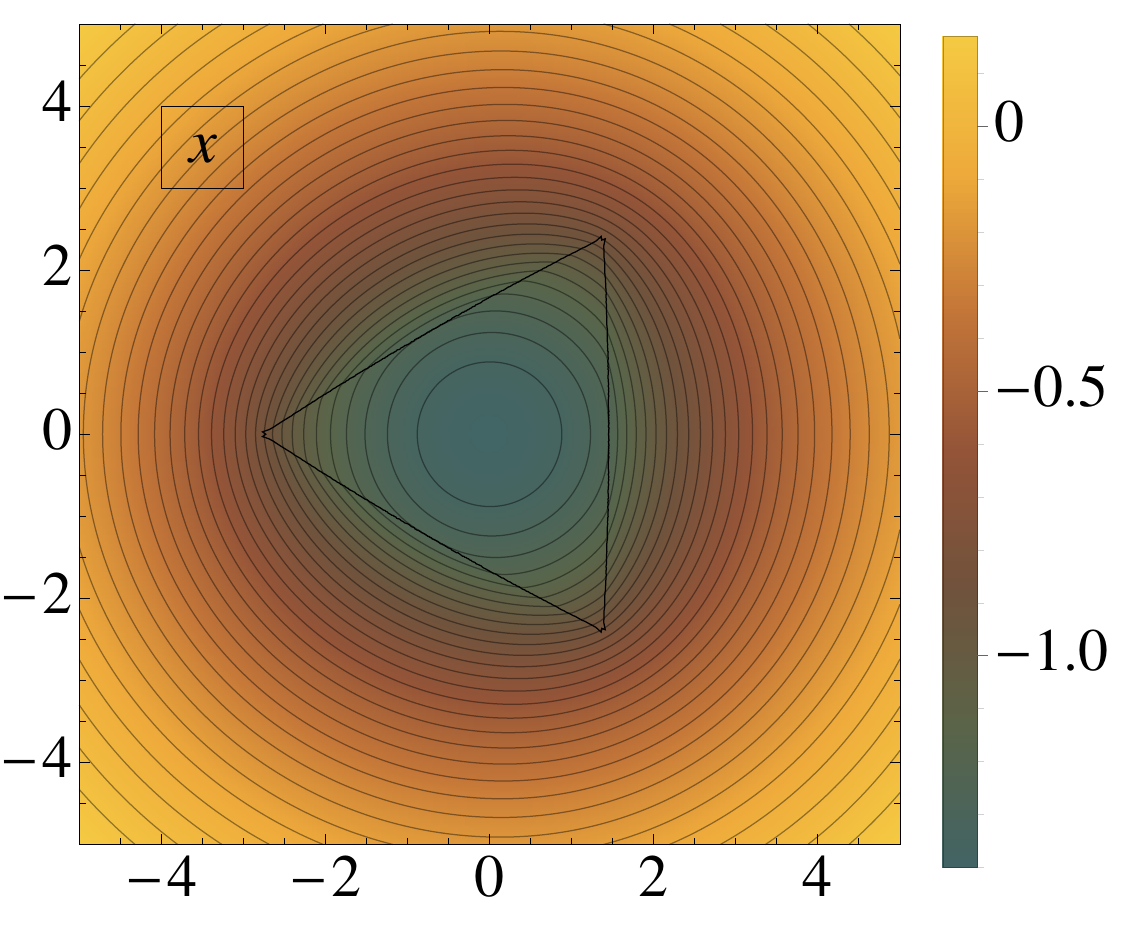}%
\includegraphics[width=2 in]{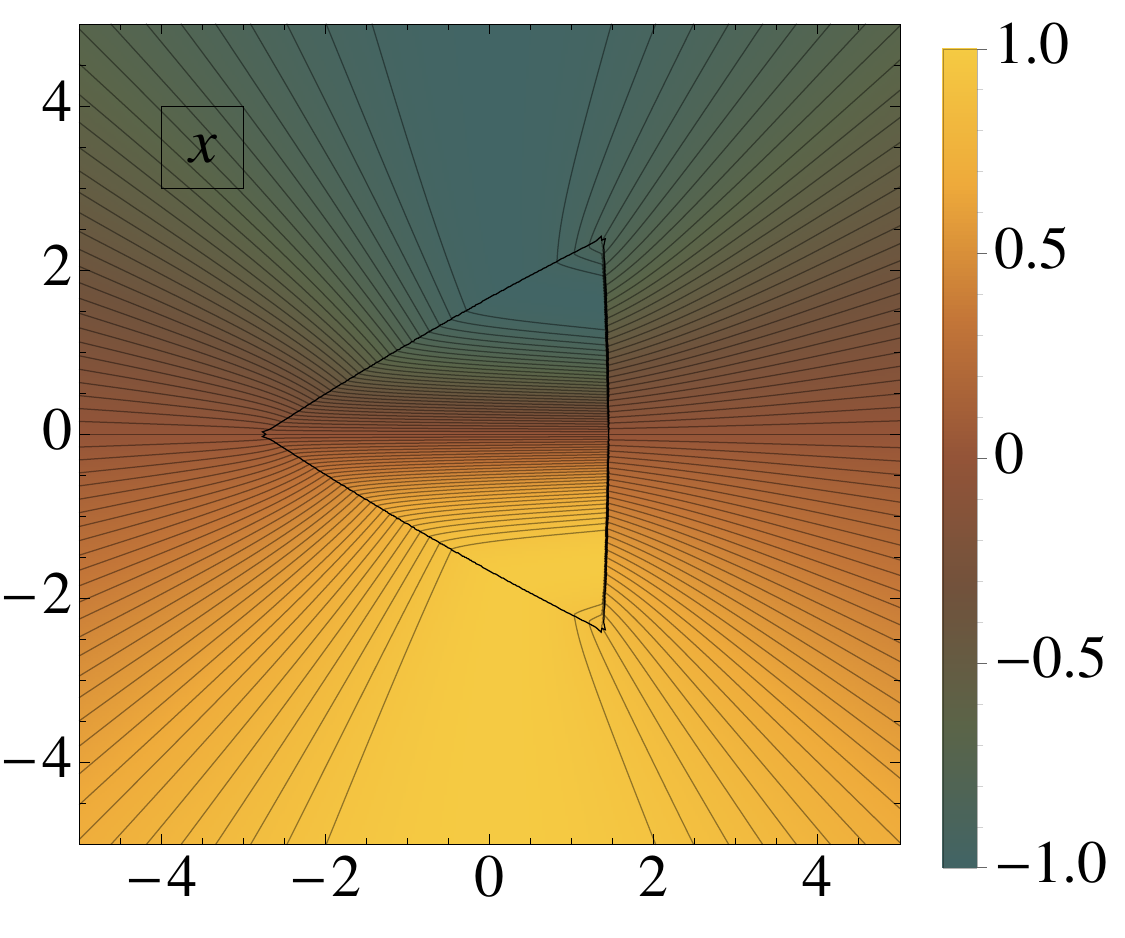}%
\includegraphics[width=2 in]{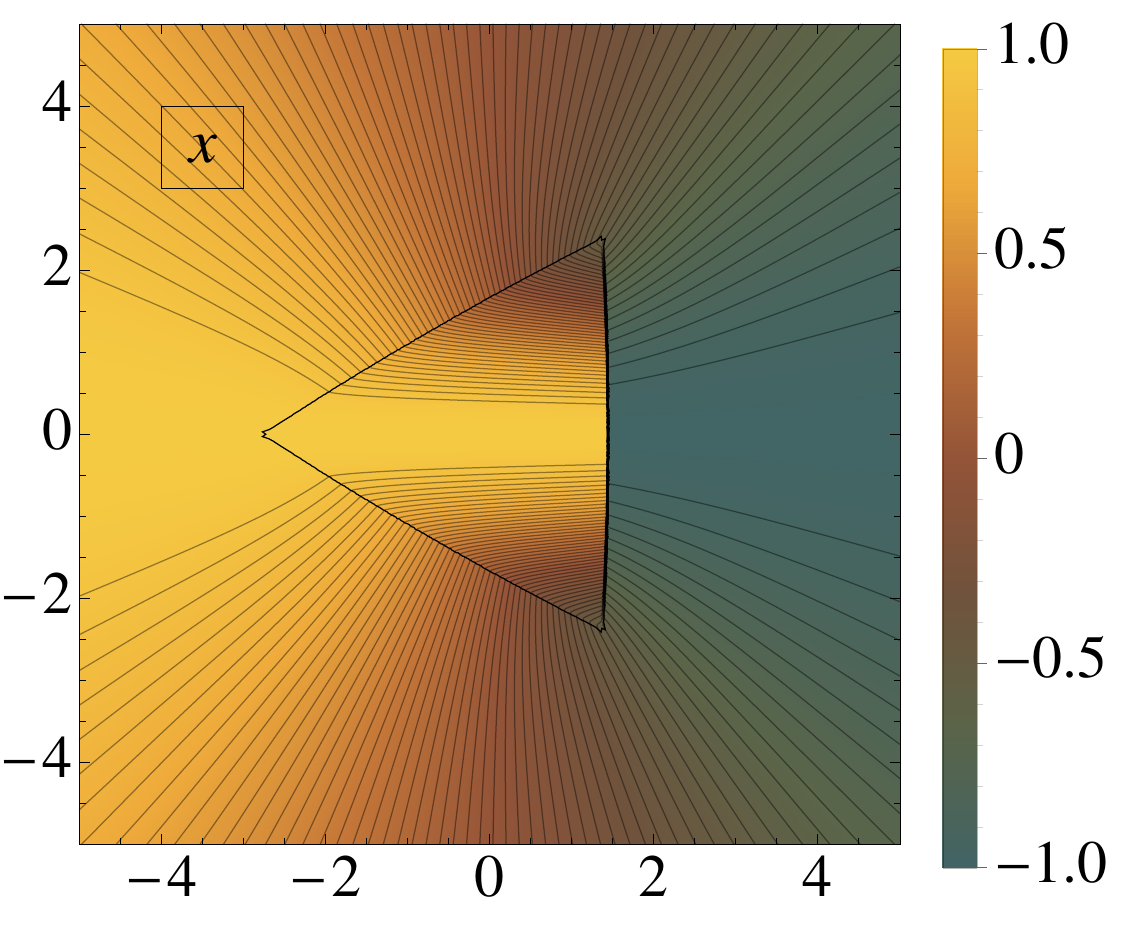}
\end{center}
\caption{\emph{Contour plots of the real part (left), the sine of the imaginary part (middle), and the cosine of the imaginary part (right) of the function defined as $\lambda(x)$ outside the region $T$ and $\Lambda(x)$ inside the region $T$.  The boundary of $T$ is superimposed on each plot for reference.}}
\label{fig:Lambdalambda}
\end{figure}
The plots clearly show what was observed in Proposition~\ref{prop:g1-Lambdalambda}; the phases of $e^{m\lambda(x)}$ and $e^{m\Lambda(x)}$ do not match along the (right) edge of $\partial T$ that crosses the positive real axis.  It turns out that the approximating function $\dot{\mathcal{U}}_m(w;x_0)$, which depends on $m$ unlike the corresponding approximating function $\dot{\pu}(x)$ valid for $x$ outside of $T$, contains an $m$-dependent phase factor as well (this can be seen by looking at the plots in Figure~\ref{fig:g1-U-compare-real} and comparing the behavior near $x=x_c$ and $x=x_e$), and this makes the product $e^{m\Lambda(x)}\dot{\mathcal{U}}_m(w;x_0)$ match onto the corresponding product for 
$x\in\mathbb{C}\setminus\overline{T}$ along the right edge of $\partial T$.  The plots in Figure~\ref{fig:Lambdalambda} also confirm that $\Lambda(x)$ is not analytic for $x\in T$, as one can see from the fact that the orthogonality of the level contours for the real and imaginary parts that holds for the analytic function $\lambda(x)$ defined outside of $T$ obviously fails to hold for $x\in T$.

It is known that all poles and zeros of the rational functions $\mathcal{U}_m$, $\mathcal{V}_m$, $\mathcal{P}_m$, and $\mathcal{Q}_m$ are simple.  The same is of course true of the approximating functions\footnote{Recall from \eqref{eq:g1-UVdot-redefine} that
$\dot{\mathcal{U}}_m(w;x_0)$ differs from the meromorphic-in-$w$ function $\dot{\mathcal{U}}_m^0(w;x_0)$ (respectively $\dot{\mathcal{V}}_m(w;x_0)$ differs from the meromorphic function $\dot{\mathcal{V}}_m^0(w;x_0)$) by a bounded and nonvanishing exponential factor that is non-analytic  in $w$.} $\dot{\mathcal{U}}^0_m(\cdot;x_0)$, $\dot{\mathcal{V}}^0_m(\cdot;x_0)$, $\dot{\mathcal{P}}_m(\cdot;x_0)$, and $\dot{\mathcal{Q}}_m(\cdot;x_0)$, and we can easily establish the following approximation result.
\begin{corollary}
Each simple pole (respectively simple zero) of $\mathcal{U}_m((m-\tfrac{1}{2})^{2/3}x)$ lies within a distance in the $w$-plane that is $\mathcal{O}(m^{-1})$ as $m\to+\infty$ from exactly one simple pole (respectively simple zero) of the approximating function $\dot{\mathcal{U}}^0_m(w;x_0)$.  The analogous result holds for poles and zeros of $\mathcal{V}_m((m-\tfrac{1}{2})^{2/3}x)$, $\mathcal{P}_m((m-\tfrac{1}{2})^{2/3}x)$, and
$\mathcal{Q}_m((m-\tfrac{1}{2})^{2/3}x)$ with corresponding approximating functions $\dot{\mathcal{V}}^0_m(w;x_0)$, $\dot{\mathcal{P}}_m(w;x_0)$, and $\dot{\mathcal{Q}}_m(w;x_0)$.
\label{cor:g1-pole-zero-approx}
\end{corollary}
\begin{proof}
This is an elementary consequence of Rouch\'e's Theorem.
\end{proof}

The approximating functions $\dot{\mathcal{P}}_m(w;x_0)$ and $\dot{\mathcal{Q}}_m(w;x_0)$ defined by \eqref{eq:g1-dotP-formula}--\eqref{eq:g1-dotQ-formula} are meromorphic functions of $w$ but they are nowhere-analytic functions of $x_0\in T$, as follows from \eqref{eq:g1-dbarPi}.  However, in a certain sense these functions are nearly analytic in $x_0$.  A precise statement can be formulated based on Theorem~\ref{theorem-g1-basic} and the Triangle Inequality.
\begin{corollary}
Suppose that $x_0\in T$, that $\zeta\in\mathbb{C}$ is bounded, and that $w+\zeta\in\mathscr{S}_m(x_0,\delta)$.  
If also $w+\zeta$ is bounded away from $\mathscr{Z}_m[\dot{\mathcal{U}}]$, then
\begin{equation}
\dot{\mathcal{P}}_m(w;x_0+\epsilon\zeta)=\dot{\mathcal{P}}_m(w+\zeta;x_0)+\mathcal{O}(m^{-1}),
\end{equation}
while if $w+\zeta$ is bounded away from $\mathscr{Z}_m[\dot{\mathcal{V}}]$, then
\begin{equation}
\dot{\mathcal{Q}}_m(w;x_0+\epsilon\zeta)=\dot{\mathcal{Q}}_m(w+\zeta;x_0)+\mathcal{O}(m^{-1}).
\end{equation}
\label{cor-kernel}
\end{corollary}

\begin{remark}
The near-analyticity in $x_0$ follows because small variations in $x_0$ can be identified to leading order with variations in the coordinate $w$, in which the functions $\dot{\mathcal{P}}_m(w;x_0)$ and $\dot{\mathcal{Q}}_m(w;x_0)$ are meromorphic.  
In fact, it is also true under some hypotheses to avoid the poles that $\dot{\mathcal{U}}_m(w;x_0+\epsilon\zeta)=\dot{\mathcal{U}}_m(w+\zeta;x_0)+\mathcal{O}(m^{-1})$ and that $\dot{\mathcal{V}}_m(w;x_0+\epsilon \zeta)=\dot{\mathcal{V}}_m(w+\zeta;x_0)+\mathcal{O}(m^{-1})$, but as $\dot{\mathcal{U}}_m$ and $\dot{\mathcal{V}}_m$ are not analytic for any $w\in\mathbb{C}$ these relations do not imply any corresponding near-analyticity in $x_0$ for $\dot{\mathcal{U}}_m(w;x_0)$ or $\dot{\mathcal{V}}_m(w;x_0)$.  
\end{remark}
A direct proof of Corollary~\ref{cor-kernel} would be somewhat uninspiring, and this is the reason why we chose from the outset of this section to resolve a given $x\in T$ into $x=x_0+\epsilon w$.  When we wish to obtain a uniform (away from poles) approximation to $\mathcal{P}_m$ or any of the other rational Painlev\'e-II functions, we can simply set $w=0$ throughout and hence $x_0=x$.  This is the easiest way to make plots comparing the approximating functions with the true rational solutions of the Painlev\'e-II equation, and it is the approach taken in Figure~\ref{fig:g1-U-compare-T}, which displays a measure of the magnitude of the uniform approximation $\dot{\mathcal{U}}_m(0;x)$ as $x$ varies in $T$.  This figure shows how remarkably accurately this approximation resolves the locations of the poles and zeros of the rational function $\mathcal{U}_m((m-\tfrac{1}{2})^{2/3}x)$.  The agreement is remarkable even for quite small values of $m$ (for example, note the accuracy in the upper left-hand panel of Figure~\ref{fig:g1-U-compare-T} where $m=2$).
On the other hand, when one wishes to understand the nature of the rational Painlev\'e-II solutions $(\mathcal{P}_m,\mathcal{Q}_m)$ as meromorphic functions, it is better to fix $x_0\in T$ and let $w$ vary, as it is quite difficult to extract the (true) local meromorphicity from the approximating functions $(\dot{\mathcal{P}}_m,\dot{\mathcal{Q}}_m)$ by expanding in the ``parametric'' argument $x_0$.
The idea we have in mind is that $x_0$ parametrizes the base space $T$, while for each $x_0\in T$, $w\in\mathbb{C}$ is a coordinate in the tangent space to $T$.  The four approximating functions $\dot{\mathcal{U}}_m$, $\dot{\mathcal{V}}_m$, $\dot{\mathcal{P}}_m$, and $\dot{\mathcal{Q}}_m$ may therefore be interpreted as functions on the tangent bundle to $T$.  For the functions $\dot{\mathcal{P}}_m$ and $\dot{\mathcal{Q}}_m$, Corollary~\ref{cor-kernel} relates the tangent space coordinates to local perturbations of the base point $x_0\in T$.  This gives rise to the idea of simpler ``tangent approximations'' of the rational functions in which a fixed point $x_0\in T$ is given and $w$ is written in the form $w=\epsilon^{-1}(x-x_0)$, which is understood to be bounded as $m\to +\infty$.  These approximations are simpler in their dependence on $x$, but they are only accurate near $x=x_0$.  

The uniform and tangent (based at $x_0=0$) approximations to $m^{-2m/3}e^{-m\Lambda(x)}\mathcal{U}_m((m-\tfrac{1}{2})^{2/3}x)$ are illustrated and compared for $x\in T\cap\mathbb{R}$ in Figure~\ref{fig:g1-U-compare-real}, and for $x\in e^{i\pi/6}\mathbb{R}\cap T$ in Figures~\ref{fig:g1-U-compare-PiBySix-RealParts} and \ref{fig:g1-U-compare-PiBySix-ImagParts}.  
In each case, the uniform accuracy of the ``uniform'' approximations (obtained by setting $w=0$) over arbitrary compact subsets (and avoiding poles) is remarkable, even for $m$ as small as $m=3$, and for $m\ge 9$ their graphs are virtually indistinguishable to the eye.  The tangent approximation based at the origin is also accurate, but obviously only in a shrinking neighborhood of $x=0$, as expected.  

\begin{remark}
To see how the various approximations improve as $m$ increases, it is
important to be able to make plots for several different values of
$m$, and this can be done efficiently because it is possible to first
calculate the $m$-independent ingredients (quantities such as the
theta-function parameter $\mathcal{H}:T\to\mathbb{C}$) on a dense grid
of points within the elliptic region $T$.  Once these values have been
computed, they can be stored and reused many times.  Generating a plot
of, say, $\dot{\pu}_m(0;x)$ for any given value of $m$ on the same
grid of points is then easy, since $m$ only enters into the theta-function 
arguments and exponents as an explicit multiplicative
parameter.  Of course, as $m$ increases, the approximation
$\dot{\pu}_m(0;x)$ oscillates more violently as a function of $x\in
T$, so the plots can only be accurate if the length scale of
oscillation is large compared to the spacing of the grid on which the
$m$-independent data has been previously computed.  So in practice the
data generated on a given grid can be used to plot the approximate
solutions up to some maximum value of $m$ determined by the grid
spacing.

To calculate the $m$-independent data on a suitable dense grid
requires first solving for the four branch points of the radical
$R(z)$, so one needs to implement an iterative solver for the Boutroux
equations \eqref{eq:g1-Boutroux} and then (once $\Pi=u+iv$ is found
for a given $x_0$) for the moment equations \eqref{eq:g1-moments}.
Since such a solver needs to have a sufficiently accurate initial
guess, it is convenient to organize the calculation so that one moves
sequentially from each grid point to a nearest neighbor, taking the
calculated values of $\Pi$ and the four branch points at one value of
$x_0$ as initial guesses for the iterative calculation at a nearby
value of $x_0$.

We used a grid consisting of equally-spaced points on a large number
of rays through the origin and implemented such an iterative
root-finding scheme to generate data files of the various
$m$-independent quantities needed to make the figures in this paper.
This is a useful approach because each ray begins at the origin
$x_0=0$, and the solution of the Boutroux equations and moment
equations is known explicitly for $x_0=0$.  All of our calculations
were done in Mathematica 9, and the Mathematica notebook that both
generated the data files and also that generated the figures in this
paper by reading back in the data from the files and varying the value
of $m$ is available from the authors on request.
\end{remark}

As a further application of the relationship between the base space coordinate $x_0$ and the ``microscopic'' tangent space coordinate $w$, we offer a proof of the following distributional convergence result.  Recall the definition \eqref{eq:g1-dotP-average-define} of the quantity $\langle\dot{\mathcal{P}}\rangle$ as a function of $x_0\in T$.
Also recall the space $\mathscr{D}'(T)$ of distributions dual to the space 
$\mathscr{D}(T)$, the latter consisting of complex-valued test functions with 
$C^\infty$ real and imaginary parts and compact support in $T$ and equipped with 
the standard seminorm-based topology.
\begin{theorem} As $m\to +\infty$, 
\begin{equation}
m^{-1/3}\mathcal{P}_m((m-\tfrac{1}{2})^{2/3}\cdot) \to \langle\dot{\mathcal{P}}\rangle(\cdot) \quad\text{in $\mathscr{D}'(T)$}.
\end{equation}
\label{theorem:g1-weak-limit}
\end{theorem}
\begin{proof}
We have to show that, given any test function $\phi\in\mathscr{D}(T)$, we have
\begin{equation}
\lim_{m\to +\infty}I_m[\phi] = \iint_T\langle\dot{\mathcal{P}}\rangle(x)\phi(x)\,dA(x)
\end{equation}
where
\begin{equation}
I_m[\phi]:=\iint_T m^{-1/3}\mathcal{P}_m((m-\tfrac{1}{2})^{2/3}x)\phi(x)\,dA(x)
\end{equation}
and where $dA(x)$ denotes the area element $d\mathrm{Re}(x)\,d\mathrm{Im}(x)$.  

We begin by subdividing the domain $T$ into a large number of small, curvilinear parallelograms in the following way.  Because the periods $2\pi i$ and $\mathcal{H}$ of $\omega$ are necessarily independent over the reals, for each $x\in T$ there exist unique real numbers $\alpha(x)$ and $\beta(x)$
such that $\kappa_1(x)U(x)=2\pi i\alpha(x)+\mathcal{H}(x)\beta(x)$; explicitly,
\begin{equation}
\beta(x):=\frac{\mathrm{Re}(\kappa_1(x)U(x))}{\mathrm{Re}(\mathcal{H}(x))}\quad\text{and}\quad
\alpha(x):=\frac{1}{2\pi}\left[\mathrm{Im}(\kappa_1(x)U(x))-\beta(x)\mathrm{Im}(\mathcal{H}(x))\right].
\end{equation}
The coordinates $(\alpha,\beta)$ locally resolve the complex-valued and rapidly-varying global nonlinear phase in the approximate formula \eqref{eq:g1-dotP-formula} into components in the direction of the fundamental periods of the local period lattice of $\dot{\mathcal{P}}_m$ viewed as a function of the tangent space coordinate $w$.  It is essentially the relationship between the macroscopic coordinate $x$ and the microscopic coordinate $w$ embodied in the statement of Corollary~\ref{cor-kernel} that implies that the mapping $\mu:(\mathrm{Re}(x),\mathrm{Im}(x))\mapsto (\alpha(x),\beta(x))$ is invertible on $T$.

With these real-valued coordinates defined on $T$, we may divide $T$ into small domains $T_{jk}$ doubly-indexed by integers $j,k\in\mathbb{Z}$ defined by simple inequalities:
\begin{equation}
T_{jk}:=\left\{x\in T: \epsilon (j-\tfrac{1}{2})\le \alpha(x) < \epsilon (j+\tfrac{1}{2}),\,\epsilon k\le\beta(x)<\epsilon(k+1)\right\},\quad \epsilon = (m-\tfrac{1}{2})^{-1}.
\end{equation}
That is, $T_{jk}$ is the preimage under $\mu$ of a small coordinate square in the $(\alpha,\beta)$-plane.
The tiling of $T$ by the subdomains $T_{jk}$ is illustrated in Figure~\ref{fig:ParallelogramPictures}, which (incidentally) also shows that the intersection points of the level curves of $\alpha$ and $\beta$ (the vertices of the curvilinear parallelograms $T_{jk}$) provide very accurate approximations to the locations of the poles of the rational function $\mathcal{U}_m$.
\begin{figure}[h]
\begin{center}
\includegraphics[width=0.4\linewidth]{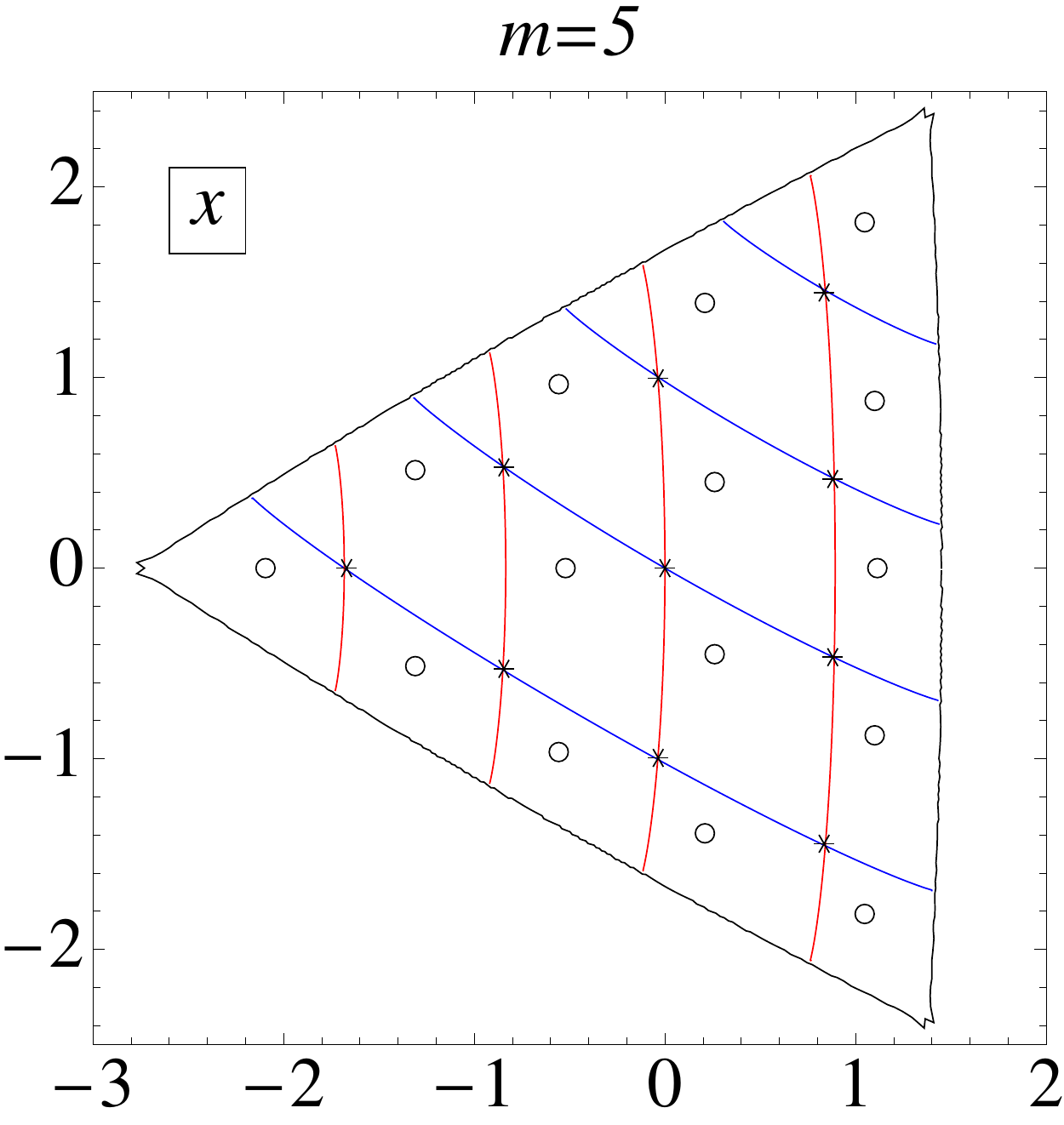}%
\hspace{0.15\linewidth}%
\includegraphics[width=0.4\linewidth]{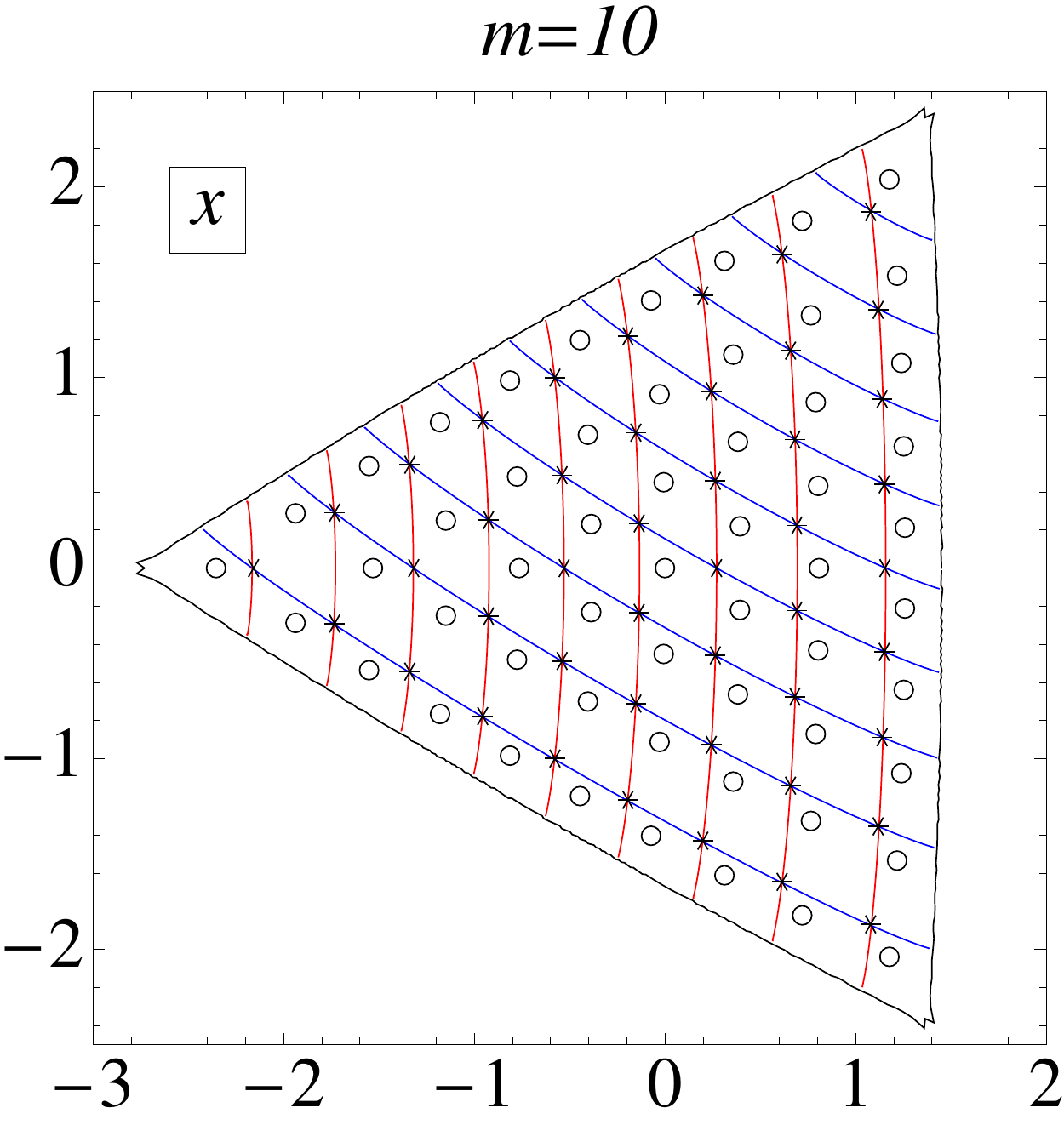}
\end{center}
\caption{\emph{The domains $T_{jk}$ are bounded by the $\epsilon(\mathbb{Z}-\tfrac{1}{2})$-level curves of $\alpha(x)$ (blue) and the $\epsilon\mathbb{Z}$-level curves of $\beta(x)$ (red).  For reference, in each figure are plotted the exact locations of the poles (with asterisks) and zeros (with circles) of $\mathcal{U}_m((m-\tfrac{1}{2})^{2/3}x)$, corresponding to poles of $\mathcal{P}_m((m-\tfrac{1}{2})^{2/3}x)$ with residues of $-1$ and $+1$ respectively. The level curves (deduced from an approximation) appear to pass nearly exactly through the poles.}}
\label{fig:ParallelogramPictures}
\end{figure}
Since $\phi$ has compact support in $T$, we can write
\begin{equation}
I_m[\phi]=\sum_{(j,k)\in\mathbb{Z}^2}\iint_{T_{jk}}m^{-1/3}\mathcal{P}_m((m-\tfrac{1}{2})^{2/3}x)\phi(x)\,dA(x),
\end{equation}
even though, strictly speaking, $T_{jk}$ is only properly defined for a finite number of pairs of indices $(j,k)$ for each given value of $m\in \mathbb{Z}_+$.  Let $x_0^{jk}$ denote an arbitrary point of $T_{jk}$.  Writing $x=x_0^{jk}+\epsilon w$ and using Theorem~\ref{theorem-g1-basic} we claim that
\begin{equation}
\iint_{T_{jk}}m^{-1/3}\mathcal{P}_m((m-\tfrac{1}{2})^{2/3}x)\phi(x)\,dA(x) = \epsilon^2\iint_{\epsilon (T_{jk}-x_0^{jk})}\dot{\mathcal{P}}_m(w;x_0^{jk})\phi(x_0^{jk}+\epsilon w)\,dA(w) + \mathcal{O}(m^{-3})
\end{equation}
holds, where the error term is uniform for all $(j,k)$ corresponding to contributions from the compact support of $\phi$ in $T$.  Now, due to the macro-micro correspondence afforded by Corollary~\ref{cor-kernel}, $\epsilon (T_{jk}-x_0^{jk})$ is a domain in the $w$-plane that is only slightly perturbed (by $\mathcal{O}(m^{-1})$) from a translate of the exact parallelogram $\pgram$ with sides $2\pi i/c_1$ and $\mathcal{H}/c_1$ (both functions of $x_0=x_0^{jk}$).  Moreover, $\phi(x_0^{jk}+\epsilon w)$ can be approximated by $\phi(x_0^{jk})$ to the same accuracy.  Therefore, we also have
\begin{equation}
\iint_{T_{jk}}m^{-1/3}\mathcal{P}_m((m-\tfrac{1}{2})^{2/3}x)\phi(x)\,dA(x)=
\epsilon^2\phi(x_0^{jk})\iint_{\pgram}\dot{\mathcal{P}}_m(w;x_0^{jk})\,dA(w) + \mathcal{O}(m^{-3})
\end{equation}
with the same caveat on the uniformity of the error term.   By definition of $\langle\dot{\mathcal{P}}\rangle$ (see \eqref{eq:g1-dotP-average-define}), we then have
\begin{equation}
\iint_{T_{jk}}m^{-1/3}\mathcal{P}_m((m-\tfrac{1}{2})^{2/3}x)\phi(x)\,dA(x)=\epsilon^2\langle\dot{\mathcal{P}}\rangle(x_0^{jk})\phi(x_0^{jk})
\iint_\pgram\,dA(w) + \mathcal{O}(m^{-3}).
\end{equation}
Again invoking Corollary~\ref{cor-kernel}, we claim that the Jacobian $J(x)$ of the mapping $\mu$ is exactly the double integral $\iint_\pgram\,dA(w)$.  By summing over $(j,k)\in\mathbb{Z}^2$, we therefore arrive at
\begin{equation}
I_m[\phi]=\epsilon^2\sum_{(j,k)\in\mathbb{Z}^2}\langle\dot{\mathcal{P}}\rangle(x_0^{jk})\phi(x_0^{jk})J(x_0^{jk}) + \mathcal{O}(m^{-1}).
\end{equation}
But this is a Riemann sum based on a partition of a region of the $(\alpha,\beta)$-plane into squares of equal area $\epsilon^2$, and consequently, as $\epsilon\to 0$ when $m\to +\infty$,
\begin{equation}
\lim_{m\to +\infty}I_m[\phi]=\iint_{\mu(T)}\left[\langle\dot{\mathcal{P}}\rangle\phi J\right]\circ\mu^{-1}(\alpha,\beta)\,dA(\alpha,\beta) = \iint_T\langle\dot{\mathcal{P}}\rangle(x)\phi(x)\,dA(x),
\end{equation}
as desired.
\end{proof}

\begin{corollary}
Let $\dot{\mathcal{P}}_\mathrm{macro}:\mathbb{C}\to\mathbb{C}$ be defined as follows:
\begin{equation}
\dot{\mathcal{P}}_\mathrm{macro}(x):=\begin{cases}
\langle\dot{\mathcal{P}}\rangle(x),\quad &x\in T,\\
-\tfrac{1}{2}S(x),\quad & x\in\mathbb{C}\setminus T.
\end{cases}
\end{equation}
(Note that $\dot{\mathcal{P}}_\mathrm{macro}$ is continuous across $\partial T$ as a consequence of Proposition~\ref{prop:g1-average-boundary}.)  Then the rescaled rational Painlev\'e-II function $m^{-1/3}\mathcal{P}_m((m-\tfrac{1}{2})^{2/3}\cdot)$ converges to $\dot{\mathcal{P}}_\mathrm{macro}(\cdot)$ in the sense of distributions in $\mathscr{D}'(\mathbb{C}\setminus\partial T)$.
\label{corollary:g1-global-weak-convergence}
\end{corollary}
\begin{proof}
Given a test function $\phi\in\mathscr{D}(\mathbb{C}\setminus\partial T)$, the boundary of $T$ divides the support of $\phi$ into two disjoint compact components, one contained in $T$ and one contained in the open exterior of $\overline{T}$.  The part of the support in $T$ is handled by Theorem~\ref{theorem:g1-weak-limit}, while the corresponding weak convergence result for the complementary part of the support follows from local uniform convergence of $m^{-1/3}\mathcal{P}_m((m-\tfrac{1}{2})^{2/3}\cdot)$ to $-S(\cdot)/2$ as guaranteed by Theorem~\ref{main-genus-zero-thm}.
\end{proof}
Note that we cannot draw a conclusion in the case that the support of the test function straddles the boundary $\partial T$, even though the macroscopic limit function $\dot{\mathcal{P}}_\mathrm{macro}(\cdot)$ is continuous there.  Different analysis is required to study the asymptotic behavior of $m^{-1/3}\mathcal{P}_m((m-\tfrac{1}{2})^{2/3}x)$ for $x$ near $\partial T$; this will be the subject of a subsequent paper \cite{Buckingham-rational-crit}.

Plots of the real and imaginary parts of the continuous function $\dot{\mathcal{P}}_\mathrm{macro}:\mathbb{C}\to\mathbb{C}$ are shown in Figure~\ref{fig:g1-weak-limit}.
Although $\dot{\mathcal{P}}_\mathrm{macro}(\cdot)$ is continuous, analytic for $x\in\mathbb{C}\setminus\overline{T}$, and smooth but not analytic for $x\in T$, it is not (real) differentiable at $\partial T$, where it exhibits sharp gradients.
\begin{figure}[h]
\begin{center}
\includegraphics[width=2.5in]{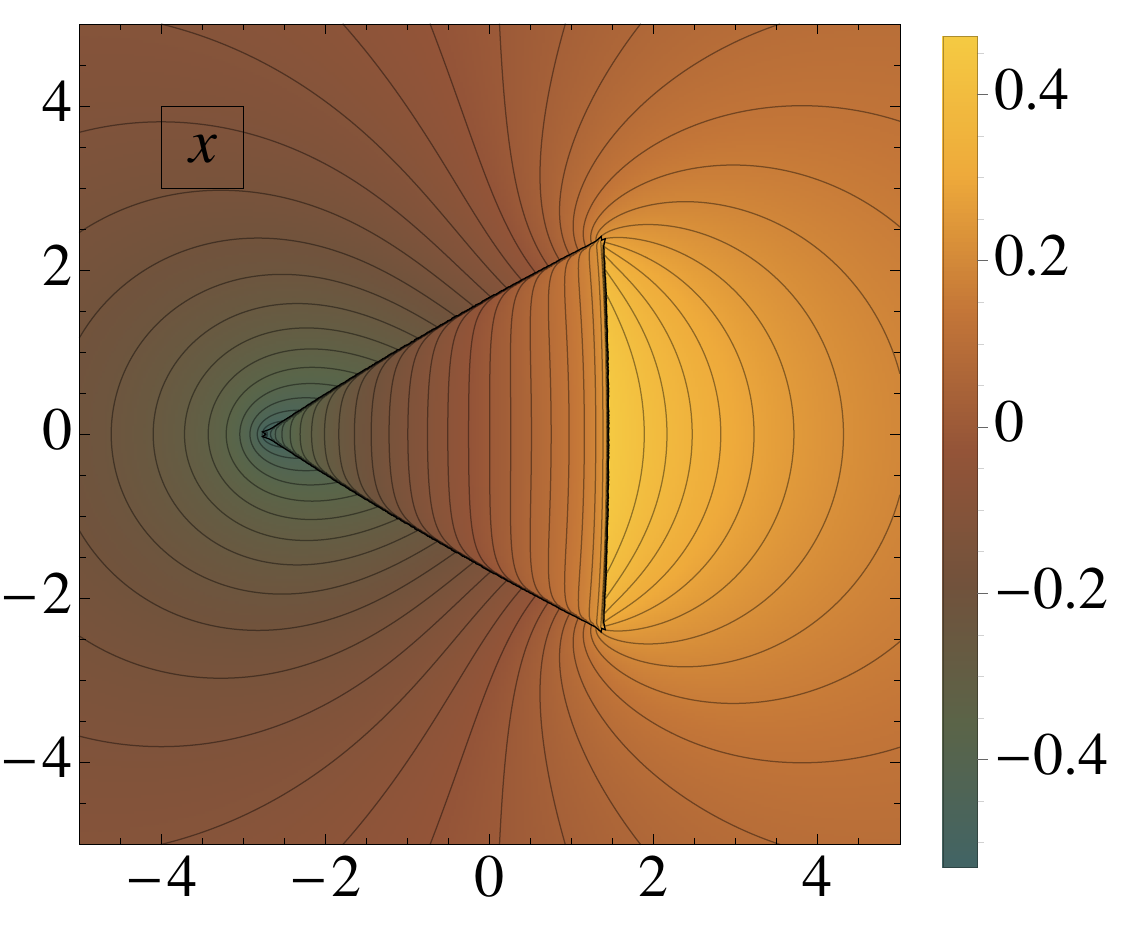}\hspace{0.5 in}%
\includegraphics[width=2.5in]{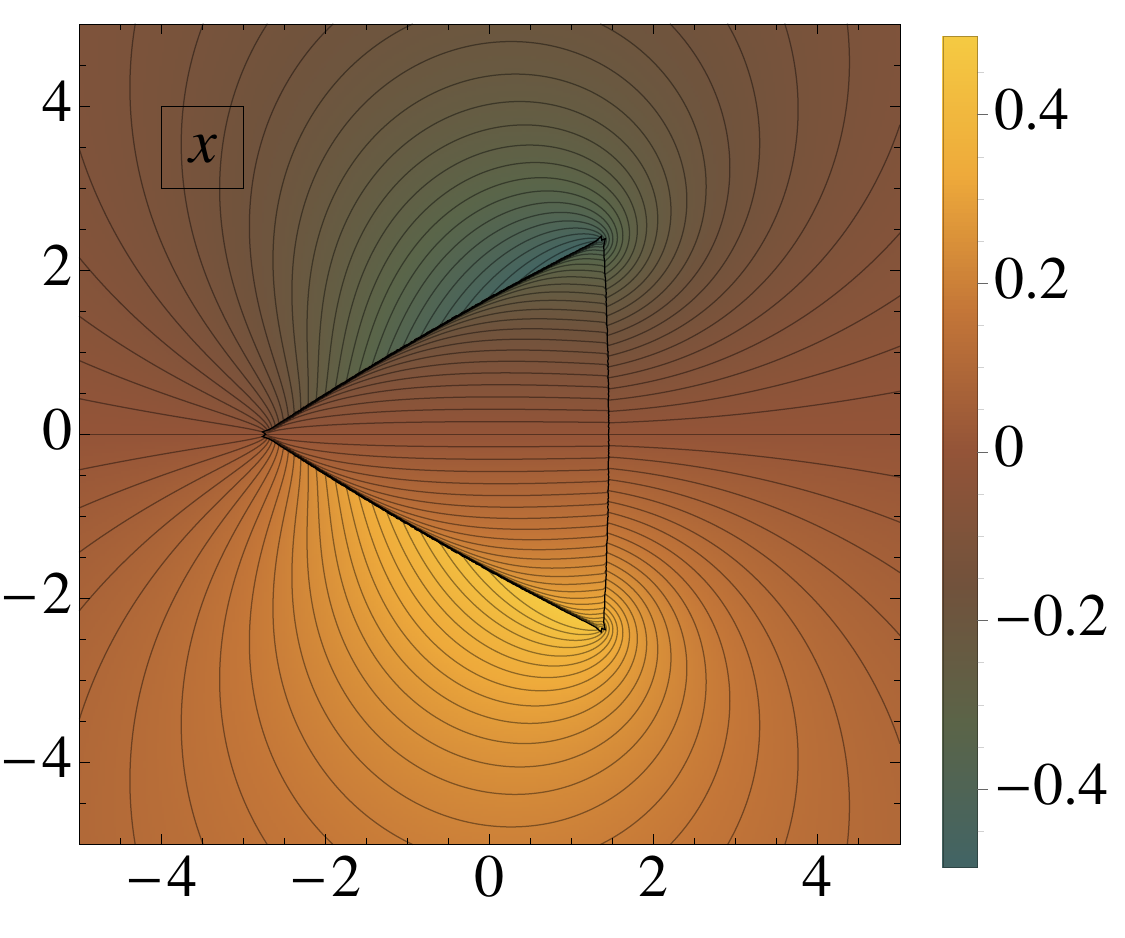}
\end{center}
\caption{\emph{Contour plots of the real (left) and imaginary (right) parts of the function $\dot{\mathcal{P}}_\mathrm{macro}:\mathbb{C}\to\mathbb{C}$.  The boundary of $T$ is superimposed for reference.}}
\label{fig:g1-weak-limit}
\end{figure}

The following results all have similar proofs.  Recall the quantity $\langle\dot{\mathcal{P}}\rangle_\mathbb{R}$ defined as a function of $x_0\in T\cap\mathbb{R}$ by \eqref{eq:g1-dotP-real-average}.  Also recall its relation with $\langle\dot{\mathcal{P}}\rangle$ as given by Proposition~\ref{prop:g1-equivalent-averages}.
\begin{theorem}
Let $\phi\in\mathscr{D}((x_c,x_e))$ be a test function with real compact support in the interval $(x_c,x_e)$.  Then
\begin{equation}
\lim_{m\to +\infty}\dashint_{x_c}^{x_e}m^{-1/3}\mathcal{P}_m((m-\tfrac{1}{2})^{2/3}x)\phi(x)\,dx = 
\int_{x_c}^{x_e}\langle\dot{\mathcal{P}}\rangle_\mathbb{R}(x)\phi(x)\,dx = 
\int_{x_c}^{x_e}\langle\dot{\mathcal{P}}\rangle (x)\phi(x)\,dx,
\end{equation}
where on the left-hand side the integral denotes the Cauchy-Hadamard principal value.
\label{theorem:g1-weak-limit-real}
\end{theorem}
The graph of the weak (distributional) limit $\langle\dot{\mathcal{P}}\rangle_\mathbb{R}$ described by Theorem~\ref{theorem:g1-weak-limit-real} is superimposed on the plots of $m^{-1/3}\mathcal{P}_m((m-\tfrac{1}{2})^{2/3}x)$ and its approximations for real $x$ given in Figure~\ref{fig:g1-P-compare-real} as a ($m$-independent) brown curve.  While we do not have a good explanation for the fact, it is evident from these plots that $\langle\dot{\mathcal{P}}\rangle_\mathbb{R}(\cdot)$ also provides an accurate approximation of the upper envelope of the lower branches of the graph of $m^{-1/3}\mathcal{P}_m((m-\tfrac{1}{2})^{2/3}\cdot)$.

It is the idea of choosing test functions that are approximate characteristic functions (say of a small disk centered at a point in $T$) that gives rise to the interpretation of $\langle\dot{\mathcal{P}}\rangle(\cdot)$ as a limiting local average (over microscopic fluctuations happening on length scales $\Delta x\sim m^{-1}$) of the rescaled rational function $m^{-1/3}\mathcal{P}_m((m-\tfrac{1}{2})^{2/3}\cdot)$ within the region $T$.  In fact, Theorem~\ref{theorem:g1-weak-limit} also holds true if the test function $\phi\in\mathscr{D}(T)$ is simply replaced by a characteristic function of a measurable set.

Another quantity of interest is the macroscopic density of poles of $\mathcal{U}_m$ in the limit $m\to\infty$.  Actually, there are two different types of density we may consider:  planar density $\sigma_\mathrm{P}$ for complex $x_0\in T$, and linear density $\sigma_\mathrm{L}$ for $x_0\in T\cap\mathbb{R}$.  These may be defined as follows:
\begin{equation}
\label{g1-planar-density}
\sigma_\mathrm{P}(x_0):=\mathop{\lim_{m\to\infty}}_{m^{-1}\ll r_m\ll 1}\frac{1}{m^2}\frac{\#\{\text{poles $x$ of $\mathcal{U}_m((m-\tfrac{1}{2})^{2/3}x)$ with $|x-x_0|<r_m$}\}}{\pi r_m^2},\quad x_0\in T
\end{equation}
and 
\begin{equation}
\label{g1-linear-density}
\sigma_\mathrm{L}(x_0):=\mathop{\lim_{m\to\infty}}_{m^{-1}\ll r_m\ll 1}\frac{1}{m}\frac{\#\{\text{real poles $x$ of $\mathcal{U}_m((m-\tfrac{1}{2})^{2/3}x)$ in $(x_0-r_m,x_0+r_m)$}\}}{2r_m},\quad x_0\in T\cap\mathbb{R}
\end{equation}
provided that the limits exist.  Again with the help of the micro-macro correspondence following from the statement $\dot{\mathcal{U}}_m(w;x_0+\epsilon\zeta)=\dot{\mathcal{U}}_m(w+\zeta;x_0)+\mathcal{O}(m^{-1})$ and the pole-confinement result of Corollary~\ref{cor:g1-pole-zero-approx}, 
it is easy to deduce the following result.  Recall the functions $\dot{\sigma}_\mathrm{P}:T\to\mathbb{R}_+$ and $\dot{\sigma}_\mathrm{L}:T\cap\mathbb{R}\to\mathbb{R}_+$ defined by \eqref{eq:dotsigmaPdefine} and \eqref{eq:dotsigmaLdefine} respectively, and for which we have obtained the explicit formulae \eqref{eq:dotsigmaPformula}--\eqref{eq:dotsigmaLformula}.
\begin{theorem}
$\sigma_\mathrm{P}=\dot{\sigma}_\mathrm{P}$ and $\sigma_\mathrm{L}=\dot{\sigma}_\mathrm{L}$ hold as identities on $T$ and $T\cap\mathbb{R}$ respectively.
\label{theorem:g1-density-equalities}
\end{theorem}
Plots of the densities $\sigma_\mathrm{P}$ and $\sigma_\mathrm{L}$ are shown in Figure~\ref{fig:g1-densities}.  
We note that one can similarly ask for the planar and linear macroscopic densities of zeros of 
the rational function $\mathcal{U}_m((m-\tfrac{1}{2})^{2/3}\cdot)$ (which equal the corresponding densities of poles), or for the planar and linear macroscopic densities of poles of the rational function $\dot{\mathcal{P}}_m((m-\tfrac{1}{2})^{2/3}\cdot)$ (which are exactly twice the corresponding densities of poles of $\mathcal{U}_m((m-\tfrac{1}{2})^{2/3}\cdot)$).

\begin{remark}
This brings us to an interesting connection with the application of the rational Painlev\'e-II functions to the construction of equilibrium configurations of interacting fluid vortices in planar flows as discussed by Clarkson \cite{Clarkson:2009}.  The poles and zeros of $\mathcal{U}_m$
turn out to correspond to the equilibrium locations in the plane (under the obvious identification $\mathbb{C}\cong\mathbb{R}^2$) of vortices of opposite unit vorticity.  The condition of fluid-mechanical equilibrium suggests a variational characterization of equilibrium configurations of vortices, and one might think that the limiting macroscopic density $\sigma_\mathrm{P}$ might 
arise in an appropriate continuum limit as the solution of a well-posed problem in the calculus of variations.  However, the fact that the ``vortex gas'' that emerges in the limit $m\to +\infty$ consists of oppositely charged vortices with equal densities means that one has to consider the limiting functional for signed measures of arbitrary mass, and it seems unlikely to us that the asymptotic variational problem will be well-posed.  In fact, as explained in \cite{Clarkson:2009}, there are many different equilibrium configurations of large numbers of oppositely-charged vortices, only some of which correspond to poles or zeros of rational Painlev\'e-II functions.
\end{remark}

\appendix

\section{The Standard Airy Parametrix}
\renewcommand{\theequation}{A-\arabic{equation}}
Recall the Airy function $y=\mathrm{Ai}(\xi)$, the solution of the differential equation $y''(\xi)-\xi y(\xi)=0$ uniquely specified by the asymptotic behavior
\begin{equation}
\mathrm{Ai}(\xi)=\frac{1}{2\xi^{1/4}\sqrt{\pi}}e^{-2\xi^{3/2}/3}\left(1-\frac{5}{48}\xi^{-3/2}+\mathcal{O}(\xi^{-3})\right),\quad\xi\to\infty,\quad |\arg(\xi)|<\pi.
\label{eq:AiryAppendix-AiryAsymp}
\end{equation}
Complete information about this special function can be found online in \cite{DLMF}.  
The Airy function is an entire function of $\xi$ that satisfies the identity
\begin{equation}
\mathrm{Ai}(\xi)+e^{-2\pi i/3}\mathrm{Ai}(\xi e^{-2\pi i/3}) + e^{2\pi i/3}\mathrm{Ai}(\xi e^{2\pi i/3})=0
\label{eq:AiryAppendix-AiryIdentity}
\end{equation}
and its derivative has the asymptotic behavior
\begin{equation}
\mathrm{Ai}'(\xi)=-\frac{\xi^{1/4}}{2\sqrt{\pi}}e^{-2\xi^{3/2}/3}\left(1+\frac{7}{48}\xi^{-3/2}+\mathcal{O}(\xi^{-3})\right),\quad
\xi\to\infty,\quad |\arg(\xi)|<\pi.
\label{eq:AiryAppendix-AiryPrimeAsymp}
\end{equation}
Now set $\xi:=(\tfrac{3}{4})^{2/3}\zeta$ and consider the matrix defined in disjoint and complementary sectors of the $\zeta$-plane by the following formulae:
\begin{equation}
\mathbf{A}(\zeta):=\sqrt{2\pi}\left(\frac{4}{3}\right)^{\sigma_3/6}
\begin{bmatrix}-\mathrm{Ai}'(\xi) & e^{2\pi i/3}\mathrm{Ai}'(\xi e^{-2\pi i/3})\\
-i\mathrm{Ai}(\xi) & ie^{-2\pi i/3}\mathrm{Ai}(\xi e^{-2\pi i/3})\end{bmatrix}
e^{2\xi^{3/2}\sigma_3/3},\quad
0<\arg(\zeta)<\frac{2\pi}{3},
\label{eq:AiryAppendix-ParametrixDef-I}
\end{equation}
\begin{equation}
\mathbf{A}(\zeta):=\sqrt{2\pi}\left(\frac{4}{3}\right)^{\sigma_3/6}
\begin{bmatrix}e^{-2\pi i/3}\mathrm{Ai}'(\xi e^{2\pi i/3}) & e^{2\pi i/3}\mathrm{Ai}'(\xi e^{-2\pi i/3})\\
ie^{2\pi i/3}\mathrm{Ai}(\xi e^{2\pi i/3}) & ie^{-2\pi i/3}\mathrm{Ai}(\xi e^{-2\pi i/3})\end{bmatrix}
e^{2\xi^{3/2}\sigma_3/3},\quad \frac{2\pi}{3}<\arg(\zeta)<\pi,
\label{eq:AiryAppendix-ParametrixDef-II}
\end{equation}
\begin{equation}
\mathbf{A}(\zeta):=\sqrt{2\pi}\left(\frac{4}{3}\right)^{\sigma_3/6}\begin{bmatrix}
e^{2\pi i/3}\mathrm{Ai}'(\xi e^{-2\pi i/3}) & -e^{-2\pi i/3}\mathrm{Ai}'(\xi e^{2\pi i/3})\\
ie^{-2\pi i/3}\mathrm{Ai}(\xi e^{-2\pi i/3}) & -ie^{2\pi i/3}\mathrm{Ai}(\xi e^{2\pi i/3})
\end{bmatrix}e^{2\xi^{3/2}\sigma_3/3},\quad -\pi<\arg(\zeta)<-\frac{2\pi}{3},
\label{eq:AiryAppendix-ParametrixDef-III}
\end{equation}
\begin{equation}
\mathbf{A}(\zeta):=\sqrt{2\pi}\left(\frac{4}{3}\right)^{\sigma_3/6}
\begin{bmatrix} -\mathrm{Ai}'(\xi) & -e^{-2\pi i/3}\mathrm{Ai}'(\xi e^{2\pi i/3})\\
-i\mathrm{Ai}(\xi) & -ie^{2\pi i/3}\mathrm{Ai}(\xi e^{2\pi i/3})\end{bmatrix}
e^{2\xi^{3/2}\sigma_3/3},\quad -\frac{2\pi}{3}<\arg(\zeta)<0.
\label{eq:AiryAppendix-ParametrixDef-IV}
\end{equation}
From the asymptotic formulae \eqref{eq:AiryAppendix-AiryAsymp} and \eqref{eq:AiryAppendix-AiryPrimeAsymp} it follows that
\begin{equation}
\mathbf{A}(\zeta)
\mathbf{V}^{-1}
\zeta^{-\sigma_3/4}=\mathbb{I} +\begin{bmatrix}\mathcal{O}(\zeta^{-3}) & \mathcal{O}(\zeta^{-1})\\
\mathcal{O}(\zeta^{-2}) & \mathcal{O}(\zeta^{-3})\end{bmatrix}\quad
\label{eq:AiryAppendix-ParametrixAsymp}
\end{equation}
as $\zeta\to\infty$ uniformly in all directions of the complex plane, where $\mathbf{V}$ is the unimodular and unitary matrix
\begin{equation}
\mathbf{V}:=\frac{1}{\sqrt{2}}\begin{bmatrix}1 & -i\\-i & 1\end{bmatrix},\quad \mathbf{V}^{-1}=\mathbf{V}^\dagger.
\label{eq:AiryAppendix-EigenvectorMatrix}
\end{equation}
The matrix function $\mathbf{A}$ is analytic in the complex $\zeta$-plane with the exception of  the four rays $\arg(\pm\zeta)=0$ and $\arg(\zeta)=\pm 2\pi/3$, along which $\mathbf{A}$ has jump discontinuities.  Taking the four rays to be oriented from the origin outward and using
\eqref{eq:AiryAppendix-AiryIdentity} yields the jump conditions satisfied by $\mathbf{A}$:
\begin{equation}
\mathbf{A}_+(\zeta)=\mathbf{A}_-(\zeta)\begin{bmatrix} 1 & e^{-\zeta^{3/2}} \\
0 & 1\end{bmatrix},\quad \arg(\zeta)=0,
\label{eq:AiryAppendix-AiryJump-I}
\end{equation}
\begin{equation}
\mathbf{A}_+(\zeta)=\mathbf{A}_-(\zeta)\begin{bmatrix}0 & -1\\1 & 0\end{bmatrix},\quad
\arg(-\zeta)=0,
\label{eq:AiryAppendix-AiryJump-II}
\end{equation}
\begin{equation}
\mathbf{A}_+(\zeta)=\mathbf{A}_-(\zeta)\begin{bmatrix} 1 & 0\\
-e^{\zeta^{3/2}} & 1\end{bmatrix},\quad\arg(\zeta)=\pm\frac{2\pi}{3}.
\label{eq:AiryAppendix-AiryJump-III}
 \end{equation}
From these jump conditions and the asymptotic condition \eqref{eq:AiryAppendix-ParametrixAsymp} it follows via a Liouville argument that $\det(\mathbf{A}(\zeta))\equiv 1$.

\section{Small-Norm Riemann-Hilbert Problems}
\label{small-norm-app}
\renewcommand{\theequation}{B-\arabic{equation}}

Here we gather together for completeness certain results necessary for the 
analysis of small-norm Riemann-Hilbert problems, in particular information 
about the moments of the solution at infinity that is often used implicitly.  

An \emph{arc} is the image of an injective continuously differentiable map $\zeta: I\to\mathbb{C}$, where $I\subset\mathbb{R}$ is an open interval (possibly infinite), for which both $\zeta$ and $\zeta'$ extend continuously to any finite endpoints of $I$ (we say the arc terminates at the image of any such endpoint).   By a \emph{contour} $\Sigma\subset\mathbb{C}$ we mean the closure of a finite union of arcs.  A point $\zeta$ on a contour $\Sigma$ is a \emph{regular point} of $\Sigma$ if there is a complex open neighborhood $U$ containing $\zeta$ such that $\Sigma\cap U$ is an arc.
The set of regular points of a contour $\Sigma$ is denoted $\Sigma^\circ$.  If each point of $\Sigma^\circ$ is assigned a definite orientation, then $\Sigma$ is called an \emph{oriented contour}.   If $\zeta\in\Sigma\setminus\Sigma^\circ$, then there is a complex open neighborhood $U$ containing $\zeta$ such that $\Sigma\cap U$ is the union of $\{\zeta\}$ and $M\ge 1$ pairwise disjoint arcs terminating at $\zeta$.  If we denote by $|\Sigma_1|$ the total length of the part of $\Sigma$ lying within the unit disk of the complex plane, then $|\Sigma_1|<\infty$ for every contour $\Sigma$.
\begin{definition}
An oriented contour $\Sigma$ (possibly unbounded) is called \emph{admissible} if 
\begin{itemize}
\item whenever $\zeta\in\Sigma\setminus\Sigma^\circ$, the (well-defined) tangents to each of the $M\ge 1$ arcs terminating at $\zeta$ are distinct, and
\item there exists a radius $R>0$ such that for $|\zeta|>R$ each unbounded arc of $\Sigma$ coincides exactly with a straight line.
\end{itemize}
\end{definition}
In some applications it is convenient to assume further that $\Sigma$ is a \emph{complete contour}, namely one for which
$\mathbb{C}\setminus\Sigma$ is the union of two complementary regions $\Omega_+$ and $\Omega_-$ such that each arc  of $\Sigma$ has been unambiguously oriented with $\Omega_+$ (respectively $\Omega_-$) lying on the left (respectively right).  The assumption of completeness can be made without loss of generality, because any admissible contour can always be ``completed'' by including a finite number of additional arcs to achieve completeness without violating any of the other requirements.  

Given an admissible contour $\Sigma$, a map $\mathbf{V}:\Sigma^\circ\to SL(2,\mathbb{C})$
is called an \emph{admissible jump matrix} if 
$\mathbf{V}-\mathbb{I}\in L^\infty(\Sigma;\mathbb{C}^{2\times 2})\cap L^2(\Sigma;\mathbb{C}^{2\times 2})$, where arc length measure $|d\zeta|$ and any (pointwise) norm on $2\times 2$ matrices serve to define the spaces.
If one is given an admissible jump matrix $\mathbf{V}$ on an admissible contour $\Sigma$ and one wishes
to complete the contour, one may extend the definition of the jump matrix $\mathbf{V}$ to the completed contour in a natural way simply by
setting
 $\mathbf{V}(z):=\mathbb{I}$ along each arc that is included to achieve completeness, and the resulting jump matrix is again admissible.
 
Consider now the following Riemann-Hilbert problem:
 \begin{rhp}
Let $\Sigma$ be an admissible contour, and let $\mathbf{V}$ be an admissible jump matrix on $\Sigma$. Find a $2\times 2$ matrix-valued function $\mathbf{M}:\mathbb{C}\setminus\Sigma\to \mathbb{C}^{2\times 2}$ satisfying the following three conditions:
\begin{itemize}
\item[]{\bf Analyticity.}  $\mathbf{M}$ is an analytic function of $z\in\mathbb{C}\setminus\Sigma$,
and $\mathbf{M}$ has well-defined boundary values $\mathbf{M}_+(\zeta)$ and $\mathbf{M}_-(\zeta)$ for almost every $\zeta\in\Sigma^\circ$
given by the non-tangential limits of $\mathbf{M}(z)$ as $z\to \zeta\in\Sigma^\circ$ with $z$ on the left (respectively right) of the oriented arc containing $\zeta$ to define $\mathbf{M}_+(\zeta)$ (respectively $\mathbf{M}_-(\zeta)$).  The differences $\mathbf{M}_\pm(\zeta)-\mathbb{I}$ are required to be square-integrable on $\Sigma$.
\item[]{\bf Jump Condition.}  The boundary values are related by $\mathbf{M}_+(\zeta)=\mathbf{M}_-(\zeta)\mathbf{V}(\zeta)$ for almost every $\zeta\in\Sigma^\circ$.
\item[]{\bf Normalization.}  $\mathbf{M}(z)$ tends to the identity matrix $\mathbb{I}$ as $z\to\infty$ 
in any direction distinct from those of the unbounded arcs (lines) of $\Sigma$.
\end{itemize}
\label{small-norm-rhp}
\end{rhp}

This Riemann-Hilbert problem is equivalent to a certain singular integral equation, the theory of which forms a basis for the study of the Riemann-Hilbert problem.  To formulate the integral equation, we need to recall some basic singular integral operators.
\begin{definition}Let $\Sigma$ be an admissible contour.  Given a function $f\in L^2(\Sigma)$, 
\eq
(\mathcal{C}^\Sigma f)(z):=\frac{1}{2\pi i}\int_\Sigma\frac{f(w)\,dw}{w-z}, \quad z\in\mathbb{C}\backslash\Sigma
\endeq
is called the \emph{Cauchy transform} of $f$ relative to $\Sigma$.
\end{definition}
Cauchy transforms can be calculated for scalar-valued, vector-valued, or matrix-valued $f$ that are square-integrable given an arbitrary pointwise norm $|\cdot|$.  The fact that $(\mathcal{C}^\Sigma f)(z)$ is well-defined for each $z\in\mathbb{C}\setminus\Sigma$ follows from the Cauchy-Schwarz inequality because $(\cdot-z)^{-1}\in L^2(\Sigma)$.  In the same way, one sees that 
the function $\mathcal{C}^\Sigma f$ is analytic in its domain of definition because $(\cdot-z)^{-2}$ is also square-integrable on $\Sigma$ when $z\in\mathbb{C}\setminus\Sigma$.  The integrand in the definition of $\mathcal{C}^\Sigma f$ tends to zero pointwise in $w$ as $z\to\infty$.  If we insist that
$z$ tends to infinity in such a way that it avoids sectors centered in the direction of each unbounded line of $\Sigma$ of opening angle $2\vartheta>0$, then some simple trigonometry shows that $|w-z|>|w|\sin(\vartheta)$ holds for $w\in\Sigma$ and $|z|$ sufficiently large, and as $|\cdot |^{-1}\in L^2(\Sigma)$ (assuming without loss of generality that $0\not\in\Sigma$) a Cauchy-Schwarz argument shows that the conditions of the Lebesgue Dominated Convergence Theorem are satisfied, and hence $(\mathcal{C}^\Sigma f)(z)$ tends to zero as $z\to\infty$ in the specified fashion.
See Muskhelishvili \cite{Muskhelishvili:1992} for more information.  

At each point $\zeta\in\Sigma^\circ$ we may consider the boundary values $(\mathcal{C}^\Sigma_+f)(\zeta)$ and
$(\mathcal{C}^\Sigma_-f)(\zeta)$ given by
\eq
(\mathcal{C}_+^\Sigma f)(\zeta):=\mathop{\lim_{z\to \zeta}}_{\text{$z$ on the left of $\zeta$}}(\mathcal{C}^\Sigma f)(z)
 \quad \text{and}\quad
 (\mathcal{C}_-^\Sigma f)(\zeta):=\mathop{\lim_{z\to \zeta}}_{\text{$z$ on the right of $\zeta$}}(\mathcal{C}^\Sigma f)(z)
\endeq
whenever the (non-tangential) limits exist.  
It turns out that, for admissible contours $\Sigma$ and square-integrable $f$,
the boundary values $(\mathcal{C}_\pm^\Sigma f)(\zeta)$ exist pointwise for 
almost every $\zeta\in\Sigma^\circ$ and these boundary values may be identified with elements of $L^2(\Sigma)$.  Therefore, $\mathcal{C}_\pm^\Sigma$ can be viewed as linear operators on $L^2(\Sigma)$.  Furthermore, combining (i) a fundamental result of Coifman, McIntosh, and Meyer 
\cite{Coifman:1982} establishing the $L^2$-boundedness of the Hilbert transform on Lipschitz curves with (ii) an argument to properly handle self-intersection points of $\Sigma$ (see, for example the proof of Lemma 8.1 in \cite{BealsC84}),
the operators $\mathcal{C}_\pm^\Sigma:L^2(\Sigma)\to L^2(\Sigma)$ are bounded.  The operator norms $\|\mathcal{C}_\pm^\Sigma\|_{L^2(\Sigma)\circlearrowleft}$ are controlled by the geometrical properties of $\Sigma$, but we will only use the fact that for each admissible contour $\Sigma$, $\|\mathcal{C}_\pm^\Sigma\|_{L^2(\Sigma)\circlearrowleft}<\infty$.  
A classical local result in the theory (see \cite{Muskhelishvili:1992}) is the following:
\begin{proposition}[Plemelj formula]
Suppose $\Sigma$ is an admissible contour and $f\in L^2(\Sigma)$.  Then
\eq
(\mathcal{C}_+^\Sigma f)(\zeta) - (\mathcal{C}_-^\Sigma f)(\zeta) = f(\zeta)
\endeq
holds for almost all regular points $\zeta\in\Sigma^\circ$.
\end{proposition}
This pointwise result implies the operator identity $\mathcal{C}_+^\Sigma-\mathcal{C}_-^\Sigma=\mathrm{id}$ on $L^2(\Sigma)$.
In the special case that $\Sigma$ is a complete contour, the 
operators $\pm\mathcal{C}_\pm^\Sigma$ are projections:  $(\pm\mathcal{C}_\pm^\Sigma)\circ (\pm\mathcal{C}_\pm^\Sigma)=\pm\mathcal{C}_\pm^\Sigma$.  The Plemelj formula then asserts that
the projections are complementary.

To set up the singular integral equation equivalent to Riemann-Hilbert Problem~\ref{small-norm-rhp}, let $\mathbf{V}$ be an admissible jump matrix on the admissible contour $\Sigma$ and define the related matrix function $\mathbf{W}:\Sigma\to \mathbb{C}^{2\times 2}$ by setting
\eq
{\bf W}(z):={\bf V}(z)-\mathbb{I}.
\endeq
Note that $\mathbf{W}\in L^\infty(\Sigma;\mathbb{C}^{2\times 2})\cap L^2(\Sigma;\mathbb{C}^{2\times 2})$ by assumption.  The expression
\begin{equation}
(\mathcal{C}_\mathbf{W}^\Sigma\mathbf{F})(\zeta):=(\mathcal{C}^\Sigma_-(\mathbf{FW}))(\zeta),\quad\zeta\in\Sigma
\end{equation}
then clearly defines a bounded linear operator on $L^2(\Sigma;\mathbb{C}^{2\times 2})$ and the operator norm obeys the estimate
\begin{equation}
\|\mathcal{C}_\mathbf{W}^\Sigma\|\le\|\mathcal{C}_-^\Sigma\|_{L^2(\Sigma)\circlearrowleft} \|\mathbf{W}\|_{L^\infty(\Sigma)}.
\label{eq:small-norm-Cw-estimate}
\end{equation}
Note also that (even though the constant functions on $\Sigma$ are not square integrable if $\Sigma$ is unbounded) the expression $(\mathcal{C}^\Sigma_\mathbf{W}\mathbb{I})(\zeta)=(\mathcal{C}_-^\Sigma\mathbf{W})(\zeta)$ for $\zeta\in\Sigma$ defines a square integrable matrix-valued function on $\Sigma$, the norm of which can be estimated as follows:
\begin{equation}
\|\mathcal{C}^\Sigma_\mathbf{W}\mathbb{I}\|_{L^2(\Sigma)}\le \|\mathcal{C}_-^\Sigma\|_{L^2(\Sigma)\circlearrowleft} \|\mathbf{W}\|_{L^2(\Sigma)}.
\label{eq:small-norm-rhs-estimate}
\end{equation}
The singular integral equation that is equivalent to Riemann-Hilbert Problem~\ref{small-norm-rhp} is then the following, for an unknown matrix function $\mathbf{X}\in L^2(\Sigma;\mathbb{C}^{2\times 2})$:
\begin{equation}
\mathbf{X}-\mathcal{C}_\mathbf{W}^\Sigma\mathbf{X}=\mathcal{C}_\mathbf{W}^\Sigma\mathbb{I}.
\label{X-sing-int-eq}
\end{equation}
The connection with Riemann-Hilbert Problem~\ref{small-norm-rhp} is illuminated by the following
result.
\begin{proposition}
\label{M-ito-Cauchy-prop}
Let $\Sigma$ be an admissible contour and  $\mathbf{V}$ an admissible jump matrix defined on $\Sigma$, and set $\mathbf{W}:=\mathbf{V}-\mathbb{I}$.  Suppose the integral equation \eqref{X-sing-int-eq} has a solution 
${\bf X}$ in $L^2(\Sigma;\mathbb{C}^{2\times 2})$ .  Then the matrix
\eq
\label{M-ito-Cauchy-ops}
{\bf M}(z):=\mathbb{I}+(\mathcal{C}^\Sigma{\bf W})(z) + (\mathcal{C}^\Sigma({\bf XW}))(z), \quad z\in\mathbb{C}\backslash\Sigma
\endeq
is a solution of
Riemann-Hilbert Problem~\ref{small-norm-rhp}.  Conversely, if $\mathbf{M}(z)$ is a solution of Riemann-Hilbert Problem~\ref{small-norm-rhp}, then the matrix $\mathbf{X}(\zeta):=\mathbf{M}_-(\zeta)-\mathbb{I}$ is a solution of \eqref{X-sing-int-eq}.
\end{proposition}
\begin{proof}
First assume that $\mathbf{X}$ solves \eqref{X-sing-int-eq}.  Because both $\mathbf{W}$ and $\mathbf{XW}$ are in $L^2(\Sigma)$, the function ${\bf M}(z)$ defined by \eqref{M-ito-Cauchy-ops} is analytic 
 for $z\in\mathbb{C}\backslash\Sigma$ and converges to $\mathbb{I}$ as $z\to\infty$ in all radial directions except possibly those of the unbounded arcs (lines) of $\Sigma$.  For the same reason, the non-tangential boundary values taken on $\Sigma^\circ$ by $\mathbf{M}(z)$ exist almost everywhere and $\mathbf{M}_\pm-\mathbb{I}$ correspond to functions in $L^2(\Sigma)$.
We now check the jump condition.  For $\zeta\in\Sigma^\circ$, using 
\eqref{X-sing-int-eq} we see 
\eq
\begin{split}
{\bf M}_-(\zeta) & = \mathbb{I}+(\mathcal{C}_-^\Sigma{\bf W})(\zeta) + (\mathcal{C}_-^\Sigma({\bf XW}))(\zeta) \\
   & = \mathbb{I}+(\mathcal{C}_-^\Sigma{\bf W})(\zeta) + {\bf X}(\zeta) - (\mathcal{C}_-^\Sigma{\bf W})(\zeta) \\
   & = \mathbb{I} + {\bf X}(\zeta).
\end{split}
\endeq
Using the Plemelj formula and \eqref{X-sing-int-eq}, we also see that, for 
$\zeta\in\Sigma^\circ$,
\eq
\begin{split}
{\bf M}_+(\zeta) & = \mathbb{I}+(\mathcal{C}_+^\Sigma{\bf W})(\zeta) + (\mathcal{C}_+^\Sigma({\bf XW}))(\zeta) \\
   & = \mathbb{I} + [(\mathcal{C}_-^\Sigma{\bf W})(\zeta) + {\bf W}(\zeta)] + [(\mathcal{C}_-^\Sigma({\bf XW}))(\zeta) + {\bf X}(\zeta){\bf W}(\zeta)] \\
   & = \mathbb{I}+(\mathcal{C}_-^\Sigma{\bf W})(\zeta) + {\bf W}(\zeta) + [{\bf X}(\zeta) - (\mathcal{C}_-^\Sigma{\bf W})(\zeta)] + {\bf X}(\zeta){\bf W}(\zeta) \\
   & = (\mathbb{I} + {\bf X}(\zeta))(\mathbb{I} + {\bf W}(\zeta)) \\
   & = {\bf M}_-(\zeta){\bf V}(\zeta),
\end{split}
\endeq
as desired.  

Reversing the argument shows that whenever $\mathbf{M}(z)$ is a solution of Riemann-Hilbert Problem~\ref{small-norm-rhp}, the function $\mathbf{X}:=\mathbf{M}_--\mathbb{I}\in L^2(\Sigma;\mathbb{C}^{2\times 2})$ satisfies the integral equation \eqref{X-sing-int-eq}.
\end{proof}

\begin{proposition}
\label{sing-int-eq-prop}
Suppose that $\Sigma$ is an admissible contour with admissible jump matrix $\mathbf{V}=\mathbb{I}+\mathbf{W}$, and suppose also that $\|\mathbf{W}\|_{L^\infty(\Sigma)}\le\rho\|\mathcal{C}_-^\Sigma\|_{L^2(\Sigma)\circlearrowleft}^{-1}$ for some $\rho<1$.  Then the singular integral equation \eqref{X-sing-int-eq} has a unique
solution $\mathbf{X}\in L^2(\Sigma;\mathbb{C}^{2\times 2})$ that can be obtained as the sum
of the infinite Neumann series
\begin{equation}
\mathbf{X}(\zeta)=(\mathcal{C}_\mathbf{W}^\Sigma\mathbb{I})(\zeta)+(\mathcal{C}_\mathbf{W}^\Sigma\circ\mathcal{C}_\mathbf{W}^\Sigma\mathbb{I})(\zeta) +\cdots
\end{equation}
and whose norm satisfies the inequality
\begin{equation}
\|\mathbf{X}\|_{L^2(\Sigma)}\le \frac{1}{1-\rho}\|\mathcal{C}_\mathbf{W}^\Sigma\mathbb{I}\|_{L^2(\Sigma)}\le\frac{\|\mathcal{C}_-^\Sigma\|_{L^2(\Sigma)\circlearrowleft}}{1-\rho}\|\mathbf{W}\|_{L^2(\Sigma)}.
\end{equation}
\end{proposition}
\begin{proof}
Taking into account the bound \eqref{eq:small-norm-Cw-estimate}, this is an elementary consequence of the contraction mapping principle (the Neumann series arises by solving \eqref{X-sing-int-eq} by iteration) and the inequality \eqref{eq:small-norm-rhs-estimate}.
\end{proof}
Therefore, if the difference $\mathbf{W}:=\mathbf{V}-\mathbb{I}$ is sufficiently small in the $L^\infty(\Sigma)$ sense compared with the reciprocal of the $L^2(\Sigma)$ operator norm of $\mathcal{C}_-^\Sigma$ (the latter being a quantity estimated in terms of geometrical properties of the contour $\Sigma$ alone that is finite for each admissible contour), there will be a unique solution to Riemann-Hilbert Problem~\ref{small-norm-rhp}.  Moreover, the analysis of the singular integral equation \eqref{X-sing-int-eq} makes available several useful estimates for the corresponding solution $\mathbf{M}(z)$.  In this situation, the Riemann-Hilbert Problem~\ref{small-norm-rhp} is frequently called a \emph{small-norm problem}.

Next, we consider the first two moments ${\bf M}_1$ and ${\bf M}_2$ of 
${\bf M}(z)$ at infinity, quantities defined by the limits
\begin{equation}
\mathbf{M}_1:=\lim_{z\to\infty} z(\mathbf{M}(z)-\mathbb{I})\quad\text{and}\quad
\mathbf{M}_2:=\lim_{z\to\infty} z^2(\mathbf{M}(z)-\mathbb{I}-\mathbf{M}_1z^{-1})
\label{eq:small-norm-momentsdefine}
\end{equation}
assuming that $z\to\infty$ radially at any angle different from those of the unbounded arcs (rays) of $\Sigma$.  The existence of the moments is equivalent to the validity of the asymptotic formula 
\eq
{\bf M}(z) = \mathbb{I} + \frac{{\bf M}_1}{z} + \frac{{\bf M}_2}{z^2} + o\left(\frac{1}{z^2}\right),\quad z\to\infty
\endeq
assuming that $z$ tends to infinity in the same sense.
\begin{proposition}
Let $\Sigma$ be an admissible contour and $\mathbf{V}$ an admissible jump matrix for which
$\|\mathbf{W}\|_{L^\infty(\Sigma)}\le\rho\|\mathcal{C}_-^\Sigma\|_{L^2(\Sigma)\circlearrowleft}^{-1}$ for some $\rho<1$, where $\mathbf{W}:=\mathbf{V}-\mathbb{I}$.  If also $\mathbf{W}\in L^1(\Sigma)$ and $\zeta\mathbf{W}\in L^1(\Sigma)\cap L^2(\Sigma)$, 
then the moments $\mathbf{M}_1$ and $\mathbf{M}_2$ associated by \eqref{eq:small-norm-momentsdefine} with the unique solution of Riemann-Hilbert Problem~\ref{small-norm-rhp} exist 
and obey the following estimates:  
\begin{equation}
|\mathbf{M}_1|\le\frac{1}{2\pi}\|\mathbf{W}\|_{L^1(\Sigma)} +\frac{1}{2\pi}
\frac{\|\mathcal{C}_-^\Sigma\|_{L^2(\Sigma)\circlearrowleft}}{1-\rho}\|\mathbf{W}\|_{L^2(\Sigma)}^2
\label{eq:small-norm-M1-estimate}
\end{equation}
and
\begin{equation}
|\mathbf{M}_2|\le\frac{1}{2\pi}\|\zeta\mathbf{W}\|_{L^1(\Sigma)}+\frac{1}{2\pi}
\frac{\|\mathcal{C}_-^\Sigma\|_{L^2(\Sigma)\circlearrowleft}}{1-\rho}\|\mathbf{W}\|_{L^2(\Sigma)}
\|\zeta\mathbf{W}\|_{L^2(\Sigma)},
\label{eq:small-norm-M2-estimate}
\end{equation}
where $|\cdot|$ denotes the pointwise matrix norm on which the integral $L^p(\Sigma)$ norms are based.  
\label{small-norm-moments-prop}
\end{proposition}
\begin{proof}
We begin by noting that the hypotheses guarantee the existence of the matrices defined by
\begin{equation}
\begin{split}
\dot{\mathbf{M}}_1&:=-\frac{1}{2\pi i}\int_\Sigma\mathbf{W}(\zeta)\,d\zeta -\frac{1}{2\pi i}\int_\Sigma
\mathbf{X}(\zeta)\mathbf{W}(\zeta)\,d\zeta,\\
\dot{\mathbf{M}}_2&:=-\frac{1}{2\pi i}\int_\Sigma\zeta\mathbf{W}(\zeta)\,d\zeta - \frac{1}{2\pi i}
\int_\Sigma \zeta\mathbf{X}(\zeta)\mathbf{W}(\zeta)\,d\zeta,
\end{split}
\end{equation}
where $\mathbf{X}(\zeta)$ is the unique solution of the singular integral equation \eqref{X-sing-int-eq} corresponding to the solution $\mathbf{M}(z)$ of Riemann-Hilbert Problem~\ref{small-norm-rhp}.
Indeed, all four integrals are absolutely convergent and, by the Cauchy-Schwarz Inequality applied to the last term in each expression, we have
\begin{equation}
|\dot{\mathbf{M}}_1|\le\frac{1}{2\pi}\|\mathbf{W}\|_{L^1(\Sigma)} +\frac{1}{2\pi}\|\mathbf{X}\|_{L^2(\Sigma)}\|\mathbf{W}\|_{L^2(\Sigma)}\le\frac{1}{2\pi}\|\mathbf{W}\|_{L^1(\Sigma)} +
\frac{1}{2\pi}\frac{\|\mathcal{C}_-^\Sigma\|_{L^2(\Sigma)\circlearrowleft}}{1-\rho}\|\mathbf{W}\|^2_{L^2(\Sigma)}
\end{equation}
and
\begin{equation}
|\dot{\mathbf{M}}_2|\le\frac{1}{2\pi}\|\zeta\mathbf{W}\|_{L^1(\Sigma)} +\frac{1}{2\pi}\|\mathbf{X}\|_{L^2(\Sigma)}\|\zeta\mathbf{W}\|_{L^2(\zeta)}\le
\frac{1}{2\pi}\|\zeta\mathbf{W}\|_{L^1(\Sigma)} + \frac{1}{2\pi}\frac{\|\mathcal{C}_-^\Sigma\|_{L^2(\Sigma)\circlearrowleft}}{1-\rho}\|\mathbf{W}\|_{L^2(\Sigma)}\|\zeta\mathbf{W}\|_{L^2(\Sigma)}.
\end{equation}
We now claim that $\mathbf{M}_1=\dot{\mathbf{M}}_1$.  Indeed, 
\begin{equation}
z(\mathbf{M}(z)-\mathbb{I})-\dot{\mathbf{M}}_1 = \frac{1}{2\pi i}\int_\Sigma\frac{\zeta}{\zeta-z}\mathbf{W}(\zeta)\,d\zeta +\frac{1}{2\pi i}\int_\Sigma\frac{\zeta}{\zeta-z}\mathbf{X}(\zeta)\mathbf{W}(\zeta)\,d\zeta
\end{equation}
and the integrands obviously tend to zero pointwise as $z\to\infty$.  Restriction of $z$ to a ray distinct from the directions of the unbounded arcs of $\Sigma$ implies the estimate $|\zeta-z|\ge \alpha |\zeta|$ for some $\alpha>0$ and hence $z(\mathbf{M}(z)-\mathbb{I})-\dot{\mathbf{M}}_1$ tends to zero as $z\to\infty$ by the Dominated Convergence Theorem because 
$\mathbf{W}\in L^1(\Sigma)\cap L^2(\Sigma)$ and $\mathbf{X}\in L^2(\Sigma)$.
Similarly, $\mathbf{M}_2=\dot{\mathbf{M}}_2$ because 
\begin{equation}
\begin{split}
z^2(\mathbf{M}(z)-\mathbb{I}-\mathbf{M}_1z^{-1})-\dot{\mathbf{M}}_2&=
z^2(\mathbf{M}(z)-\mathbb{I}-\dot{\mathbf{M}}_1z^{-1})-\dot{\mathbf{M}}_2\\&{} = 
\frac{1}{2\pi i}\int_\Sigma\frac{\zeta}{\zeta-z}\zeta\mathbf{W}(\zeta)\,d\zeta +\frac{1}{2\pi i}
\int_\Sigma\frac{\zeta}{\zeta-z}\zeta\mathbf{X}(\zeta)\mathbf{W}(\zeta)\,d\zeta
\end{split}
\end{equation}
tends to zero by the Dominated Convergence Theorem as $\zeta\mathbf{W}\in L^1(\Sigma)\cap L^2(\Sigma)$
and $\mathbf{X}\in L^2(\Sigma)$.
\end{proof}
Finally, we observe that a small number of estimates on $\mathbf{W}$ suffice to provide existence and uniqueness for Riemann-Hilbert Problem~\ref{small-norm-rhp} and control of the first two moments.

\begin{proposition}
\label{prop:small-norm-moments-bound}
Let $\Sigma$ be an admissible contour and suppose that $\mathbf{V}:\Sigma^\circ\to\mathbb{C}^{2\times 2}$ is a mapping such that $\mathbf{W}\in L^\infty(\Sigma)$ and $\zeta^2\mathbf{W}\in L^1(\Sigma)$, where $\mathbf{W}:=\mathbf{V}-\mathbb{I}$.
Then $\mathbf{V}$ is an admissible jump matrix on $\Sigma$.  Furthermore, there exist positive constants $K_\Sigma$, $K'_\Sigma$, and $K''_\Sigma$ such that Riemann-Hilbert Problem~\ref{small-norm-rhp} has a unique solution if
\begin{equation}
\|\mathbf{W}\|_{L^\infty(\Sigma)}\le K_\Sigma,
\label{eq:small-norm-W-infty-bound}
\end{equation}
in which case the first two moments $\mathbf{M}_1$ and $\mathbf{M}_2$ both exist and satisfy
\begin{equation}
|\mathbf{M}_1|\le K'_\Sigma\max\{\|\mathbf{W}\|_{L^\infty(\Sigma)},\|\zeta^2\mathbf{W}\|_{L^1(\Sigma)}\} \quad\text{and}\quad |\mathbf{M}_2|\le K''_\Sigma\max\{\|\mathbf{W}\|_{L^\infty(\Sigma)},\|\zeta^2\mathbf{W}\|_{L^1(\Sigma)}\}
\label{eq:small-norm-M12-simple}
\end{equation}
as long as $\|\zeta^2\mathbf{W}\|_{L^1(\Sigma)}$ is also sufficiently small.
\label{prop:moments-control}
\end{proposition}
\begin{proof}
This follows from some elementary inequalities.  Firstly, by splitting integrals at the unit circle in the complex plane one obtains
\begin{equation}
\|\zeta\mathbf{W}\|_{L^1(\Sigma)}\le |\Sigma_1|\|\mathbf{W}\|_{L^\infty(\Sigma)}+
\|\zeta^2\mathbf{W}\|_{L^1(\Sigma)}
\label{eq:small-norm-ineq-1}
\end{equation}
and similarly
\begin{equation}
\|\mathbf{W}\|_{L^1(\Sigma)}\le |\Sigma_1|\|\mathbf{W}\|_{L^\infty(\Sigma)}+
\|\zeta^2\mathbf{W}\|_{L^1(\Sigma)}.
\label{eq:small-norm-ineq-2}
\end{equation}
Recall that $|\Sigma_1|$ denotes the total arc length of the part of $\Sigma$ within the unit circle (necessarily finite for admissible contours $\Sigma$).  By the $L^\infty$-$L^1$ H\"older Inequality,
\begin{equation}
\|\zeta\mathbf{W}\|_{L^2(\Sigma)}\le\sqrt{\|\mathbf{W}\|_{L^\infty(\Sigma)}}\sqrt{\|\zeta^2\mathbf{W}\|_{L^1(\Sigma)}}
\label{eq:small-norm-ineq-3}
\end{equation}
and similarly (also using \eqref{eq:small-norm-ineq-2})
\begin{equation}
\|\mathbf{W}\|_{L^2(\Sigma)}\le\sqrt{\|\mathbf{W}\|_{L^\infty(\Sigma)}}\sqrt{\|\mathbf{W}\|_{L^1(\Sigma)}}\le \sqrt{\|\mathbf{W}\|_{L^\infty(\Sigma)}}\sqrt{|\Sigma_1|\|\mathbf{W}\|_{L^\infty(\Sigma)}+\|\zeta^2\mathbf{W}\|_{L^1(\Sigma)}}.
\label{eq:small-norm-ineq-4}
\end{equation}
The inequality \eqref{eq:small-norm-ineq-4} shows that $\mathbf{W}\in L^2(\Sigma)$, and hence $\mathbf{V}$ is admissible.  By Proposition~\ref{sing-int-eq-prop}, the condition \eqref{eq:small-norm-W-infty-bound} guarantees unique solvability if, say, $K_\Sigma:=(2\|\mathcal{C}_-^\Sigma\|_{L^2(\Sigma)\circlearrowleft})^{-1}$ (choosing $\rho=\tfrac{1}{2}$).
Using \eqref{eq:small-norm-ineq-2} and \eqref{eq:small-norm-ineq-4} in \eqref{eq:small-norm-M1-estimate} with $\rho=\tfrac{1}{2}$ from Proposition~\ref{small-norm-moments-prop}, along with the inequality \eqref{eq:small-norm-W-infty-bound} with the above value of $K_\Sigma$ yields the estimate for $\mathbf{M}_1$ in \eqref{eq:small-norm-M12-simple} with 
\begin{equation}
\label{small-norm-K'}
K'_\Sigma:=\frac{1}{\pi}\left(|\Sigma_1|+1\right).
\end{equation}
Similarly, using \eqref{eq:small-norm-ineq-1} and \eqref{eq:small-norm-ineq-3} in \eqref{eq:small-norm-M2-estimate} with $\rho=\tfrac{1}{2}$ from Proposition~\ref{small-norm-moments-prop}, along with \eqref{eq:small-norm-W-infty-bound} and the condition $\|\zeta^2\mathbf{W}\|_{L^1(\Sigma)}\le M$, say, yields the estimate for $\mathbf{M}_2$ in \eqref{eq:small-norm-M12-simple} with
\begin{equation}
\label{small-norm-K''}
K''_\Sigma:=\frac{1}{2\pi}\left(|\Sigma_1|+1+2\|\mathcal{C}_-^\Sigma\|_{L^2(\Sigma)\circlearrowleft}\sqrt{|\Sigma_1|K_\Sigma M+M^2}\right).
\end{equation}
\end{proof}

\end{document}